\documentclass[acmsmall,screen]{acmart}

\AtBeginDocument{%
  \providecommand\BibTeX{{%
    \normalfont B\kern-0.5em{\scshape i\kern-0.25em b}\kern-0.8em\TeX}}}

\newcommand{\macrospath}{.}
\usepackage{ifthen}
\usepackage{xspace}

\newboolean{talk}
\setboolean{talk}{false}
\newboolean{paper}
\setboolean{paper}{false}

\newboolean{IEEEstyle}
\setboolean{IEEEstyle}{false}
\newboolean{lipicsstyle}
\setboolean{lipicsstyle}{false}
\newboolean{eptcsstyle}
\setboolean{eptcsstyle}{false}
\newboolean{entcsstyle}
\setboolean{entcsstyle}{false}
\newboolean{sigplanstyle}
\setboolean{sigplanstyle}{false}
\newboolean{easychairstyle}
\setboolean{easychairstyle}{false}
\newboolean{scrartclstyle}
\setboolean{scrartclstyle}{false}
\newboolean{lncsstyle}
\setboolean{lncsstyle}{false}

\newboolean{acmmart}
\setboolean{acmmart}{false}

\newboolean{needstheorems}
\setboolean{needstheorems}{false}
\newboolean{withimages}
\setboolean{withimages}{false}
\newboolean{withproofs}
\setboolean{withproofs}{true}

\newboolean{french}
\setboolean{french}{false}

\setboolean{acmmart}{true}

 \usepackage{booktabs}   
\usepackage{subcaption} 
\usepackage[all]{xy}

\usepackage{ebproof}
\ebproofset{label separation=0.3em}

\usepackage{booktabs}   
\usepackage{subcaption} 
\usepackage{graphicx}
\usepackage[normalem]{ulem}
\usepackage{multirow}

\usepackage[utf8]{inputenc}

\usepackage{multirow}
\usepackage{cmll}
\usepackage{listings}

\usepackage{amsmath, amsfonts, 
	amsthm, mathtools}
\usepackage{stmaryrd}
\usepackage{mathtools}
\usepackage{enumerate}
\usepackage{listings}

\usepackage{fancybox}
\usepackage{hhline}
\usepackage{marginnote}
\usepackage{soul}
\usepackage{xcolor}
\usepackage{cleveref}

\newcommand{\sem}[1]{\llbracket#1\rrbracket}
\newcommand{\semsolv}[1]{\sem{#1}^\solvredsym}

\newcommand{\ignore}[1]{}

\newcommand{\myinput}[1]{\ifthenelse{\boolean{withimages}}{\input{#1}}{}}

\newcommand{\reflemma}[1]{Lemma~\ref{l:#1}}
\newcommand{\reflemmap}[2]{Lemma~\ref{l:#1}.\ref{p:#1-#2}}

\newcommand{\refrmk}[1]{Remark~\ref{rmk:#1}} 

\newcommand{\ie}{\textit{i.e.}\xspace}
\newcommand{\eg}{\textit{e.g.}\xspace}
\newcommand{\ih}{\textit{i.h.}\xspace}
\newcommand{\resp}{\textnormal{resp.}\xspace}

\newcommand{\ES}{\text{ES}\xspace}

\newcommand{\full}{\text{full}\xspace}
\newcommand{\Full}{\text{Full}\xspace}
\renewcommand{\full}{\text{strong}\xspace}
\renewcommand{\Full}{\text{Strong}\xspace}

\newcommand{\defeq}{\coloneqq} 
\newcommand{\eqdef}{\eqqcolon} 
\newcommand{\grameq}{\Coloneqq} 
\newcommand{\set}[1]{\{#1\}}

\newcommand{\nat}{\mathbb{N}}
\newcommand{\size}[1]{|#1|}

\newcommand{\mulsym}{\msym} 

\renewcommand{\l}{\lambda}
\newcommand{\isub}[2]{\{#1/#2\}}
\newcommand{\replace}[2]{#1{\shortleftarrow}#2}
\renewcommand{\isub}[2]{\{\replace{#1}{#2}\}}
\newcommand{\esub}[2]{[\replace{#1}{#2}]}
\renewcommand{\esub}[2]{[#1{\shortleftarrow}#2]}
\newcommand{\subs}[4]{\{\replace{#1}{#2}, \ldots, \replace{#3}{#4}\}}

\newcommand{\fv}[1]{{\sf fv}(#1)}

\newcommand{\bv}[1]{{\sf bv}(#1)}

\newcommand{\rootRew}[1]{\mapsto_{#1}}

\newcommand{\Rew}[1]{\rightarrow_{#1}}

\newcommand{\lRew}[1]{\; \mbox{}_{#1}{\leftarrow}\ }

\newcommand{\lto}{\lRew{}}

\newcommand{\rtom}{\rootRew{\msym}}
\newcommand{\rtoe}{\rootRew{\esym}}

\newcommand{\rtoeabs}{\rootRew{\expoabs}} 
\newcommand{\rtoevar}{\rootRew{\expovar}} 

\newcommand{\slsym}{\sigma_1}
\newcommand{\srsym}{\sigma_3}

\newcommand{\rtosl}{\rootRew{\slsym}}
\newcommand{\rtosr}{\rootRew{\srsym}}

\newcommand{\tob}{\Rew{\beta}}

\newcommand{\betaplot}{\beta_v}
\newcommand{\fbetaplot}{\fullsym\betaplot}
\newcommand{\obetaplot}{\osym\betaplot}
\newcommand{\abssym}{\lambda} 

\newcommand{\tobvplot}{\Rew{\betaplot}} 
\newcommand{\tofbvplot}{\Rew{\fbetaplot}} 
\newcommand{\toobvplot}{\Rew{\obetaplot}} 
\newcommand{\rtobvplot}{\rootRew{\betaplot}} 
\newcommand{\tobvploto}{\Rew{\obetaplot}} 

\newcommand{\esym}{{\mathsf e}}

\newcommand{\erulesym}{{\esym_\Rule}}
\newcommand{\msym}{\mathsf{m}}

\newcommand{\ssym}{{\mathsf s}}
\newcommand{\wsym}{\osym} 
\newcommand{\wmsym}{{\wsym\msym}} 
\newcommand{\wesym}{{\wsym\esym}} 
\newcommand{\omsym}{{\wmsym}} 
\newcommand{\oesym}{{\wesym}} 

\newcommand{\vmsym}{\mathsf{shuf}} 
\newcommand{\shuf}{\vmsym} 

\newcommand{\shufeqext}{\shufeqext} 

 \newcommand{\tom}{\Rew{\msym}}
 \newcommand{\toe}{\Rew{\esym}}
\newcommand{\varsym}{\mathrm{var}}
\newcommand{\nvarsym}{\lambda}
\newcommand{\onvarsym}{\osym_{\nvarsym}}
\newcommand{\solvnvarsym}{\solvredsym_{\nvarsym}}
\newcommand{\vsubnvarsym}{\vsub_{\nvarsym}}

\newcommand{\toeabs}{\Rew{\expoabs}}
\newcommand{\toevar}{\Rew{\expovar}}
\newcommand{\toevaro}{\Rew{\wsym\evarsym}}
\newcommand{\toevarsolv}{\Rew{\solvredsym \evarsym}}

\newcommand{\toeruleo}{\Rew{\wsym\erulesym}}

\newcommand{\toeabssolv}{\Rew{\solvredsym \esym_{\l}}}

\newcommand{\tosolvnvar}{\Rew{\solvnvarsym}}
\newcommand{\tovsubnvar}{\Rew{\vsubnvarsym}}

\newcommand{\tovsubo}{\Rew{\wsym}}

\newcommand{\tomo}{\Rew{\omsym}}

\newcommand{\toeo}{\Rew{\wsym{\esym}}}
\newcommand{\toeabso}{\Rew{\wsym{\eabssym}}}

\newcommand{\msolvsym}{\solvredsym\msym}
\newcommand{\esolvsym}{\solvredsym\esym}

\newcommand{\tovsubsolv}{\tosolv}
\newcommand{\tomsolv}{\Rew{\msolvsym}}
\newcommand{\toesolv}{\Rew{\esolvsym}}

\newcommand{\eqstruct}{\equiv}
\newcommand{\streq}{\eqstruct}
\newcommand{\tostruct}{\eqstruct}

\newcommand{\aplsym}{@\textup{l}}
\newcommand{\aprsym}{@\textup{r}}
\newcommand{\essym}{[\cdot]}
\newcommand{\comsym}{\textup{com}}
\newcommand{\tostructapl}{\tostruct_{\aplsym}}
\newcommand{\tostructapr}{\tostruct_{\aprsym}}
\newcommand{\tostructes}{\tostruct_{\essym}}
\newcommand{\tostructcom}{\tostruct_{\comsym}}

\newcommand{\tm}{t}
\newcommand{\tmtwo}{u}
\newcommand{\tmthree}{r}
\newcommand{\tmfour}{q}
\newcommand{\tmfive}{p}
\newcommand{\tmsix}{s}

\newcommand{\tmp}{\tm'}
\newcommand{\tmtwop}{\tmtwo'}

\newcommand{\tmfourp}{\tmfour'}
\newcommand{\tmfivep}{\tmfive'}

\newcommand{\var}{x}
\newcommand{\vartwo}{y}
\newcommand{\varthree}{z}
\newcommand{\varfour}{w}

\newcommand{\val}{v}
\newcommand{\valtwo}{\val'}
\newcommand{\valthree}{\val''}

\newcommand{\ctxholep}[1]{\langle #1\rangle}
\newcommand{\ctxhole}{\ctxholep{\cdot}}

\newcommand{\ctx}{C}
\newcommand{\ctxtwo}{\ctx'}

\newcommand{\ctxp}[1]{\ctx\ctxholep{#1}}
\newcommand{\ctxtwop}[1]{\ctxtwo\ctxholep{#1}}

\newcommand{\sctx}{L}

\newcommand{\sctxp}[1]{\sctx\ctxholep{#1}}

\newcommand{\arbctxp}[1]{\arbctxp{#1}}
\newcommand{\arbctxtwop}[1]{\arbctxtwop{#1}}

\newcommand{\fctx}{F}
\newcommand{\fctxtwo}{\fctx'}

\newcommand{\fctxp}[1]{\fctx\ctxholep{#1}}
\newcommand{\fctxtwop}[1]{\fctxtwo\ctxholep{#1}}

\ifthenelse{\boolean{acmmart}
}{
  
  }{
  
  }

\newcommand{\deriv}{d}
\newcommand{\derivp}{d'} 

\newcommand{\sizehole}[2]{|#2|_{#1}}

\newcommand{\sizem}[1]{\sizehole{\msym}{#1}} 
\newcommand{\sizes}[1]{\sizehole{\ssym}{#1}} 
\newcommand{\sizeo}[1]{\sizehole{\osym}{#1}} 

\ifthenelse{\boolean{IEEEstyle}}{
\usepackage{amssymb}
\usepackage{amsthm}
    \newtheorem{theorem}{Theorem}[section]
    \newtheorem{lemma}[theorem]{Lemma}
    \newtheorem{corollary}[theorem]{Corollary}
    \newtheorem{proposition}[theorem]{Proposition}
    \newtheorem{definition}[theorem]{Definition}
}{}

\ifthenelse{\boolean{lncsstyle}}{}{\ifthenelse{\boolean{lipicsstyle}}{}{\newtheorem{remark}{Remark}[section]}}

\newcommand{\itm}{i}
\newcommand{\itmtwo}{\itm'} 
\newcommand{\itmthree}{\itm''}

\newcommand{\fire}{f}
\newcommand{\firetwo}{\fire'}

\newcommand{\sfire}{\fire_\fullsym}

\newcommand{\vsub}{\mathsf{vsc}} 
\newcommand{\VSC}{\textnormal{VSC}\xspace}

\newcommand{\tovsub}{\Rew{\vsub}}

\newcommand{\tovsubonvar}{\Rew{\onvarsym}}

\newcommand{\tow}{\Rew{\wsym}} 
\newcommand{\towm}{\Rew{\wmsym}} 
\newcommand{\towe}{\Rew{\wsym{\esym}}} 

\newcommand{\osym}{{\mathsf o}}

\newcommand{\la}[1]{\lambda #1.}

\newcommand{\myproof}[1]{
\ifthenelse{\boolean{omitproofs}}{\begin{IEEEproof} Proof available but omitted for readability. \end{IEEEproof}}{#1}}

\newcommand{\gregoire}{Gr{\'{e}}goire\xspace}

\newcommand{\withproofs}[1]{\ifthenelse{\boolean{withproofs}}{#1}{}}

\newcommand{\withoutproofs}[1]{\ifthenelse{\boolean{withproofs}}{}{#1}}

\newcommand{\NoteProof}[1]{
	\marginnote{{\normalfont\scriptsize{Proof\,p.\,{\pageref{#1}}\,}}}}
\newcommand{\NoteState}[1]{
	\marginnote{{\normalfont\scriptsize{See p.\,{\pageref{#1}}}}}}
\renewcommand{\NoteProof}[1]{\marginnote{{Proof\,p.\,{\pageref{#1}}}}}
\renewcommand{\NoteState}[1]{\marginnote{{See\,p.\,{\pageref{#1}}\\\cref{#1}}}}

\crefname{proposition}{Prop.}{Props.}
\crefname{theorem}{Thm.}{Thms.}
\crefname{lemma}{Lemma}{Lemmas}
\crefname{corollary}{Cor.}{Cors.}
\crefname{section}{Sect.}{Sects.}
\Crefname{section}{Section}{Sections}

\newcommand{\vsubterms}{\Lambda_\vsub}

\newcommand{\shufcalc}{\lambda_\shuf}

\newcommand{\plotsym}{\mathsf{Plot}}
\newcommand{\plotcalc}{\lambda_{\plotsym}}

\newcommand{\moggicalc}{\lambda_{\mathsf{c}}}

\newcommand{\doubt}[1]{}

\newcommand{\letexp}{\mathsf{let}}
\newcommand{\letin}[3]{{\sf let}\ #1=#2\ {\sf in}\ #3}

\newcommand{\lambdamucalc}{\overline\lambda\mu\tilde{\mu}}

\newcommand{\Rule}{\mathsf{r}}

\newcounter{numberone}
\newcounter{numberoneroman}
\newcounter{numberonealph}

\newcommand{\cbn}{CbN\xspace}

\newcommand{\cbv}{CbV\xspace}

\newcommand{\ocbv}{Open \cbv}

\newcommand{\ccbv}{Closed \cbv}

\newcommand{\mset}[1]{[#1]}
\newcommand{\emptymset}{\mset{\,}}
\renewcommand{\emptymset}{\zero}
\newcommand{\zero}{\mathbf{0}}

\newcommand{\ltype}{\typefont{A}}
\newcommand{\ltypetwo}{\typefont{B}}

\newcommand{\typctx}{\Gamma}
\newcommand{\typctxtwo}{\Delta}
\newcommand{\typctxthree}{\Sigma}
\newcommand{\typctxfour}{\Pi}

\newcommand{\hastype}{\!:\!}

\newcommand{\domain}[1]{\mathsf{dom}(#1)}

\newcommand{\mytr}[1]{\underline{#1}}
\newcommand{\auxtr}[1]{\overline{#1}}

\makeatletter
\newcommand\Copy[2]{
        \marginpar{\scriptsize \ \ \hyperlink{hl-appendix-#1}{Proof p.\,{\pageref*{appendix-#1}}}}
	\immediate\write\@auxout{\unexpanded{\global\long\@namedef{mytext@#1}{#2}
  }}%
	#2%
}

\newcommand\Paste[1]{%
        \hypertarget{hl-appendix-#1}{}\label{appendix-#1}
	\renewcommand{\inappendix}[1]{}
	\ifcsname mytext@#1\endcsname
	\@nameuse{mytext@#1}%
	\else
	``??''
	\fi
	\renewcommand{\inappendix}[1]{#1}
}
\makeatother

\newcommand{\inappendix}[1]{#1}

\newcommand{\weakctx}{O}
\newcommand{\weakctxtwo}{\weakctx'}
\newcommand{\weakctxp}[1]{\weakctx\ctxholep{#1}}
\newcommand{\weakctxtwop}[1]{\weakctxtwo\ctxholep{#1}}

\newcommand{\openctx}{O}

\newcommand{\subctx}{\sctx}
\newcommand{\subctxtwo}{\sctx'}

\newcommand{\solvctx}{S}
\newcommand{\solvctxtwo}{\solvctx'}

\newcommand{\solvctxp}[1]{\solvctx\ctxholep{#1}}
\newcommand{\solvctxtwop}[1]{\solvctxtwo\ctxholep{#1}}

\newcommand{\larrow}[2]{#1 \multimap #2}
\newcommand{\ground}{\typefont{X}}

\newcommand{\Ax}{\mathsf{ax}}
\newcommand{\Es}{\mathsf{es}}

\newcommand{\derive}[2]{#1 \vartriangleright #2}
\newcommand{\concl}[4]{\derive{#1}{#2 \vdash #3 \hastype #4}}

\newcommand{\subctxp}[1]{\subctx\ctxholep{#1}}

\newcommand{\fullsym}{{\mathsf{f}}}
\renewcommand{\fullsym}{{\mathsf{s}}}
\newcommand{\sitm}{\itm_\fullsym}
\newcommand{\sitmtwo}{\itmtwo_\fullsym}

\newcommand{\sval}{\val_\fullsym}

\newcommand\Crumb\mytr
\newcommand\CrumbAux\auxtr

\newcommand{\evarsym}{\esym_{\varsym}}
\newcommand{\eabssym}{\esym_{\l}}
\renewcommand{\rtoevar}{\rootRew\evarsym}
\renewcommand{\toevar}{\Rew\evarsym}

\makeatletter
\newcommand{\exder}{%
  \def\exderW[##1]{\triangleright_{##1}\ }%
  \def\exderWO{\triangleright\ }%
  \@ifnextchar[\exderW\exderWO%
  }
\makeatother

\newcommand{\tderiv}{\Phi}
\newcommand{\tderivp}{\tderiv'} 
\newcommand{\tderivtwo}{\Psi}
\newcommand{\tderivtwop}{\tderivtwo'} 
\newcommand{\tderivthree}{\Theta}
\newcommand{\tderivthreep}{\tderivthree'} 
\newcommand{\tderivfour}{\Upsilon}

\newcommand{\typefont}[1]{{\mathsf{#1}}}
\renewcommand{\typefont}[1]{#1}
\newcommand{\mtype}{\typefont{M}}
\newcommand{\mtypetwo}{\typefont{N}}
\newcommand{\mtypethree}{\typefont{O}}

\newcommand\emptytype{\mathbf{0}}
\newcommand{\type}{\typefont{T}}
\newcommand{\imtype}{\inertop{\mtype}}

\newcommand{\inltype}{\inertop{\ltype}}
\newcommand\mplus{\uplus}

\newcommand{\tyjp}[4]{{#3} \vdash^{#1} #2 \hastype #4}
\newcommand{\namedtyjp}[5]{#1 \vartriangleright \tyjp{#2}{#3}{#4}{#5}}

\newcommand{\dom}[1]{\mathsf{dom}(#1)}

\newcommand{\ruleApp}{@}
\newcommand{\ruleFun}{\lambda}
\newcommand{\ruleES}{\mathsf{es}}
\newcommand{\ruleMany}{\mathsf{many}}
\newcommand{\ruleManyVar}{\ruleMany}
\newcommand{\ruleManyVal}{\ruleMany}
\newcommand{\ruleAp}{@}
\newcommand{\ruleAx}{\mathsf{ax}}

\newcommand{\I}{I}
\newcommand{\J}{J}
\newcommand{\K}{K}
\newcommand{\iI}{{i \in \I}}
\newcommand{\jJ}{{j \in \J}}
\newcommand{\kK}{{k \in \K}}

\newcommand{\tarrow}[2]{#1 \multimap #2}
\newcommand{\ty}[2]{\tarrow{#1}{#2}}

\newcommand{\mult}[1]{[ #1 ] }

\newcommand{\bigmplus}{\biguplus}

\newcommand{\Id}{{\mathsf{I}}}
\newcommand{\gluesym}{\mathsf{glue}}
\newcommand{\rtoglue}{\rootRew{\gluesym}}
\newcommand{\toglue}{\Rew{\gluesym}}

\newcommand{\aptm}{a}
\newcommand{\solvredsym}{\mathsf{s}}
\newcommand{\tosolv}{\Rew{\solvredsym}}
\newcommand{\tonsolv}{\Rew{\neg\solvredsym}}
\newcommand{\tosolvm}{\Rew{\solvredsym\mulsym}}
\newcommand{\tosolve}{\Rew{\solvredsym\esym}}

\newcommand{\sltype}{\ltype^{\solvsym}}

\newcommand{\smtype}{\mtype^{\solvsym}}
\newcommand{\smtypetwo}{\mtypetwo^{\solvsym}}

\newcommand{\unitarysym}{{\textsc{u}}}

\newcommand{\usltype}{\ltype^{\unitarysym\solvsym}}

\newcommand{\usmtype}{\mtype^{\unitarysym\solvsym}}

\newcommand{\inertsym}{{\textsc{i}}}
\newcommand{\inertop}[1]{#1^{\mathsf{i}}}
\newcommand{\isltype}{\ltype^{\inertsym\solvsym}}

\newcommand{\ismtype}{\mtype^{\inertsym\solvsym}}

\newcommand{\precisesym}{{\textsc{p}}}

\newcommand{\psmtype}{\mtype^{\precisesym\solvsym}}
\newcommand{\psmtypetwo}{\mtypetwo^{\precisesym\solvsym}}

\newcommand{\solvsym}{\mathsf{s}}
\newcommand{\solvnf}{\fire_{\solvredsym}}
\newcommand{\solvnftwo}{\solvnf'}

\newcommand\mydots{\hbox to .6em{.\hss.}}

\newcommand{\bctx}{B}
\newcommand{\bctxtwo}{\bctx'}

\newcommand{\bctxp}[1]{\bctx\ctxholep{#1}}
\newcommand{\bctxtwop}[1]{\bctxtwo\ctxholep{#1}}

\newcommand{\mysubparagraph}[1]{\vspace{-9pt} \subparagraph{#1.}}
\renewcommand{\mysubparagraph}[1]{\paragraph{#1}}
\newcommand{\mybigsubparagraph}[1]{\vspace{-13pt} \subparagraph{#1.}}
\renewcommand{\mybigsubparagraph}[1]{\paragraph{#1}}

\newcommand{\tovsc}{\tovsub}

\newcommand{\hctx}{H}
\newcommand{\hctxtwo}{\hctx'}

\newcommand{\hctxp}[1]{\hctx\ctxholep{#1}}
\newcommand{\hctxtwop}[1]{\hctxtwo\ctxholep{#1}}

\newcommand{\solvfire}{\solvnf}

\newcommand{\eqth}{\mathcal{T}}
\renewcommand{\Full}{Full\xspace}
\renewcommand{\full}{full\xspace}
\renewcommand{\fullsym}{{\mathsf{f}}}
\renewcommand{\rtoeabs}{\rootRew{\esym_{\l}}}
\renewcommand{\toeabs}{\Rew{\esym_{\l}}}

\newcommand{\ctxeq}{=_{\mathcal{C}}}

\renewcommand{\mydots}{\dots}
\newcommand{\balanced}{\text{testing}\xspace}

\newcommand{\shallow}{\text{\ES-shallow}\xspace}

\renewcommand{\bctx}{T}

\renewcommand{\tmthree}{s}

\usepackage{microtype}

\begin{document}

\title{The Theory of Call-by-Value Solvability}

%
\author{Beniamino Accattoli}
\affiliation{%
  \institution{Inria \& LIX, École Polytechnique}
  \country{France}}
\email{beniamino.accattoli@inria.fr}

\author{Giulio Guerrieri}
\affiliation{%
  \institution{Edinbugh Research Centre, Central Software Institute, Huawei}
  \country{United Kingdom}
}
\email{giulio.guerrieri@huawei.com}
%
%
%
%
%
%

\begin{abstract}

The denotational semantics of the untyped $\lambda$-calculus is a well developed field built around the concept of solvable terms, which are elegantly characterized in many different ways. In particular, unsolvable terms provide a consistent notion of meaningless term. The semantics of the untyped \textit{call-by-value} $\lambda$-calculus (CbV) is instead still in its infancy, because of some inherent difficulties but also because CbV solvable terms are less studied and understood than in call-by-name. On the one hand, we show that a carefully crafted presentation of CbV allows us to recover many of the properties that solvability has in call-by-name, in particular qualitative and quantitative characterizations via multi types. On the other hand, we stress that,  in CbV, solvability plays a different role: identifying unsolvable terms as meaningless induces an inconsistent theory. 
\end{abstract}



\maketitle


\section{Introduction}
\label{sect:intro}
A semantics of the $\l$-calculus can be simply seen as an equational theory over $\l$-terms. The $\l$-calculus is Turing-complete, thus there should be notions of terminating/defined/meaningful and diverging/undefined/meaningless computations corresponding to the ones of partial recursive functions. At the level of the equational theory, it is natural to have many different equivalence classes of meaningful terms, while one would expect to have a unique equivalence class of meaningless terms. That is, all meaningless terms should be equated, or \emph{collapsed}.

Instinctively, one would identify being meaningful with \emph{being ($\beta$-)normalizable} (a normal form being the result of computation), and thus, dually, being meaningless with being ($\beta$-)divergent. As it is often the case in the theory of $\l$-calculus, things are not as simple as that. 

Theories that collapse all divergent terms, called here \emph{normalizable} theories, have two related drawbacks. Firstly, the representation of partial recursive functions mapping the \emph{everywhere undefined function} to the class of divergent terms is problematic, as it is not stable by composition. The crucial point is that such a notion of meaningless term is not stable by substitution. 
Secondly, and more importantly, normalizable theories are \emph{inconsistent}, that is, because of the closure properties of theories, they end up equating \emph{all} $\l$-terms. Therefore, there is a unique and \emph{trivial} normalizable theory. We then say that divergent terms are not \emph{collapsible}.

One might then be led to think that being normalizable is not a meaningful predicate. But be careful: the issue is rather that it is too coarse to see all divergent terms as meaningless, that is, there is meaning to be found also in some divergent terms. Therefore, the point is rather that \emph{normalizing terms are not the only meaningful ones}.

\paragraph{Solvability} 
These issues were first studied by Wadsworth \cite{Wad:SemPra:71,DBLP:journals/siamcomp/Wadsworth76} and Barendregt \cite{DBLP:books/daglib/0016519,solvability-barendregt} in the '70s. They showed that both drawbacks of the normalizable theory disappear if meaningful/meaningless terms are rather identified with solvable/unsolvable terms. (Un)solvable terms can be defined in many ways. The official definition is: \emph{$\tm$ is solvable if it exists a head context $\hctx$ sending $\tm$ to the identity} $\Id \defeq \la{\varthree}\varthree$, that is, such that $\hctxp\tm \tob^{*} \Id$. 
The idea is that a solvable term $\tm$ might be divergent but all its diverging sub-terms are removable via interactions with an environment that cannot simply discard $\tm$ (enforced by the restriction to \emph{head} contexts). As an example, a divergent term such as $\var\Omega$ (with $\Omega \defeq \delta \delta$ and  $ \delta \defeq \la{\vartwo}\vartwo\vartwo$) is solvable, because the head context $(\la\var\ctxhole) \la\vartwo\Id$ sends it on the identity by erasing the diverging argument $\Omega$. 
Consider instead $\Omega$: no head contexts can erase its divergence and produce the identity, thus it is unsolvable. More generally, unsolvable terms are a strict subset of the diverging ones. 

A compositional representation of partial recursive functions can then be given, as shown by Barendregt, and the equational theory extending $\beta$-conversion with the collapse of all unsolvable terms---known as the \emph{minimal sensible theory} $\mathcal{H}$---is \emph{consistent}, that is, it does not equate all $\l$-terms. In particular, the equational theories of important models of the $\l$-calculus such as Scott's $D_{\infty}$ or the one induced by the relational semantics of linear logic do collapse all unsolvable terms.

\paragraph{Characterizations of Solvability} A natural question is whether the external ingredient in the definition of solvability, represented by the head context, can be somehow internalized. Wadsworth showed that it can \cite{DBLP:journals/siamcomp/Wadsworth76}: \emph{a term $\tm$ is solvable if and only if the head reduction of $\tm$ terminates}. This is often referred to as the \emph{operational characterization} of solvability. The characterization shows that, internally, \emph{meaningful} should be associated to \emph{head normalizable} rather than \emph{normalizable}. Since \emph{head normalizable} is a weaker predicate than \emph{normalizable}, the solvable approach is \emph{more meaningful}, in the sense that, as expected, it accepts more terms as meaningful. Additionally, the head normalizable predicate is a \emph{refinement} of the normalizable one, as normalizable terms can be seen as \emph{hereditarily head normalizable terms}, that is, terms that are head normalizable and the head arguments of which are hereditarily head normalizable. 

Solvability can also be characterized via Coppo and Dezani's \emph{intersection types} \cite{DBLP:journals/aml/CoppoD78,DBLP:journals/ndjfl/CoppoD80}, which are a theoretical notion of type mediating between semantic and operational properties. A term $\tm$ is solvable if and only if $\tm$ is typable with intersection types. Moreover, by adopting Gardner-de Carvalho's \emph{non-idempotent} intersection types \cite{DBLP:conf/tacs/Gardner94,Carvalho07,deCarvalho18}, also known as \emph{multi types}, one can additionally extract quantitative operational information about solvable terms. Namely, the number of head reduction steps, which is a reasonable measure of time complexity for $\l$-terms (see Accattoli and Dal Lago \cite{DBLP:conf/rta/AccattoliL12}), as well as the size of the head normal form, as first shown by de Carvalho \cite{Carvalho07,deCarvalho18}. Multi types are also relevant because the set of multi type judgments for a term $\tm$ is a syntactic presentation of the relational semantics of $\tm$, a paradigmatic denotational model of the $\l$-calculus.

\paragraph{Semantics of Call-by-Value} Many variants of the $\l$-calculus have emerged. What is usually referred to as \emph{the} $\l$-calculus could nowadays be more precisely referred to as the \emph{(strong) call-by-name (\cbn for short) $\l$-calculus}. Somewhat embarrassingly, it is the most studied of $\l$-calculi, and yet it is the one that it is \emph{never} used in applications. 
Functional programming languages, in particular, often prefer Plotkin's \emph{call-by-value} (\emph{\cbv} for short) $\l$-calculus \cite{DBLP:journals/tcs/Plotkin75}, where $\beta$-redexes can fire only when the argument is a \emph{value} (\ie, not an application) and usually further restrict it to \emph{weak reduction} (\ie, out of abstractions) and to closed terms---what we shall refer to as \mbox{\emph{Closed \cbv}~($\l$-calculus).}

The denotational semantics of the \cbv $\l$-calculus is less studied and understood than the \cbn one (some notable exceptions are \cite{DBLP:journals/fuin/EgidiHR92,DBLP:journals/mscs/PravatoRR99,DBLP:conf/csl/Ehrhard12,DBLP:journals/fuin/ManzonettoPR19}). 
This is not by accident: as  first shown by Paolini and Ronchi Della Rocca
\cite{DBLP:journals/ita/PaoliniR99,DBLP:conf/ictcs/Paolini01,parametricBook}, there are some inherent complications in trying to adapt semantic notions from \cbn to \cbv. They stem from two key facts:
\begin{itemize}
\item \emph{Difficulties with open terms}: while \ccbv is an elegant setting, denotational semantics has to deal with \emph{open} terms, and Plotkin's operational semantics is not adequate for that because of \emph{premature} normal forms---see Accattoli and Guerrieri for extensive discussions~\cite{DBLP:conf/aplas/AccattoliG16}.
\item \emph{Inability to erase some divergent subterms}: while in \cbn every term is erasable, in \cbv only values are erasable. 
Therefore, \cbv solvability identifies a different set of terms than in \cbn. 
In particular, the example of (\cbn) solvable term $\var \Omega$ given above is not solvable in \cbv, because $\Omega$ cannot be erased.
\end{itemize}
The difficulty with open terms has the consequence that \cbv solvability does not admit an \emph{internal} operational characterization akin to Wadsworth's one for \cbn, and thus it is not really an easily manageable notion. 
Additionally, some of the properties that solvable terms have in \cbn are rather verified, in \cbv,  by another, larger set of terms, called here \emph{scrutable terms}\footnotemark
\footnotetext{Introduced by \citet{DBLP:journals/ita/PaoliniR99}, scrutable terms are those terms for which there is a (certain kind of) head context sending them to a \emph{value} (rather than the identity as in solvability). They are called \emph{potentially valuable} in the literature, but we prefer to use a lighter terminology. Inscrutable terms are also called \emph{unsolvable of order 0} in the literature.}. In particular, a term is typable with \cbv intersection/multi types if and only if it is scrutable (instead of solvable).

\paragraph{Two Approaches to Call-by-Value Solvability} The literature has focused more on solvability than scrutability, exploring two opposite approaches towards the difficulties of studying it in \cbv:
\begin{enumerate}
\item \emph{Disruptive}: replacing Plotkin's \cbv calculus with another, extended \cbv calculus so as to obtain a smoother framework, and in particular an easier theory of solvability;
\item \emph{Conservative}: considering Plotkin's \cbv calculus as untouchable and striving harder to characterize semantic notions, and potentially build new ones.
\end{enumerate}
One of the achievements of the disruptive approach is the operational characterization of both \cbv scrutability and solvability due to Accattoli and Paolini \cite{AccattoliPaolini12}. They introduce a \cbv $\l$-calculus with $\letexp$ expressions which is isomorphic to the proof-nets \cbv representation of $\l$-calculus, called \emph{value substitution calculus} (shortened to \VSC), together with a \emph{solving reduction} (called \emph{stratified-weak} in \cite{AccattoliPaolini12}) and prove that (for possibly open terms):
\begin{itemize}
\item $\tm$ is \VSC-scrutable if and only if the weak reduction of $\tm$ terminates;
\item $\tm$ is \VSC-solvable if and only if the solving reduction of $\tm$ terminates.
\end{itemize}
This is akin to Wadsworth's characterization in \cbn, where solvable terms are 
those for which head reduction terminates. 
In particular, the \cbv characterizations show that a term is \cbv solvable if and only if it is \emph{hereditarily scrutable}, since solving reduction is defined by iterating weak reduction (under head abstractions).

The conservative approach is explored by Garc\'ia-P\'erez and Nogueira \cite{DBLP:journals/corr/Garcia-PerezN16}. Inspired by a fine analysis \cbn solvability, they strive to adapt some of its properties to \cbv. 
They propose alternative notions of \cbv solvability and scrutability (that we distinguish adding the prefix \emph{GPN})\footnote{They do not define GPN-scrutability, but their \emph{GPN-unsolvable terms of order 0} can be taken as \emph{GPN-inscrutable terms}.} that are not equivalent to the usual ones in Plotkin's \cbv $\lambda$-calculus. Because of the difficulties of Plotkin's framework, their results are strictly weaker than in \cbn and considerably more complex. They also lack semantic or type-theoretic justifications, \eg via intersection types.

Garc\'ia-P\'erez and Nogueira are also the first ones to clearly mention that---surprisingly---equational theories collapsing all \cbv unsolvable terms are inconsistent\footnote{For this result, they point to \citet{DBLP:journals/ita/PaoliniR99}, where it is mentioned but it is not stated nor proved. It follows instead from results in \citet{DBLP:journals/fuin/EgidiHR92}, where however it is not stated nor mentioned.}. Such a non-collapsibility is relevant because it shows that in \cbv---in contrast to \cbn---\emph{unsolvable} does not mean \emph{meaningless}, that is, there is meaning to be found in (some) \cbv unsolvable terms. Such an essential point seems to have been neglected instead by the disruptive approach. Garc\'ia-P\'erez and Nogueira also show that GPN-inscrutable terms are instead collapsible.

\paragraph{Closing the Schism} A reconciliation of the disruptive and the conservative approaches is obtained by Guerrieri et al. \cite{DBLP:journals/lmcs/GuerrieriPR17}. On the one hand, they embrace the disruptive approach, as they study Carraro and Guerrieri's \emph{shuffling calculus} \cite{DBLP:conf/fossacs/CarraroG14}, another extensions of Plotkin's calculus which can be seen as a variant over the \VSC (where $\letexp$ expressions are replaced by commuting conversions) and where Accattoli and Paolini's operational characterization of solvability smoothly transfers. On the other hand, they prove that \emph{a $\l$-term $\tm$ is solvable in the shuffling calculus if and only if it is solvable in Plotkin's \cbv $\l$-calculus}. Therefore, the disruptive extension becomes a way to study the conservative notion of solvability for Plotkin's calculus.


\paragraph{Open Questions about \cbv Solvability} These works paved the way for a theory of \cbv solvability analogous to the one in \cbn. Such a theory however is still lacking. 

Semantically, the literature has somewhat neglected the important issue of collapsibility in \cbv.
Operationally, the proofs of equivalence of various definitions of \cbn solvability do not carry over Plotkin's calculus, as pointed out by Garc\'ia-P\'erez and Nogueira \cite{DBLP:journals/corr/Garcia-PerezN16}. 

Further delicate points concern the characterization of \cbv solvable terms via intersection types. In \cbv, there exist characterizations of solvable terms via intersection and multi types \cite{DBLP:journals/ita/PaoliniR99,DBLP:conf/fscd/KerinecMR21}. Those type systems, however, are \emph{defective}: contrarily to what claimed in those papers, their systems do not verify subject reduction  (for \cite{DBLP:journals/ita/PaoliniR99} subject expansion also fails), as we detail in Appendix \ref{sect:counter}. 
Carraro and Guerrieri \cite{DBLP:conf/fossacs/CarraroG14} characterize \cbv solvability using relational semantics, but their characterization is not purely semantic (or type-theoretic) because it also needs the syntactic notion of \emph{\cbv Taylor-Ehrhard expansion} \cite{DBLP:conf/csl/Ehrhard12}. Additionally, from none of these characterizations it is possible to extract quantitative operational information. They all rely, indeed, on the shuffling calculus, for which it is unclear how to extract (from type derivations) the number of commuting conversion steps, and the time cost model of which is also unclear, see 
Accattoli and Guerrieri \cite{DBLP:conf/aplas/AccattoliG16}. Accattoli and Guerrieri provide a quantitative characterization of \cbv scrutability via multi types \cite{DBLP:conf/aplas/AccattoliG18}, but not of \cbv solvability.

\paragraph{Contributions} 
In this paper we study all these questions, providing also a \emph{quantitative} analysis of solvability via intersection types. Because of the quantitative aspect, we study solvability \emph{in the VSC} rather than in the shuffling calculus. The \VSC is indeed a better fit than the shuffling calculus for quantitative analyses, because its  number of $\beta$ steps  \emph{is} a reasonable time cost model and can be  extracted from multi type derivations, as shown by \citet{DBLP:conf/lics/AccattoliCC21,DBLP:journals/corr/abs-2104-13979}.

The paper is divided in \emph{two} parts. The first part deals with providing evidence for the robustness of our approach and clarifying some key aspects in the literature. Our contributions are:
\begin{enumerate}
\item \emph{Robustness: solvabilities coincide}. Following Guerrieri et al. \cite{DBLP:journals/lmcs/GuerrieriPR17}, we prove that both solvability and contextual equivalence in the VSC coincide with the corresponding notions in Plotkin's calculus. Thus, similarly to the shuffling calculus, also the \VSC is a disruptive tool which can be used to study conservative notions.
\item \emph{Operationally: alternative definitions of solvability}. We show how to catch, for solvability in the \VSC, the various equivalent definitions of solvability holding in \cbn. Additionally, we give a further new equivalent definition that captures \cbv solvability at the \emph{open} level, instead than at the \emph{strong} one.
\item \emph{Semantically: \cbv collapsibility}. We point out that \cbv scrutable terms are collapsible and show why \cbv unsolvable terms instead are not. Showing this crucial facts simply amounts to collect results already in the literature but which were never presented in this way. 

\end{enumerate}
In the second part, we provide an in-depth study of the relationship between \cbv solvability and multi types. The contributions are:
\begin{enumerate}
\item \emph{Multi types and \cbv solvability}: we characterize \cbv solvability using Ehrhard's \cbv \emph{multi types} \cite{DBLP:conf/csl/Ehrhard12}, which are strongly related to linear logic. 
Namely, we prove that a term is \cbv solvable if and only if it is typable with a certain kind of multi types deemed \emph{solvable} and inspired by Paolini and Ronchi Della Rocca \cite{DBLP:journals/ita/PaoliniR99};
\item \emph{Bounds from types}: refining our  solvable types, we extract the number of steps of the solving reduction on a solvable term, together with the size of the solving normal form. This study re-casts de Carvalho's results in a \cbv setting, but it is more than a simple adaptation, as the \cbv case requires new concepts.
\end{enumerate}
The  contributions of the second part of the paper are the most elaborate. The beginning of \Cref{sect:types} provides an introduction to multi types, references to the literature, and an overview of our results.

\paragraph{The Big Picture} While solvability is certainly subtler in \cbv than in \cbn, our contributions show that, if the presentation of \cbv is carefully crafted, then a solid theory of \cbv solvability is possible. In fact, we obtain a theory comparable to the one in \cbn. 

Because of the duality between \cbn and \cbv in classical logic, the literature tends to see these two settings as mirror images of each other. While we show that solvability can indeed be defined and characterized in similar ways in \cbn and \cbv, we also find that it has inherently different roles for the semantics of the two settings, because of the non-collapsibility of \cbv unsolvable terms. It might look as a negative result, but it actually sheds a positive light on \cbv. It shows indeed that the semantic theory of \cbv is \emph{strictly finer} than the \cbn one, as  unsolvability is too coarse for capturing meaningless terms in \cbv, where the right approach is the finer one of inscrutability. 

It is important to stress that the non-collapsibility of \cbv unsolvable terms does not mean that \cbv solvability is uninteresting, similarly to how the inconsistency of the normalizable theory does not mean that normalization is uninteresting. 

\paragraph{Methodology} The paper is built around a methodology which in our opinion is a further contribution to the theory of \cbv. There are two correlated points:
\begin{enumerate}
\item \emph{Irrelevance}: we use the \VSC as a core calculus, which is sufficient for computing results and characterizing solvability. We also introduce the concept of \emph{irrelevant} extension or subtraction. The idea is that there are some rules and equivalences that 
\begin{enumerate}
\item can be added or removed without breaking termination and confluence, and 
\item can be postponed.
\end{enumerate}
As extensions, we consider a structural equivalence and in the Appendix we discuss a rule related to Moggi's computational $\l$-calculus. As subtractions, we consider a sub-relation of the core \VSC---the substitution of variables---inspired by work on the study of cost models for \cbv \cite{DBLP:journals/iandc/AccattoliC17,fireballs,DBLP:conf/aplas/AccattoliG16,DBLP:conf/aplas/AccattoliG18,DBLP:conf/ppdp/AccattoliCGC19,DBLP:conf/lics/AccattoliCC21}. The VSC without the substitution of variables is a complete operational \emph{sub-core}. It actually turns out that such a sub-core has some operational properties \emph{not} available in the \VSC, and playing a role in the proof of some properties of \cbv inscrutable and unsolvable terms.

\item \emph{Normal forms}: additionally, the sub-core admits a neat  inductive description of \cbv normal forms in terms of \emph{inert terms} and \emph{fireballs}, akin to the one for \cbn and used throughout the paper. The role of inert terms, in particular, is crucial in the study of multi types, and it is also used to give a new alternative definition of \cbv solvability. 
\end{enumerate}

These points are key technical differences between our study of \cbv solvability and the other ones in the literature \cite{DBLP:journals/ita/PaoliniR99,AccattoliPaolini12,DBLP:conf/fossacs/CarraroG14,DBLP:journals/corr/Garcia-PerezN16,DBLP:journals/lmcs/GuerrieriPR17,DBLP:conf/fscd/KerinecMR21}. They seem minor details but they can also be understood from more conceptual points of view. 

Firstly, \cbv is a modular setting organized in \emph{two} levels: there is a core which can be safely extended with further rewriting rules, enriching the equational theory and the flexibility of the calculus. A similar point of view is also advocated by \citet{DBLP:journals/fuin/ManzonettoPR19}. Our core, however, is smaller, as their core (that is, the shuffling calculus) contains part of our structural equivalence.

Secondly, according to the disruptive approach to \cbv, the problem with Plotkin's calculus for \cbv is about premature normal forms, and one needs to extend such a calculus in order to solve it. This is undeniable, and already studied at length, see Accattoli and Guerrieri \cite{DBLP:conf/aplas/AccattoliG16} for an overview. There is however a second essential ingredient for obtaining a good semantic theory, the importance of which---we believe---has not been stressed enough so far:  \emph{having a neat inductive description of normal forms}. The various extensions of Plotkin's calculus do not necessarily have neat grammars for normal forms. Our contribution here is to show both the relevance of neat normal forms and the fact that they are connected to the (non-)substitution of variables. 

\paragraph{Further Related Work} Another framework where solvability is subtler than in \cbn is \cbn extended with pattern matching, as shown by Bucciarelli et al. \cite{DBLP:journals/lmcs/BucciarelliKR21}.
The literature contains many other proposal of \cbv calculi extending Plotkin's, for instance \cite{DBLP:journals/lisp/SabryF93,DBLP:journals/toplas/SabryW97,DBLP:journals/tcs/MaraistOTW99,
DBLP:conf/icfp/CurienH00,DBLP:journals/logcom/DyckhoffL07,DBLP:conf/tlca/HerbelinZ09}. 
Fireballs and inert terms (under other names) were first considered by Paolini and Ronchi Della Rocca \cite{DBLP:journals/ita/PaoliniR99,parametricBook}, and then by  \gregoire and Leroy \cite{DBLP:conf/icfp/GregoireL02}. The recognition of their importance, however, is due to the study of cost models for \cbv.

\paragraph{Proofs} Many proofs are in the Appendix.
This is the long version of a paper accepted at ICFP~2022.

\section{Preliminaries and Notations in Rewriting}
\label{sect:preliminaries}

In this technical section we recall some well-known notions and facts in rewrite theory, and we introduce some notations used in the rest of the paper. We suggest merely skimming over this section on the first reading.\medskip

For a binary relation $\Rew{\Rule}$ on a set of terms, $\Rew{\Rule}^*$ is its reflexive-transitive closure, $\Rew{\Rule}^+$ is its transitive closure, $=_{\Rule}$ is its symmetric, transitive, reflexive closure.
The \emph{transpose} of $\Rew{\Rule}$ is denoted by $\!\lRew{\Rule}$.

Given a binary relation $\Rew{\Rule}$, an $\Rule$-\emph{reduction sequence}---or simply \emph{reduction sequence} if unambiguous---is a finite sequence of terms $\deriv = (\tm_i)_{0 \leq i \leq n}$ (for some $n \geq 0$) such that $\tm_i \Rew{\Rule} \tm_{i+1}$ for all $1 \leq i < n$; we write $\deriv \colon \tm \Rew{\Rule}^* \tmtwo$ if $\tm_0 = \tm$ and $\tm_n = \tmtwo$, and we then say that $\tm$ $\Rule$-\emph{reduces} to $\tmtwo$. 
The \emph{length} $n$ of $\deriv$ is denoted by $\size{\deriv}$, and $\size{\deriv}_a$ is the number of $a$-\emph{steps} (\ie the number of $\tm_i \Rew{a} \tm_{i+1}$ for some $1 \leq i < n$) in $\deriv$, for a given sub-relation $\Rew{a} \,\subseteq\, \Rew{\Rule}$.
We write $\tm \Rew{\Rule}^k \tmtwo$ if there exists $\deriv \colon \tm \Rew{\Rule}^* \tmtwo$ with $\size{\deriv} = k \geq 0$.

A term $\tm$ is $\Rule$-\emph{normal} if there is no $\tmtwo$ such that $\tm \Rew{\Rule} \tmtwo$.
A reduction sequence $\deriv \colon \tm \Rew{\Rule}^* \tmtwo$ is \emph{$\Rule$-normalizing} if $\tmtwo$ is $\Rule$-normal.
A term $\tm$ is \emph{(weakly) $\Rule$-normalizing} if there is a $\Rule$-normalizing reduction sequence $\deriv \colon \tm \Rew{\Rule}^* \tmtwo$; and $\tm$ is \emph{strongly $\Rule$-normalizing} if there is no \emph{diverging reduction sequence} from $\tm$, that is, there is no infinite sequence $(\tm_i)_{i \in \nat}$  such that $\tm_0  = \tm$ and $\tm_i \Rew{\Rule} \tm_{i+1}$ for all $i \in \nat$ (in this case we also say that $\Rew{\Rule}$ is \emph{terminating} on $\tm$).
Clearly, strong $\Rule$-normalization implies weak $\Rule$-normalization.
A relation $\Rew{\Rule}$ is \emph{strongly normalizing} if every term $\tm$ is strongly $\Rule$-normalizing.

A relation $\Rew{\Rule}$ is \emph{confluent} if $\tmtwo_1 \,{}_\Rule^*\!\!\lto \tm \Rew{\Rule}^* \tmtwo_2$  implies $\tmtwo_1 \Rew{\Rule}^* \tmthree \, {}_\Rule^*\!\!\lto \tmtwo_2$ for some $\tmthree$. 
It is well-known that if $\Rew{\Rule}$ is confluent then:
\begin{enumerate}
	\item\emph{Uniqueness of the normal form:} any term $\tm$ has at most one normal form (\ie if $\tm \Rew{\Rule}^* \tmtwo$  and $\tm \Rew{\Rule}^* \tmthree$ with $\tmtwo$ and $\tmthree$ $\Rule$-normal, then $\tmtwo = \tmthree$);
	\item\emph{Church-Rosser:} for every terms $\tm$ and $\tmtwo$, if $\tm =_\Rule \tmtwo$ then $\tm \Rew{\Rule}^* \tmthree \, {}_\Rule^*\!\!\lto \tmtwo$ for some $\tmthree$.
\end{enumerate}

A relation $\Rew{\Rule}$ is \emph{diamond} if $\tmtwo_1 \,{}_\Rule\!\!\lto \tm \Rew{\Rule} \tmtwo_2$ and $\tmtwo_1 \neq \tmtwo_2$ imply $\tmtwo_1 \Rew{\Rule} \tmthree \, {}_\Rule\!\!\lto \tmtwo_2$ for some $\tmthree$. 
It is well-known that if $\Rew{\Rule}$ is diamond then:
\begin{enumerate}
	\item\emph{Confluence:} $\Rew{\Rule}$ is confluent; 
	\item\emph{Random descent:} all $\Rule$-reduction sequences with the same start and end terms have the same length (\ie if $\deriv \colon \tm \Rew{\Rule}^* \tmtwo$  and $\deriv' \colon \tm \Rew{\Rule}^* \tmtwo$ then $\size{\deriv} = \size{\deriv'}$);
	\item\emph{Uniformity:} for any term $\tm$, $\tm$ is weakly $\Rule$-normalizing if and only if $\tm$ is strongly $\Rule$-normalizing.
\end{enumerate}

Two relations $\Rew{\Rule_1}$ and $\Rew{\Rule_2}$ \emph{strongly commute} if $\tmtwo_1 \,{}_{\Rule_1}\!\!\!\lto \tm \Rew{\Rule_2} \tmtwo_2$ implies $\tmtwo_1 \Rew{\Rule_2} \tmthree \, {}_{\Rule_2}\!\!\!\lto \tmtwo_2$ for some $\tmthree$. 
If $\Rew{\Rule_1}$ and $\Rew{\Rule_2}$ strongly commute and are diamond, then 
\begin{enumerate}
	\item\emph{Diamond of the union:} $\Rew{\Rule} \, = \, \Rew{\Rule_1} \!\cup \Rew{\Rule_2}$ is diamond,
	\item\emph{Modular random descent:} all $\Rule$-reduction sequences with the same start and end terms have the same number of any kind of steps (\ie if $\deriv \colon \tm \Rew{\Rule}^* \tmtwo$  and $\deriv' \colon \tm \Rew{\Rule}^* \tmtwo$ then $\size{\deriv}_{\Rule_1} = \size{\deriv'}_{\Rule_1}$ and $\size{\deriv}_{\Rule_2} = \size{\deriv'}_{\Rule_2}$).
\end{enumerate}

\section{Value Substitution Calculus}
\label{sect:vsc}

\begin{figure}
\begin{tabular}{c}
\begin{tabular}{cc}

$\arraycolsep=3pt\begin{array}{rrll}
\multicolumn{4}{c}{\textsc{Language}}
\\
\textsc{Terms } & \vsubterms \ni \tm,\tmtwo, \tmthree & \grameq& \val \mid \tm\tmtwo 
\mid \tm \esub\var\tmtwo 
\\
\textsc{Values } & \val,\valtwo & \grameq & \var \mid  \la\var\tm 
\\[4pt]
\textsc{Sub. ctxs } &\subctx,\subctxtwo  &\grameq &\ctxhole \mid \subctx \esub\var\tm

\end{array}$
&
$\begin{array}{rr@{\ }l@{\ }l}
\multicolumn{4}{c}{\textsc{Root rules}}
\\
    \textsc{Mult. } & \subctxp{\la\var\tm}\tmtwo &  \rtom  & \subctxp{\tm\esub{\var}{\tmtwo}} 
    \\
    \textsc{Exp.}  & \tm\esub\var{\subctxp{\val}} &  \rtoe  & \subctxp{\tm\isub{\var}{\val}} 
    \\[4pt]
    \textsc{Exp. abs}  & \tm\esub\var{\subctxp{\la\vartwo\tmtwo}} &  \rtoeabs  & \subctxp{\tm\isub{\var}{\la\vartwo\tmtwo}} 
\\
    \textsc{Exp. var}  & \tm\esub\var{\subctxp{\vartwo}} &  \rtoevar  & \subctxp{\tm\isub{\var}{\vartwo}} 
\end{array}$    
\end{tabular}
\\[28pt]
\hline
\\[-10pt]
\tabcolsep = 2pt
	\begin{tabular}{c}
\textsc{Open reduction + fireballs}
	\\
	\begin{tabular}{c@{\hspace{.3cm}}c}
	$\begin{array}{r@{\hspace{.15cm}}r@{\hspace{.1cm}}l@{\hspace{.1cm}}ll}
	\textsc{Open ctxs} & \weakctx & \grameq &\ctxhole \mid \weakctx \tm \mid \tm \weakctx \mid \weakctx \esub\var\tm \mid 
\tm\esub\var \weakctx 

	\\
	\textsc{Notations} & \ \tovsubo   & \defeq  & \tomo \cup \toeo \quad\ \  \tovsubonvar    \ \defeq\   \tomo \cup \toeabso
	
	\end{array}$
&
	\begin{tabular}{r@{\hspace{.15cm}}cc}
		\textsc{Open rules:}
		&
		\multirow{2}{*}
		{\begin{prooftree}
	\hypo{\tm \rootRew{a} \tm'}		
					\infer1{\weakctxp{\tm} \Rew{\wsym a} \weakctxp{\tm'}}		\end{prooftree}}
		\\
		$ a \in \set{\msym,\esym,\eabssym,\evarsym}$
	\end{tabular}
	\end{tabular}
	\\[10pt]
	$\arraycolsep = 2pt\begin{array}{l@{\hspace{.3cm}}lll@{\hspace{1cm}}l@{\hspace{.3cm}}llll}
\textsc{Inert terms}  &
 \itm & \grameq & \var \mid \itm \fire \mid \itm \esub{\var}{\itmtwo}
&
\textsc{Fireballs} 
& \fire &\grameq &\val \mid \itm \mid \fire \esub{\var}{\itm}
\end{array}$
\end{tabular}
\\[23pt]
\hline
\\[-10pt]
\tabcolsep = 2pt
	\begin{tabular}{c}
\textsc{Full reduction + full fireballs}
	\\
	\begin{tabular}{c@{\hspace{.3cm}}c}
	$\begin{array}{r@{\hspace{.15cm}}r@{\hspace{.1cm}}l@{\hspace{.1cm}}ll}
	\textsc{Full ctxs} & \fctx &\grameq & \ctxhole \mid \fctx \tm \mid \tm \fctx \mid \la{\var}{\fctx} \mid 
\fctx \esub\var\tm \mid \tm \esub \var \fctx 
	\\
	\textsc{Notations} & \ \tovsub  & \defeq  &\tom \cup \toe \quad \ \ \tovsubnvar \  \defeq  \ \tom \cup \toeabs
	\end{array}$
&
	\begin{tabular}{r@{\hspace{.15cm}}cc}
		\textsc{Full rules:}
		&
		\multirow{2}{*}
		{\begin{prooftree}
    \hypo{\tm \rootRew{a} \tm'}		
    \infer1{\fctxp\tm \Rew{a} \fctxp{\tm'}}
		\end{prooftree}}
		\\
		$a \in \set{\msym,\esym, \eabssym,\evarsym}$
	\end{tabular}
	\end{tabular}
	\\[10pt]
	$\begin{array}{c@{\hspace{1cm}}c@{\hspace{1cm}}ccc}
	\textsc{Full values}
	&\textsc{Full inert terms} 
	&\textsc{Full fireballs} 
	\\[-3pt]
	\sval \grameq \var \mid \la{\var}\sfire
	& \sitm \grameq \var \mid \sitm \sfire \mid \sitm \esub{\var}{\sitmtwo}
	& \sfire \grameq \sval \mid \sitm \mid \sfire \esub{\var}{\sitm}
	\end{array}$
\end{tabular}
\\[28pt]
\hline
\\[-10pt]

\tabcolsep = 2pt
	\begin{tabular}{c}
\textsc{Solving reduction + solved fireballs}
	\\
	\begin{tabular}{c@{\hspace{.3cm}}c}
	$\begin{array}{rl@{\hspace{.1cm}}l@{\hspace{.1cm}}ll}
	\textsc{Solving ctxs} & \solvctx &\grameq &\openctx \mid \la{\var}\solvctx \mid \solvctx \tm \mid \solvctx\esub{\var}{\tm}
	\\
	\textsc{Notations} & \tosolv  & \defeq & \tomsolv \cup \toesolv \quad \ \ \tosolvnvar  \ \defeq \ \tomsolv \cup \toeabssolv
	\end{array}$
&
	\begin{tabular}{rcc}
		\textsc{Solving rules:}
		&
		\multirow{2}{*}
		{\begin{prooftree}
				\hypo{\tm \Rew{\wsym a} \tm'}		
				\infer1{\solvctxp\tm \Rew{\solvredsym a} \solvctxp{\tm'}}
		\end{prooftree}}
		\\
		$a \in \set{\msym,\esym,\eabssym,\evarsym}$
	\end{tabular}
	\end{tabular}
	\\[10pt]
	$\textsc{Solved fireballs} \qquad \solvnf \grameq  \itm \mid \la{\var}\solvnf \mid \solvnf \esub{\var}{\itm}$
\end{tabular}
\end{tabular}
\vspace{-10pt}
\caption{Value Substitution Calculus and its 3 context closures and fireballs (open, full, solving).}
\label{fig:vsc}
\end{figure}

%
%

In this section we define Accattoli and Paolini's \emph{value substitution calculus} (shortened to \VSC) \cite{AccattoliPaolini12}.
Intuitively, the \VSC is a $\lambda$-calculus extended with $\letexp$-expressions, as is common for \cbv $\l$-calculi, such as for instance Moggi's. 
We do however replace a $\letexp$-expression $\letin\var\tmtwo\tm$ with a more compact  \emph{explicit substitution} (\ES for short) notation $\tm\esub{\var}{\tmtwo}$, which binds $\var$ in $\tm$.

The reduction rules of \VSC are slightly unusual as they use \emph{contexts} both to allow one to reduce redexes located in sub-terms, which is standard, \emph{and} to define the redexes themselves, which is less standard---these kind of rules is 
called \emph{at a distance}. The rewriting rules in fact mimic exactly cut-elimination on proof nets, via Girard's \cbv translation $(A \Rightarrow B)^v = \oc (A^v \multimap B^v)$ of intuitionistic logic into linear logic, see Accattoli \cite{DBLP:journals/tcs/Accattoli15}.  
We shall endow the terms of \VSC with various notions of reductions. All the definitions are in \Cref{fig:vsc}, the next paragraphs explain them in order.

\paragraph{Terms and Contexts} Terms may be \emph{applications} $\tm\tmtwo$, \emph{values} (\ie variables $\var, \vartwo, \varthree, \dots$, and \emph{abstractions} $\la{\var}\tm$) and \emph{explicit substitutions} $\tm\esub{\var}{\tmtwo}$.
The set of \emph{free} (\resp \emph{bound}) \emph{variables} of a term $\tm$,  denoted by $\fv{\tm}$ (\resp $\bv{\tm}$), is defined as expected, abstractions and \ES being the only binding constructors. 
Terms are identified up to renaming of bound variables.
We use $\tm \isub{\var}{\tmtwo}$ for the capture-avoiding substitution of
$\tm$ for each free occurrence of $\var$ in $\tm$.

All along the
paper we use (many notions of) \emph{contexts}, \ie terms with exactly one hole, noted $\ctxhole$. Plugging a term $\tm$ in a context $\ctx$, noted $\ctxp{\tm}$, possibly~captures free variables of $\tm$. For instance $(\la\var\ctxhole)\ctxholep\var = \la\var\var$, while $(\la\var\vartwo)\isub\vartwo\var = \la\varthree\var$.
\Cref{fig:vsc} defines the notions of context that we use.

\paragraph{Root rewriting rules}
In \VSC, $\beta$-redexes are decomposed via ES, and the
\emph{by-value} restriction is on $ES$-redexes, not on $\beta$-redexes, because only values
can be substituted.
There are two main rewrite rules, the \emph{multiplicative} one $\tom$ and the \emph{exponential} one $\toe$ (the terminology comes from the connection between \VSC and linear logic), and both work \emph{at a distance}: they use contexts even in the definition of their \emph{root} rules (that is, before the contextual closure). Their definition is based on \emph{substitution contexts} $\sctx$, which are lists of~\ES. 
In \Cref{fig:vsc}, the root rule $\rtom$ (resp. $\rtoe$) is assumed to be capture-free, so no free
 variable of $\tmtwo$ (resp. $\tm$) is captured by the substitution context $\subctx$ (by possibly $\alpha$-renaming on-the-fly). 

Examples: $(\la\var\vartwo)\esub\vartwo\tm\tmtwo \rtom \vartwo\esub\var\tmtwo\esub\vartwo\tm$ and $(\la\varthree\var\var)\esub\var{\vartwo\esub\vartwo\tm} \rtoe (\la\varthree\vartwo\vartwo)\esub\vartwo\tm$. An example with on-the-fly $\alpha$-renaming is $(\la\var\vartwo)\esub\vartwo\tm\vartwo \rtom \varthree\esub\var\vartwo\esub\varthree\tm$.

The multiplicative rule
$\rtom$ fires a $\beta$-redex at a distance and generates an \ES even when the argument is not a value.
	The \cbv discipline is entirely encoded in the exponential rule $\toe$ (see \Cref{fig:vsc}): it can fire an \ES performing a substitution only when its argument is a \emph{value} (\ie a variable or an abstraction) up to a list of \ES. This means that only values can be duplicated or erased.
It is useful to split the exponential root rule $\rtoe$ in two disjoint rules, depending on whether it is an abstraction (rule $\rtoeabs$) or a variable ($\rtoevar$) that it is substituted.

We shall consider 3 different contextual closures for the given rules. For all of them, rule $\rtoevar$ shall be postponable without altering the properties of the calculus (\Cref{prop:irrelevance}). Actually, the reductions without $\rtoevar$ shall have stronger properties, crucial for some of our results.

With respect to the explanations in the introduction, our \emph{core} calculus is the \VSC with its three contextual closures. The irrelevant extension shall be considered in the next section. The \emph{sub-core} is instead obtained by removing $\rtoevar$ from the three contextual closures, and it is justified at the end of this section.
\subsection{The Open VSC} 
The first contextual closure is the \emph{open} one,  where rewriting is forbidden under abstraction and terms are possibly open (but not necessarily). It is obtained via (possibly) open contexts $\weakctx$ (see \Cref{fig:vsc}). 
We consider both the reduction that substitutes variables, noted $\tovsubo$, and the one that does not, noted $\tovsubonvar$ (indeed, note that $\tovsubonvar \,=\, \tovsubo\smallsetminus \toevar$). Examples:
\begin{align*}
		\tm \esub\var{(\la\vartwo\tmtwo)\esub\varthree\tmthree \tmfour}
		& \tomo \tm \esub\var{\tmtwo\esub\vartwo\tmfour\esub\varthree\tmthree}
		&
		\tm ((\var\var) \esub\var{\vartwo\esub\varthree\tmtwo}) & \toevaro\! \tm ((\vartwo\vartwo)\esub\varthree\tmtwo)
		\\[-2pt]
		((\var\var) \esub\var{\la\vartwo\varthree} \tm)\esub\varfour\tmtwo 
		& \toeabso\! ((\la\vartwo\varthree) (\la\vartwo\varthree) \tm)\esub\varfour\tmtwo
		&
		\la\varthree ((\var\var) \esub\var{\la\vartwo\varthree}) & \not\toeabso\!  \la\varthree((\la\vartwo\varthree) \la\vartwo\varthree)
\end{align*}

Normal forms for $\tovsubonvar$ admit a neat inductive description via \emph{inert terms} and \emph{fireballs}.

\paragraph{Inert Terms and Fireballs} \cbv is about \emph{values}, and, if terms are closed, normal forms are abstractions. In going beyond the closed setting,  a finer view is required. First, the notion of normal form in the Open \VSC is more generally given by the mutually defined notions of \emph{inert terms} and \emph{fireballs} in \Cref{fig:vsc}. 
Second, variables are \emph{both} values and inert terms. This is on purpose, because they have the properties of both kinds of term.

Examples: $\la\var\vartwo$ is a fireball 
as an abstraction, while $\vartwo(\la\var\var)$, $\var\vartwo$, and $(\varthree(\la\var\var))(\varthree\varthree) (\la\vartwo(\varthree\vartwo))$ are fireballs as inert terms. The grammars also allow to have ES containing inert terms around abstractions and applications: $(\la\var\vartwo)\esub\vartwo{\varthree\varthree}$ is a fireball and $\var\esub\var{\vartwo(\la\var\var)}\vartwo$ is an inert term. One of the key points of inert terms is that they have a \emph{free} head variable (in particular they are open). In \gregoire and Leroy 
\cite{DBLP:conf/icfp/GregoireL02}, inert terms are called \emph{accumulators}, and fireballs are simply called \emph{values}.

Normal forms for $\tovsubonvar$ are exactly fireballs. Note that $\var\esub\var\vartwo$ is an inert term and it is not $\toevaro$ normal, thus not $\tovsubo$ normal. Normal forms for $\tovsubo$ are a slightly stricter subset of fireballs (they are fireballs without ES of shape $\esub\var{\sctxp\vartwo}$), with a similar but less neat and omitted inductive description. We shall show that $\toevaro$ is postponable and strongly normalizing, allowing us to take fireballs as our reference notion of open normal form. The same approach shall be followed for the other contextual closures.

In the literature, fireballs have been considered for different open CbV calculi (the VSC and the fireball calculus), and their definition depends on the calculus. They are however characterized by the same operational property (fireballs are the normal forms for the open reduction of the chosen CbV calculus) and they correspond to each other, see \citet{DBLP:conf/aplas/AccattoliG16}. 
The name \emph{fireball}, due to \citet{fireballs}, is a pun: in the fireball calculus, a $\beta$-redex can be \emph{fired} only when the argument is a fireball, so fireballs are the \emph{fireable} terms, more catchily called \emph{fireballs}.

\begin{proposition}[Basic properties of open reduction]
	\label{prop:ovsc-diamond}\label{prop:properties-open-reduction}
	\NoteProof{propappendix:properties-open-reduction}
	\begin{enumerate}

		\item \label{p:properties-open-reduction-commutation} 
		\emph{Strong commutation:} reductions $\tomo$, $\toeabso$, and $\toevaro$ are pairwise strongly commuting.
	
		\item \label{p:properties-open-reduction-diamond}
		\emph{Diamond}: reductions $\tovsubo$ and $\tovsubonvar$ are diamond (separately).
				
		\item \label{p:properties-open-reduction-harmony}
		\emph{Normal forms}:  $\tm$ is $\onvarsym$-normal if and only if $\tm$ is a fireball. 
		If $\tm$ is $\osym$-normal then it is a fireball.

\end{enumerate}
\end{proposition}

Diamond of $\tovsubo$ and strong commutation of $\tomo$ and $\toeo$ are technical facts (see  \Cref{sect:preliminaries} for definitions) with relevant consequences:
$\tovsubo$ is confluent and its non-determinism is only apparent, because if an $\osym$-reduction sequence from $\tm$ reaches
a $\osym$-normal form $\tmtwo$, then \emph{every} $\osym$-sequence from $\tm$ eventually ends in $\tmtwo$;
and all these sequences have the \emph{same length} and \emph{same number} of $\msym$-steps and $\esym$-steps. This is essential for measuring them via multi types in the second part of the paper.
The same properties shall hold for solving reduction.	

We shall use also the valuability property of $\tovsubo$, that is, that $\tovsubo$ is enough to reach a value. 
It is slightly weaker than a normalization theorem, because it concerns a specific kind of open normal form, values, and not all open normal forms, as it leaves out inert terms. 
Two more general normalization theorems also hold, proved independently using type theoretic means (\Cref{sect:denotation}) and the irrelevance of $\toevar$ (\Cref{ssect:irrelevance}), two notions that we shall introduce in the next sections.

\begin{proposition}[Further properties of open reduction]
		\label{prop:properties-open-extra} 
		\NoteProof{propappendix:properties-open-extra}
		\hfill
		\begin{enumerate}
			
			\item\label{p:properties-open-extra-valuability} \emph{Valuability \cite{AccattoliPaolini12}:} if $\tm \tovsub^{*}\val$ for some value $\val$, then $\tm \tovsubo^{*} \valtwo$ for some value $\valtwo$.
			
			\item\label{p:properties-open-extra-normalization} \emph{Normalization:} if $\tm \tovsub^{*} \tmtwo$ for some $\osym$-normal $\tmtwo$, then $\tm \tovsubo^{*} \tmthree$ for some $\osym$-normal $\tmthree$.
			
			\item\label{p:properties-open-extra-normalization-bis} \emph{Normalization 2:} if $\tm \tovsubnvar^{*} \fire$ for some fireball $\fire$, then $\tm \tovsubonvar^{*} \firetwo$ for some fireball $\firetwo$.
			
		\end{enumerate}
	\end{proposition}

\subsection{The Strong/Full VSC} 

To avoid notation clashes between the solving and strong reductions (both would start with '\emph{s}'), we refer to the \emph{strong} one as to the \emph{full} one. The \Full \VSC allows rewrite rules  to fire everywhere in a term, via full contexts $\fctx$ (see \Cref{fig:vsc}). It is here for the sake of completeness, it does not really play a role in our study. 
Note that $\tovsubnvar \,=\, \tovsub \!\smallsetminus \toevar$, \ie, $\tovsubnvar$ is the full reduction that does not substitute variables.
Full fireballs are obtained by iterating the fireball construction under all abstractions. Full reductions $\tovsub$ and $\tovsubnvar$ are confluent but not diamond, just consider the example below ($\Id \defeq \la{\varthree}\varthree$):
\begin{center}
$\begin{array}{ccccccccc}
(\var\var) \esub{\var}{\la{\vartwo} \Id \Id} & &\tom && (\var\var) \esub{\var}{\la{\vartwo} \varthree\esub\varthree\Id}
\\
\downarrow_{\eabssym} &&&& \downarrow_{\eabssym}
\\
(\la{\vartwo} \Id \Id) (\la{\vartwo} \Id \Id) & \tom & (\la{\vartwo} \varthree\esub\varthree\Id) (\la{\vartwo} \Id \Id) & \tom & (\la{\vartwo} \varthree\esub\varthree\Id) (\la{\vartwo} \varthree\esub\varthree\Id)
\end{array}$
\end{center}

\begin{proposition}[Basic properties of the full reduction]
	\label{prop:properties-full-reduction} 
	\NoteProof{propappendix:properties-full-reduction}
	\hfill
	\begin{enumerate}
		\item\label{p:properties-full-reduction-confluence} \emph{Confluence \cite{AccattoliPaolini12}:}  reductions $\tovsub$ and $\tovsubnvar$ are confluent.
		
		\item\label{p:properties-full-reduction-harmony} \emph{Normal forms:} a term is $\vsubnvarsym$-normal if and only if it is a \full fireball. If a term is $\vsub$-normal then it is a \full fireball.
	\end{enumerate}
\end{proposition}

\subsection{Solving Reduction} 
Accattoli and Paolini's  solving reduction 
 $\tosolv$ from  \cite{AccattoliPaolini12} is in between the full one $\tovsub$ and the open one $\tovsubo$, that is, it restricts $\tovsub$ but extends $\tovsubo$. 
 It iterates open reduction under head 
abstractions \emph{only}, via the notion of \emph{solving~context} $\solvctx$ (see \Cref{fig:vsc}).
For instance, the extension under head abstractions gives
$\la\var (\Id \Id) \tomsolv \la\var(\varthree\esub\varthree\Id) \toesolv \la\var\Id$.
But reduction under non-head abstractions is forbidden:
$\vartwo (\la\var(\Id \Id)) \not\tomsolv  \vartwo (\la\var(\varthree\esub\varthree\Id))$.
Solved fireballs iterate the fireball structure under head abstractions only.

\begin{proposition}[Properties of solving reduction]
	\label{prop:solvsc-diamond}\label{prop:properties-solvable-reduction}
	\NoteProof{propappendix:properties-solvable-reduction} 
	\hfill
	\begin{enumerate}

		\item \label{p:properties-solvable-reduction-commutation} 
		\emph{Strong commutation:} reductions $\tomsolv$, $\toeabssolv$, and $\toevarsolv$ are pairwise strongly commuting.

		\item \label{p:properties-solvable-reduction-diamond}
		\emph{Diamond}: reductions $\tovsubsolv$ and $\tosolvnvar$ are diamond (separately).

		\item \label{p:properties-solvable-reduction-harmony} 
		\emph{Normal forms}: a term is $\solvnvarsym$-normal if and only if it is a solved fireball. 
		If a term is $\solvredsym$-normal then it is a solved fireball.
		\end{enumerate}
\end{proposition}

\cbv solvability shall be introduced in \Cref{sect:solving-strat}. From that point on, solving reduction shall play a crucial role.
In particular, we shall also use its following properties.

\begin{proposition}[Further properties of solving reduction]
	\label{prop:properties-solving} 
	\NoteProof{propappendix:properties-solving}
	\hfill
	\begin{enumerate}
		
		\item\label{p:properties-solving-normalization} \emph{Normalization:} if $\tm \tovsub^{*} \tmtwo$ for some $\solvsym$-normal $\tmtwo$, then $\tm \tosolv^{*} \tmthree$ for some $\solvsym$-normal $\tmthree$.
		
		\item\label{p:properties-solving-normalization-bis} \emph{Normalization 2:} if $\tm \tovsubnvar^{*} \solvfire$  with $\solvfire$ solved fireball, then $\tm \tosolvnvar^{*} \solvfire'$ for some solved fireball $\solvfire'$.
		
		\item\label{p:properties-solving-decomposition} \emph{Stability by extraction from a head context:} if $\hctxp{\tm} \tosolv^* \tmtwo$ for some head context $\hctx$ and $\solvsym$-normal $\tmtwo$, then $\tm \tosolv^* \tmthree$ for some $\solvsym$-normal $\tmthree$.
	\end{enumerate}
\end{proposition}	

\subsection{Irrelevance of Variable Exponential Steps}
\label{ssect:irrelevance}
In \VSC some sub-reductions of $\tovsub$ can be neglected because they are computationally \emph{irrelevant}, in that they can be postponed and do not jeopardize normalization. Typically, this is the case for $\toevar$ and its variants, but there shall also be other cases.

\begin{definition}[Irrelevance]
	\label{def:irrelevance}
	Let $R, S$ be binary relations on $\vsubterms$.
	We say $R$ is $S$-\emph{irrelevant} if for every $\tm,\tmtwo \in \vsubterms$:
	\begin{itemize}
		\item \emph{Postponement:} if $\deriv \colon \tm \,(S \cup R)^*\, \tmtwo$ then there is $\deriv' \colon \tm \  S^* R^* \, \tmtwo$ with $\sizem{\deriv'} = \sizem{\deriv}$; and
		\item \emph{Termination:} $S$ is weakly (\resp~strongly) normalizing on $\tm$ if and only if so is $S \cup R$.
	\end{itemize}
\end{definition}

\begin{proposition}[Irrelevance of $\toevar$, $\toevaro$, and $\toevarsolv$]
	\label{prop:irrelevance}
	\NoteProof{propappendix:irrelevance}
	Reduction $\toevar$ is $\tovsubnvar$-irrelevant, reduction $\toevaro$ is $\tovsubonvar$-irrelevant, and reduction $\toevarsolv$ is $\tosolvnvar$-irrelevant.
\end{proposition}

Distinguishing between $\tovsub$ and $\tovsubnvar$ (or between $\tosolv$ and $\tosolvnvar$, or between $\tovsubo$ and $\tovsubonvar$) is important for at least two reasons. 
Firstly, as we have already seen it, if reduction excludes $\toevar$ then its normal forms have a neat inductive description. 
Secondly, if reduction excludes $\toevar$ then it is stable under substitution. 

\begin{lemma}[Stability of reductions without $\toevar$ under substitution]
	\label{l:stability-substitution}
	\NoteProof{lappendix:stability-substitution}
			Let $a \in \{\onvarsym, \solvnvarsym, \vsubnvarsym\}$.
		If $\tm \Rew{a} \tmtwo$ then $\tm\isub\var\tmthree \Rew{a} \tmtwo\isub\var\tmthree$ for every $\tmthree$.
\end{lemma}

Note that stability under substitution of \emph{values} (that is, when $\tmthree$ is a value in the statement of \Cref{l:stability-substitution}) holds also for $\tovsubo$, $\tosolv$ and $\tovsub$. 
The problem is that it breaks for $\toevar$ steps when $\tmthree$ is not a value: $(\vartwo\vartwo)\esub\vartwo\var \toevaro \var\var$ but $(\vartwo\vartwo)\esub\vartwo\var\isub\var\tmthree = (\vartwo\vartwo)\esub\vartwo\tmthree  \not\toevaro \tmthree\tmthree= (\var\var)\isub\var\tmthree$ if $\tmthree$ is not a value. 
The same remark applies to $\tosolv$ and $\tosolvnvar$, and to $\tovsub$~and~$\tovsubnvar$.

\subsection{Time Cost Model}
Here we recall the reasonable time cost model for the \VSC. It is not used anywhere in the paper, but it is part of the motivations for adopting the \VSC and pursuing a quantitative study via multi types.

\paragraph{External Reduction and Reasonable Time} \citet{DBLP:conf/lics/AccattoliCC21} introduce the \emph{external reduction} $\tovsubsolv$ of \VSC, a sub-reduction of $\tovsubnvar$ that is diamond, extends $\tovsubonvar$ and computes full normal forms. They prove that the number of multiplicative steps of external reduction to full normal form is a reasonable time cost model for full \cbv. The exclusion of $\toevar$ is a detail and does not affect the result, as the postponement of $\toevar$ (due to its irrelevance) actually preserves the number of multiplicative steps.

\citet{DBLP:journals/corr/abs-2104-13979} additionally prove that external reduction is \emph{normalizing} (in the untyped calculus), that is, that it reaches the normal form whenever it exists, thus giving to external reduction the same status of leftmost-outermost evaluation in \cbn.

\paragraph{Solving Reduction and Reasonable Time} Solving reduction $\tosolv$ is a strict sub-relation of external reduction, thus its number of multiplicative steps also is a reasonable time cost model (the same holds for $\tosolvnvar$, and for $\tovsubo$ and $\tovsubonvar$).

\section{Plotkin and Shuffling}
\label{sect:other-calculi}

In this section, we compare the \VSC with two untyped \cbv calculi in the literature, namely Plotkin's $\plotcalc$ and Carraro and Guerrieri's shuffling calculus. A further case comparison, with Moggi's computational $\l$-calculus, is in \Cref{sect:app-moggi} for lack of space. 

For relationships with further calculi, see Accattoli and Paolini \cite{AccattoliPaolini12}, where the relationship with a calculus by Herbelin and Zimmerman \cite{DBLP:conf/tlca/HerbelinZ09} is studied, or Accattoli and Guerrieri \cite{DBLP:conf/aplas/AccattoliG16}, where the relationship with the intuitionistic \cbv fragment of Curien and Herbelin $\lambdamucalc$ \cite{DBLP:conf/icfp/CurienH00} is studied. 

\paragraph{Theories} We introduce here the notion of \emph{equational theory}, referred to several calculi, which shall be used for the comparisons of these section and also in the study of collapsibility in \Cref{sect:collapsibility}. 

\begin{definition}[(Equational) theories]
\label{def:eq-theory}
Let $\mathsf{X}$ be a calculus, that is, a set of terms with a binary relation $R$ on it.
\begin{itemize}
\item An \emph{$\mathsf{X}$-theory} $\eqth$ is an equivalence relation  containing $R$ and closed by all contexts of $\mathsf{X}$.
\item The \emph{equational theory} $\eqth_\mathsf{X}$ of $\mathsf{X}$ is the smallest $\mathsf{X}$-theory, that is, it is the symmetric, reflexive, transitive, and contextual closure of $R$.
\end{itemize}
\end{definition}

\subsection{Plotkin}
Plotkin's original \cbv $\lambda$-calculus $\plotcalc$ \cite{DBLP:journals/tcs/Plotkin75} can be easily simulated in the \VSC.
The syntax of $\plotcalc$ is simply the same as in the \VSC but without \ES.
Coherently with our notations, we define $\toobvplot$ and $\tofbvplot$ in $\plotcalc$ as the closures under open and \full contexts (without \ES) of the $\betaplot$-rule:
\[(\la{\var}\tm)\val \rtobvplot \tm\isub{\var}{\val} \qquad \text{where $\val$ is a value (without \ES)}.\]
\begin{proposition}[Simulation]
	\label{prop:plotkin-vsc}
	\NoteProof{propappendix:plotkin-vsc}
	Let $\tm$ and $\tm'$ be terms without \ES. 
	If $\tm \toobvplot \tm'$ then $\tm \tomo \!\cdot \toeo  \tm'$; and if $\tm \tofbvplot \tm'$ then $\tm \tom \!\cdot \toe  \tm'$.
\end{proposition}

There is no sensible way to simulate \VSC into $\plotcalc$.
Indeed \VSC is a \emph{proper} extension of $\plotcalc$:  terms such as $(\la{\var}\delta)(\vartwo\vartwo)\delta$ and $\delta ((\la{\var}\delta) (\vartwo\vartwo))$ ($\delta \defeq \la{\var}\var\var$) diverge in \VSC, but they are $\fullsym\betaplot$-normal. 

\begin{corollary}[Plotkin $\subsetneq$ \VSC]
\label{coro:plot-inside-vsc}
	The equational theory of $\plotcalc$ is strictly contained in the equational theory of \VSC, that is, $\eqth_{\plotcalc} \, \subsetneq \, \eqth_\vsub$.
\end{corollary}

\begin{proof}
From the simulation of $\plotcalc$ into \VSC (\Cref{prop:plotkin-vsc}), it follows immediately that $=_{\fbetaplot} \, \subseteq \, =_\vsub$.
	The inclusion is strict because, in $\plotcalc$, $(\la{\var}\delta)(\vartwo\vartwo)\delta$ is $\fbetaplot$-normal while the only $\fbetaplot$-reduction sequence from  $(\la{\var}\delta\delta)(\vartwo\vartwo)$ is $(\la{\var}\delta\delta)(\vartwo\vartwo) \tofbvplot (\la{\var}\delta\delta)(\vartwo\vartwo) \tofbvplot \cdots$, so by Church-Rosser $(\la{\var}\delta)(\vartwo\vartwo)\delta \neq_{\fbetaplot} (\la{\var}\delta\delta)(\vartwo\vartwo)$;
	in \VSC instead, $(\la{\var}\delta)(\vartwo\vartwo)\delta =_\vsub (\la{\var}\delta\delta)(\vartwo\vartwo)$ because $(\la{\var}\delta)(\vartwo\vartwo)\delta \tom \delta\esub{\var}{\vartwo\vartwo}\delta \tom (\varthree\varthree)\esub{\varthree}{\delta}\esub{\var}{\vartwo\vartwo} \toe (\delta\delta)\esub{\var}{\vartwo\vartwo} \lRew{\msym} (\la{\var}\delta\delta)(\vartwo\vartwo)$.
\end{proof}

Despite extending $\plotcalc$, VSC does not lose the \cbv essence, as $(\la\var\vartwo) \Omega$ has no normalizing reduction sequence in both $\plotcalc$ and VSC, while in \cbn it normalizes in one step, erasing $\Omega \defeq \delta\delta$.

\paragraph{Valuability and Contextual Equivalence} While the equational theory of the VSC is strictly larger than the one of $\plotcalc$, in some respects the two calculi are equivalent, as we now show. First, in the special case where the (open) \VSC turns a term into a value $\val$ then (the open) $\plotcalc$ can do it as well.

%
%
\begin{lemma}[Lifting valuability]
	\label{l:valuing-sequences-lift-to-plotkin}
	\NoteProof{lappendix:valuing-sequences-lift-to-plotkin}
	If $\tm \tovsubo^{*} \val$ and $\tm$ is without \ES, then $\val$ is  without~\ES and $\tm\tobvploto^* \!\!\val$.
\end{lemma}

Such a property allows us to show that the contextual equivalences of the two calculi coincide. Contextual equivalence shall play a role in \Cref{sect:collapsibility}.

\begin{definition}[\cbv contextual equivalence]
\label{def:ctx-eq}
Let $\tm$ and $\tmtwo$ be two terms in a \cbv calculus $\mathsf{X}$.
We say that $\tm$ is \emph{contextually equivalent} to $\tmtwo$ in $\mathsf{X}$ if and only if for every context $\ctx$ (of $\mathsf{X}$) such that $\ctxp\tm$ and $\ctxp\tmtwo$ are closed we have that $\ctxp\tm \Rew{\mathsf{X}}^* \val$ if and only if $\ctxtwop\tm \Rew{\mathsf{X}}^* \valtwo$, for some values $\val$ and $\valtwo$ (of $\mathsf{X}$).
\end{definition}

\begin{proposition}[Equivalence of contextual equivalences]
\label{prop:ctx-eq-eq} 
Let $\tm$ and $\tmtwo$ be $\l$-terms. Then  $\tm \ctxeq^{\plotcalc} \tmtwo$ if and only if $\tm \ctxeq^{\VSC} \tmtwo$.
\end{proposition}

\begin{proof}
 Direction $\Rightarrow$ holds because the VSC simulates $\plotcalc$ (\Cref{prop:plotkin-vsc}). 
Direction $\Leftarrow$ follows from the fact that if $\tm \tovsub^* \val$ for some value $\val$ then $\tm \tobvplot \valtwo$ for some value $\valtwo$\!. 
Indeed, by valuability (\Cref{prop:properties-open-extra}.\ref{p:properties-open-extra-valuability}) $\tm \tovsubo^* \valtwo$ for some value $\valtwo$; 
by lifting (\Cref{l:valuing-sequences-lift-to-plotkin}), $\valtwo$ is without \ES and $\tm \tobvploto^*\! \valtwo$\!.
\end{proof}

\Cref{prop:ctx-eq-eq} deals with terms with no ES, what about contextual equivalence on terms with ES? We need a way of expanding ES into $\beta$-redexes, that shall be used also in the following sections, and that preserves contextual equivalence.

\begin{definition}[ES expansion]
Given a term $\tm$ with ES, the expansion of all the ES of $\tm$ into $\beta$-redexes is obtained by applying $\tom$ backwards, obtaining a term $\tm^\bullet$ without ES. Formally, $(\tmtwo\esub{\var}{\tmthree})^\bullet \defeq (\la{\var}\tmtwo^\bullet)\tmthree^\bullet$, and in the other cases $\tm^\bullet$ is defined as expected.
\end{definition}

\begin{lemma}[Stability of contextual equivalence by ES expansion]
\label{l:es-exp-and-ctx-eq}
\NoteProof{lappendix:es-exp-and-ctx-eq}
Let $\tm, \tmtwo \in \vsubterms$: one has
$\tm \ctxeq^{\VSC} \tmtwo$ if and only if $\tm^\bullet \ctxeq^{\VSC} \tmtwo^\bullet$ (if and only if $\tm^\bullet \ctxeq^{\plotcalc} \tmtwo^\bullet$).
\end{lemma}

\subsection{Shuffling and Structural Equivalence} 
To relate the \VSC to Carraro and Guerrieri's shuffling calculus $\shufcalc$ \cite{DBLP:conf/fossacs/CarraroG14}, we need a concept. 

\paragraph{Structural Equivalence} The VSC comes with a notion of \emph{structural equivalence} $\eqstruct$, that equates terms differing only in the position of ES. 
A strong justification comes from the \cbv linear logic interpretation of $\l$-terms with ES, in which structurally equivalent terms 
translate to the same (recursively typed) proof net, see Accattoli \cite{DBLP:journals/tcs/Accattoli15}. 
Structural equivalence $\eqstruct$ is defined as the least equivalence relation on terms closed by all contexts 
and generated by the following root cases:
%
\begin{alignat*}{3}
	\tm\esub\var\tmtwo\tmthree &\tostructapl (\tm\tmthree)\esub\var\tmtwo  &\ \textrm{ if }\var\not\in\fv\tmthree \qquad
	&
	\tm\esub{\var}{\tmtwo\esub{\vartwo}{\tmthree}} &\tostructes \tm\esub{\var}{\tmtwo}\esub{\vartwo}{\tmthree} &\ \mbox{ 
		if $\vartwo\not\in\fv{\tm}$}  
	\\[-2pt]
	\tm\,\tmthree\esub\var\tmtwo &\tostructapr  (\tm\tmthree)\esub\var\tmtwo  &\ \textrm{if }\var\not\in\fv{\tm} \qquad
	&
	\tm\esub{\vartwo}{\tmthree}\esub{\var}{\tmtwo} &\tostructcom \tm\esub{\var}{\tmtwo}\esub{\vartwo}{\tmthree} &\ \mbox{ if $\vartwo\notin\fv{\tmtwo}$, $\var\notin\fv{\tmthree}$}
\end{alignat*}
\label{eq:eqstruct}
Pleasantly, adding $\eqstruct$ results in a smooth system, as $\eqstruct$ commutes with the rewriting rules, and can thus be 
postponed. Additionally, the commutation is \emph{strong}, as it preserves the number and kind of steps (thus the cost model)---one says that 
it is a \emph{strong bisimulation} (with respect to $\tovsub$). 
 Being a strong bisimulation in particular implies that $\eqstruct$ is \emph{irrelevant}, as it is the case for $\toevar$.
\begin{proposition}[Operational properties of $\eqstruct$]
	\label{prop:strong-bisimulation}
	\NoteProof{propappendix:strong-bisimulation} 
	\hfill
	\begin{enumerate}
	\item \emph{$\eqstruct$ is a strong bisimulation}: if $\tm\eqstruct\tmtwo$ and $\tm\Rew{\mathsf{a}}\tmp$ then there exists $\tmtwop \in \vsubterms$ such that 
$\tmtwo\Rew{\mathsf{a}}\tmtwop$ and $\tmp\eqstruct\tmtwop$, for $\mathsf{a}\in\set{\msym,\esym, \vsub, \omsym,\oesym, \osym, \solvredsym\msym,\solvredsym\esym,\solvsym}$.

 	\item $\eqstruct$ is $\Rew{\mathsf{a}}$-irrelevant, for $\mathsf{a}\in\set{\msym,\esym, \vsub, \omsym,\oesym, \osym, \solvredsym\msym,\solvredsym\esym,\solvsym}$.
\end{enumerate}
\end{proposition}
From strong bisimulation of $\streq$, it immediately follows that rewriting modulo $\streq$ is confluent. Let $\tm \Rew{\vsub/\eqstruct} \tmtwo$ be defined as $\tm \streq \tm' \tovsub \tmtwo' \streq \tmtwo$ for some $\tm'$ and $\tmtwo'$.

\begin{lemma}[Reduction modulo $\streq$ is confluent]
	\label{l:confluence-modulo-eqstruct}
	\NoteProof{lappendix:confluence-modulo-eqstruct} %
Reduction $\Rew{\vsub/\eqstruct}$ is confluent.
\end{lemma}

The same reasoning also applies to any other contextual closure of \VSC, with or without $\rtoevar$. 
We shall also show that typability with multi types is invariant by structural equivalence~(\Cref{prop:qual-subject}).

\paragraph{Shuffling} The equational theory of the shuffling calculus \cite{DBLP:conf/fossacs/CarraroG14,Guerrieri15,GuerrieriPR15} is contained in \VSC modulo $\eqstruct$. The $\shufcalc$  extends Plotkin's calculus with two rules, $\sigma_{1}$ and $\sigma_{3}$:
\begin{center}
\arraycolsep = 2pt
$\begin{array}{rll@{\hspace{2cm}}rllllllll}
 ((\l\var.\tm)\tmtwo)\tmthree &\rtosl& (\l\var.\tm \tmthree)\tmtwo 
 &
 \val ((\l\var.\tmthree)\tmtwo) &\rtosr& (\l\var.\val \tmthree)\tmtwo 
\end{array}$
\end{center}

\begin{proposition}[Shuffling $\subseteq \VSC/\eqstruct$]
The equational theory of Carraro and Guerrieri's $\shufcalc$ is contained in the one of the \VSC extended with $\eqstruct$, that is, $\eqth_{\shufcalc} \, \subsetneq \, \eqth_{\vsub/\eqstruct}$.

\end{proposition}
\begin{proof}
The containment is proved by \Cref{coro:plot-inside-vsc} for $\betaplot$ and as follows for $\sigma_{1}$ and $\sigma_{3}$:
\begin{center}
\arraycolsep = 2pt
$\begin{array}{rll@{\hspace{.2cm}}l@{\hspace{.2cm}}llllllll}
\tmfour \defeq ((\l\var.\tm)\tmtwo)\tmthree &\rtosl& (\l\var.\tm \tmthree)\tmtwo \eqdef \tmfour'
& \mbox{is captured by} &
  \tmfour & \tom& \tm\esub\var\tmtwo\tmthree & \tostructapl & (\tm\tmthree)\esub\var\tmtwo & \lRew{\msym} & \tmfour'
\\
\tmfour \defeq\val ((\l\var.\tmthree)\tmtwo) &\rtosr& (\l\var.\val \tmthree)\tmtwo \eqdef \tmfour'
& \mbox{is captured by} &
  \tmfour & \tom& \val \tmthree\esub\var\tmtwo & \tostructapr & (\val\tmthree)\esub\var\tmtwo & \lRew{\msym} & \tmfour'
\end{array}$
\end{center}
The following different $\shufcalc$-normal    
terms are equated in $\vsub/\eqstruct$, so the containment is strict:
\begin{center}
\arraycolsep = 2pt\small
$\begin{array}{rll@{\hspace{.2cm}}l@{\hspace{.2cm}}llllllll}
(\la\vartwo((\la\var\varthree)(\varthree\varthree))) (ww) \ \tom^2 \ \varthree\esub\var{\varthree\varthree}\esub\vartwo{ww} \ \tostructcom \ \varthree \esub\vartwo{ww}\esub\var{\varthree\varthree} \ {}\; \mbox{}_{\msym}^2{\leftarrow}\ \ (\la\var((\la\vartwo\varthree)(ww))) (\varthree\varthree) \qedhere
\end{array}$
\end{center}
\end{proof}

\paragraph{Shape of Inert Terms} Extending \VSC with $\eqstruct$ allows a further simplification of the structure of inert terms (which is however not used in the paper, to confirm the irrelevance of $\eqstruct$). Because of  $\tostructapl$, substitutions can be grouped together, obtaining that inert terms have the following shape:
\[ (\var \fire_{1} \ldots \fire_{n}) \esub{\var_{1}}{\itm_{1}}\ldots \esub{\var_{m}}{\itm_{m}} \ \ \ \ \ \ \mbox{with }n,m\geq 0.\]
Additionally, by repeatedly applying $\tostructapr$ and $\tostructes$ one can assume that $\fire_{1}, \ldots, \fire_{n}$ and $\itm_{1}, \ldots, \itm_{m}$ do not contain ES at the open level.

\section{Call-by-Value Solvability and Scrutability}
\label{sect:solving-strat}

In the $\lambda$-calculus, the notion of solvability identifies \emph{meaningful} terms.
This notion is well studied in the \cbn $\lambda$-calculus, with an elegant theory, see Barendregt \cite{Barendregt84}.
In \cbv, as first observed by Ronchi Della Rocca and Paolini \cite{DBLP:journals/ita/PaoliniR99,parametricBook}, there are \emph{two} notions that are semantically relevant, solvability and scrutability (which they call \emph{potential valuability}), and neither can be characterized operationally in Plotkin's calculus. Accattoli and Paolini \cite{AccattoliPaolini12} show that instead the \VSC admits natural operational characterizations of both \cbv scrutability and solvability.

The definitions of solvability and scrutability depend on the calculus and are \emph{interactive} in the sense they are based on the behavior of a term inside a testing context. For solvability, head contexts are used. The intuition is that they are contexts that cannot discard the plugged term without interacting with it. For scrutability, we need head contexts that additionally cannot turn the plugged term into a value without interacting with it, simply called \emph{testing (head) contexts}.

\begin{definition}[Head context, scrutability, solvability]
	Let $\mathsf{X}$ be a calculus containing the $\l$-calculus.
	
	A \emph{head context} in $\mathsf{X}$ is a context defined by the grammar $\hctx \grameq \ctxhole \mid \la{\var}\hctx \mid \hctx\tm$. 
	
	A \emph{\balanced (head) context} in $\mathsf{X}$ is a (head) context defined by the grammar $\bctx \grameq \ctxhole \mid (\la{\var}\bctx) \tm \mid \bctx \tm$. 

	A term $\tm$ in $\mathsf{X}$ is \emph{$\mathsf{X}$-scrutable} (or \emph{$\mathsf{X}$-potentially valuable}) if there is a \balanced  context $\bctx$ and a value $\val$ in $\mathsf{X}$ such that 
	$\bctxp{\tm}$ $\mathsf{X}$-reduces to $\val$, and it is \emph{$\mathsf{X}$-inscrutable} otherwise.
	
	A term $\tm$ in $\mathsf{X}$ is \emph{$\mathsf{X}$-solvable} if
	$\hctxp{\tm}$ $\mathsf{X}$-reduces to the identity $\Id \defeq \la{\var}\var$ for some head context $\hctx$, and it is \emph{$\mathsf{X}$-unsolvable} otherwise.
\end{definition}

%
%
 Accattoli and Paolini \cite{AccattoliPaolini12} give characterizations of \VSC-solvability and \VSC-scrutability akin to Wadsworth's characterization of \cbn\  solvability \cite{Wad:SemPra:71,DBLP:journals/siamcomp/Wadsworth76}.
 
%

\begin{proposition}[Operational characterization of \VSC scrutability/solvability, \cite{AccattoliPaolini12}]
	\label{prop:operational-characterization-cbv-var}
	\hfill
	\begin{enumerate}
		\item\label{p:operational-characterization-cbv-var-scrut} \emph{\VSC-Scrutability via $\tovsubo$:} a term $\tm$ is \VSC-scrutable if and only if $\tovsubo$ terminates on $\tm$.
		
		\item\label{p:operational-characterization-cbv-var-solv} \emph{\VSC-Solvability via $\tosolv$:} a term $\tm$ is \VSC-solvable if and only if $\tosolv$ terminates on $\tm$.
	\end{enumerate}
\end{proposition}

By irrelevance of $\toevaro$ and $\toevarsolv$ (\Cref{{prop:irrelevance}}), the operational characterizations above of \VSC-scrutability and \VSC-solvability can be reformulated in terms of $\tovsubonvar$ and $\tosolvnvar$, respectively.

\begin{corollary}[Operational characterization of \VSC scrutability/solvability, Bis]
	\label{prop:operational-characterization-cbv-novar}
	\NoteProof{propappendix:operational-characterization-cbv-novar}
	\begin{enumerate}
		\item\label{p:operational-characterization-cbv-novar-scrut} \emph{\VSC-Scrutability via  $\tovsubonvar$:} a term $\tm$ is \VSC-scrutable if and only if $\tovsubonvar$ terminates on $\tm$.
		
		\item\label{p:operational-characterization-cbv-novar-solv} \emph{\VSC-Solvability via  $\tosolvnvar$:} a term $\tm$ is \VSC-solvable if and only if $\tosolvnvar$ terminates on $\tm$.
	\end{enumerate}
\end{corollary}


Since solving reduction $\tosolv$ is a strict extension of open reduction $\tovsubo$, an immediate consequence of the characterizations in \Cref{prop:operational-characterization-cbv-var} is that \emph{every \VSC-solvable term is \VSC-scrutable}, but the converse fails (see the term $\tmtwo$ just below). 
That is, 
the set of \VSC-inscrutable terms is strictly contained in the set of \VSC-unsolvable terms.

Open reduction captures the fact that $\tm \defeq \Omega$ (with $\Omega \defeq \delta\delta$ and $\delta \defeq \la{\varthree}\varthree\varthree$) is \VSC-inscrutable, as $\tovsubo$ diverges on $\tm$, while $\tmtwo \defeq \la\var\Omega$ is \VSC-scrutable, indeed $\tovsubo$ terminates on $\tmtwo$ (as it does not reduce under abstractions).
Solving reduction captures the fact that $\tmtwo \defeq \la\var\Omega$ is \VSC-unsolvable, as $\tovsubsolv$ diverges on $\tmtwo$, while $\tmthree \defeq \var (\la\var\Omega)$ is \VSC-solvable, indeed $\tovsubsolv$ terminates on $\tmthree$. 
Note that $(\la{\var}\delta)(\vartwo\vartwo)\delta$ and $\delta ((\la{\var}\delta) (\vartwo\vartwo))$ are \VSC-unsolvable and \VSC-inscrutable, while they are normal---but still unsolvable and inscrutable---in Plotkin's calculus.

\paragraph{Scrutability} Ronchi Della Rocca and Paolini \cite{DBLP:journals/ita/PaoliniR99,parametricBook} define scrutability in a slightly different way\footnotemark
	\footnotetext{In \cite{DBLP:journals/ita/PaoliniR99,parametricBook}, 
	 and---with minor variations---in \citet{AccattoliPaolini12,DBLP:conf/fossacs/CarraroG14}), potential valuability is defined as follows: $\tm$ is $\mathsf{X}$-potentially valuable if there are variables $\var_1, \dots, \var_n$ and values $\val, \val_1, \dots, \val_n$ (with $n \geq 0$) such that the simultaneous substitution $\tm\subs{\var_1}{\val_1}{\var_n}{\val_n}$ $\mathsf{X}$-reduces to $\val$.}, which is however proved to be equivalent to ours in \Cref{sect:solving-strat-proofs} (\Cref{prop:equiv-def-scrutability}). To our knowledge, scrutability has been studied only in \cbv, but it also make sense in \cbn, where can easily be characterized operationally via \emph{weak} (\ie not reducing under abstractions) head reduction.

\paragraph{Equivalence with Scrutability and Solvability  in Plotkin's Calculus} 
Solvability and scrutability depend on the calculus in which they are defined.
Then, what is the relationship between these notions in Plotkin's $\plotcalc$ and in the \VSC? 
We here show that the two variants of each property coincide, adapting an argument from Guerrieri et al. \cite{DBLP:journals/lmcs/GuerrieriPR17}.



 \begin{theorem}[Robustness of \cbv solvability and scrutability]\label{thm:robust}
 	Let $\tm$ be a term without \ES.
 	
 	\begin{enumerate}
 		\item \label{p:robust-scrutable} \emph{\cbv Scrutability:} $\tm$ is \VSC-scrutable if and only if $\tm$ is $\plotcalc$-scrutable.
 		
 		\item \label{p:robust-solvable} \emph{\cbv Solvability:} $\tm$ is \VSC-solvable if and only if $\tm$ is $\plotcalc$-solvable.
		
		\item \label{p:robust-ES} \emph{With/without ES}: for every term $\tm \in \vsubterms$, $\tm$ is \cbv scrutable (resp. solvable) if and only if $\tm^\bullet$ is \cbv scrutable (resp. solvable). 
 	\end{enumerate} 
\end{theorem}

\begin{proof}
	The right-to-left direction of both \Cref{p:robust-scrutable,p:robust-solvable} is obvious, since $\tofbvplot^* \,\subseteq\, \tovsub^*$ (\Cref{prop:plotkin-vsc}).
	 Let us prove the left-to-right directions of \Cref{p:robust-scrutable,p:robust-solvable} separately.
	
	\begin{enumerate}
		\item By definition of \VSC-scrutability, there is a \balanced head context $\bctx$ and a value $\val$ such that $\bctxp{\tm} \tovsub^* \val$. 
		By valuability (\Cref{prop:properties-open-extra}.\ref{p:properties-open-extra-valuability}), $\bctxp{\tm} \tovsubo^{*} \valtwo$ for some value $\valtwo$.
		By \Cref{l:valuing-sequences-lift-to-plotkin},  $\valtwo$ is without \ES and $\bctxp{\tm} \tobvploto^* \valtwo$.
		Thus, $\tm$ is $\plotcalc$-scrutable, since $\tobvploto \subseteq \tofbvplot$.
		
		\item By definition of \VSC-solvability, there is a head context $\hctx$ such that
		$\tmthree \defeq \hctxp{\tm}  \tovsub^{*} \Id$. 
		By valuability (\Cref{prop:properties-open-extra}.\ref{p:properties-open-extra-valuability}), $\tmthree \tovsubo^{*} \val$ for some value $\val$. 
		By confluence, $\val \tovsub^{*}\Id$. Clearly, $\val$ must be an abstraction $\la\var\tmfour$ such that $\tmfour \tovsub^{*}\var$, so that $\val =\la\var\tmfour \tovsub^{*} \la\var\var = \Id$. Again by valuability,  $\tmfour \tovsubo^{*} \valtwo$ for some value $\valtwo$, and $\valtwo \tovsub^{*}\var$ by confluence. 
		Note that $\valtwo$ cannot be an abstraction because it would not reduce to a variable. Then $\valtwo = \var$. 
		Summing up, we have $\tmthree \tovsubo^{*} \la\var\tmfour$ and $\tmfour \tovsubo^{*} \var$. By lifting of valuability (\Cref{l:valuing-sequences-lift-to-plotkin}), we obtain both $\tmthree \tobvploto^{*} \la\var\tmfour$ and $\tmfour\tobvploto^{*} \var$, and putting the two sequences together we obtain $\tmthree \tobvploto^{*} \la\var\tmfour\tofbvplot^{*} \la\var\var$, that is, $\tmthree \tofbvplot^{*} \Id$. Thus, $\tm$ is $\plotcalc$-solvable.
		
		\item See \Cref{sect:solving-strat-proofs} (\Cref{thmappendix:robust}.\ref{pappendix:robust-ES}).
		\qedhere
	\end{enumerate}
\end{proof}


%
%

For both solvability and scrutability, the equivalence holds also with the \VSC extended structural equivalence $\eqstruct$.
This is an easy consequence of the irrelevance of $\equiv$.
These results corroborate the idea that solvability and scrutability in \cbv are \emph{robust} notions that are independent from the particular \cbv calculus used to define them. 
Thus, we can talk about \emph{\cbv solvability} and \emph{\cbv scrutability}, instead of $\mathsf{X}$-solvability and $\mathsf{X}$-scrutability for each \cbv calculus $\mathsf{X}$.
	Pushing things even further, one could take \Cref{thm:robust} as a \emph{criterion} for good \cbv calculi: the notions of $\mathsf{X}$-solvability and $\mathsf{X}$-scrutability must coincide with those in $\plotcalc$.

\paragraph{Differences between \cbv and \cbn}
There is a crucial difference between \cbv and \cbn solvability: a term such as  $ \var\Omega$ is \cbv unsolvable (and $\tovsubsolv$ indeed diverges) while it is \cbn solvable (it is head normal), because plugging in a head context can erase $\Omega$ in \cbn but instead cannot in \cbv (similarly, it is \cbn scrutable but \cbv inscrutable). Every \cbv solvable term is also \cbn solvable, as solving reduction is an extension of head reduction, because it reduces arguments  both out of abstractions and under head abstractions (and similarly for scrutability).

\subsection{Equivalent Definitions} 
\label{ssect:equiv-defs}
As nicely surveyed by Garc\'ia-P\'erez and Nogueira \cite{DBLP:journals/corr/Garcia-PerezN16}, in \cbn there are many equivalent definitions of solvability. 
Here we focus on three of them, given for a generic calculus $\mathsf{X}$.
A term $\tm$ in $\mathsf{X}$ is \emph{solvable} in the sense of \textsc{SOL-FE}, \textsc{SOL-ID}, \textsc{SOL-EX} if respectively: 
\begin{enumerate}
	\item \textsc{SOL-FE}: \emph{for every} full normal form $\tmtwo$ there exists a head context $\hctx_{\tmtwo}$ such that $\hctx_{\tmtwo}\ctxholep\tm \Rew{\mathsf{X}}^* \tmtwo$.
	\item \textsc{SOL-ID}: there exists a head context $\hctx$ such that $\hctxp\tm \Rew{\mathsf{X}}^*\Id$, where $\Id \defeq \la{\var}{\var}$ (the identity).
	\item \textsc{SOL-EX}: \emph{there exists} a full normal form $\tmtwo$ and a head context $\hctx$ such that $\hctxp\tm \Rew{\mathsf{X}}^* \tmtwo$.
\end{enumerate}
The implications \textsc{SOL-FE} $\Rightarrow$ \textsc{SOL-ID} $\Rightarrow$ \textsc{SOL-EX} are obvious in every calculus $\mathsf{X}$. 

In \cbn, the direction \textsc{SOL-EX} $\Rightarrow$ \textsc{SOL-ID} follows easily from the properties of the reduction characterizing solvability (namely, the head normalization theorem, and the stability of head termination by extraction from a head context). Since these properties hold true also for solving reduction (see \Cref{prop:properties-solving}.\ref{p:properties-solving-normalization} and \Cref{prop:properties-solving}.\ref{p:properties-solving-decomposition} above), the same implication holds in the \VSC.

The implication \textsc{SOL-ID} $\Rightarrow$ \textsc{SOL-FE} in \cbn is immediate: one has $\Id \tmtwo \tob \tmtwo$ for every term $\tmtwo$, and so if $\hctx$ is the context for \textsc{SOL-ID} then $\hctx_{\tmtwo} \defeq \hctx \tmtwo$ is the context proving \textsc{SOL-FE}. 
Garc\'ia-P\'erez and Nogueira point out that, in \cbv, $\Id \tmtwo$ does not necessarily reduce to $\tmtwo$, if $\tmtwo$ is not a value \cite{DBLP:journals/corr/Garcia-PerezN16}.  They do not point out, however, that nonetheless the implication \textsc{SOL-ID} $\Rightarrow$ \textsc{SOL-FE} \emph{does} hold in $\plotcalc$ (and thus in the \VSC) via a simple argument, due to Xavier Montillet and given in the next proof.

Therefore, in the \VSC the three definitions of solvability are equivalent, \emph{exactly as in \cbn}.

\begin{theorem}[Equivalent notions of solvability]
	\label{prop:solv-E-I-F}
	In the \VSC, $\textsc{SOL-EX} \Leftrightarrow \textsc{SOL-ID} \Leftrightarrow \textsc{SOL-FE}$.
\end{theorem}

\begin{proof}
	The non-trivial implications to prove are \textsc{SOL-EX} $\Rightarrow$ \textsc{SOL-ID} and \textsc{SOL-ID} $\Rightarrow$ \textsc{SOL-FE}.
	
	For \textsc{SOL-EX} $\Rightarrow$ \textsc{SOL-ID}, suppose that $\tm$ fulfills \textsc{SOL-EX} in \VSC, that is, there is a full normal form $\tmtwo$ and a head context $\hctx$ such that $\hctxp\tm \tovsub^* \tmtwo$.
	By normalization (\Cref{prop:properties-solving}.\ref{p:properties-solving-normalization}), $\hctxp{\tm} \tosolv^* \tmthree$ for some $\solvsym$-normal $\tmthree$.
	By stability by extraction from a head context (\Cref{prop:properties-solving}.\ref{p:properties-solving-decomposition}), $\tm \tosolv^* \tmfour$ for some $\solvsym$-normal $\tmfour$.
	Then, according to the operational characterization of \textsc{SOL-ID} (\Cref{prop:operational-characterization-cbv-var}.\ref{p:operational-characterization-cbv-var-solv}), $\tm$ verifies \textsc{SOL-ID}.
	
	For \textsc{SOL-ID} $\Rightarrow$ \textsc{SOL-FE}, suppose that $\tm$ is solvable in the sense of \textsc{SOL-ID}, that is, there is a head context $\hctx$ such that $\hctxp{\tm} \tovsub^* \Id$.
	Let $\tmtwo$ be a full normal form with $\var \notin \fv{\tmtwo}$ and let $\hctx_{\tmtwo} \defeq (\hctx \, \la{\var}\tmtwo) \Id$. Then, $\hctx_{\tmtwo}\ctxholep{\tm}  \tovsub^* (\Id \la{\var}\tmtwo) \Id \tovsub^+ (\la{\var}\tmtwo) \Id \tovsub^+ \tmtwo$. 
	As $\hctx_{\tmtwo}$ is a head context, $\tm$ verifies  \textsc{SOL-FE}.
\end{proof}

\paragraph{One More New Definition} The relevance of inert terms can be stressed by showing that they can be used to provide yet another characterization of \cbv solvability, which shows that solvability can be captured at the open level.
\begin{proposition}[Yet another definition of \cbv solvability]
\label{prop:op-char-cbv-solv-alt}
	 A term $\tm$ is \VSC-solvable if
	 \begin{itemize}
	 \item \emph{SOL-IN}: there is a head context $\hctx$ and an inert term $\itm$ such that $\hctxp\tm \tovsub^{*} \itm$.
	 \end{itemize}
\end{proposition}

\begin{proof}
Direction SOL-ID $\Rightarrow$ SOL-IN is straightforward: if $\hctx$ is the context such that $\hctxp\tm \tovsc^{*} \Id$ then $\hctxtwo\defeq \hctx \var$ is such that $\hctxtwop\tm \tovsc^{*} \Id \var \tovsc^{*} \var$, which is inert. For  SOL-IN $\Rightarrow$ SOL-ID, let $\hctx$ be the head context such that $\hctxp\tm\tovsub^{*}\itm$. 
Since inert terms are $\solvnvarsym$-normal (\Cref{prop:properties-solvable-reduction}.\ref{p:properties-solvable-reduction-harmony}), by the derived operational characterization of SOL-ID (\Cref{prop:operational-characterization-cbv-novar}.\ref{p:operational-characterization-cbv-novar-solv}) there is a context $\hctxtwo$ such that $\hctxtwop\itm\tovsub^{*}\Id$. Then the head context $\hctxtwop\hctx$ is such that  $\hctxtwop{\hctxp\tm}\tovsub^{*}\hctxtwop\itm\tovsub^{*}\Id$.
\end{proof}

Note that, of the many definitions of \cbv solvability that we discussed, SOL-IN is the only one using as target \emph{open} normal forms (inert terms), and \emph{not} fully normal terms. 
Thus, solvability can be captured at the open level, without requiring full reduction. Additionally, SOL-IN can be equivalently defined using $\hctxp\tm \tovsubonvar^{*} \itm$ instead of $\hctxp\tm \tovsub^{*} \itm$ (the proof is in \Cref{sect:solving-strat-proofs}, \Cref{cor:op-char-cbv-solv-open} ).

Last, SOL-IN can  be  adapted to \cbn, by replacing inert terms with terms of the form $\var \tm_{1} \ldots \tm_{n}$ with $n\geq 0$ (and no hypotheses on $\tm_{1}, \ldots, \tm_{n}$), sometimes called \emph{neutral terms} (the literature is inconsistent with the terminology, at times the definition of neutral terms requires $\tm_{1}, \ldots, \tm_{n}$ to be normal). This fact is both positive and negative: it is good that the open characterization can be adapted, but it shows that the open characterization depends on the calculus (inert/neutral terms in \cbv/\cbn), while SOL-ID is calculus-independent.

\section{(Non-)Collapsibility}
\label{sect:collapsibility}

In \cbn, unsolvable terms are \emph{collapsible}, that is, the equational theory $\mathcal{H}$, extending $\beta$-conversion by equating all unsolvable terms, is \emph{consistent}, \ie it does not equate all terms. Here we show that \cbv inscrutable terms are collapsible, while \cbv unsolvable terms are not. This section mostly adapts results from \citet{DBLP:journals/fuin/EgidiHR92}, presenting them in a different way.

\paragraph{\cbv Inscrutable Terms Are Collapsible} The collapsibility of  inscrutable terms is obtained by exhibiting a consistent theory that equates them, namely \cbv contextual equivalence (\Cref{def:ctx-eq}).

Showing that $\ctxeq^{\plotcalc}$ (resp. $\ctxeq^{\VSC}$) is a $\plotcalc$-theory---see \Cref{def:eq-theory}---(resp. VSC-theory) is immediate, in particular context closure follows immediately from the clause defining it.

\begin{proposition}[Consistency of \cbv contextual equivalence]
\label{prop:consistency-ctx-eq} 
\cbv contextual equivalence is consistent in both Plotkin's calculus $\plotcalc$ and the VSC.
\end{proposition}

\begin{proof}
Simply note that $\Omega \not\ctxeq \la\var\Omega$ in $\plotcalc$ and \VSC, since the two terms are closed and the empty context distinguishes them: $\la\var\Omega$ reduces to a value (itself) in 0 steps, while $\Omega$ diverges.
\end{proof}

In analogy to the \cbn \emph{sensible theories}, which are  $\l$-theories collapsing all \cbn unsolvable terms, and \emph{semi-sensible} ones, which do not equate solvable and unsolvable terms, we introduce the corresponding scrutable notions.
 
\begin{definition}[Scrutable theories]
A $\plotcalc$-theory (resp. VSC-theory) is \emph{scrutable} if it equates all \cbv inscrutable terms without \ES (\resp terms in $\vsubterms$) and \emph{semi-scrutable} if it does not equate \cbv scrutable and inscrutable terms without \ES (\resp terms in $\vsubterms$). 
\end{definition}

The fact that contextual equivalence in $\plotcalc$ is a scrutable theory easily follows from a result in the literature, the full abstraction of \cbv applicative bisimilarity (\citet{DBLP:journals/fuin/EgidiHR92,DBLP:books/cu/12/Pitts12}).

\begin{proposition}[$\plotcalc$ contextual equivalence is scrutable]
	\label{prop:contextual-scrutable-plotkin}
	\NoteProof{propappendix:contextual-scrutable-plotkin}
$\ctxeq^{\plotcalc}$ is a scrutable $\plotcalc$-theory.
\end{proposition}

The fact that contextual equivalence in VSC is a scrutable theory is proved via the scrutability of the theory $\ctxeq^{\plotcalc}$ and the robustness of \cbv solvability with or without \ES (\Cref{thm:robust}.\ref{p:robust-ES}).

\begin{corollary}[VSC contextual equivalence is scrutable]
	\label{cor:contextual-scrutable-vsc}
	\NoteProof{corappendix:contextual-scrutable-vsc}
$\ctxeq^{\VSC}$ is a scrutable \VSC-theory.
\end{corollary}

\paragraph{\cbv Unsolvable Terms Are Not Collapsible} Perhaps surprisingly, in \cbv unsolvable terms are not collapsible. This crucial fact is referred to by Garc\'ia-P\'erez and Nogueira \cite{DBLP:journals/corr/Garcia-PerezN16} by pointing to \citet{DBLP:journals/ita/PaoliniR99}, where however it is not stated. To our knowledge, it is never formally stated anywhere in the literature, which is why we present it here. 
The argument in the next theorem adapts the idea in the proof by \citet{DBLP:journals/fuin/EgidiHR92} that $\plotcalc$ contextual equivalence $\ctxeq^{\plotcalc}$ is a maximal consistent $\plotcalc$-theory (Proposition 35, therein).
\begin{theorem}[Non-collapsibility of unsolvable terms]
	\label{prop:inconsistency} 
\hfill
\begin{enumerate}
\item\label{p:inconsistency-abstract} Any scrutable $\plotcalc$-theory (or \VSC-theory) $\eqth$ that is not semi-scrutable is inconsistent.
\item \label{p:inconsistency-unsolvable} The set of \cbv unsolvable terms is not collapsible.
\end{enumerate}
\end{theorem}

\begin{proof}
\begin{enumerate}
\item Since $\eqth$ is not semi-scrutable, there are $\tm$ (\cbv) scrutable and $\tmtwo$ (\cbv) inscrutable such that $\tm =_{\eqth} \tmtwo$. 
Since $\tm$ is scrutable, there is a \balanced context $\bctx$ sending it to a value $\val$. 
Since $\tmtwo$ i inscrutable, $\bctxp\tmtwo$ is also inscrutable (as the composition $\bctxtwop{\bctx}$ of two \balanced contexts $\bctx, \bctxtwo$ is a \balanced context). 
By the definition of  $\plotcalc$-theory, we have $\bctxp\tmtwo =_{\eqth} \bctxp\tm =_{\eqth} \val$. 
Now, let $\tmthree$ be a term and $\vartwo\notin\fv\tmthree$. Then $\tmthree =_{\eqth} (\la\vartwo\tmthree) \val$ because $=_{\betaplot}\subseteq \eqth$ by definition of $\plotcalc$-theory. By the context closure of theories and $\bctxp\tmtwo =_{\eqth} \val$, we obtain $(\la\vartwo\tmthree) \val =_{\eqth} (\la\vartwo\tmthree) \bctxp\tmtwo$. Since $\eqth$ is scrutable and both $\bctxp\tmtwo$ and $(\la\vartwo\tmthree) \bctxp\tmtwo$ are inscrutable, $(\la\vartwo\tmthree) \bctxp\tmtwo =_{\eqth} \bctxp\tmtwo$. Therefore, $\tmthree =_{\eqth} \bctxp\tmtwo$ for every term $\tmthree$, that is, $\eqth$ is inconsistent.

\item Any $\plotcalc$-theory $\eqth$ 
equating all \cbv unsolvable terms is scrutable (because inscrutable terms are unsolvable) and not semi-scrutable, because \eg $\Omega =_{\eqth} \la\var\Omega$, where $\Omega$ is inscrutable, $\la\var\Omega$ is scrutable, and both are unsolvable. 
By Point \ref{p:inconsistency-abstract}, $\eqth$ is inconsistent.\qedhere
\end{enumerate}
\end{proof}

%
%

\subsection{Axioms for Collapsibility} 
Kennaway et al. \cite{DBLP:journals/jflp/KennawayOV99} provide three axioms in order for a set of terms $U$ of the $\l$-calculus to be collapsible and also satisfy a genericity lemma\footnote{Their notion of genericity however is not equivalent to the one in Barendregt's book \cite{Barendregt84}, because in \cite{DBLP:journals/jflp/KennawayOV99} plugging in a context---which is part of the statement of genericity---is a capture-avoiding operation, while for Barendregt it is not.}. Because of the unusual rewriting rules at a distance of the VSC, it is unclear to us whether it fits into the class of rewriting systems covered by the axiomatics, which is not clearly specified in \cite{DBLP:journals/jflp/KennawayOV99}. It is nonetheless instructive to see how the axioms are instantiated in our setting by taking $U$ as the set of either inscrutable or unsolvable terms.

\paragraph{Axiom 1} The first axiom asks the stability of the terms in $U$ by substitution, that is,  \emph{if $\tm\in U$ then $\tm\isub\var\tmtwo \in U$ for every $\tmtwo$}. In our setting, both inscrutable and unsolvable terms verify this axiom. The axiom is proved in its contrapositive form via the operational characterizations.

\begin{proposition}[Stability of \cbv scrutability/solvability under removal]
\label{prop:first-axiom}
If there exist $\tmtwo$ such that $\tm\isub\var\tmtwo$ is \cbv scrutable (resp. solvable) then $\tm$ is \cbv scrutable (resp. solvable). 
\end{proposition}

\begin{proof}
By contradiction, suppose that $\tm$ is \cbv unsolvable. 
According to the operational characterization of \cbv solvability (\Cref{prop:operational-characterization-cbv-novar}.\ref{p:operational-characterization-cbv-novar-solv}), $\tosolvnvar$ diverges on $\tm$. 
By stability of $\tosolvnvar$ by substitution (\Cref{l:stability-substitution}), $\tosolvnvar$ diverges on $\tm \isub{\var}{\tmtwo}$ for every term $\tmtwo$, so $\tm \isub{\var}{\tmtwo}$ is \cbv unsolvable by \Cref{prop:operational-characterization-cbv-novar}.\ref{p:operational-characterization-cbv-novar-solv}.

The proof concerning \cbv scrutability is analogous, just replace the properties for $\tosolvnvar$ with their analogue for $\tovsubonvar$, in particular \Cref{prop:operational-characterization-cbv-novar}.\ref{p:operational-characterization-cbv-novar-solv} with \Cref{prop:operational-characterization-cbv-novar}.\ref{p:operational-characterization-cbv-novar-scrut}
\end{proof}

The proof of \Cref{prop:first-axiom} relies on the stability under substitution for $\tovsubonvar$ and $\tosolvnvar$. Note that, as we have seen in \Cref{ssect:irrelevance}, such a property  \emph{fails} instead for $\toevar$ steps. Therefore, \Cref{prop:first-axiom} is a point where the irrelevance of $\toevar$ plays a crucial role.

Note also that, as pointed out in the introduction, (\cbn) diverging terms are not collapsible because they are not stable by substitution, that is, they violate axiom 1. While \cbv unsolvable terms are also not collapsible, they do satisfy axiom 1.

\paragraph{Axiom 2} The second axiom is the stability of terms in $U$ by reduction, which in our setting is an easy consequence of the normalization theorem for open/solving reduction.

\begin{lemma}[Stability of unsolvable terms by reduction]
Let $\tm$ be \cbv inscrutable (resp. unsolvable) and $\tm \tovsub \tmtwo$. 
Then $\tmtwo$ is \cbv inscrutable (resp. unsolvable).
\end{lemma}

\begin{proof}
By contradiction. If $\tmtwo$ is \cbv solvable then by the operational characterization of \cbv solvability (\Cref{prop:operational-characterization-cbv-var}.\ref{p:operational-characterization-cbv-var-solv}), $\tmtwo\tosolv^{*} \tmthree$ for some $\solvsym$-normal $\tmthree$. 
	Then $\tm \tovsub^{*} \tmthree$. 
	By the normalization property for $\tosolv$ (\Cref{prop:properties-solving}.\ref{p:properties-solving-normalization}), $\tosolv$ terminates on $\tm$, which then is \cbv solvable (\Cref{prop:operational-characterization-cbv-var}.\ref{p:operational-characterization-cbv-var-solv} again)---absurd.

The proof concerning \cbv (in)scrutability is analogous, just  replace the properties for $\tosolv$ with their analogue for $\tovsubo$: \Cref{prop:operational-characterization-cbv-var}.\ref{p:operational-characterization-cbv-var-solv} and \Cref{prop:properties-solving}.\ref{p:properties-solving-normalization}  with \Cref{prop:operational-characterization-cbv-var}.\ref{p:operational-characterization-cbv-var-scrut} and \Cref{prop:properties-open-extra}.\ref{p:properties-open-extra-normalization}, respectively.
\end{proof}

\paragraph{Axiom 3} The third axiom is more technical and about overlappings of redex patterns with terms in $U$. Roughly, in our case it amounts to prove that in the two root rules:
\begin{center}
$\begin{array}{r@{\ }l@{\ }l@{\hspace{2cm}}r@{\ }l@{\ }l}
\subctxp{\la\var\tm}\tmtwo &  \rtom  & \subctxp{\tm\esub{\var}{\tmtwo}} 
    
    & \tm\esub\var{\subctxp{\val}} &  \rtoe  & \subctxp{\tm\isub{\var}{\val}} 
\end{array}$ 
\end{center}
if $\subctxp{\la\var\tm} \in U$ then $\subctxp{\la\var\tm}\tmtwo \in U$, and if $\subctxp{\val} \in U$ then $ \tm\esub\var{\subctxp{\val}} \in U$. Interestingly, both conditions hold when taking as $U$ the set of  inscrutable terms, while the second one \emph{fails} for unsolvable terms. A counter-example is obtained by taking the unsolvable term $\la\vartwo\Omega$ and noting that $\Id\esub\var{\la\vartwo\Omega} \rtoe \Id$ is instead solvable. This fact recasts in Kennaway et al.'s axiomatics the non-collapsibility of \cbv unsolvable terms. 
\section{Multi Types by Value}
\label{sect:types}
This section starts the second part of the paper, where the \VSC is studied via a  \emph{multi type system}. We first recall the background about multi types and provide an overview of our results.

\subsection{From Multi Types to Call-by-Value Solvability}
Intersection types are a standard and flexible tool to study $\l$-calculi, mainly used to characterize termination properties, see Coppo and Dezani \cite{DBLP:journals/aml/CoppoD78,DBLP:journals/ndjfl/CoppoD80}, Pottinger \cite{Pottinger80}, and Krivine \cite{Kri}, as well as to study $\l$-models (\citet{DBLP:journals/jsyml/BarendregtCD83,DBLP:journals/iandc/CoppoDM87,DBLP:journals/fuin/EgidiHR92,DBLP:journals/tcs/Plotkin93,DBLP:journals/jcss/HonsellR92,DBLP:journals/apal/Abramsky91}).  
Among several variants of intersection types, the \emph{non-idempotent} ones, where the
intersection $A \cap A$ is not equivalent to $A$, were introduced by Gardner \cite{DBLP:conf/tacs/Gardner94}. Then Kfoury \cite{DBLP:journals/logcom/Kfoury00}, Neergaard and Mairson \cite{DBLP:conf/icfp/NeergaardM04}, and de Carvalho \cite{Carvalho07,deCarvalho18} provided a first wave of works about them. A survey can be
found in Bucciarelli et al.~\cite{BKV17}. Non-idempotent intersections can be seen as \emph{multisets}, which is why, to ease the
language, we prefer to call them \emph{multi types} rather than
\emph{non-idempotent intersection types}.
Multi types refine intersection types with multiplicities, giving rise to a \emph{quantitative} approach that reflects 
resource consumption, and that it turns out to coincide exactly with the one at work in linear logic.  Neergaard and Mairson prove that type inference for multi types is equivalent to normalization. Therefore, multi types hide a computational mechanism.

\paragraph{De Carvalho's Bounds from Multi Types} An insightful use of multi types and of their computational mechanism is de Carvalho's extraction of bounds for the
\cbn $\l$-calculus \cite{Carvalho07,deCarvalho18}: from certain  type derivations, he extracts \emph{exact} bounds about the length of reduction sequences and the size of the normal form of a term, according to various notions of reduction. In particular, for \emph{head} reduction, which in \cbn is the reduction characterizing solvability.
De Carvalho's seminal work has been extended to many notions of reduction and formalisms.
  A first wave was inspired directly from his original work \cite{DBLP:journals/tcs/CarvalhoPF11,DBLP:journals/iandc/CarvalhoF16,DBLP:journals/corr/BernadetL13,Guerrieri18,DBLP:journals/fuin/ManzonettoPR19}, and a second wave \cite{DBLP:conf/aplas/AccattoliG18,DBLP:conf/esop/AccattoliGL19,DBLP:conf/fossacs/KesnerPV21,DBLP:conf/flops/BucciarelliKRV20,DBLP:conf/csl/KesnerV22,DBLP:conf/lics/KesnerV20,DBLP:conf/types/AlvesKV19,DBLP:journals/pacmpl/LagoFR21,DBLP:journals/pacmpl/AccattoliLV21,DBLP:conf/lics/AccattoliLV21} started after the revisitation of de Carvalho's technique by Accattoli et al. \cite{DBLP:journals/pacmpl/AccattoliGK18}.  

\paragraph{\ccbv and Multi Types} Ehrhard \cite{DBLP:conf/csl/Ehrhard12} introduces a \cbv system of multi types to study Plotkin's $\plotcalc$ with closed terms. His system is the \cbv 
version of Gardner-de Carvalho system for \cbn \cite{DBLP:conf/tacs/Gardner94,Carvalho07,deCarvalho18}. Both systems can be seen as the restrictions of the relational semantics of linear logic \cite{Girard88,DBLP:journals/apal/BucciarelliE01} to the \cbn/\cbv translations of the $\l$-calculus. 

\paragraph{\ocbv and Multi Types} Accattoli and Guerrieri \cite{DBLP:conf/aplas/AccattoliG18} use Ehrhard's system to study \emph{\ocbv} (that is, weak call-by-value with possibly open terms). They show that the open reduction of a term $\tm$ terminates in \ocbv  if and only if $\tm$ is typable with \cbv multi types. Moreover, they show how to extract exact bounds from type derivations, adapting de Carvalho's technique. 
Since termination of open reduction characterizes \cbv scrutability (\Cref{prop:operational-characterization-cbv-var}.\ref{p:operational-characterization-cbv-var-scrut}), their results provide a quantitative characterization of \cbv scrutability via multi types.

\paragraph{\cbv Solvability and Multi Types, Qualitatively} Here, we build over their work, using Ehrhard's \cbv multi types to study Accattoli and Paolini's solving reduction. 
Since solving reduction \emph{extends} open reduction, the terms that are solving terminating---that is, solvable terms---form a subset of the open terminating ones and so cannot be characterized simply as the typable ones. We characterize them as those typable with certain \emph{solvable types}, inspired by Paolini and Ronchi Della Rocca \cite{DBLP:journals/ita/PaoliniR99} and at the same time fixing some technical issues of similar characterizations in \cite{DBLP:journals/ita/PaoliniR99,DBLP:conf/fscd/KerinecMR21} (see Appendix \ref{sect:counter} for details).


\paragraph{\cbv Solvability and Multi Types, Quantitatively} A further contribution is that, for the first time in the literature, we provide a \emph{quantitative} characterization of \cbv solvability, adapting once more de Carvalho's technique. First, we show that every solvable derivation provides bounds to the length of solving reduction sequences \emph{and} the size of the solving normal form. Second, we characterize solvable derivations that provide \emph{exact} bounds. This last part requires introducing \emph{two} refinements of solvable types, detailed in \Cref{sect:solvable}.


%
 
\subsection{Introducing Multi Types by Value}
\paragraph{Multi Types} There are two mutually defined layers of types, \emph{linear} and \emph{multi types}, their grammars are in \Cref{fig:cbvtypes}. 
We use $\ground$ for a fixed unspecified ground type, and $\mset{\ltype_1, \mydots, \ltype_n}$ is our notation for finite 
multisets. 
The \emph{empty} multi type $\mset{\,}$ (obtained taking $n = 0$) is also denoted by $\emptymset$.  A multi type is \emph{ground} if it is of the form $n\mset{\ground} \defeq \mset{ \ground, \mydots, \ground}$  ($n$  times $\ground$) for some $n \geq 0$ (so, $0\mset{\ground} = \emptytype$).
A generic (multi or linear) type is noted $\type$.
A multi type $\mset{\ltype_1, \mydots, \ltype_n}$ has to be intended as a conjunction $\ltype_1 \cap \mydots \cap 
\ltype_n$ of linear types $\ltype_1, \mydots, \ltype_n$, for a commutative, associative, non-idempotent conjunction 
$\cap$ (morally a tensor $\otimes$), whose neutral element~is~$\emptymset$.

Intuitively, a linear type corresponds to a single use of a term $\tm$, and $\tm$ is typed with a 
multiset 
$\mtype$ of $n$ linear types if it is going to be used (at most) $n$ times. The  meaning of \emph{using a term} is not 
easy to define precisely. Roughly, it means that if $\tm$ is part of a larger term $\tmtwo$, then (at most) $n$ copies 
of $\tm$ shall end up in evaluation positions---where they are applied to some terms---while evaluating $\tmtwo$.

The derivation rules for the multi types system are in \Cref{fig:cbvtypes} (explanation follows).  The rules are the same as in Ehrhard \cite{DBLP:conf/csl/Ehrhard12}, up to the fact that they are extended to \ES. 

A \emph{multi} (\resp \emph{linear}) \emph{judgment} has the shape $\typctx \vdash \tm \hastype \type$ where $\tm$ is a term, $\type$ is a 
multi (\resp linear) type and $\typctx$ is a \emph{type context}, that is, a total function from variables to multi types 
such that  the set $\domain{\typctx} \defeq \{\var \mid \typctx(\var) \neq \emptymset\}$ is finite.  

\begin{figure*}[t!]
\scalebox{0.9}{
\begin{tabular}{ccc}
\arraycolsep = 3pt
$\begin{array}{cccc@{\hspace{1cm}}cccccc}
	\textsc{Linear types}& \ltype, \ltypetwo &\grameq &\ground \mid \larrow{\mtype}{\mtypetwo}
	&
	\textsc{Multi types} & \mtype, \mtypetwo &\grameq & \mset{\ltype_1, \mydots, \ltype_n} \quad n \geq 0
\end{array}$
\\[5pt]
\begin{tabular}{c@{\hspace{1cm}}c@{\hspace{1cm}}c@{\hspace{1cm}}c}
	$\begin{prooftree}[label separation = .1em]
					\infer0
					[\scriptsize$\ruleAx$]
					{\var \hastype [\ltype] \vdash \var \hastype \ltype}
	\end{prooftree}$
	&
	$\begin{prooftree}[label separation = .1em]
				\hypo{\tyjp{}{\tm}{\typctx, \var \hastype \mtype}{\mtypetwo}}
				\infer1[\scriptsize$\ruleFun$]
				{\tyjp{}{\la{\var}{\tm}}{\typctx}{\ty{\mtype}{\mtypetwo}}}
		\end{prooftree}$
&				
		$\begin{prooftree}[separation=1em, label separation = .1em]
				\hypo{\left[\tyjp{}{\val}{\typctx_{\!i}}{\ltype_{i}}\right]_{\iI}}
				\hypo{I \text{\small\ finite}}
				\infer2
				[\scriptsize$\ruleMany$]
				{\tyjp{}{\val}{\biguplus_{\iI}\typctx_{\!i} }{ 
						\biguplus_{\iI} \mult{\ltype_{i}}
				}}
		\end{prooftree}$
	\end{tabular}
\\[15pt]
\begin{tabular}{c@{\hspace{1cm}}c}
			$\begin{prooftree}[separation=1em, label separation = .1em]
					\hypo{\typctx \vdash \tm \hastype \mset{ \larrow{\mtype\!}{\!\mtypetwo} }}
					\hypo{\typctxtwo \vdash \tmtwo \hastype \mtype}
					\infer2[\scriptsize$\ruleAp$]
					{\typctx \uplus \typctxtwo \vdash \tm\tmtwo \hastype \mtypetwo}
			\end{prooftree}$
			&
			$\begin{prooftree}[separation=1em, label separation = .1em]
					\hypo{\typctx, \var \hastype \mtype \vdash \tm \hastype \mtypetwo}
					\hypo{\typctxtwo \vdash \tmtwo \hastype \mtype}
					\infer2
					[\scriptsize$\ruleES$]
					{\typctx \uplus \typctxtwo \vdash \tm \esub\var\tmtwo \hastype \mtypetwo}
			\end{prooftree}$
		\end{tabular}
\end{tabular}
}
\vspace{-2pt}
	\caption{Call-by-Value Multi Type System for \VSC.}
	\label{fig:cbvtypes}
\end{figure*}	
\paragraph{Explanations about the Inference Rules}
All rules but $\ruleAx$ and $\ruleFun$ assign a multi type to the term on the
right-hand side of a judgment.
Values are the only terms that can be typed by a linear type, via $\ruleAx$ and $\ruleFun$.
Rule $\ruleMany$ can be applied only to values, turning linear types into multi types: it has as many premises as the elements in the (possibly
empty) set of indices $I$ (when $I = \emptyset$, the rule has no premises, and it gives an empty multi type $\emptymset$). Note that \emph{every} value can then be typed with $\zero$.
The $\ruleMany$ rule says how many ``copies'' of one occurrence of a value in a term $\tm$ are needed to evaluate~$\tm$.
It corresponds to the promotion rule of linear logic,
which, in the \cbv representation of the $\lambda$-calculus, is indeed used for typing values.

\paragraph{Example of Type Derivation} 
Let $\mtype \defeq \mset{\ty{\mset{\ground}}{\mset{\ground}}}$. Consider the following derivation:
\begin{center}
	\footnotesize
\begin{prooftree}[separation=1.0em,label separation=0.2em]
	\infer0[\scriptsize$\ruleAx$]{\var \hastype \mset{\ty{\mtype}{\mtype}} \vdash \var \hastype \ty{\mtype}{\mtype} }
	\infer1[\scriptsize$\ruleMany$]{\var \hastype \mset{\ty{\mtype}{\mtype}} \vdash \var \hastype \mset{\ty{\mtype}{\mtype}} }
	\infer0[\scriptsize$\ruleAx$]{\var \hastype \mtype \vdash \var \hastype \ty{\mset{\ground}}{\mset{\ground}}}
	\infer1[\scriptsize$\ruleMany$]{\var \hastype \mtype \vdash \var \hastype \mtype}
	\infer2[\scriptsize$\ruleAp$]{\var \hastype \mset{\ty{\mtype}{\mtype}} \mplus \mtype \vdash \var\var \hastype \mtype}
	\infer1[\scriptsize$\ruleFun$]{\vdash \la{\var}\var\var \hastype \ty{(\mset{\ty{\mtype}{\mtype}}\mplus \mtype)}{\mtype}}
	\infer1[\scriptsize$\ruleMany$]{\vdash \la{\var}\var\var \hastype \mset{\ty{(\mset{\ty{\mtype}{\mtype}}\mplus \mtype)}{\mtype}}}
	\infer0[\scriptsize$\ruleAx$]{\vartwo \hastype \mtype \vdash \vartwo \hastype \ty{\mset{\ground}}{\mset{\ground}}}
	\infer1[\scriptsize$\ruleMany$]{\vartwo \hastype \mtype \vdash \vartwo \hastype \mtype}
	\infer1[\scriptsize$\ruleFun$]{\vdash \Id \hastype \ty{\mtype}{\mtype}}
	\infer0[\scriptsize$\ruleAx$]{\vartwo \hastype \mset{\ground} \vdash \vartwo \hastype \ground}
	\infer1[\scriptsize$\ruleMany$]{\vartwo \hastype \mset{\ground} \vdash \vartwo \hastype \mset{\ground}}
	\infer1[\scriptsize$\ruleFun$]{\vdash \Id \hastype \ty{\mset{\ground}}{\mset{\ground}}}
	\infer2[\scriptsize$\ruleMany$]{\vdash \Id \hastype \mset{\ty{\mtype}{\mtype}}\mplus \mtype}
	\infer2[\scriptsize$\ruleAp$]{\vdash (\la{\var}\var\var)\Id \hastype {\mtype}}
\end{prooftree}
\end{center}
Note that the argument identity $\Id \defeq \la{\vartwo}\vartwo$ is typed \emph{twice}, and with different types. It is typed once with $\mset{\ty{\mtype}{\mtype}}$, when it is used as a function, and once with $\mtype$, when it is used as a value. Thus, multi types account for a form of finite polymorphism. Moreover, the finite polymorphism of multi types allows us to type the term $\la\var\var\var$, which is not typable with simple types.

\paragraph{Technicalities about Types} The type context $\typctx$ is \emph{empty} if $\dom{\typctx} = \emptyset$.  
\emph{Multi-set sum} $\mplus$ is extended to type contexts point-wise,
\ie\  $(\typctx \mplus \typctxtwo)(\var) \defeq \typctx(\var) \mplus \typctxtwo(\var)$ for each variable $\var$.
This notion is extended to a finite family of type contexts as expected, 
in particular $\bigmplus_{i \in J\!} \typctx_i$ is the empty context  when $J = \emptyset$.
A type context $\typctx$ is denoted by $\var_1 \hastype \mtype_1, \mydots, \var_n \hastype \mtype_n$ (for some $n \in 
\nat$) if $\dom{\typctx} \subseteq \{\var_1, \mydots, \var_n\}$ and $\typctx(\var_i) = \mtype_{i}$ for all $1 \leq i \leq 
n$.
Given two type contexts $\typctx$ and $\typctxtwo$ such that $\dom{\typctx} \cap \dom{\typctxtwo} = \emptyset$, the 
type 
context $\typctx, \typctxtwo$ is defined by $(\typctx, \typctxtwo)(\var) \defeq \typctx(\var)$ if $\var \in 
\dom{\typctx}$, $(\typctx, \typctxtwo)(\var) \defeq \typctxtwo(\var)$ if $\var \in \dom{\typctxtwo}$, and $(\typctx, 
\typctxtwo)(\var) \defeq \emptymset$ otherwise.
Note that $\typctx, \var \hastype \emptymset = \typctx$, where we implicitly assume $\var \notin \dom{\typctx}$. 

We write $\concl{\tderiv}{\typctx}{\tm}{\mtype}$ if $\tderiv$ is a (\emph{type}) \emph{derivation} (\ie a tree built up from the rules in \Cref{fig:cbvtypes}) with conclusion the multi judgment $\typctx \vdash \tm \hastype \mtype$.
In particular,  we write $\concl{\tderiv}{\,}{\tm}{\mtype}$ when $\typctx$ is empty.
We write $\derive{\tderiv}{\tm}$ if $\concl{\tderiv}{\typctx}{\tm}{\mtype}$ for some type context $\typctx$ and multi type $\mtype$.  

\paragraph{The Sizes of Type Derivations} Our study being quantitative, we need a notion of size of type derivations. In fact, we shall use \emph{two} notions of size.

\begin{definition}[Derivation size(s)]
\label{def:two-sizes}
	Let $\tderiv$ be a derivation. 
	The \emph{(general) size} $\size{\tderiv}$ of $\tderiv$ is the number of 
	rule occurrences in $\tderiv$ except for the rule $\ruleMany$.
	The \emph{multiplicative size} $\sizem{\tderiv}$ of $\tderiv$ is the number of occurrences of the rules $\lambda$ and $\ruleAp$ in $\tderiv$.
\end{definition}
The two sizes for derivations play different roles. Qualitatively,  to prove that typability implies termination of solving reduction, we need a measure that 
decreases for all solving steps; 
this role is played by the general size $\size{\!\cdot\!}$.
	Quantitatively, we want to measure the number of $\tom$ steps in 
	solving reduction sequences, because it is the time cost model of \VSC, see Accattoli et al. \cite{DBLP:conf/lics/AccattoliCC21};
	this role is played by the multiplicative size $\sizem{\!\cdot\!}$.

\paragraph{Substitution and Removal Lemmas}
The two next lemmas establish a key feature of this type system: in a typed term $\tm$,
substituting a value for
a variable as in the exponential step, or, dually, removing a value, preserves the type of $\tm$ and \emph{consumes} (dually, \emph{adds}) the multi type of the variable.
The statements also provide  quantitative information about the type derivation for $\tm$ before and after the substitution/removal.

\begin{lemma}[Substitution]
	\label{l:substitution}	
	\NoteProof{lappendix:substitution}
	Let $\tm$ be a term, $\val$ be a value and $\namedtyjp{\tderiv}{}{\tm}{\typctx, \var \hastype 
		\mtypetwo}{\mtype}$ and $\namedtyjp{\tderivtwo}{}{\val}{\typctxtwo}{\mtypetwo}$ be derivations.
	Then there is a derivation $\namedtyjp{\tderivthree}{}{\tm \isub{\var}{\val}}{\typctx \mplus \typctxtwo}{\mtype}$ 
	with $\sizem{\tderivthree} = \sizem{\tderiv} + \sizem{\tderivtwo}$ and $\size{\tderivthree} \leq \size{\tderiv} + 
	\size{\tderivtwo}$. 
\end{lemma}

\begin{lemma}[Removal]
	\label{l:anti-substitution}
	\NoteProof{lappendix:anti-substitution}
	Let $\tm$ be a term, $\val$ be a value, and 
	$\namedtyjp{\tderiv}{}{\tm\isub{\var}{\val}}{\typctx}{\mtype}$
	be a derivation. 
	Then there are two  derivations $\namedtyjp{\tderivtwo}{}{\tm}{\typctxtwo, \var \hastype 
		\mtypetwo}{\mtype}$ 
	and $\namedtyjp{\tderivthree}{}{\val}{\typctxthree}{\mtypetwo}$ such that $\typctx = \typctxtwo \mplus \typctxthree$ 
	with $\sizem{\tderiv} = \sizem{\tderivtwo} + \sizem{\tderivthree}$ and $\size{\tderiv} \leq \size{\tderivtwo} + 
	\size{\tderivthree}$.
\end{lemma}

\cref{l:substitution,l:anti-substitution} are needed to prove subject reduction and expansion, respectively, which mean that the type is preserved after and before any reduction step.
It holds not only for $\tovsub$ but also for $\eqstruct$.
Here we state a \emph{qualitative} version. 
Quantitative versions of subject reduction are in the next sections, they hold for some restrictions of the reduction. 

\begin{proposition}[Qualitative subject reduction and expansion]
	\label{prop:qual-subject}
	\NoteProof{propappendix:qual-subject}
	Let $\tm \,(\tovsub \!\cup \eqstruct)\, \tm'$.
	There is a derivation $\namedtyjp{\tderiv}{}{\tm}{\typctx}{\mtype}$ if and only if there is a derivation $\namedtyjp{\tderiv'}{}{\tm'}{\typctx}{\mtype}$.
\end{proposition}

By \Cref{prop:qual-subject}, our type system does not suffer from Kesner's  counterexample to subject reduction for the type system of \cite{DBLP:conf/fscd/KerinecMR21} in \Cref{ssect:counter-example-kerinec}. Indeed the counterexample concerns the step $\sigma_{3}$ which is subsumed by $\tovsub \!\cup  \eqstruct$ as shown in \Cref{sect:other-calculi}, and for which \Cref{prop:qual-subject} proves subject reduction.

\mybigsubparagraph{The Special Role of Inert Terms} In the characterizations via multi types of the following two sections, inert terms play a crucial role. In statements about solvable normal forms, they usually satisfy stronger properties, essential for the induction to go through. 

\section{Multi Types for Open \cbv}
\label{sect:open}

Here we recall the relationship between \cbv multi types and \ocbv developed by Accattoli and Guerrieri in 
\cite{DBLP:conf/aplas/AccattoliG18}. The reason is threefold: 
\begin{enumerate}
\item \emph{Building block}: the solvable case of the next section relies on the open one, because solving reduction is an iteration under head abstractions of open reduction. 
\item \emph{Blueprint}: the open case provides the blueprint for the solvable case. 
\item \emph{Adapting a few details}: the development in \cite{DBLP:conf/aplas/AccattoliG18} needs to be slightly adapted to our present framework. 
Namely, here we use the Open VSC 
instead of the \emph{split fireball calculus} used in \cite{DBLP:conf/aplas/AccattoliG18} (another formalism for \ocbv),
and we include a \emph{ground type} $\ground$---absent in \cite{DBLP:conf/aplas/AccattoliG18}---required to deal with  solving reduction in the next section.
\end{enumerate}

\mybigsubparagraph{The Open Size of Terms} For our quantitative study, we need a notion of term size, introduced here. We actually need a notion of size for \emph{each} notion of reduction (open here, solving in the next section) that we aim at measuring via multi types. Essentially, the size counts the constructors of a term that can be traversed by the reduction. The \emph{open size} $\sizeo{\tm}$ of a term $\tm$, then, is its number of applications out of abstractions, \ie 
\begin{center} \arraycolsep=2pt	$\begin{array}{rcl@{\hspace{.7cm}}rcl@{\hspace{.7cm}}rcl@{\hspace{.7cm}}rcl}
	\sizeo{\var} &\defeq &0 
	&
	\sizeo{\la{\var}{\tm}} & \defeq  &0
	& 
	\sizeo{\tm\tmtwo} & \defeq & \sizeo{\tm} + \sizeo{\tmtwo} + 1 
	&
	\sizeo{\tm \esub{\var}{\tmtwo}} & \defeq & \sizeo{\tm} + \sizeo{\tmtwo}.
	\end{array}$\end{center}

\mybigsubparagraph{Overview of the Characterization}
Qualitatively, the open reduction of $\tm$ terminates if and only if $\tm$ is typable.
Since $\tovsubo$ does not reduce under abstractions, every abstraction is $\osym$-normal (even unsolvable ones) and hence must be typable: for this reason, $\la{\var}\delta\delta$ is typable with $\emptytype$ (take the derivation only made of one rule $\ruleManyVal$ with $0$ premises), though $\delta\delta$ is not.

Quantitatively, the multiplicative size 
$\sizem{\tderiv}$ of \emph{every} type derivation $\tderiv$ for $\tm$ provides \emph{upper} bounds to the sum of the length of the open reduction of $\tm$ plus 
the open size of its open normal form. To obtain \emph{exact} bounds, one has to avoid typing parts of the term that cannot be touched by open reduction, that is, the body of abstractions (out of other abstractions). Types control in different ways  the possibly many abstractions of an inert term or a term that is itself an abstraction. The former is controlled by the typing context, via a \emph{inert} predicate, the latter by the right-hand type, which needs to not be an arrow type. For the constraint to hold for fireballs, independently of whether they are inert terms or values, the two constraint are put together in the \emph{tight} predicate.

\mysubparagraph{Inert and Tight Derivations} 
\emph{Inert types} are defined as follows, with $n \geq 0$.
\begin{align*}
	\textsc{Inert multi type }& \imtype \grameq 	\mset{ \inltype_1, \mydots, \inltype_n}
	&
	\textsc{Inert linear type }& \inltype \grameq \ground \mid \larrow{n\mset{\ground}}{\imtype} 
\end{align*}
Note that 
every ground multi type $n\mset{\ground}$ is inert. 

\begin{definition}[Inert and tight derivations]
A type context $\typctx=\var_1 \hastype \mtype_1, \mydots, \var_n \hastype \mtype_n$ is \emph{inert} if $\mtype_1, \mydots, 
\mtype_n$ are inert multi types. A derivation $\namedtyjp{\tderiv}{}{\tm}{\typctx}{\mtype}$ is \emph{inert} if $\typctx$ is an inert type context
, and it is \emph{tight} if moreover $\mtype$ is ground.
\end{definition}
Note that the definitions of inert and tight derivations depend only on their final judgment.


The next lemma states the first key property of inert terms, that the inertness of their typing context spreads to the right-hand type. It is used to propagate inertness and tightness from the final judgment to the internal ones, allowing us to apply the \ih in proofs.

\begin{lemma}
	[Spreading of inertness on judgments]
	\label{l:spread-inert}
	\NoteProof{lappendix:spread-inert}
	Let $\concl{\tderiv}{\typctx}{\itm}{\mtype}$ be a inert derivation and $\itm$ be an inert term. 
	Then, $\mtype$ is a inert multi type.
\end{lemma}

\mybigsubparagraph{Correctness}
Open correctness establishes that all typable terms $\osym$-normalize and the multiplicative size of the derivation bounds the number of $\tomo$ steps plus the open size of the $\osym$-normal form; this bound is exact if the derivation is tight. 
Open correctness is proved following a standard scheme in two stages: $(a)$ \emph{quantitative} subject reduction states that every $\tovsubo$ step preserves types and decreases the general size of a derivation, and that any $\tomo$ step decreases by an exact quantity the multiplicative size of a derivation;
$(b)$ a lemma states that the multiplicative size of any derivation typing a $\osym$-normal form $\tm$ provides an \emph{upper bound} to the open size of $\tm$, and if moreover the derivation is tight then the bound~is~\emph{exact}. 
\newcounter{l:size-fireballs}
\addtocounter{l:size-fireballs}{\value{theorem}}
\begin{lemma}[Size of fireballs]
	\label{l:size-fireballs}
	\NoteProof{lappendix:size-fireballs}
Let $\namedtyjp{\tderiv}{}{\tm}{\typctx}{\mtype}$.
%
	\begin{enumerate}
		\item\label{p:size-fireballs-inert} If $\tm = \itm$ is an inert term then $\sizem{\tderiv} \geq \sizeo{\itm}$. 
		If moreover $\tderiv$ is inert, then  $\sizem{\tderiv} = \sizeo{\itm}$.
		
		\item\label{p:size-fireballs-tight} If $\tm =\fire$ is a fireball then $\sizem{\tderiv} \geq \sizeo{\fire}$. 
		If moreover $\tderiv$ is tight, then $\sizem{\tderiv} = \sizeo{\fire}$.
			
	\end{enumerate}
\end{lemma}

Note that for inert terms the equality of sizes is ensured by the weaker inert predicate.
Let us show how tightness enforces the equality of sizes. We have that $\delta \defeq \la{\var}{\var\var}$ is typable, has size $\sizeo\delta = 0$, and any derivation $\concl{\tderiv}{\typctx}{\delta}{\mtype}$ ends with rule 
$\ruleManyVal$.
If $\mtype$ is not ground (and $\tderiv$ not tight) then $\ruleManyVal$ has at least one premise that types the subterm $\var\var$, so $\sizem{\tderiv} > 0=\sizeo{\delta}$.
If $\mtype$ is ground, then $\mtype = \emptytype$ and $\ruleManyVal$ has no premises, that is, $\sizem{\tderiv} = 0 = \sizeo{\delta}$.

Now, we can prove quantitative subject reduction, from which open correctness follows. Note that quantitative subject reduction does not need the inert nor the tight predicate.

\begin{proposition}[Open quantitative subject reduction]
	\label{prop:weak-subject-reduction}
	\NoteProof{propappendix:weak-subject-reduction}
	Let $\namedtyjp{\tderiv}{}{\tm}{\typctx}{\mtype}$ be a derivation.
	\begin{enumerate}
		\item\emph{Multiplicative step:} if $\tm \towm \tm'$ then there is a derivation 
$\namedtyjp{\tderiv'}{}{\tm'}{\typctx}{\mtype}$ with
		$\sizem{\tderiv'} = \sizem{\tderiv} - 2$ and $\size{\tderiv'} = \size{\tderiv} - 1$; 
		\item\emph{Exponential step:} if $\tm \towe \tm'$ then there is a derivation 
$\namedtyjp{\tderiv'}{}{\tm'}{\typctx}{\mtype}$ such that
		$\sizem{\tderiv'} = \sizem{\tderiv}$ and $\size{\tderiv'} < \size{\tderiv}$.
	\end{enumerate}
\end{proposition}

\begin{theorem}[Open correctness]
	\label{thm:open-correctness}
	Let $\concl{\tderiv}{\typctx}{\tm}{\mtype}$.
	Then there is a $\osym$-normalizing reduction $\deriv \colon \tm \tovsubo^* \tmtwo$ with $2\sizem{\deriv} + 
	\sizeo{\tmtwo} \leq \sizem{\tderiv}$.
	And if $\tderiv$ is tight, then $2\sizem{\deriv} + \sizeo{\tmtwo} = \sizem{\tderiv}$.
\end{theorem}

\begin{proof}
	Given the derivation (\resp tight derivation) $\concl{\tderiv}{\typctx}{\tm}{\mtype}$, we proceed by induction on the general size $\size{\tderiv}$ of $\tderiv$.
	
	If $\tm$ is normal for $\tovsubo$, then $\tm = \fire$ is a fireball.
	Let $\deriv$ be the empty reduction sequence (so $\sizem{\deriv} = 0$), thus $\sizem{\tderiv} \geq \sizeo{\fire} = \sizeo{\fire} + 2\sizem{\deriv}$ 
	(\resp $\sizem{\tderiv} = \sizeo{\fire} = \sizeo{\fire} + 2\sizem{\deriv}$) by \reflemma{size-fireballs}.
	
	Otherwise, $\tm$ is not normal for $\tomo$ and so $\tm \tovsubo \tmtwo$.
	According to open subject reduction (\Cref{prop:weak-subject-reduction}), there is a derivation $\concl{\tderivtwo}{\typctx}{\tmtwo}{\mtype}$ such that $\size{\tderivtwo} < \size{\tderiv}$  and 
	\begin{itemize}
		\item $\sizem{\tderivtwo} \leq \sizem{\tderiv} - 2$ (\resp $\sizem{\tderivtwo} = \sizem{\tderiv} - 2$) if $\tm \tomo \tmtwo$,
		\item $\sizem{\tderivtwo} = \sizem{\tderiv}$ if $\tm \toeo \tmtwo$.
	\end{itemize}
	By \ih, there exists a fireball $\fire$ and a reduction sequence $\deriv' \colon \tmtwo \tovsubo^* \fire$ with 
	$2\sizem{\deriv'} + \sizeo{\fire} \leq \sizem{\tderivtwo}$ (\resp $2\sizem{\deriv'} + \sizeo{\fire} = \sizem{\tderivtwo}$).
	Let $\deriv$ be the $\osym$-reduction sequence obtained by concatenating the first step $\tm \tovsubo \tmtwo$ and $\deriv'$.
	There are two cases:
	\begin{itemize}
		\item \emph{Multiplicative:} if $\tm \tomo \tmtwo$ then $\sizem{\tderiv} \geq \sizem{\tderivtwo} + 2 \geq \sizeo{\fire} 
		+ 2\sizem{\deriv'} + 2 = \sizeo{\fire} + 2\sizem{\deriv}$ (\resp $\sizem{\tderiv} = \sizem{\tderivtwo} + 2 = \sizeo{\fire} + 
		2\sizem{\deriv'} + 2 = \sizeo{\fire} + 2\size{\deriv}$), since $\sizem{\deriv} = \sizem{\deriv'} + 1$.
		\item \emph{Exponential:} if $\tm \toeo \tmtwo$ then $\sizem{\tderiv} = \sizem{\tderivtwo} \geq \sizeo{\fire} + 2\sizem{\deriv'} = \sizeo{\fire} + 2\sizem{\deriv}$ (\resp $\sizem{\tderiv} = \sizem{\tderivtwo} = \sizeo{\fire} + 2\sizem{\deriv'} = \sizeo{\fire} + 2\sizem{\deriv}$),  since $\sizem{\deriv} = \sizem{\deriv'}$.
		\qedhere
	\end{itemize} 
\end{proof}

By the operational characterization of \cbv scrutability (\Cref{prop:operational-characterization-cbv-var}.\ref{p:operational-characterization-cbv-var-scrut}), open correctness says in particular that \emph{only} \VSC-scrutable terms are typable (with a multi type).

\mybigsubparagraph{Completeness}
Open completeness states that every $\osym$-normalizing term is typable, and with a tight derivation $\tderiv$ such that $\sizem\tderiv$ is exactly the number of $\tomo$ steps plus the open size of the $\osym$-normal form. 
The proof technique is standard: $(a)$ a lemma states that every $\osym$-normal form is typable with a \emph{tight} derivation; $(b)$ subject expansion (\Cref{prop:qual-subject}) pulls back typability along $\tovsubo$ steps;
the exact bound is inherited from open correctness. 
A notable point is that, again, inert terms verify a special property: they can be given \emph{any} multi~type~$\mtype$.

\newcounter{prop:precise-open-typability-nf}
\addtocounter{prop:precise-open-typability-nf}{\value{theorem}}
\begin{lemma}[Tight typability of open normal forms]
	\label{prop:precise-open-typability-nf} 
	\NoteProof{propappendix:precise-open-typability-nf} 
	\begin{enumerate}
		\item \emph{Inert:}\label{p:precise-open-typability-nf-inert} if $\tm$ is an inert term then, for any multi type $\mtype$, there is a type context $\typctx$ and a derivation $\concl{\tderiv}{\typctx}{\tm}{\mtype}$; if, moreover, $\mtype$ is inert then $\tderiv$ is inert.
		\item \emph{Fireball:}\label{p:precise-open-typability-nf-fireball} if $\tm$ is a fireball then there is a tight derivation $\concl{\tderiv}{\typctx}{\tm}{\emptytype}$.		
	\end{enumerate}
\end{lemma}



\begin{theorem}[Open completeness]
	\label{thm:open-completeness}
	Let $\deriv \colon \tm \tovsubo^* \tmtwo$ be an $\osym$-normalizing reduction sequence. 
	Then there is a tight derivation $\concl{\tderiv}{\typctx}{\tm}{\emptytype}$ such that $2\sizem{\deriv} + \sizeo{\tmtwo} = \sizem{\tderiv}$.
\end{theorem}

\begin{proof}
	It is enough to prove that there is a tight derivation $\concl{\tderiv}{\typctx}{\tm}{\emptytype}$.
	Indeed, by open correctness (\Cref{thm:open-correctness}), from this it follows that there is an $\osym$-normalizing reduction sequence $\deriv' \colon \tm \tovsubo^* \tmtwo'$ such that $2\sizem{\deriv'} + \sizeo{\tmtwo'} = \sizem{\tderiv}$.
	By diamond and strong commutation (\Cref{prop:properties-open-reduction}.\ref{p:properties-open-reduction-diamond}), $\tmtwo' = \tmtwo$ and $\sizem{\deriv'} = \sizem{\deriv}$.
	Let us prove that there is a tight derivation $\concl{\tderiv}{\typctx}{\tm}{\emptytype}$ by induction on the length $\size{\deriv}$ of $\osym$-normalizing reduction sequence $\deriv \colon \tm \tovsubo^* \tmtwo$.
	
	If $\size{\deriv} = 0$ then $\sizem{\deriv} = 0$ and $\tm = \tmtwo$ is $\osym$-normal and hence $\onvarsym$-normal.
	By \Cref{prop:properties-open-reduction}.\ref{p:properties-open-reduction-harmony}, $\tm$ is a fireball.
	By tight typability of fireballs (\Cref{prop:precise-open-typability-nf}), there is a tight derivation $\concl{\tderiv}{\typctx}{\tm}{\emptytype}$.
	
	Otherwise, $\size{\deriv} > 0$ and $\deriv$ is the concatenation of a first step $\tm \tovsubo \tmthree$ and a reduction sequence $\deriv' \colon \tmthree \tovsubo^* \tmtwo$, with $\size{\deriv} = 1 + \size{\deriv'}$.
	By \ih, there is a tight derivation $\concl{\tderivtwo}{\typctx}{\tmthree}{\emptytype}$.
	According to subject expansion (\Cref{prop:qual-subject}, as $\tovsubo \,\subseteq\, \tovsub$), there is a (tight) derivation~$\concl{\tderiv}{\typctx}{\tm}{\emptytype}$.
\end{proof}

%

By the operational characterization of \cbv scutability (\Cref{prop:operational-characterization-cbv-var}.\ref{p:operational-characterization-cbv-var-scrut}), open completeness says that \emph{every} \VSC-scrutable term is typable with $\emptymset$ and an inert type context.
\section{Multi Types for \cbv Solvability}
\label{sect:solvable}
Here we provide both qualitative and quantitative characterizations of \VSC solvable terms by studying the relationship between multi types and solving reduction $\tosolv$.

\mysubparagraph{Solvable size} We need a notion of size for normal forms of solving reduction. The \emph{solvable size} $\sizes{\tm}$ of a term $\tm$ is its number of applications plus its number of head abstractions.
\begin{center} \arraycolsep=2pt	$\begin{array}{rcl@{\hspace{.7cm}}rcl@{\hspace{.7cm}}rcl@{\hspace{.7cm}}rcl}
\sizes{\var} &\defeq &0 
&
\sizes{\la{\var}{\tm}} & \defeq  &\sizes{\tm} + 1 
& 
\sizes{\tm\tmtwo} & \defeq  &\sizes{\tm} + \sizeo{\tmtwo} + 1 
&
\sizes{\tm \esub{\var}{\tmtwo}} & \defeq & \sizes{\tm} + \sizeo{\tmtwo}.
	\end{array}$\end{center}
\begin{figure*}[!t]
\begin{center}
	\scalebox{.9}{
$\begin{array}{rr@{\ }c@{\ }l@{\qquad} rr@{\ }c@{\ }l}
 \text{Solvable multi type } & \smtype &\grameq &	\mset{\sltype_1, \dots, \sltype_n} \ \ n > 0
 &
 \text{Solvable linear type} & \sltype &\grameq & \ground \mid  \larrow{\mtype}{\smtype}
 \\
 \text{Unitary s. multi type} & \usmtype &\grameq &	\mset{ \usltype}
 &
 \text{Unitary s. linear type} &\usltype &\grameq & \ground \mid \larrow{\mtype}{\usmtype}
 \\
 \text{Inertly s. multi type} & \ismtype &\grameq & \mset{\isltype_1, \dots, \isltype_n} \ \ n > 0
 &
 \text{Inertly s. linear type} & \isltype &\grameq & \ground \mid \larrow{\imtype}{\ismtype}
 \end{array}$
}
 \end{center}\vspace{-8pt}
 \caption{Kinds of solvable types. A 
 	type is \emph{precisely solvable} if it is unitary and inertly solvable.}
 \label{fig:solvable-types}
\end{figure*}
\mybigsubparagraph{Solvable Multi Types}
The (qualitative) characterization of solvable terms with multi types is  simple: they are those terms typable with a \emph{solvable multi type}, defined in \Cref{fig:solvable-types}. 
The idea is that an unsolvable term such as $\tm \defeq \la{\var}{\delta\delta}$ should not be typable. 
It is typable only with $\zero$, so we have to forbid the right-hand type to be $\zero$. But then $\la{\vartwo}\tm$, which is also unsolvable, is still typable, with \eg $\mset{\larrow\zero\zero}$. Now,  the problem is the $\zero$ on the \emph{right} of $\multimap$, which is used to type $\tm$, and \emph{not} the $\zero$ on the left of $\multimap$, as it is needed to type solvable terms such as $\la\vartwo\var$ (which is typable with $\mset{\larrow\zero\mtype}$ for any $\mtype$). 
Therefore, solvable types forbids the right-hand type to be $\zero$, and recursively to have $\zero$ on the right of $\multimap$ inside the right-hand type. Such a constraint ultimately requires
a ground multi type $n\mset{\ground}$ different from $\emptytype$ in the type system (in contrast to the open case, which does not need $\ground$). 

\mysubparagraph{Precisely Solvable Multi Types} Every solvable type derivation shall provide bounds, but for \emph{exact} bounds two orthogonal predicates refining solvable types, namely \emph{unitary solvable} and \emph{inertly solvable types} (see \Cref{fig:solvable-types}), are required. 

The \emph{unitary} predicate ensures that each solving multiplicative step is counted \emph{exactly} once. Solvable types guarantee that each such step is counted, but it might be counted more than once. The constraint amounts to 
asking that the topmost and right-hand multisets are singletons. This is the key requirement for obtaining that in the statement of subject reduction the general size of the derivation decreases by exactly one at each multiplicative~step. 

The \emph{inert} predicate (generalizing the one for the open case) ensures that the type derivation does not type sub-terms not accessible to solving reduction. The constraint is that the left-hand multisets have to be inert. As for the open case, the inert predicate enforces the matching of the size of solving normal forms with the size of their type derivation.

Solvable types that are both unitary and inert are called \emph{precisely solvable}, and provide exact bounds, when the type context is also inert (to avoid typing the body of non-head abstractions). 


\mysubparagraph{Correctness}
Solving correctness claims that  solving reduction terminates for all terms typable with a solvable type $\mtype$, and that the multiplicative size of a derivation bounds the number of $\tosolvm$ steps plus the solvable size of the $\solvredsym$-normal form. This bound is exact if the type context is inert and $\mtype$ is precisely solvable. Modulo the new predicates, the proof follows the blueprint of the open~case.

\newcounter{l:size-solvable-nf}
\addtocounter{l:size-solvable-nf}{\value{theorem}}
\begin{lemma}[Size of solved fireballs]
	\label{l:size-solvable-nf}
	\NoteProof{lappendix:size-solvable-nf}
	Let $\solvnf$ be a solved fireball and $\namedtyjp{\tderiv}{}{\solvnf}{\typctx}{\mtype}$.
	\begin{enumerate}
		\item\label{p:size-solvable-nf-bound} \emph{Bounds}: if $\mtype$ is solvable then $\sizem{\tderiv} \geq \sizes{\solvnf}$.
		\item\label{p:size-solvable-nf-exact} \emph{Exact bounds}: if $\typctx$ is inert and $\mtype$ is precisely solvable then $\sizem{\tderiv} = \sizes{\solvnf}$.
	\end{enumerate}
\end{lemma}


The only difference with the open case is that for quantitative solving subject reduction we also need the predicates. The sizes of type derivations decrease only if the right-hand type $\mtype$ is \emph{solvable}, and decrease of the exact quantity only if $\mtype$ is \emph{unitary} solvable. 
For the need for solvable types, consider the unsolvable term $\la{\vartwo}{\delta\delta}$: it is typable only with $\zero$ (which is \emph{not} a solvable type) using a derivation $\tderiv$ that does not type the body $\delta\delta$ of the abstraction (it is made of a $\ruleManyVal$ rule without premises). Its reduct, obtained by reducing the body, is still an abstraction, typable in the same way, and then the size of the derivation does not decrease. 

\begin{proposition}[Solving quantitative subject reduction]
	\label{prop:solvable-subject-reduction}
	\NoteProof{propappendix:solvable-subject-reduction}
Let $\concl{\tderiv}{\typctx\!}{\tm}{\mtype}$ with $\mtype$ solvable.
	\begin{enumerate}
		\item \emph{Multiplicative step:} if $\tm \tosolvm \tm'$ then there is a derivation 
$\concl{\tderiv'}{\typctx}{\tm'}{\mtype}$ such that $\sizem{\tderiv'} \leq \sizem{\tderiv}-2$ and $\size{\tderiv'} < 
\size{\tderiv}$. If moreover  $\mtype$ is unitary solvable then 
$\sizem{\tderiv'} = \sizem{\tderiv}-2$ and $\size{\tderiv'} = \size{\tderiv}-1$.
		
		\item \emph{Exponential step:} if $\tm \tosolve \tm'$ then there is a derivation 
$\concl{\tderiv'}{\typctx}{\tm'}{\mtype}$ such that
		$\sizem{\tderiv'} = \sizem{\tderiv}$ and $\size{\tderiv'} < \size{\tderiv}$.
	\end{enumerate}
\end{proposition}

\begin{theorem}[Solving correctness]
	\label{thm:solvable-correctness}
	Let $\concl{\tderiv}{\typctx}{\tm}{\mtype}$ be a derivation with $\mtype$ solvable.
	Then, 
	there is an $\solvredsym$-normalizing reduction sequence $\deriv \colon \tm \tosolv^* \tmtwo$ 
	with $2\sizem{\deriv} + \sizes{\tmtwo} \leq \sizem{\tderiv}$. If moreover $\typctx$ is a inert type context and $\mtype$ is precisely solvable then $2\sizem{\deriv} + \sizes{\tmtwo} = \sizem{\tderiv}$.
\end{theorem}

\begin{proof}
	By induction on the general size $\size{\tderiv}$ of $\tderiv$.
	
	If $\tm$ is normal for $\tosolv$, then $\tm = \solvnf$ is a solved fireball. 
	Let $\deriv$ be the empty reduction sequence  (so $\sizem{\deriv} = 0$), thus $\sizem{\tderiv} \geq \sizes{\solvnf} = \sizes{\solvnf} + 2\sizem{\deriv}$ 
	(\resp $\sizem{\tderiv} = \sizes{\solvnf} = \sizes{\solvnf} + 2\sizem{\deriv}$) by \reflemma{size-solvable-nf}.
	
	Otherwise, $\tm$ is not normal for $\tosolv$ and so $\tm \tosolv \tmtwo$.
	According to solvable subject reduction (\Cref{prop:solvable-subject-reduction}), there is a derivation $\concl{\tderivtwo}{\typctx}{\tmtwo}{\mtype}$ such that $\size{\tderivtwo} < \size{\tderiv}$  and 
	\begin{itemize}
		\item $\sizem{\tderivtwo} \leq \sizem{\tderiv} - 2$ (\resp $\sizem{\tderivtwo} = \sizem{\tderiv} - 2$) if $\tm \tosolvm \tmtwo$,
		\item $\sizem{\tderivtwo} = \sizem{\tderiv}$ if $\tm \tosolve \tmtwo$.
	\end{itemize}
	By \ih, there is a solved fireball $\solvnf$ and a reduction sequence $\deriv' \colon \tmtwo \tosolv^* \solvnf$ with 
	$2\sizem{\deriv'} + \sizes{\solvnf} \leq \sizem{\tderivtwo}$ (\resp $2\sizem{\deriv'} + \sizes{\solvnf} = \sizem{\tderivtwo}$).
	Let $\deriv$ be the $\solvredsym$-reduction sequence obtained by concatenating the first step $\tm \tosolv \tmtwo$ and $\deriv'$.
	There are two cases:
	\begin{itemize}
		\item \emph{Multiplicative:} if $\tm \tosolvm \tmtwo$ then $\sizem{\tderiv} \geq \sizem{\tderivtwo} + 2 \geq \sizes{\solvnf} 
		+ 2\sizem{\deriv'} + 2 = \sizes{\solvnf} + 2\sizem{\deriv}$ (\resp $\sizem{\tderiv} = \sizem{\tderivtwo} + 2 = \sizes{\solvnf} + 
		2\sizem{\deriv'} + 2 = \sizes{\solvnf} + 2\sizem{\deriv}$), since $\sizem{\deriv} = \sizem{\deriv'} + 1$.
		\item \emph{Exponential:} if $\tm \tosolve \tmtwo$ then $\sizem{\tderiv} = \sizem{\tderivtwo} \geq \sizes{\solvnf} + 2\sizem{\deriv'} = \sizes{\solvnf} + 2\sizem{\deriv}$ (\resp $\sizem{\tderiv} = \sizem{\tderivtwo} = \sizes{\solvnf} + 2\sizem{\deriv'} = \sizes{\solvnf} + 2\sizem{\deriv}$),  since $\sizem{\deriv} = \sizem{\deriv'}$.
		\qedhere
	\end{itemize} 
\end{proof}

By the operational characterization of \cbv solvability (\Cref{prop:operational-characterization-cbv-var}.\ref{p:operational-characterization-cbv-var-solv}), solving correctness says in particular that \emph{only} \VSC-solvable terms are typable with a solvable multi type.

\mybigsubparagraph{Completeness} 
Solving completeness claims that every term such that its solving reduction terminates is typable with a precisely solvable type and an inert type context, and that the multiplicative size of the derivation is equal to the number of $\tosolvm$ steps plus the solvable size of the $\solvredsym$-normal form. Modulo the new predicates, the proof essentially follows the blueprint of the open case.
In particular, completeness follows easily from the typability of solved fireballs.

\newcounter{prop:solvable-typability-nf}
\addtocounter{prop:solvable-typability-nf}{\value{theorem}}
\begin{lemma}[Precisely solvable typability of solved fireballs]
	\label{prop:precise-solvable-typability-nf}
	\NoteProof{propappendix:precise-solvable-typability-nf}
		If $\tm$ is a solved fireball, then there is a derivation 
$\concl{\tderiv}{\typctx}{\tm}{\mtype}$ with $\typctx$ inert type context and $\mtype$ precisely solvable.
\end{lemma}


\begin{theorem}[Solving completeness]
	\label{thm:solvable-completeness}
	Let $\deriv \colon \tm \tosolv^* \tmtwo$ be a $\solvredsym$-normalizing reduction sequence. 
	Then there is a derivation $\concl{\tderiv}{\typctx}{\tm}{\mtypetwo}$ with $\typctx$ inert, $\mtypetwo$ precisely solvable 	and \mbox{$2\sizem{\deriv} + \sizes{\tmtwo} = \sizem{\tderiv}$}.
\end{theorem}

\begin{proof}
	It suffices to prove that there is a derivation $\concl{\tderiv}{\typctx}{\tm}{\mtypetwo}$ with $\typctx$ inert and $\mtypetwo$ precisely solvable.
	Indeed, by solvable correctness (\Cref{thm:solvable-correctness}), it follows then that there is an $\solvsym$-normalizing reduction sequence $\deriv' \colon \tm \tovsubo^* \tmtwo'$ such that $2\sizem{\deriv'} + \sizes{\tmtwo'} = \sizem{\tderiv}$.
	By diamond and strong commutation (\Cref{prop:properties-solvable-reduction}.\ref{p:properties-solvable-reduction-diamond}), $\tmtwo' = \tmtwo$ and $\sizem{\deriv'} = \sizem{\deriv}$.
	Let us prove that there is a derivation $\concl{\tderiv}{\typctx}{\tm}{\mtypetwo}$ with $\typctx$ inert and $\mtypetwo$ precisely solvable, by induction on the length $\size{\deriv}$ of the $\solvsym$-normalizing reduction~$\deriv \colon \tm \tosolv^* \tmtwo$.
	
	If $\size{\deriv} = 0$ then $\sizem{\deriv} = 0$ and $\tm = \tmtwo$ is $\solvredsym$-normal and hence $\solvredsym_\nvarsym$-normal.
	By \Cref{prop:properties-solvable-reduction}.\ref{p:properties-solvable-reduction-harmony}, $\tm$ is a solved fireball.
	By precisely solvable typability of solved fireballs (\Cref{prop:precise-solvable-typability-nf}), there is a derivation $\concl{\tderiv}{\typctx}{\tm}{\mtypetwo}$ with $\typctx$ inert and $\mtypetwo$ precisely solvable.
	
	Otherwise, $\size{\deriv} > 0$ and $\deriv$ is the concatenation of a first step $\tm \tosolv \tmthree$ and a reduction sequence $\deriv' \colon \tmthree \tosolv^* \tmtwo$, with $\size{\deriv} = 1 + \size{\deriv'}$.
	By \ih, there is a derivation $\concl{\tderivtwo}{\typctx}{\tmthree}{\mtypetwo}$ with $\typctx$ inert and $\mtypetwo$ precisely solvable.
	By subject expansion (\Cref{prop:qual-subject}, as $\tosolv \,\subseteq\, \tovsub$), there is a derivation~$\concl{\tderiv}{\typctx}{\tm}{\mtypetwo}$.
\end{proof}

By the operational characterization of \cbv solvability (\Cref{prop:operational-characterization-cbv-var}.\ref{p:operational-characterization-cbv-var-solv}), solving completeness says that \emph{every} \VSC-solvable term is typable with a precisely solvable type and an inert type context.

\section{Normalization and denotational semantics}
\label{sect:denotation}

In this section we show how our type-theoretic investigation can be used to study other operational properties of the \VSC, and can be lifted to a semantic level.

\paragraph{Normalizations}
Our study of multi types for \ocbv and \cbv solvability also allows us to prove two normalization results: 
reductions $\tovsubo$ and  $\tosolv$ are \emph{complete} with respect to their own normal forms, in the sense that if a term $\vsub$-reduces to a $\osym$-normal (\resp~$\solvsym$-normal) form, reduction $\tovsubo$ (\resp $\tosolv$) is enough to reach a possibly different $\osym$-normal (\resp~$\solvsym$-normal) form.
The proof exploits an elegant technique already used by \citet{DBLP:journals/tcs/CarvalhoPF11} and \citet{MazzaPellissierVial18}.

\begin{theorem}[Normalization]
	\label{thm:normalization}
	Let $\tm$ be a term in the \VSC.
	\begin{enumerate}
		\item \label{p:normalization-open}\emph{Open reduction:} if $\tm \tovsub^* \tmtwo$ where  $\tmtwo$ is $\osym$-normal, then $\tm \tovsubo^* \tmthree$ for some $\osym$-normal $\tmthree$.
		\item \label{p:normalization-solv}\emph{Solving reduction:} if $\tm \tovsub^* \tmtwo$ where  $\tmtwo$ is $\solvsym$-normal, then $\tm \tosolv^* \tmthree$ for some $\solvsym$-normal $\tmthree$.
	\end{enumerate}
\end{theorem}

\begin{proof}
	\begin{enumerate}
		\item Every $\osym$-normal form $\tmtwo$ is a fireball (\Cref{prop:properties-open-reduction}.\ref{p:properties-open-reduction-harmony}) and hence has a derivation $\concl{\tderiv}{\typctx}{\tmtwo}{\mtype}$ (\Cref{prop:precise-open-typability-nf}.\ref{p:precise-open-typability-nf-fireball}). 
		Subject
expansion (\Cref{prop:qual-subject}) iterated along $\tm \tovsub^* \tmtwo$ gives a derivation $\concl{\tderivtwo}{\typctx}{\tm}{\mtype}$ for $\tm$.
		Open correctness (\Cref{thm:open-correctness}) gives $\tm \tovsubo^* \tmthree$ with $\tmthree$ $\osym$-normal. 
		
		\item Every $\solvsym$-normal form $\tmtwo$ is a solved fireball (\Cref{prop:properties-solvable-reduction}.\ref{p:properties-solvable-reduction-harmony}) and hence has a derivation $\concl{\tderiv}{\typctx}{\tmtwo}{\mtype}$ (\Cref{prop:precise-solvable-typability-nf}). 
		Subject
expansion (\Cref{prop:qual-subject}) iterated along $\tm \tovsub^* \tmtwo$ gives a derivation $\concl{\tderivtwo}{\typctx}{\tm}{\mtype}$ for $\tm$.
		Solving correctness (\Cref{thm:solvable-correctness}) gives $\tm \tosolv^* \tmthree$ with $\tmthree$ $\solvsym$-normal. 
		\qedhere
	\end{enumerate}
\end{proof}

\Cref{thm:normalization}.\ref{p:normalization-open} is a generalization of the valuability result (\Cref{prop:properties-open-extra}.\ref{p:properties-open-extra-valuability}) and it is the same as \Cref{prop:properties-open-extra}.\ref{p:properties-open-extra-normalization}.
\Cref{thm:normalization}.\ref{p:normalization-solv} is the same as \Cref{prop:properties-solving}.\ref{p:properties-solving-normalization}, but proved by type-theoretic means rather \mbox{than~operational}.

\mybigsubparagraph{Multi Types as (Sensible) Relational Semantics}
Multi types induce a relational model\footnotemark
\footnotetext{Such a model is 
	the restriction of the relational model for lineal logic to the image of  Girard's \cbv translation $(A \Rightarrow B)^\mathsf{v} = \oc (A^\mathsf{v} \multimap B^\mathsf{v})$ of the intuitionistic arrow into linear logic \cite{DBLP:journals/tcs/Girard87}.}
by interpreting a term 
as the set of its type judgments. 
More precisely, let $\tm$ be a term and $\var_1, \dots, \var_n$ (with $n \geq 0$) be pairwise distinct variables. 
If $\fv{\tm} \subseteq \{\var_1, \dots, \var_n\}$, we say that the list $\vec{\var} = (\var_1, \dots, \var_n)$ is \emph{suitable for} $\tm$.
If $\vec{\var} = (\var_1, \dots, \var_n)$ is suitable for $\tm$, the (\emph{plain}) \emph{semantics} $\sem{\tm}_{\vec{\var}}$ \emph{of} $\tm$ \emph{for} $\vec{\var}$ and the \emph{solvable semantics} $\semsolv{\tm}_{\vec{\var}}$ \emph{of} $\tm$ \emph{for} $\vec{\var}$ are defined by:
\begin{align*}
	\sem{\tm}_{\vec{\var}} &\defeq \{((\mtypetwo_1,\dots, \mtypetwo_n),\mtype) \mid 
	\exists 
	\, 
	\concl{\tderiv}{\var_1 \hastype \mtypetwo_1, \dots, \var_n \hastype \mtypetwo_n}{\tm}{\mtype} \} \,
	\\[-2pt]
	\semsolv{\tm}_{\vec{\var}} &\defeq \{((\mtypetwo_1,\dots, \mtypetwo_n),\mtype) \mid \exists 
	\concl{\tderiv}{\var_1 \hastype \mtypetwo_1, \dots, \var_n \hastype \mtypetwo_n}{\tm}{\mtype} 
	\mbox{ such that $\mtype$ is solvable} \} .
\end{align*}
Subject reduction and expansion (\Cref{prop:qual-subject}) guarantee that $\sem{\tm}_{\vec{\var}}$ and $\semsolv{\tm}_{\vec{\var}}$ are \emph{invariant} by $\tovsub  \!\cup \eqstruct$. 
So, we provide two distinct \emph{denotational} semantics
not only for the (core) \VSC, 
but also for its extension considered in \Cref{sect:other-calculi}, obtained by adding structural equivalence $\eqstruct$ to the core \VSC.

\begin{proposition}[Invariance]
	Let $\tm, \tmtwo$ be terms in the \VSC with $\vec{\var}  = (\var_1, \dots, \var_n)$ suitable for both of them.
	If $\tm \,(\tovsub \!\cup \eqstruct) \, \tmtwo$ then $\sem{\tm}_{\vec{\var}} = \sem{\tmtwo}_{\vec{\var}}$ and $\semsolv{\tm}_{\vec{\var}} = \semsolv{\tmtwo}_{\vec{\var}}$.
\end{proposition}

Open and solving correctness (\Cref{thm:open-correctness,thm:solvable-correctness}) and completeness (\Cref{thm:open-completeness,thm:solvable-completeness}) guarantee \emph{adequacy} results for these semantics, \ie a semantic characterization of \cbv\ scrutability/solvability.

\begin{theorem}[Adequacy]
	\label{thm:adequacy}
	Let $\tm$ be a term in the \VSC with $\vec{\var}  = (\var_1, \dots, \var_n)$ suitable for it.
\begin{enumerate}
	\item\label{p:adequacy-open} \emph{Open:} $\sem{\tm}_{\vec{\var}}$ is non-empty if and only if $\tm$ is $\osym$-normalizing if and only if $\tm$ is \VSC-scrutable.
	\item\label{p:adequacy-solv} \emph{Solvable:} $\semsolv{\tm}_{\vec{\var}}$ is non-empty if and only if $\tm$ $\solvredsym$-normalizing if and only if $\tm$ is \VSC-solvable.
\end{enumerate}
\end{theorem}
Open adequacy (\Cref{thm:adequacy}.\ref{p:adequacy-open}) implies that the equational theory $\eqth_{\osym}$ induced by $\sem{\tm}_{\vec{\var}}$ (which equates terms having the same semantics)  is scrutable.
The equational theory $\eqth_{\solvsym}$ induced by the solving semantics, instead, collapses all \cbv unsolvable terms, and is thus \emph{inconsistent} (\Cref{prop:inconsistency}).
Thus---unlike $\eqth_{\osym}$---the study of $\eqth_{\solvsym}$ turns out to be pointless, although the solving semantics which induces that theory characterizes interesting operational properties. 

\paragraph{Relational Semantics and \cbv Models} 
\newcommand{\Mal}{\mathcal{M}}
\newcommand{\Db}{\mathbb{D}}
\newcommand{\Vb}{\mathbb{V}}
\newcommand{\Eb}{\mathbb{E}}

Inspired by \citet{HindleyLongo80}, \citet{DBLP:journals/fuin/EgidiHR92} proposed a set-theoretic and axiomatic definition of a \cbv denotational model, later used and simplified by Ronchi Della Rocca et al. \cite{DBLP:journals/mscs/PravatoRR99,parametricBook,DBLP:journals/fuin/ManzonettoPR19}.
\citet{DBLP:journals/fuin/ManzonettoPR19} showed that a certain family of multi type systems for \cbv induce a family of \cbv models (in the sense of \citet{DBLP:journals/fuin/EgidiHR92}).
Ehrhard's multi type system (\Cref{fig:cbvtypes}) used here does not belong to that family, it has different rules, but it shares the same philosophy based on two kinds of type, linear and multi. 
So, the proof in \cite{DBLP:journals/fuin/ManzonettoPR19} can be easily adapted to show that our multi type system in \Cref{fig:cbvtypes} induces a \cbv model.
\section{Conclusions}
\label{sect:conclusions}
This paper shows that \cbv solvability in the \VSC has a rich theory,  comparable to the one of \cbn solvability in terms of characterizations, and yet different, as \cbv unsolvable terms are not collapsible. A natural future direction is the refinement of behavioral equivalences such as Lassen's open \cbv bisimilarity \cite{DBLP:conf/lics/Lassen05}, which is not a scrutable theory: inscrutable terms such as $\Omega$, $(\var\vartwo) \Omega$, and $(\la\var\delta) (\vartwo\vartwo)\delta$ (where $\delta \defeq \la\varthree\varthree\varthree$) are all distinct for his bisimilarity. 
At a more technical level, Ghilezan \cite{DBLP:journals/jcss/Ghilezan01} develops an interesting technique for proving the genericity lemma, based on a topology over $\l$-terms defined via intersection types. It would be interesting to see if it can be adapted to Ehrhard's multi types to prove genericity for \cbv inscrutable terms. 

\bibliographystyle{ACM-Reference-Format}
\bibliography{\macrospath/biblio}


\begin{thebibliography}{77}


\ifx \showCODEN    \undefined \def \showCODEN     #1{\unskip}     \fi
\ifx \showDOI      \undefined \def \showDOI       #1{#1}\fi
\ifx \showISBNx    \undefined \def \showISBNx     #1{\unskip}     \fi
\ifx \showISBNxiii \undefined \def \showISBNxiii  #1{\unskip}     \fi
\ifx \showISSN     \undefined \def \showISSN      #1{\unskip}     \fi
\ifx \showLCCN     \undefined \def \showLCCN      #1{\unskip}     \fi
\ifx \shownote     \undefined \def \shownote      #1{#1}          \fi
\ifx \showarticletitle \undefined \def \showarticletitle #1{#1}   \fi
\ifx \showURL      \undefined \def \showURL       {\relax}        \fi
\providecommand\bibfield[2]{#2}
\providecommand\bibinfo[2]{#2}
\providecommand\natexlab[1]{#1}
\providecommand\showeprint[2][]{arXiv:#2}

\bibitem[Abramsky(1991)]%
        {DBLP:journals/apal/Abramsky91}
\bibfield{author}{\bibinfo{person}{Samson Abramsky}.}
  \bibinfo{year}{1991}\natexlab{}.
\newblock \showarticletitle{Domain Theory in Logical Form}.
\newblock \bibinfo{journal}{\emph{Ann. Pure Appl. Log.}} \bibinfo{volume}{51},
  \bibinfo{number}{1-2} (\bibinfo{year}{1991}), \bibinfo{pages}{1--77}.
\newblock
\urldef\tempurl%
\url{https://doi.org/10.1016/0168-0072(91)90065-T}
\showDOI{\tempurl}


\bibitem[Accattoli(2015)]%
        {DBLP:journals/tcs/Accattoli15}
\bibfield{author}{\bibinfo{person}{Beniamino Accattoli}.}
  \bibinfo{year}{2015}\natexlab{}.
\newblock \showarticletitle{{Proof nets and the call-by-value
  {\(\lambda\)}-calculus}}.
\newblock \bibinfo{journal}{\emph{Theor. Comput. Sci.}}  \bibinfo{volume}{606}
  (\bibinfo{year}{2015}), \bibinfo{pages}{2--24}.
\newblock


\bibitem[Accattoli et~al\mbox{.}(2019a)]%
        {DBLP:conf/ppdp/AccattoliCGC19}
\bibfield{author}{\bibinfo{person}{Beniamino Accattoli},
  \bibinfo{person}{Andrea Condoluci}, \bibinfo{person}{Giulio Guerrieri}, {and}
  \bibinfo{person}{Claudio Sacerdoti~Coen}.} \bibinfo{year}{2019}\natexlab{a}.
\newblock \showarticletitle{Crumbling Abstract Machines}. In
  \bibinfo{booktitle}{\emph{Proceedings of the 21st International Symposium on
  Principles and Practice of Programming Languages, {PPDP} 2019, Porto,
  Portugal, October 7-9, 2019}}. \bibinfo{pages}{4:1--4:15}.
\newblock
\urldef\tempurl%
\url{https://doi.org/10.1145/3354166.3354169}
\showDOI{\tempurl}


\bibitem[Accattoli et~al\mbox{.}(2021a)]%
        {DBLP:conf/lics/AccattoliCC21}
\bibfield{author}{\bibinfo{person}{Beniamino Accattoli},
  \bibinfo{person}{Andrea Condoluci}, {and} \bibinfo{person}{Claudio
  Sacerdoti~Coen}.} \bibinfo{year}{2021}\natexlab{a}.
\newblock \showarticletitle{Strong Call-by-Value is Reasonable, Implosively}.
  In \bibinfo{booktitle}{\emph{{LICS}}}. \bibinfo{publisher}{{IEEE}},
  \bibinfo{pages}{1--14}.
\newblock


\bibitem[Accattoli and {Dal Lago}(2012)]%
        {DBLP:conf/rta/AccattoliL12}
\bibfield{author}{\bibinfo{person}{Beniamino Accattoli} {and}
  \bibinfo{person}{Ugo {Dal Lago}}.} \bibinfo{year}{2012}\natexlab{}.
\newblock \showarticletitle{{On the Invariance of the Unitary Cost Model for
  Head Reduction}}. In \bibinfo{booktitle}{\emph{{RTA}}}.
  \bibinfo{pages}{22--37}.
\newblock


\bibitem[Accattoli et~al\mbox{.}(2021b)]%
        {DBLP:journals/pacmpl/AccattoliLV21}
\bibfield{author}{\bibinfo{person}{Beniamino Accattoli}, \bibinfo{person}{Ugo
  Dal~Lago}, {and} \bibinfo{person}{Gabriele Vanoni}.}
  \bibinfo{year}{2021}\natexlab{b}.
\newblock \showarticletitle{The (In)Efficiency of interaction}.
\newblock \bibinfo{journal}{\emph{Proc. {ACM} Program. Lang.}}
  \bibinfo{volume}{5}, \bibinfo{number}{{POPL}} (\bibinfo{year}{2021}),
  \bibinfo{pages}{1--33}.
\newblock
\urldef\tempurl%
\url{https://doi.org/10.1145/3434332}
\showDOI{\tempurl}


\bibitem[Accattoli et~al\mbox{.}(2021c)]%
        {DBLP:conf/lics/AccattoliLV21}
\bibfield{author}{\bibinfo{person}{Beniamino Accattoli}, \bibinfo{person}{Ugo
  Dal~Lago}, {and} \bibinfo{person}{Gabriele Vanoni}.}
  \bibinfo{year}{2021}\natexlab{c}.
\newblock \showarticletitle{The Space of Interaction}. In
  \bibinfo{booktitle}{\emph{{LICS}}}. \bibinfo{publisher}{{IEEE}},
  \bibinfo{pages}{1--13}.
\newblock


\bibitem[Accattoli et~al\mbox{.}(2018)]%
        {DBLP:journals/pacmpl/AccattoliGK18}
\bibfield{author}{\bibinfo{person}{Beniamino Accattoli},
  \bibinfo{person}{St{\'{e}}phane Graham{-}Lengrand}, {and}
  \bibinfo{person}{Delia Kesner}.} \bibinfo{year}{2018}\natexlab{}.
\newblock \showarticletitle{Tight typings and split bounds}.
\newblock \bibinfo{journal}{\emph{{PACMPL}}} \bibinfo{volume}{2},
  \bibinfo{number}{{ICFP}} (\bibinfo{year}{2018}),
  \bibinfo{pages}{94:1--94:30}.
\newblock
\urldef\tempurl%
\url{https://doi.org/10.1145/3236789}
\showDOI{\tempurl}


\bibitem[Accattoli and Guerrieri(2016)]%
        {DBLP:conf/aplas/AccattoliG16}
\bibfield{author}{\bibinfo{person}{Beniamino Accattoli} {and}
  \bibinfo{person}{Giulio Guerrieri}.} \bibinfo{year}{2016}\natexlab{}.
\newblock \showarticletitle{{Open Call-by-Value}}. In
  \bibinfo{booktitle}{\emph{Programming Languages and Systems - 14th Asian
  Symposium, {APLAS} 2016}} \emph{(\bibinfo{series}{Lecture Notes in Computer
  Science}, Vol.~\bibinfo{volume}{10017})}. \bibinfo{publisher}{Springer},
  \bibinfo{pages}{206--226}.
\newblock
\urldef\tempurl%
\url{https://doi.org/10.1007/978-3-319-47958-3_12}
\showDOI{\tempurl}


\bibitem[Accattoli and Guerrieri(2018)]%
        {DBLP:conf/aplas/AccattoliG18}
\bibfield{author}{\bibinfo{person}{Beniamino Accattoli} {and}
  \bibinfo{person}{Giulio Guerrieri}.} \bibinfo{year}{2018}\natexlab{}.
\newblock \showarticletitle{Types of Fireballs}. In
  \bibinfo{booktitle}{\emph{Programming Languages and Systems - 16th Asian
  Symposium, {APLAS} 2018, Wellington, New Zealand, December 2-6, 2018,
  Proceedings}}. \bibinfo{pages}{45--66}.
\newblock
\urldef\tempurl%
\url{https://doi.org/10.1007/978-3-030-02768-1\_3}
\showDOI{\tempurl}


\bibitem[Accattoli et~al\mbox{.}(2019b)]%
        {DBLP:conf/esop/AccattoliGL19}
\bibfield{author}{\bibinfo{person}{Beniamino Accattoli},
  \bibinfo{person}{Giulio Guerrieri}, {and} \bibinfo{person}{Maico Leberle}.}
  \bibinfo{year}{2019}\natexlab{b}.
\newblock \showarticletitle{Types by Need}. In
  \bibinfo{booktitle}{\emph{Programming Languages and Systems - 28th European
  Symposium on Programming, {ESOP} 2019, Held as Part of the European Joint
  Conferences on Theory and Practice of Software, {ETAPS} 2019, Prague, Czech
  Republic, April 6-11, 2019, Proceedings}}. \bibinfo{pages}{410--439}.
\newblock
\urldef\tempurl%
\url{https://doi.org/10.1007/978-3-030-17184-1\_15}
\showDOI{\tempurl}


\bibitem[Accattoli et~al\mbox{.}(2021d)]%
        {DBLP:journals/corr/abs-2104-13979}
\bibfield{author}{\bibinfo{person}{Beniamino Accattoli},
  \bibinfo{person}{Giulio Guerrieri}, {and} \bibinfo{person}{Maico Leberle}.}
  \bibinfo{year}{2021}\natexlab{d}.
\newblock \showarticletitle{Semantic Bounds and Strong Call-by-Value
  Normalization}.
\newblock \bibinfo{journal}{\emph{CoRR}}  \bibinfo{volume}{abs/2104.13979}
  (\bibinfo{year}{2021}).
\newblock


\bibitem[Accattoli and Paolini(2012)]%
        {AccattoliPaolini12}
\bibfield{author}{\bibinfo{person}{Beniamino Accattoli} {and}
  \bibinfo{person}{Luca Paolini}.} \bibinfo{year}{2012}\natexlab{}.
\newblock \showarticletitle{Call-by-Value Solvability, Revisited}. In
  \bibinfo{booktitle}{\emph{Functional and Logic Programming - 11th
  International Symposium, {FLOPS} 2012, Kobe, Japan, May 23-25, 2012.
  Proceedings}}. \bibinfo{pages}{4--16}.
\newblock
\urldef\tempurl%
\url{https://doi.org/10.1007/978-3-642-29822-6\_4}
\showDOI{\tempurl}


\bibitem[Accattoli and {Sacerdoti Coen}(2015)]%
        {fireballs}
\bibfield{author}{\bibinfo{person}{Beniamino Accattoli} {and}
  \bibinfo{person}{Claudio {Sacerdoti Coen}}.} \bibinfo{year}{2015}\natexlab{}.
\newblock \showarticletitle{On the Relative Usefulness of Fireballs}. In
  \bibinfo{booktitle}{\emph{30th Annual {ACM/IEEE} Symposium on Logic in
  Computer Science, {LICS} 2015, Kyoto, Japan, July 6-10, 2015}}.
  \bibinfo{pages}{141--155}.
\newblock
\urldef\tempurl%
\url{https://doi.org/10.1109/LICS.2015.23}
\showDOI{\tempurl}


\bibitem[Accattoli and Sacerdoti~Coen(2017)]%
        {DBLP:journals/iandc/AccattoliC17}
\bibfield{author}{\bibinfo{person}{Beniamino Accattoli} {and}
  \bibinfo{person}{Claudio Sacerdoti~Coen}.} \bibinfo{year}{2017}\natexlab{}.
\newblock \showarticletitle{On the value of variables}.
\newblock \bibinfo{journal}{\emph{Information and Computation}}
  \bibinfo{volume}{255} (\bibinfo{year}{2017}), \bibinfo{pages}{224--242}.
\newblock
\urldef\tempurl%
\url{https://doi.org/10.1016/j.ic.2017.01.003}
\showDOI{\tempurl}


\bibitem[Alves et~al\mbox{.}(2019)]%
        {DBLP:conf/types/AlvesKV19}
\bibfield{author}{\bibinfo{person}{Sandra Alves}, \bibinfo{person}{Delia
  Kesner}, {and} \bibinfo{person}{Daniel Ventura}.}
  \bibinfo{year}{2019}\natexlab{}.
\newblock \showarticletitle{A Quantitative Understanding of Pattern Matching}.
  In \bibinfo{booktitle}{\emph{25th International Conference on Types for
  Proofs and Programs, {TYPES} 2019, June 11-14, 2019, Oslo, Norway}}.
  \bibinfo{pages}{3:1--3:36}.
\newblock
\urldef\tempurl%
\url{https://doi.org/10.4230/LIPIcs.TYPES.2019.3}
\showDOI{\tempurl}


\bibitem[Barendregt et~al\mbox{.}(1983)]%
        {DBLP:journals/jsyml/BarendregtCD83}
\bibfield{author}{\bibinfo{person}{Henk Barendregt}, \bibinfo{person}{Mario
  Coppo}, {and} \bibinfo{person}{Mariangiola Dezani{-}Ciancaglini}.}
  \bibinfo{year}{1983}\natexlab{}.
\newblock \showarticletitle{A Filter Lambda Model and the Completeness of Type
  Assignment}.
\newblock \bibinfo{journal}{\emph{J. Symb. Log.}} \bibinfo{volume}{48},
  \bibinfo{number}{4} (\bibinfo{year}{1983}), \bibinfo{pages}{931--940}.
\newblock
\urldef\tempurl%
\url{https://doi.org/10.2307/2273659}
\showDOI{\tempurl}


\bibitem[Barendregt(1971)]%
        {DBLP:books/daglib/0016519}
\bibfield{author}{\bibinfo{person}{Hendrik~Pieter Barendregt}.}
  \bibinfo{year}{1971}\natexlab{}.
\newblock \emph{\bibinfo{title}{Some extensional term models for combinatory
  logics and l - calculi}}.
\newblock \bibinfo{thesistype}{Ph.\,D. Dissertation}. \bibinfo{school}{Univ.
  Utrecht}.
\newblock


\bibitem[Barendregt(1974)]%
        {solvability-barendregt}
\bibfield{author}{\bibinfo{person}{Hendrik~Pieter Barendregt}.}
  \bibinfo{year}{1974}\natexlab{}.
\newblock \showarticletitle{Solvability in lambda-calculi}.
\newblock \bibinfo{journal}{\emph{Journal of Symbolic Logic - JSYML}}
  (\bibinfo{date}{01} \bibinfo{year}{1974}), \bibinfo{pages}{372--372}.
\newblock


\bibitem[Barendregt(1984)]%
        {Barendregt84}
\bibfield{author}{\bibinfo{person}{Hendrik~Pieter Barendregt}.}
  \bibinfo{year}{1984}\natexlab{}.
\newblock \bibinfo{booktitle}{\emph{{The Lambda Calculus -- Its Syntax and
  Semantics}}}. Vol.~\bibinfo{volume}{103}.
\newblock \bibinfo{publisher}{North-Holland}.
\newblock


\bibitem[Bernadet and Lengrand(2013)]%
        {DBLP:journals/corr/BernadetL13}
\bibfield{author}{\bibinfo{person}{Alexis Bernadet} {and}
  \bibinfo{person}{St{\'{e}}phane Lengrand}.} \bibinfo{year}{2013}\natexlab{}.
\newblock \showarticletitle{Non-idempotent intersection types and strong
  normalisation}.
\newblock \bibinfo{journal}{\emph{Logical Methods in Computer Science}}
  \bibinfo{volume}{9}, \bibinfo{number}{4} (\bibinfo{year}{2013}).
\newblock


\bibitem[Bucciarelli and Ehrhard(2001)]%
        {DBLP:journals/apal/BucciarelliE01}
\bibfield{author}{\bibinfo{person}{Antonio Bucciarelli} {and}
  \bibinfo{person}{Thomas Ehrhard}.} \bibinfo{year}{2001}\natexlab{}.
\newblock \showarticletitle{On phase semantics and denotational semantics: the
  exponentials}.
\newblock \bibinfo{journal}{\emph{Ann. Pure Appl. Logic}}
  \bibinfo{volume}{109}, \bibinfo{number}{3} (\bibinfo{year}{2001}),
  \bibinfo{pages}{205--241}.
\newblock


\bibitem[Bucciarelli et~al\mbox{.}(2020)]%
        {DBLP:conf/flops/BucciarelliKRV20}
\bibfield{author}{\bibinfo{person}{Antonio Bucciarelli}, \bibinfo{person}{Delia
  Kesner}, \bibinfo{person}{Alejandro R{\'{\i}}os}, {and}
  \bibinfo{person}{Andr{\'{e}}s Viso}.} \bibinfo{year}{2020}\natexlab{}.
\newblock \showarticletitle{The Bang Calculus Revisited}. In
  \bibinfo{booktitle}{\emph{Functional and Logic Programming - 15th
  International Symposium, {FLOPS} 2020, Akita, Japan, September 14-16, 2020,
  Proceedings}}. \bibinfo{pages}{13--32}.
\newblock
\urldef\tempurl%
\url{https://doi.org/10.1007/978-3-030-59025-3\_2}
\showDOI{\tempurl}


\bibitem[Bucciarelli et~al\mbox{.}(2021)]%
        {DBLP:journals/lmcs/BucciarelliKR21}
\bibfield{author}{\bibinfo{person}{Antonio Bucciarelli}, \bibinfo{person}{Delia
  Kesner}, {and} \bibinfo{person}{Simona Ronchi~Della Rocca}.}
  \bibinfo{year}{2021}\natexlab{}.
\newblock \showarticletitle{Solvability = Typability + Inhabitation}.
\newblock \bibinfo{journal}{\emph{Log. Methods Comput. Sci.}}
  \bibinfo{volume}{17}, \bibinfo{number}{1} (\bibinfo{year}{2021}).
\newblock


\bibitem[Bucciarelli et~al\mbox{.}(2017)]%
        {BKV17}
\bibfield{author}{\bibinfo{person}{Antonio Bucciarelli}, \bibinfo{person}{Delia
  Kesner}, {and} \bibinfo{person}{Daniel Ventura}.}
  \bibinfo{year}{2017}\natexlab{}.
\newblock \showarticletitle{Non-idempotent intersection types for the
  Lambda-Calculus}.
\newblock \bibinfo{journal}{\emph{Logic Journal of the IGPL}}
  \bibinfo{volume}{25}, \bibinfo{number}{4} (\bibinfo{year}{2017}),
  \bibinfo{pages}{431--464}.
\newblock


\bibitem[Carraro and Guerrieri(2014)]%
        {DBLP:conf/fossacs/CarraroG14}
\bibfield{author}{\bibinfo{person}{Alberto Carraro} {and}
  \bibinfo{person}{Giulio Guerrieri}.} \bibinfo{year}{2014}\natexlab{}.
\newblock \showarticletitle{A Semantical and Operational Account of
  Call-by-Value Solvability}. In \bibinfo{booktitle}{\emph{Foundations of
  Software Science and Computation Structures - 17th International Conference,
  {FOSSACS} 2014, Grenoble, France, April 5-13, 2014, Proceedings}}.
  \bibinfo{pages}{103--118}.
\newblock
\urldef\tempurl%
\url{https://doi.org/10.1007/978-3-642-54830-7\_7}
\showDOI{\tempurl}


\bibitem[Coppo and Dezani{-}Ciancaglini(1978)]%
        {DBLP:journals/aml/CoppoD78}
\bibfield{author}{\bibinfo{person}{Mario Coppo} {and}
  \bibinfo{person}{Mariangiola Dezani{-}Ciancaglini}.}
  \bibinfo{year}{1978}\natexlab{}.
\newblock \showarticletitle{A new type assignment for {\(\lambda\)}-terms}.
\newblock \bibinfo{journal}{\emph{Arch. Math. Log.}} \bibinfo{volume}{19},
  \bibinfo{number}{1} (\bibinfo{year}{1978}), \bibinfo{pages}{139--156}.
\newblock


\bibitem[Coppo and Dezani{-}Ciancaglini(1980)]%
        {DBLP:journals/ndjfl/CoppoD80}
\bibfield{author}{\bibinfo{person}{Mario Coppo} {and}
  \bibinfo{person}{Mariangiola Dezani{-}Ciancaglini}.}
  \bibinfo{year}{1980}\natexlab{}.
\newblock \showarticletitle{An extension of the basic functionality theory for
  the {\(\lambda\)}-calculus}.
\newblock \bibinfo{journal}{\emph{Notre Dame Journal of Formal Logic}}
  \bibinfo{volume}{21}, \bibinfo{number}{4} (\bibinfo{year}{1980}),
  \bibinfo{pages}{685--693}.
\newblock


\bibitem[Coppo et~al\mbox{.}(1987)]%
        {DBLP:journals/iandc/CoppoDM87}
\bibfield{author}{\bibinfo{person}{Mario Coppo}, \bibinfo{person}{Mariangiola
  Dezani{-}Ciancaglini}, {and} \bibinfo{person}{Maddalena Zacchi}.}
  \bibinfo{year}{1987}\natexlab{}.
\newblock \showarticletitle{Type Theories, Normal Forms and
  $D_{\infty}$-Lambda-Models}.
\newblock \bibinfo{journal}{\emph{Inf. Comput.}} \bibinfo{volume}{72},
  \bibinfo{number}{2} (\bibinfo{year}{1987}), \bibinfo{pages}{85--116}.
\newblock
\urldef\tempurl%
\url{https://doi.org/10.1016/0890-5401(87)90042-3}
\showDOI{\tempurl}


\bibitem[Curien and Herbelin(2000)]%
        {DBLP:conf/icfp/CurienH00}
\bibfield{author}{\bibinfo{person}{Pierre{-}Louis Curien} {and}
  \bibinfo{person}{Hugo Herbelin}.} \bibinfo{year}{2000}\natexlab{}.
\newblock \showarticletitle{The duality of computation}. In
  \bibinfo{booktitle}{\emph{Proceedings of the Fifth {ACM} {SIGPLAN}
  International Conference on Functional Programming {(ICFP} '00), Montreal,
  Canada, September 18-21, 2000}}. \bibinfo{pages}{233--243}.
\newblock
\urldef\tempurl%
\url{https://doi.org/10.1145/351240.351262}
\showDOI{\tempurl}


\bibitem[Dal~Lago et~al\mbox{.}(2021)]%
        {DBLP:journals/pacmpl/LagoFR21}
\bibfield{author}{\bibinfo{person}{Ugo Dal~Lago}, \bibinfo{person}{Claudia
  Faggian}, {and} \bibinfo{person}{Simona Ronchi Della~Rocca}.}
  \bibinfo{year}{2021}\natexlab{}.
\newblock \showarticletitle{Intersection types and (positive) almost-sure
  termination}.
\newblock \bibinfo{journal}{\emph{Proc. {ACM} Program. Lang.}}
  \bibinfo{volume}{5}, \bibinfo{number}{{POPL}} (\bibinfo{year}{2021}),
  \bibinfo{pages}{1--32}.
\newblock
\urldef\tempurl%
\url{https://doi.org/10.1145/3434313}
\showDOI{\tempurl}


\bibitem[de~Carvalho(2007)]%
        {Carvalho07}
\bibfield{author}{\bibinfo{person}{Daniel de Carvalho}.}
  \bibinfo{year}{2007}\natexlab{}.
\newblock \emph{\bibinfo{title}{S\'emantiques de la logique lin\'eaire et temps
  de calcul}}.
\newblock {T}h\`ese de Doctorat. \bibinfo{school}{Universit\'e Aix-Marseille
  II}.
\newblock


\bibitem[de~Carvalho(2018)]%
        {deCarvalho18}
\bibfield{author}{\bibinfo{person}{Daniel de Carvalho}.}
  \bibinfo{year}{2018}\natexlab{}.
\newblock \showarticletitle{Execution time of {\(\lambda\)}-terms via
  denotational semantics and intersection types}.
\newblock \bibinfo{journal}{\emph{Math. Str. in Comput. Sci.}}
  \bibinfo{volume}{28}, \bibinfo{number}{7} (\bibinfo{year}{2018}),
  \bibinfo{pages}{1169--1203}.
\newblock


\bibitem[de~Carvalho et~al\mbox{.}(2011)]%
        {DBLP:journals/tcs/CarvalhoPF11}
\bibfield{author}{\bibinfo{person}{Daniel de Carvalho},
  \bibinfo{person}{Michele Pagani}, {and} \bibinfo{person}{Lorenzo {Tortora de
  Falco}}.} \bibinfo{year}{2011}\natexlab{}.
\newblock \showarticletitle{{A semantic measure of the execution time in linear
  logic}}.
\newblock \bibinfo{journal}{\emph{Theor. Comput. Sci.}} \bibinfo{volume}{412},
  \bibinfo{number}{20} (\bibinfo{year}{2011}), \bibinfo{pages}{1884--1902}.
\newblock


\bibitem[de~Carvalho and {Tortora de Falco}(2016)]%
        {DBLP:journals/iandc/CarvalhoF16}
\bibfield{author}{\bibinfo{person}{Daniel de Carvalho} {and}
  \bibinfo{person}{Lorenzo {Tortora de Falco}}.}
  \bibinfo{year}{2016}\natexlab{}.
\newblock \showarticletitle{A semantic account of strong normalization in
  linear logic}.
\newblock \bibinfo{journal}{\emph{Inf. Comput.}}  \bibinfo{volume}{248}
  (\bibinfo{year}{2016}), \bibinfo{pages}{104--129}.
\newblock


\bibitem[Dyckhoff and Lengrand(2007)]%
        {DBLP:journals/logcom/DyckhoffL07}
\bibfield{author}{\bibinfo{person}{Roy Dyckhoff} {and}
  \bibinfo{person}{St{\'e}phane Lengrand}.} \bibinfo{year}{2007}\natexlab{}.
\newblock \showarticletitle{{Call-by-Value lambda-calculus and {LJQ}}}.
\newblock \bibinfo{journal}{\emph{J. Log. Comput.}} \bibinfo{volume}{17},
  \bibinfo{number}{6} (\bibinfo{year}{2007}), \bibinfo{pages}{1109--1134}.
\newblock


\bibitem[Egidi et~al\mbox{.}(1992)]%
        {DBLP:journals/fuin/EgidiHR92}
\bibfield{author}{\bibinfo{person}{Lavinia Egidi}, \bibinfo{person}{Furio
  Honsell}, {and} \bibinfo{person}{Simona Ronchi Della~Rocca}.}
  \bibinfo{year}{1992}\natexlab{}.
\newblock \showarticletitle{Operational, denotational and logical descriptions:
  a case study}.
\newblock \bibinfo{journal}{\emph{Fundam. Inform.}} \bibinfo{volume}{16},
  \bibinfo{number}{1} (\bibinfo{year}{1992}), \bibinfo{pages}{149--169}.
\newblock


\bibitem[Ehrhard(2012)]%
        {DBLP:conf/csl/Ehrhard12}
\bibfield{author}{\bibinfo{person}{Thomas Ehrhard}.}
  \bibinfo{year}{2012}\natexlab{}.
\newblock \showarticletitle{{Collapsing non-idempotent intersection types}}. In
  \bibinfo{booktitle}{\emph{{CSL}}}. \bibinfo{pages}{259--273}.
\newblock


\bibitem[Flanagan et~al\mbox{.}(1993)]%
        {DBLP:conf/pldi/FlanaganSDF93a}
\bibfield{author}{\bibinfo{person}{Cormac Flanagan}, \bibinfo{person}{Amr
  Sabry}, \bibinfo{person}{Bruce~F. Duba}, {and} \bibinfo{person}{Matthias
  Felleisen}.} \bibinfo{year}{1993}\natexlab{}.
\newblock \showarticletitle{The essence of compiling with continuations (with
  retrospective)}. In \bibinfo{booktitle}{\emph{20 Years of the {ACM} {SIGPLAN}
  Conference on Programming Language Design and Implementation 1979-1999, {A}
  Selection, PLDI 1993}}. \bibinfo{publisher}{{ACM}},
  \bibinfo{pages}{502--514}.
\newblock
\urldef\tempurl%
\url{https://doi.org/10.1145/989393.989443}
\showDOI{\tempurl}


\bibitem[Garc{\'{\i}}a{-}P{\'{e}}rez and Nogueira(2016)]%
        {DBLP:journals/corr/Garcia-PerezN16}
\bibfield{author}{\bibinfo{person}{{\'{A}}lvaro Garc{\'{\i}}a{-}P{\'{e}}rez}
  {and} \bibinfo{person}{Pablo Nogueira}.} \bibinfo{year}{2016}\natexlab{}.
\newblock \showarticletitle{No solvable lambda-value term left behind}.
\newblock \bibinfo{journal}{\emph{Logical Methods in Computer Science}}
  \bibinfo{volume}{12}, \bibinfo{number}{2} (\bibinfo{year}{2016}).
\newblock
\urldef\tempurl%
\url{https://doi.org/10.2168/LMCS-12(2:12)2016}
\showDOI{\tempurl}


\bibitem[Gardner(1994)]%
        {DBLP:conf/tacs/Gardner94}
\bibfield{author}{\bibinfo{person}{Philippa Gardner}.}
  \bibinfo{year}{1994}\natexlab{}.
\newblock \showarticletitle{Discovering Needed Reductions Using Type Theory}.
  In \bibinfo{booktitle}{\emph{TACS '94}} \emph{(\bibinfo{series}{Lecture Notes
  in Computer Science}, Vol.~\bibinfo{volume}{789})}.
  \bibinfo{publisher}{Springer}, \bibinfo{pages}{555--574}.
\newblock
\showISBNx{3-540-57887-0}


\bibitem[Ghilezan(2001)]%
        {DBLP:journals/jcss/Ghilezan01}
\bibfield{author}{\bibinfo{person}{Silvia Ghilezan}.}
  \bibinfo{year}{2001}\natexlab{}.
\newblock \showarticletitle{Full Intersection Types and Topologies in Lambda
  Calculus}.
\newblock \bibinfo{journal}{\emph{J. Comput. Syst. Sci.}} \bibinfo{volume}{62},
  \bibinfo{number}{1} (\bibinfo{year}{2001}), \bibinfo{pages}{1--14}.
\newblock


\bibitem[Girard(1987)]%
        {DBLP:journals/tcs/Girard87}
\bibfield{author}{\bibinfo{person}{Jean-Yves Girard}.}
  \bibinfo{year}{1987}\natexlab{}.
\newblock \showarticletitle{{Linear Logic}}.
\newblock \bibinfo{journal}{\emph{Theoretical Computer Science}}
  \bibinfo{volume}{50} (\bibinfo{year}{1987}), \bibinfo{pages}{1--102}.
\newblock


\bibitem[Girard(1988)]%
        {Girard88}
\bibfield{author}{\bibinfo{person}{Jean-Yves Girard}.}
  \bibinfo{year}{1988}\natexlab{}.
\newblock \showarticletitle{Normal functors, power series and the
  $\lambda$-calculus}.
\newblock \bibinfo{journal}{\emph{Annals of Pure and Applied Logic}}
  \bibinfo{volume}{37} (\bibinfo{year}{1988}), \bibinfo{pages}{129–177}.
\newblock


\bibitem[Gr{\'e}goire and Leroy(2002)]%
        {DBLP:conf/icfp/GregoireL02}
\bibfield{author}{\bibinfo{person}{Benjamin Gr{\'e}goire} {and}
  \bibinfo{person}{Xavier Leroy}.} \bibinfo{year}{2002}\natexlab{}.
\newblock \showarticletitle{{A compiled implementation of strong reduction}}.
  In \bibinfo{booktitle}{\emph{Proceedings of the Seventh {ACM} {SIGPLAN}
  International Conference on Functional Programming, {{ICFP} '02}}}.
  \bibinfo{publisher}{{ACM}}, \bibinfo{pages}{235--246}.
\newblock
\urldef\tempurl%
\url{https://doi.org/10.1145/581478.581501}
\showDOI{\tempurl}


\bibitem[Guerrieri(2015)]%
        {Guerrieri15}
\bibfield{author}{\bibinfo{person}{Giulio Guerrieri}.}
  \bibinfo{year}{2015}\natexlab{}.
\newblock \showarticletitle{{Head reduction and normalization in a
  call-by-value lambda-calculus}}. In \bibinfo{booktitle}{\emph{{{WPTE}
  2015}}}. \bibinfo{pages}{3--17}.
\newblock


\bibitem[Guerrieri(2019)]%
        {Guerrieri18}
\bibfield{author}{\bibinfo{person}{Giulio Guerrieri}.}
  \bibinfo{year}{2019}\natexlab{}.
\newblock \showarticletitle{Towards a Semantic Measure of the Execution Time in
  Call-by-Value lambda-Calculus}. In \bibinfo{booktitle}{\emph{Proceedings
  Twelfth Workshop on Developments in Computational Models and Ninth Workshop
  on Intersection Types and Related Systems, {DCM/ITRS} 2018.}}
  \emph{(\bibinfo{series}{{EPTCS}}, Vol.~\bibinfo{volume}{293})}.
  \bibinfo{pages}{57--72}.
\newblock
\urldef\tempurl%
\url{https://doi.org/10.4204/EPTCS.293.5}
\showDOI{\tempurl}


\bibitem[Guerrieri et~al\mbox{.}(2015)]%
        {GuerrieriPR15}
\bibfield{author}{\bibinfo{person}{Giulio Guerrieri}, \bibinfo{person}{Luca
  Paolini}, {and} \bibinfo{person}{Simona {Ronchi Della Rocca}}.}
  \bibinfo{year}{2015}\natexlab{}.
\newblock \showarticletitle{{Standardization of a Call-By-Value
  Lambda-Calculus}}. In \bibinfo{booktitle}{\emph{{{TLCA} 2015}}}.
  \bibinfo{pages}{211--225}.
\newblock


\bibitem[Guerrieri et~al\mbox{.}(2017)]%
        {DBLP:journals/lmcs/GuerrieriPR17}
\bibfield{author}{\bibinfo{person}{Giulio Guerrieri}, \bibinfo{person}{Luca
  Paolini}, {and} \bibinfo{person}{Simona {Ronchi Della Rocca}}.}
  \bibinfo{year}{2017}\natexlab{}.
\newblock \showarticletitle{Standardization and Conservativity of a Refined
  Call-by-Value lambda-Calculus}.
\newblock \bibinfo{journal}{\emph{Logical Methods in Computer Science}}
  \bibinfo{volume}{13}, \bibinfo{number}{4} (\bibinfo{year}{2017}).
\newblock
\urldef\tempurl%
\url{https://doi.org/10.23638/LMCS-13(4:29)2017}
\showDOI{\tempurl}


\bibitem[Herbelin and Zimmermann(2009)]%
        {DBLP:conf/tlca/HerbelinZ09}
\bibfield{author}{\bibinfo{person}{Hugo Herbelin} {and}
  \bibinfo{person}{St{\'e}phane Zimmermann}.} \bibinfo{year}{2009}\natexlab{}.
\newblock \showarticletitle{{An operational account of Call-by-Value Minimal
  and Classical $\lambda$-calculus in Natural Deduction form}}. In
  \bibinfo{booktitle}{\emph{{TLCA}}}. \bibinfo{pages}{142--156}.
\newblock


\bibitem[Hindley and Longo(1980)]%
        {HindleyLongo80}
\bibfield{author}{\bibinfo{person}{Roger Hindley} {and}
  \bibinfo{person}{Giuseppe Longo}.} \bibinfo{year}{1980}\natexlab{}.
\newblock \showarticletitle{Lambda-Calculus Models and Extensionality}.
\newblock \bibinfo{journal}{\emph{Mathematical Logic Quarterly}}
  \bibinfo{volume}{26}, \bibinfo{number}{19-21} (\bibinfo{year}{1980}),
  \bibinfo{pages}{289--310}.
\newblock
\urldef\tempurl%
\url{https://doi.org/10.1002/malq.19800261902}
\showDOI{\tempurl}


\bibitem[Honsell and Rocca(1992)]%
        {DBLP:journals/jcss/HonsellR92}
\bibfield{author}{\bibinfo{person}{Furio Honsell} {and} \bibinfo{person}{Simona
  Ronchi~Della Rocca}.} \bibinfo{year}{1992}\natexlab{}.
\newblock \showarticletitle{An Approximation Theorem for Topological Lambda
  Models and the Topological Incompleteness of Lambda Calculus}.
\newblock \bibinfo{journal}{\emph{J. Comput. Syst. Sci.}} \bibinfo{volume}{45},
  \bibinfo{number}{1} (\bibinfo{year}{1992}), \bibinfo{pages}{49--75}.
\newblock
\urldef\tempurl%
\url{https://doi.org/10.1016/0022-0000(92)90040-P}
\showDOI{\tempurl}


\bibitem[Kennaway et~al\mbox{.}(1999)]%
        {DBLP:journals/jflp/KennawayOV99}
\bibfield{author}{\bibinfo{person}{Richard Kennaway}, \bibinfo{person}{Vincent
  van Oostrom}, {and} \bibinfo{person}{Fer{-}Jan de Vries}.}
  \bibinfo{year}{1999}\natexlab{}.
\newblock \showarticletitle{Meaningless Terms in Rewriting}.
\newblock \bibinfo{journal}{\emph{J. Funct. Log. Program.}}
  \bibinfo{volume}{1999}, \bibinfo{number}{1} (\bibinfo{year}{1999}).
\newblock


\bibitem[Kerinec et~al\mbox{.}(2021)]%
        {DBLP:conf/fscd/KerinecMR21}
\bibfield{author}{\bibinfo{person}{Axel Kerinec}, \bibinfo{person}{Giulio
  Manzonetto}, {and} \bibinfo{person}{Simona Ronchi Della~Rocca}.}
  \bibinfo{year}{2021}\natexlab{}.
\newblock \showarticletitle{Call-By-Value, Again!}. In
  \bibinfo{booktitle}{\emph{{FSCD}}} \emph{(\bibinfo{series}{LIPIcs},
  Vol.~\bibinfo{volume}{195})}. \bibinfo{publisher}{Schloss Dagstuhl -
  Leibniz-Zentrum f{\"{u}}r Informatik}, \bibinfo{pages}{7:1--7:18}.
\newblock


\bibitem[Kesner et~al\mbox{.}(2021)]%
        {DBLP:conf/fossacs/KesnerPV21}
\bibfield{author}{\bibinfo{person}{Delia Kesner}, \bibinfo{person}{Lo{\"{\i}}c
  Peyrot}, {and} \bibinfo{person}{Daniel Ventura}.}
  \bibinfo{year}{2021}\natexlab{}.
\newblock \showarticletitle{The Spirit of Node Replication}. In
  \bibinfo{booktitle}{\emph{FoSSaCS}} \emph{(\bibinfo{series}{Lecture Notes in
  Computer Science}, Vol.~\bibinfo{volume}{12650})}.
  \bibinfo{publisher}{Springer}, \bibinfo{pages}{344--364}.
\newblock


\bibitem[Kesner and Vial(2020)]%
        {DBLP:conf/lics/KesnerV20}
\bibfield{author}{\bibinfo{person}{Delia Kesner} {and} \bibinfo{person}{Pierre
  Vial}.} \bibinfo{year}{2020}\natexlab{}.
\newblock \showarticletitle{Consuming and Persistent Types for Classical
  Logic}. In \bibinfo{booktitle}{\emph{{LICS} '20: 35th Annual {ACM/IEEE}
  Symposium on Logic in Computer Science, Saarbr{\"{u}}cken, Germany, July
  8-11, 2020}}. \bibinfo{pages}{619--632}.
\newblock
\urldef\tempurl%
\url{https://doi.org/10.1145/3373718.3394774}
\showDOI{\tempurl}


\bibitem[Kesner and Viso(2022)]%
        {DBLP:conf/csl/KesnerV22}
\bibfield{author}{\bibinfo{person}{Delia Kesner} {and}
  \bibinfo{person}{Andr{\'{e}}s Viso}.} \bibinfo{year}{2022}\natexlab{}.
\newblock \showarticletitle{Encoding Tight Typing in a Unified Framework}. In
  \bibinfo{booktitle}{\emph{{CSL}}} \emph{(\bibinfo{series}{LIPIcs},
  Vol.~\bibinfo{volume}{216})}. \bibinfo{publisher}{Schloss Dagstuhl -
  Leibniz-Zentrum f{\"{u}}r Informatik}, \bibinfo{pages}{27:1--27:20}.
\newblock


\bibitem[Kfoury(2000)]%
        {DBLP:journals/logcom/Kfoury00}
\bibfield{author}{\bibinfo{person}{Assaf~J. Kfoury}.}
  \bibinfo{year}{2000}\natexlab{}.
\newblock \showarticletitle{A linearization of the Lambda-calculus and
  consequences}.
\newblock \bibinfo{journal}{\emph{J. Log. Comput.}} \bibinfo{volume}{10},
  \bibinfo{number}{3} (\bibinfo{year}{2000}), \bibinfo{pages}{411--436}.
\newblock


\bibitem[Krivine(1990)]%
        {Kri}
\bibfield{author}{\bibinfo{person}{Jean-Louis Krivine}.}
  \bibinfo{year}{1990}\natexlab{}.
\newblock \bibinfo{booktitle}{\emph{$\lambda$-calcul, types et mod\`eles}}.
\newblock \bibinfo{publisher}{Masson}.
\newblock


\bibitem[Lassen(2005)]%
        {DBLP:conf/lics/Lassen05}
\bibfield{author}{\bibinfo{person}{S{\o}ren~B. Lassen}.}
  \bibinfo{year}{2005}\natexlab{}.
\newblock \showarticletitle{{Eager Normal Form Bisimulation}}. In
  \bibinfo{booktitle}{\emph{20th {IEEE} Symposium on Logic in Computer Scienc,
  {{LICS} 2005}}}. \bibinfo{publisher}{{IEEE} Computer Society},
  \bibinfo{pages}{345--354}.
\newblock
\urldef\tempurl%
\url{https://doi.org/10.1109/LICS.2005.15}
\showDOI{\tempurl}


\bibitem[Manzonetto et~al\mbox{.}(2019)]%
        {DBLP:journals/fuin/ManzonettoPR19}
\bibfield{author}{\bibinfo{person}{Giulio Manzonetto}, \bibinfo{person}{Michele
  Pagani}, {and} \bibinfo{person}{Simona Ronchi Della~Rocca}.}
  \bibinfo{year}{2019}\natexlab{}.
\newblock \showarticletitle{New Semantical Insights Into Call-by-Value
  {\(\lambda\)}-Calculus}.
\newblock \bibinfo{journal}{\emph{Fundam. Inform.}} \bibinfo{volume}{170},
  \bibinfo{number}{1-3} (\bibinfo{year}{2019}), \bibinfo{pages}{241--265}.
\newblock
\urldef\tempurl%
\url{https://doi.org/10.3233/FI-2019-1862}
\showDOI{\tempurl}


\bibitem[Maraist et~al\mbox{.}(1999)]%
        {DBLP:journals/tcs/MaraistOTW99}
\bibfield{author}{\bibinfo{person}{John Maraist}, \bibinfo{person}{Martin
  Odersky}, \bibinfo{person}{David~N. Turner}, {and} \bibinfo{person}{Philip
  Wadler}.} \bibinfo{year}{1999}\natexlab{}.
\newblock \showarticletitle{{Call-by-name, Call-by-value, Call-by-need and the
  Linear $\lambda$-Calculus}}.
\newblock \bibinfo{journal}{\emph{TCS}} \bibinfo{volume}{228},
  \bibinfo{number}{1-2} (\bibinfo{year}{1999}), \bibinfo{pages}{175--210}.
\newblock


\bibitem[Mazza et~al\mbox{.}(2018)]%
        {MazzaPellissierVial18}
\bibfield{author}{\bibinfo{person}{Damiano Mazza}, \bibinfo{person}{Luc
  Pellissier}, {and} \bibinfo{person}{Pierre Vial}.}
  \bibinfo{year}{2018}\natexlab{}.
\newblock \showarticletitle{Polyadic Approximations, Fibrations and
  Intersection Types}.
\newblock \bibinfo{journal}{\emph{Proceedings of the ACM on Programming
  Languages}} \bibinfo{volume}{2}, \bibinfo{number}{POPL:6}
  (\bibinfo{year}{2018}).
\newblock


\bibitem[Moggi(1988)]%
        {Moggi88tech}
\bibfield{author}{\bibinfo{person}{Eugenio Moggi}.}
  \bibinfo{year}{1988}\natexlab{}.
\newblock \bibinfo{booktitle}{\emph{{Computational $\lambda$-Calculus and
  Monads}}}.
\newblock \bibinfo{type}{LFCS report} ECS-LFCS-88-66.
  \bibinfo{institution}{University of {E}dinburgh}.
\newblock
\urldef\tempurl%
\url{http://www.lfcs.inf.ed.ac.uk/reports/88/ECS-LFCS-88-66/ECS-LFCS-88-66.pdf}
\showURL{%
\tempurl}


\bibitem[Neergaard and Mairson(2004)]%
        {DBLP:conf/icfp/NeergaardM04}
\bibfield{author}{\bibinfo{person}{Peter~M{\o}ller Neergaard} {and}
  \bibinfo{person}{Harry~G. Mairson}.} \bibinfo{year}{2004}\natexlab{}.
\newblock \showarticletitle{Types, potency, and idempotency: why nonlinearity
  and amnesia make a type system work}. In \bibinfo{booktitle}{\emph{{ICFP}
  2004}}. \bibinfo{pages}{138--149}.
\newblock


\bibitem[Paolini(2001)]%
        {DBLP:conf/ictcs/Paolini01}
\bibfield{author}{\bibinfo{person}{Luca Paolini}.}
  \bibinfo{year}{2001}\natexlab{}.
\newblock \showarticletitle{Call-by-Value Separability and Computability}. In
  \bibinfo{booktitle}{\emph{Theoretical Computer Science, 7th Italian
  Conference, {ICTCS} 2001, Torino, Italy, October 4-6, 2001, Proceedings}}.
  \bibinfo{pages}{74--89}.
\newblock
\urldef\tempurl%
\url{https://doi.org/10.1007/3-540-45446-2\_5}
\showDOI{\tempurl}


\bibitem[Paolini and Ronchi Della~Rocca(1999)]%
        {DBLP:journals/ita/PaoliniR99}
\bibfield{author}{\bibinfo{person}{Luca Paolini} {and} \bibinfo{person}{Simona
  Ronchi Della~Rocca}.} \bibinfo{year}{1999}\natexlab{}.
\newblock \showarticletitle{Call-by-value Solvability}.
\newblock \bibinfo{journal}{\emph{{RAIRO} Theor. Informatics Appl.}}
  \bibinfo{volume}{33}, \bibinfo{number}{6} (\bibinfo{year}{1999}),
  \bibinfo{pages}{507--534}.
\newblock
\urldef\tempurl%
\url{https://doi.org/10.1051/ita:1999130}
\showDOI{\tempurl}


\bibitem[Pitts(2012)]%
        {DBLP:books/cu/12/Pitts12}
\bibfield{author}{\bibinfo{person}{Andrew~M. Pitts}.}
  \bibinfo{year}{2012}\natexlab{}.
\newblock \showarticletitle{Howe's method for higher-order languages}.
\newblock In \bibinfo{booktitle}{\emph{Advanced Topics in Bisimulation and
  Coinduction}}, \bibfield{editor}{\bibinfo{person}{Davide Sangiorgi} {and}
  \bibinfo{person}{Jan J. M.~M. Rutten}} (Eds.). \bibinfo{series}{Cambridge
  tracts in theoretical computer science}, Vol.~\bibinfo{volume}{52}.
  \bibinfo{publisher}{Cambridge University Press}, \bibinfo{pages}{197--232}.
\newblock


\bibitem[Plotkin(1975)]%
        {DBLP:journals/tcs/Plotkin75}
\bibfield{author}{\bibinfo{person}{Gordon~D. Plotkin}.}
  \bibinfo{year}{1975}\natexlab{}.
\newblock \showarticletitle{{Call-by-Name, Call-by-Value and the
  lambda-Calculus}}.
\newblock \bibinfo{journal}{\emph{Theoretical Computer Science}}
  \bibinfo{volume}{1}, \bibinfo{number}{2} (\bibinfo{year}{1975}),
  \bibinfo{pages}{125--159}.
\newblock
\urldef\tempurl%
\url{https://doi.org/10.1016/0304-3975(75)90017-1}
\showDOI{\tempurl}


\bibitem[Plotkin(1993)]%
        {DBLP:journals/tcs/Plotkin93}
\bibfield{author}{\bibinfo{person}{Gordon~D. Plotkin}.}
  \bibinfo{year}{1993}\natexlab{}.
\newblock \showarticletitle{Set-Theoretical and Other Elementary Models of the
  lambda-Calculus}.
\newblock \bibinfo{journal}{\emph{Theor. Comput. Sci.}} \bibinfo{volume}{121},
  \bibinfo{number}{1{\&}2} (\bibinfo{year}{1993}), \bibinfo{pages}{351--409}.
\newblock
\urldef\tempurl%
\url{https://doi.org/10.1016/0304-3975(93)90094-A}
\showDOI{\tempurl}


\bibitem[Pottinger(1980)]%
        {Pottinger80}
\bibfield{author}{\bibinfo{person}{Garrel Pottinger}.}
  \bibinfo{year}{1980}\natexlab{}.
\newblock \showarticletitle{A type assignment for the strongly normalizable
  $\lambda$-terms}. In \bibinfo{booktitle}{\emph{To HB Curry: essays on
  combinatory logic, $\lambda$-calculus and formalism}}.
  \bibinfo{pages}{561--577}.
\newblock


\bibitem[Pravato et~al\mbox{.}(1999)]%
        {DBLP:journals/mscs/PravatoRR99}
\bibfield{author}{\bibinfo{person}{Alberto Pravato}, \bibinfo{person}{Simona
  {Ronchi Della Rocca}}, {and} \bibinfo{person}{Luca Roversi}.}
  \bibinfo{year}{1999}\natexlab{}.
\newblock \showarticletitle{{The call-by-value $\l$-calculus: a semantic
  investigation}}.
\newblock \bibinfo{journal}{\emph{Math. Str. in Comput. Sci.}}
  \bibinfo{volume}{9}, \bibinfo{number}{5} (\bibinfo{year}{1999}),
  \bibinfo{pages}{617--650}.
\newblock


\bibitem[{Ronchi Della Rocca} and Paolini(2004)]%
        {parametricBook}
\bibfield{author}{\bibinfo{person}{Simona {Ronchi Della Rocca}} {and}
  \bibinfo{person}{Luca Paolini}.} \bibinfo{year}{2004}\natexlab{}.
\newblock \bibinfo{booktitle}{\emph{{The Parametric $\l$-Calculus -- A
  Metamodel for Computation}}}.
\newblock \bibinfo{publisher}{Springer}.
\newblock
\urldef\tempurl%
\url{https://doi.org/10.1007/978-3-662-10394-4}
\showDOI{\tempurl}


\bibitem[Sabry and Felleisen(1993)]%
        {DBLP:journals/lisp/SabryF93}
\bibfield{author}{\bibinfo{person}{Amr Sabry} {and} \bibinfo{person}{Matthias
  Felleisen}.} \bibinfo{year}{1993}\natexlab{}.
\newblock \showarticletitle{{Reasoning about Programs in Continuation-Passing
  Style}}.
\newblock \bibinfo{journal}{\emph{Lisp and Symbolic Computation}}
  \bibinfo{volume}{6}, \bibinfo{number}{3-4} (\bibinfo{year}{1993}),
  \bibinfo{pages}{289--360}.
\newblock


\bibitem[Sabry and Wadler(1997)]%
        {DBLP:journals/toplas/SabryW97}
\bibfield{author}{\bibinfo{person}{Amr Sabry} {and} \bibinfo{person}{Philip
  Wadler}.} \bibinfo{year}{1997}\natexlab{}.
\newblock \showarticletitle{{A Reflection on Call-by-Value}}.
\newblock \bibinfo{journal}{\emph{{ACM} Trans. Program. Lang. Syst.}}
  \bibinfo{volume}{19}, \bibinfo{number}{6} (\bibinfo{year}{1997}),
  \bibinfo{pages}{916--941}.
\newblock


\bibitem[Wadsworth(1971)]%
        {Wad:SemPra:71}
\bibfield{author}{\bibinfo{person}{Christopher~P. Wadsworth}.}
  \bibinfo{year}{1971}\natexlab{}.
\newblock \emph{\bibinfo{title}{{Semantics and pragmatics of the
  lambda-calculus}}}.
\newblock {Ph{D} Thesis}. \bibinfo{school}{Oxford}.
\newblock
\newblock
\shownote{Chapter 4}.


\bibitem[Wadsworth(1976)]%
        {DBLP:journals/siamcomp/Wadsworth76}
\bibfield{author}{\bibinfo{person}{Christopher~P. Wadsworth}.}
  \bibinfo{year}{1976}\natexlab{}.
\newblock \showarticletitle{{The Relation Between Computational and
  Denotational Properties for Scott's $D_\infty$-Models of the
  Lambda-Calculus}}.
\newblock \bibinfo{journal}{\emph{SIAM J. Comput.}} \bibinfo{volume}{5},
  \bibinfo{number}{3} (\bibinfo{year}{1976}), \bibinfo{pages}{488--521}.
\newblock


\end{thebibliography}

\clearpage
\appendix

\section*{Technical Appendix}

\section{Counterexamples}
\label{sect:counter}

\subsection{Counterexample to subject reduction and expansion in the type system used by Paolini and Ronchi Della Rocca \cite{DBLP:journals/ita/PaoliniR99}}

In \cite[Definitions 6.1--6.2]{DBLP:journals/ita/PaoliniR99}, the idempotent intersection type system introduced to characterize \cbv solvability is defined as follows.

\emph{Types} and \emph{intersection types} are defined by mutual induction according to the grammar below, where $\alpha$ and $\nu$ are two distinct constants, and $\{\sigma_1, \dots, \sigma_n\}$ is a non-empty finite set of types:
\begin{align*}
\text{types} \qquad \sigma, \tau &\Coloneqq \alpha \mid \nu \mid S \Rightarrow \tau 
& &&
\text{intersection types} \qquad S &\Coloneqq \{\sigma_1, \dots, \sigma_n\}  \qquad (n \geq 1)
\end{align*}

An \emph{environment} $B$ is a (total) function mapping variables to finite sets of types such that $\dom{B} = \{\var \mid B(\var) \neq \emptyset\}$ is finite.
We write $B = \var_1 : S_1, \dots, \var_n : S_n$ if $\dom{B} = \{\var_1, \dots, \var_n\}$ and $\var_1, \dots, \var_n$ are pairwise disjoint. 
Given two environments $B$ and $B'$, we write $B \cup B'$ for their pointwise union, \ie, $(B \cup B')(x) = B(x) \cup B'(x)$ for every variable $x$.

The inference rules of the type system are the following (see \cite[Definition 6.2]{DBLP:journals/ita/PaoliniR99}):\footnote{In \cite[Definition 6.2]{DBLP:journals/ita/PaoliniR99}, the rule $\Rightarrow_{\nu E}$ is not included, but it is needed otherwise the \cbv solvble term $(\la{\varthree}\var)\la{\vartwo}\Omega$ (with $\Omega \defeq \delta\delta$ and $\delta \defeq \la{\var}\var$) would not be typable.}

\begin{small}
\begin{gather*}
\!\!\!\!\!\!\!\!
\begin{prooftree}[label separation=0.1em]
	\infer0[\footnotesize$\text{var}$]{\var: \{\sigma\} \vdash \var : \sigma}
\end{prooftree}
\quad
\begin{prooftree}[separation=1em, label separation=0.1em]
	\hypo{B \vdash \tm : \{\sigma_1, \dots, \sigma_n\} \Rightarrow \tau}
	\hypo{(B_i \vdash \tmtwo : \sigma_i)_{1 \leq i \leq n}}
	\hypo{n \geq 1}
	\infer3[\footnotesize$\Rightarrow_E$]{B \cup \bigcup_{i=1}^n B_i \vdash  \tm\tmtwo : \tau}
\end{prooftree}
\quad
\begin{prooftree}[separation=1em, label separation=0.1em]
	\hypo{B \vdash \tm : \{\nu\} \Rightarrow \tau}
	\hypo{B' \vdash \tmtwo : \nu}
	\infer2[$\Rightarrow_{\nu E}$]{B \cup B' \vdash  \tm\tmtwo : \tau}
\end{prooftree}
\\
\begin{prooftree}[label separation=0.1em]
\infer0[\footnotesize$\nu$]{\vdash \la{\var}{\tm} : \nu}
\end{prooftree}
\quad
\begin{prooftree}[separation=1em, label separation=0.1em]
\hypo{B \vdash \tm : \tau}
\hypo{\var \notin \dom{B}}
\infer2[\footnotesize$\Rightarrow_{\nu I}$]{B \vdash \la{\var}{\tm} : \{\nu\} \Rightarrow \tau}
\end{prooftree}
\quad
\begin{prooftree}[separation=1em, label separation=0.1em]
\hypo{B \vdash \tm : \tau}
\hypo{\var \notin \dom{B}}
\infer2[\footnotesize$\Rightarrow_{0 I}$]{B \vdash \la{\var}{\tm} : \tau}
\end{prooftree}	
\quad
\begin{prooftree}[separation=1em, label separation=0.1em]
\hypo{B, x : S \vdash \tm : \tau}
\infer1[\footnotesize$\Rightarrow_{I}$]{B \vdash \la{\var}{\tm} : S \Rightarrow \tau}
\end{prooftree}
\end{gather*}
\end{small}

Let $\tm \defeq \la{\var}(\la{\varthree}{\var})({\var\var)}$ and $\tm' \defeq \la{\var}\var$.
According to the reduction defined in \cite{DBLP:journals/ita/PaoliniR99} to characterize operationally \cbv solvability, we have $\tm \to \tm'$.

The only possible judgments for $\tm$ are $\vdash \tm : \nu$ and $\vdash \tm : \{\{\sigma\} \Rightarrow \nu, \sigma, \tau\} \Rightarrow \tau$, for every types $\sigma, \tau$.
Indeed, 
\begin{align*}
\!\!\!\!\!\!\!
\begin{prooftree}
\infer0[$\nu$]{\vdash \la{\var}(\la{\varthree}\var)(\var\var) : \nu}
\end{prooftree}
&&
\begin{prooftree}[separation=1em, label separation=0.1em]
	\infer0[\footnotesize$\text{var}$]{x : \{\tau\} \vdash \var : \tau}
	\infer1[\footnotesize$\Rightarrow_{\nu I}$]{x : \{\tau\} \vdash \la{\varthree}\var : \{\nu\} \Rightarrow \tau}
	\infer0[\footnotesize$\text{var}$]{x : \{\{\sigma\} \Rightarrow \nu\} \vdash x : \{\sigma\} \Rightarrow \nu}
	\infer0[\footnotesize$\text{var}$]{x : \{\sigma\} \vdash x : \sigma}
	\infer2[\footnotesize$\Rightarrow_E$]{x : \{\{\sigma\} \Rightarrow \nu, \sigma\} \vdash \var\var : \nu}
	\infer2[\footnotesize$\Rightarrow_{\nu E}$]{x : \{\{\sigma\} \Rightarrow \nu, \sigma, \tau\} \vdash (\la{\varthree}\var)(\var\var) : \tau}
	\infer1[\footnotesize$\Rightarrow_I$]{\vdash \la{\var}(\la{\varthree}\var)(\var\var) : \{\{\sigma\} \Rightarrow \nu, \sigma, \tau\} \Rightarrow \tau}
\end{prooftree}
\end{align*}
and no other judgments are derivable for $\tm$.

The only possible judgments for $\tm'$ are $\vdash \tm' : \nu$ and $\vdash \tm' : \{\tau\} \Rightarrow \tau$, for every type $\tau$.
Indeed, 
\begin{align*}
\begin{prooftree}
\infer0[$\nu$]{\vdash \la{\var}(\la{\varthree}\var)(\var\var) : \nu}
\end{prooftree}
&&
\begin{prooftree}
\infer0[$\text{var}$]{x : \{\tau\} \vdash \var : \tau}
\infer1[$\Rightarrow_{\nu I}$]{\vdash \la{\var}\var : \{\tau\} \Rightarrow \tau}
\end{prooftree}
\end{align*}
and no other judgments are derivable for $\tm'$.

Summing up, subject reduction (as it is stated in \cite[Lemma 6.4]{DBLP:journals/ita/PaoliniR99}) does not hold because $\vdash \tm : \{\{\sigma\} \Rightarrow \nu, \sigma, \tau\} \Rightarrow \tau$ but there is no environment $B$ such that $B \vdash \tm' : \{\{\sigma\} \Rightarrow \nu, \sigma, \tau\} \Rightarrow \tau$.
Subject expansion (as it is stated in \cite[Lemma 6.6]{DBLP:journals/ita/PaoliniR99}) does not hold because $\vdash \tm' : \{\tau\} \Rightarrow \tau$ for every type $\tau$, but there is no environment $B$ and no type $\tau$ such that $B \vdash \tm : \{\tau\} \Rightarrow \tau$.
Note that $\tau$ can be \emph{any} type, in particular any proper type, so subject reduction and expansion do not even hold when restricted to the type system used to characterize \cbv solvability \cite[Theorem 6.8]{DBLP:journals/ita/PaoliniR99}.

\subsection{Counterexample to subject reduction in the type system used by Kerinec et al. \cite{DBLP:conf/fscd/KerinecMR21} (due to Delia Kesner)}
\label{ssect:counter-example-kerinec}

In \cite[Definition 9, Figure 1]{DBLP:conf/fscd/KerinecMR21}, the non-idempotent intersection type system used to characterize \cbv solvability is defined as follows.
The type system was first introduced in \cite{DBLP:journals/fuin/ManzonettoPR19}. 

Given a countable set of constants $a, b, c, \dots$, \emph{types} and \emph{multiset types} are defined by mutual induction according to the grammar below, where $\mset{\alpha_1, \dots, \alpha_n}$ with $n \geq 0$ is a (possibly empty) finite multiset of types:
\begin{align*}
\text{types} \ \ \alpha, \beta &\Coloneqq a \mid \mset{\,} \mid M \Rightarrow \alpha 
&&& &&
\text{multiset types} \ \ M &\Coloneqq \mset{\alpha_1, \dots, \alpha_n}  \ \ (n \geq 0, \ \alpha_i \neq \mset{\,} \text{ for all } 1 \leq i \leq n)
\end{align*}

An \emph{environment} $\typctx$ is a (total) function mapping variables to multiset types such that $\dom{\typctx} = \{\var \mid \typctx(\var) \neq \mset{\,}\}$ is finite.
We write $\typctx = \var_1 : S_1, \dots, \var_n : S_n$ if $\dom{\typctx} \subseteq \{\var_1, \dots, \var_n\}$ and $\var_1, \dots, \var_n$ are pairwise disjoint. 
Given two environments $\typctx$ and $\typctx'$, we write $\typctx + \typctx'$ for their pointwise multiset union, \ie, $(\typctx + \typctx')(x) = \typctx(x) + \typctx'(x)$ for every variable $x$.

The inference rules of the type system are the following (see \cite[Definition 10]{DBLP:conf/fscd/KerinecMR21}):
\begin{gather*}
	\begin{prooftree}
		\infer0[$\text{var}$]{\var: \mset{\alpha} \vdash \var : \alpha}
	\end{prooftree}
	\qquad
	\begin{prooftree}
		\hypo{\typctx \vdash \tm : M \Rightarrow \alpha}
		\hypo{\typctx' \vdash \tmtwo : M}
	\infer2[$\text{app}$]{\Gamma + \Gamma' \vdash  \tm\tmtwo : \alpha}
	\end{prooftree}
		\qquad
	\begin{prooftree}
	\hypo{\typctx, x : M \vdash \tm : \alpha}
	\infer1[$\text{lam}$]{\typctx \vdash \la{\var}{\tm} : M \Rightarrow \alpha}
	\end{prooftree}
	\\
	\begin{prooftree}
		\hypo{\val \text{ variable or abstraction}}
		\infer1[$\text{val}_0$]{\vdash \val : \mset{\,}}
	\end{prooftree}
	\qquad
	\begin{prooftree}
		\hypo{\typctx_1 \vdash \tm : \alpha_1}
		\hypo{\cdots}
		\hypo{\typctx_n \vdash \tm : \alpha_n}
		\hypo{n > 0}
		\infer4[$\text{val}_{>0}$]{\sum_{i=1}^n \typctx_i \vdash \tm : \mset{\alpha_1, \dots, \alpha_n}}
	\end{prooftree}	
\end{gather*}

Let $\tm \defeq w((\lambda x.w')(zy))$ and $\tm' \defeq (\lambda x.ww')(zy)$.
According to the reduction defined in \cite{DBLP:conf/fscd/KerinecMR21} to characterize operationally \cbv solvability, we have $\tm \to \tm'$ (via rule $\sigma_{3}$).

The judgment $w:[\mset{a_1,a_2} \Rightarrow \alpha], z:\mset{\mset{b_1} \Rightarrow \mset{}, \mset{b_2} \Rightarrow \mset{}}, y:\mset{b_1,b_2}, w': \mset{a_1,a_2} \vdash \tm : \alpha$ is derivable for $\tm$, for every type $\alpha$ and every pairwise distinct constants $a_1, a_2, b_1, b_2$.
Indeed, 
\begin{center}
	\small
	\!\!\!\!\!\!\!\!\!\!\!\!\!
	\begin{prooftree}[label separation=0.1em, separation =1.0em]
		\infer0[\footnotesize{$\text{var}$}]{w\!:\!\mset{\mset{a_1,a_2} \Rightarrow \alpha} \vdash w \!:\! \mset{a_1,a_2} \Rightarrow \alpha}
		\hypo{}
		\ellipsis{$\Pi_1$}{\typctx_1 \vdash (\lambda x.w')(zy) \!:\! a_1}
		\hypo{}
		\ellipsis{$\Pi_2$}{\typctx_2 \vdash (\lambda x.w')(zy) \!:\! a_2}
		\infer[separation=5em]2[\footnotesize{$\text{val}_{>0}$}]{z \!:\! [\mset{b_1} \Rightarrow \mset{}, \mset{b_2} \Rightarrow \mset{}], y \!:\! \mset{b_1,b_2}, w' \!:\! \mset{a_1,a_2} \vdash (\lambda x.w')(zy) \!:\! \mset{a_1,a_2}}
		\infer2[\footnotesize{$\text{app}$}]{w \!:\! [\mset{a_1,a_2} \Rightarrow \alpha], z \!:\! \mset{\mset{b_1} \Rightarrow \mset{}, \mset{b_2} \Rightarrow \mset{}}, y \!:\! \mset{b_1,b_2}, w' \!:\! \mset{a_1,a_2} \vdash w((\lambda x.w')(zy)) \!:\! \alpha}
	\end{prooftree}
\end{center}
\noindent where, for $i \in \{1,2\}$, we set  $\typctx_i \defeq z:[\mset{b_i} \Rightarrow \mset{}], y:\mset{b_i}, w': \mset{a_i}$ (hence $\typctx_1+ \typctx_2 = z:[\mset{b_1} \Rightarrow \mset{}, \mset{b_2} \Rightarrow \mset{}], y:\mset{b_1,b_2}, w': \mset{a_1,a_2}$)  and  
\begin{gather*}
\!\!\!\!
	\Pi_i \defeq 
	\begin{prooftree}[label separation=0.1em]
	\infer0[\footnotesize$\text{var}$]{w': \mset{a_i} \vdash w' : a_i}
	\infer1[\footnotesize$\text{lam}$]{w': \mset{a_i} \vdash \lambda x.w' : \mset{} \Rightarrow a_i}
	\infer0[\footnotesize$\text{var}$]{z:\mset{\mset{b_i} \Rightarrow \mset{}} \vdash z : \mset{b_i} \Rightarrow \mset{}}
	\infer0[\footnotesize$\text{var}$]{y:\mset{b_i} \vdash y : b_i}
	\infer1[\footnotesize$\text{val}_{>0}$]{y:\mset{b_i} \vdash y : \mset{b_i}}
	\infer2[\footnotesize$\text{app}$]{z:\mset{\mset{b_i} \Rightarrow \mset{}}, y:\mset{b_i} \vdash zy : \mset{}}
	\infer2[\footnotesize$\text{app}$]{z:[\mset{b_i} \Rightarrow \mset{}], y:\mset{b_i}, w': \mset{a_i} \vdash (\lambda x.w')(zy) : a_i}
	\end{prooftree}
\end{gather*}

But the judgment $w:[\mset{a_1,a_2} \Rightarrow \alpha], z:\mset{\mset{b_1} \Rightarrow \mset{}, \mset{b_2} \Rightarrow \mset{}}, y:\mset{b_1,b_2}, w': \mset{a_1,a_2} \vdash \tm' : \alpha$ is not derivable for $\tm'$, essentially because it is impossible to derive the judgment $y : \mset{b_1, b_2}, z : \mset{\mset{b_1} \Rightarrow \mset{}, \mset{b_2} \Rightarrow \mset{}} \vdash zy : \mset{}$.
Hence, subject reduction \cite[Proposition 14.i]{DBLP:conf/fscd/KerinecMR21} does not hold.
Note that $\alpha$ can be \emph{any} type, in particular any proper type, so subject reduction does not even hold when restricted to the type system used to characterize \cbv solvability \cite[Theorem~36]{DBLP:conf/fscd/KerinecMR21}.
\section{Moggi and the Glueing Rule} 
\label{sect:app-moggi}
Here we discuss the relationship between the VSC and Moggi computational $\l$-calculus $\moggicalc$ \cite{Moggi88tech}, which is based on a irrelevant extension of the VSC. Such an extension allows to recast in \cbv the exact same proof that in \cbn gives the implication \textsc{SOL-ID} $\Rightarrow$ \textsc{SOL-FE} in the equivalence  of definitions of solvability in \Cref{ssect:equiv-defs} (but we provided an alternative proof that does not need $\moggicalc$).

Moggi's computational $\l$-calculus $\moggicalc$ \cite{Moggi88tech} is also subsumed by VSC, via a further extension. Adapted to our syntax (so seeing $\letin \var\tmtwo\tm$ as $\tm\esub\var\tmtwo$), the (root) rewriting rules of $\moggicalc$ are  (using $\aptm$ for applications, and forgetting $\eta$-equivalence):
\begin{center}
$\begin{array}{rll@{\hspace{1cm}}rll@{\hspace{.5cm}} l}
  (\la\var\tm) \val & \rootRew{\beta_{\val}} & \tm \isub\var{\val}
&
  \var\esub\var\tm & \rootRew{\mathsf{id}} & \tm
  \\
  \tm\esub\var{\val} & \rootRew{\letexp_{\val}} & \tm\isub\var{\val}
  &
  \aptm \tm &  \rootRew{\letexp1} & (\var\tm)\esub\var\aptm & \mbox{$\var$ fresh}  
  \\
  	\tm\esub{\var}{\tmtwo\esub{\vartwo}{\tmthree}} & \rootRew{\mathsf{comp}} 
&\tm\esub{\var}{\tmtwo}\esub{\vartwo}{\tmthree} 
&
  \val \aptm &  \rootRew{\letexp2} & (\val\var)\esub\var\aptm & \mbox{$\var$ fresh}  
\end{array}$
\end{center}
Note that rule $\rootRew{\letexp_{\val}}$ is simulated in \VSC by $\toe$, while $\tobvplot$ by $\tom\cdot\toe$ (as in \Cref{prop:plotkin-vsc}), and $\rootRew{\letexp_{\val}}$ is simply $\tostructes$ oriented from left to right. 

The VSC, as well as other \cbv calculi, do not include Moggi's rule 
$\Rew{\mathsf{id}}$ corresponding to the \emph{unit law for monads}, because $\Rew{\mathsf{id}}$ seems to have a \cbn flavor, since it applies also when $\tm$ is not a value.
However, $\Rew{\mathsf{id}}$ does \emph{not duplicate or erase} $\tm$: in this respect it is compatible with the \cbv paradigm.
Moreover, note that  $\Rew{\letexp1}$, 
and $\Rew{\letexp2}$, if considered in the \emph{opposite} direction, do something very similar to $\Rew{\mathsf{id}}$: a substitution of a \emph{non-value} for a variable that occurs \emph{linearly} in an \emph{open} context.

 It is possible to extend the VSC with a single \emph{glueing} rule capturing at the same time
$\Rew{\mathsf{id}}$ and the opposite versions of $\Rew{\letexp1}$, and $\Rew{\letexp2}$:
\begin{center}
	$\begin{array}{c@{\hspace{1.3cm}}c}
	\textsc{Root rule}
&
	\textsc{Contextual closure}
	\\
	\weakctxp\var\esub\var\aptm  \rtoglue    \weakctxp\aptm  \ \ \  \textrm{~~~if } \var\notin\fv\weakctx
&
	\fctxp \tm \toglue \fctxp \tmtwo \ \ \ \textrm{~~~if } \tm \rtoglue \tmtwo 
	\end{array}$
\end{center}

The definition of the rule implicitly assumes that $\weakctx$ does not capture $\var$. Together with $\var\notin\fv\weakctx$, it means that there is exactly one free occurrence of $\var$ in $\weakctxp{\var}$.
Rule $\toglue$ is restricted to applications $\aptm$ because for values $\toe$ suffices. 
\begin{proposition}[Moggi $\subseteq \VSC\ + \toglue +\eqstruct$]
The equational theory of Moggi's $\moggicalc$ is contained in the one of the \VSC extended with $\toglue$ and $\eqstruct$.
\end{proposition}

Note that in the \VSC extended with $\toglue$, the implication \textsc{SOL-ID} $\Rightarrow$ \textsc{SOL-FE} (see \Cref{ssect:equiv-defs}) can also be proved exactly as in \cbn, since  $\Id \tmtwo \tom \var\esub\var\tmtwo \toglue \tmtwo$.

\emph{Is rule $\toglue$ ad-hoc?} No. From a linear logic point of view, 
$\toglue$ is natural: it corresponds exactly to $\eta$-contraction of exponential axioms (in the proof 
nets representation of the VSC \cite{DBLP:journals/tcs/Accattoli15}). Denotationally, the \cbv relational semantics induced by multi types is invariant by $\toglue$, as we shall prove. In the realm of program transformations, $\toglue$ simply inverts the transformation into \emph{administrative normal forms} for pure $\l$-terms, see \cite{DBLP:conf/pldi/FlanaganSDF93a,DBLP:journals/lisp/SabryF93,DBLP:conf/ppdp/AccattoliCGC19}.  
\emph{Does $\toglue$ break properties of the \VSC?} No. As for the substitution of variables, gluing is \emph{irrelevant} in the VSC. This nicely explains why most \cbv calculi do not include it.

\begin{proposition}[$\toglue$ Irrelevance]
	\label{prop:glue-irrelevance}
	\NoteProof{propappendix:glue-irrelevance}
	The reduction $\toglue$ is $\tovsub$-irrelevant.
\end{proposition}

\begin{proposition}[$\eqstruct$ is a strong bisimulation for $\toglue$]
	\label{prop:strong-bisimulation-glue}
	\NoteProof{propappendix:strong-bisimulation-glue} 
	If $\tm\eqstruct\tmtwo$ and $\tm\toglue\tmp$ then there exists $\tmtwop \in \vsubterms$ such that 
	$\tmtwo\toglue\tmtwop$ and $\tmp\eqstruct\tmtwop$. 
\end{proposition}    

About confluence, there is a glitch, as the interaction of $\toglue$ and $\tovsub$ behaves well only modulo structural equivalence. Consider for instance the following diagram:
\begin{center}
$\begin{array}{ccccccccc}
((\la\var\varthree)\esub\varthree\vartwo\tm)\esub\vartwo{\tmtwo\tmthree} & \toglue ((\la\var\varthree)\esub\varthree{\tmtwo\tmthree}\tm)\\
\downarrow_{\evarsym} 
\\
((\la\var\vartwo)\tm)\esub\vartwo{\tmtwo\tmthree}
\end{array}$
\end{center}
Note that the diagram cannot be closed via reduction because in the bottom term $\vartwo$ is not at the same 'abstraction level' as $\esub\vartwo{\tmtwo\tmthree}$. It can be closed using structural equivalence, as follows:
\[((\la\var\varthree)\esub\varthree{\tmtwo\tmthree}\tm) =_{\alpha} ((\la\var\vartwo)\esub\vartwo{\tmtwo\tmthree}\tm) \streq ((\la\var\vartwo)\tm)\esub\vartwo{\tmtwo\tmthree}\]
The proof of confluence of $\tovsub\cup\toglue$ modulo $\streq$ is left to future work, because of the tedious technicalities of working modulo $\streq$, which prevent to use standard rewriting lemmas.

\section{Proofs of Section \ref{sect:preliminaries} (Preliminaries in rewriting)}
\label{sect:preliminaries-proofs}

\begin{lemma}[Swaps, abstractly]
	\label{l:swaps-abstract}
	Let $\to_{1}$ and $\to_{2}$ be two binary relations on a same set. 
	If $\to_{1} \cup \to_{2}$ is strongly normalizing and $\to_{1}\to_{2} \,\subseteq\, \to_{2}\to_{1}^* \!\cup \to_{2}\to_{2}$, then $\to_{1}^*\to_{2}^+ \,\subseteq\, \to_{2}^+\to_{1}^*$.
\end{lemma}

\begin{proof}
	The only non trivial point is the third one, that we now prove. The hypothesis is $\tm \to_{1}^k\to_{2}^h\tmtwo$ for some $k \geq 0$ and $h > 0$. 
	The proof is by lexicographic induction on $(n,k)$, where $n$ is the length of the longest $\to_{1}\cup\to_{2}$  sequence from $\tm$ (it exists because $\to_{1} \!\cup \to_{2}$ is strongly normalizing by hypothesis). 
	The case $k = 0$ is trivial.
	The case $k=h=1$ is given by the hypothesis $\to_{1}\to_{2} \,\subseteq\, \to_{2}\to_{1}^* \!\cup \to_{2}\to_{2}$. 
	Consider now $\tm\to_{1}^k\to_{1}\to_{2}\to_{2}^h\tmtwo$, and apply the swap of the second point to the central pair. There are two cases:
	\begin{enumerate}
		\item $\to_{1}\to_{2} \,\subseteq\, \to_{2}\to_{1}^*$, and so  $\tm\to_{1}^k\to_{2}\to_{1}^*\to_{2}^h\tmtwo$. 
		Now by \ih (on $k$) applied to the prefix $\to_{1}^k\to_{2}$ we obtain $\tm\to_{2}^+\to_{1}^*\to_{1}^*\to_{2}^h\tmtwo$. 
		By \ih (on the longest sequence), the suffix $\to_{1}^*\to_{1}^*\to_{2}^h$ turns into $\to_{2}^+ \to_{1}^*$, giving $\tm\to_{2}^+\to_{2}^+ \to_{1}^*\tmtwo$, \ie~the statement~holds.
		
		\item $\to_{1}\to_{2} \,\subseteq\, \to_{2}\to_{2}$, and so $\tm\to_{1}^k\to_{2}\to_{2}\to_{2}^h\tmtwo$. 
		By \ih (on $k$) applied to the whole sequence, $\tm\to_{2}^+ \to_{1}^*\tmtwo$.
		\qedhere
	\end{enumerate}
\end{proof}

\begin{lemma}[Sufficient condition for irrelevance]
	\label{l:irrelevance-abstract}
		Let $\to_{1}$ and $\to_{2}$ be binary relations on $\vsubterms$, let $\to \,=\, \to_{1} \!\cup \to_{2}$. 
		Assume that:
		\begin{enumerate}
			\item\label{p:irrelevance-abstract-postpone} $\to_{1}$ postpones after $\to_{2}$ (\ie~if $\deriv \colon \tm \to^* \tmtwo$ then there is $\deriv' \colon \tm \to_{2}^* \cdot \to_{1}^* \tmtwo$ with $\sizem{\deriv} = \sizem{\deriv'}$),
						\item\label{p:irrelevance-abstract-strong-normalizing} $\to_{1}$ is strongly normalizing, 
			\item\label{p:irrelevance-abstract-preserve-normal} $\to_{1}$ preserves $\to_{2}$-normal forms (\ie~given $\tm \to_{1} \tmtwo$, we have that $\tm$ is $\to_{2}$-normal if and only if $\tmtwo$ is $\to_{2}$-normal).
	\end{enumerate}
	Then, $\to_{1}$ is $\to_{2}$-irrelevant.
\end{lemma}

\begin{proof}
	According to the definition of irrelevance (\Cref{def:irrelevance}), we have to prove that, for every term $\tm$, $\to_{2}$ is weakly (\resp~strongly) normalizing on $\tm$ if and only if so is $\to$ (indeed, the postponement condition in \Cref{def:irrelevance} holds by \Cref{p:irrelevance-abstract-postpone}).
	
	Clearly, if $\to$ is strongly normalizing on $\tm$ then $\to_{2}$ is strongly normalizing on $\tm$, since $\to_{2} \,\subseteq\, \to$. 
		
		Conversely, suppose that $\to$ is not strongly normalizing on $\tm$. 
		Then, there is an infinite sequence of terms $(\tm_i)_{i \in \nat}$ such that $\tm_i \to \tm_{i+1}$ for all $i \in \nat$, and $\tm = \tm_0$. 
		So, for every $n \in \nat$ there is a reduction sequence  $\deriv_n \colon \tm \to^* \tm_{n}$ with $\sizem{\deriv_n} = n$ and, by postponement (\Cref{p:irrelevance-abstract-postpone}), there is a reduction sequence $\deriv'_n \colon \tm \to_{2}^* \tmtwo_{n}$ with $\sizem{\deriv'_n} = n$.
		Thus, $\to_{2}$ is not strongly normalizing on $\tm$ (by Konig's lemma). 
		
		If $\to$ is weakly normalizing on $\tm$ then $\tm \to^* \tmtwo$ for some $\tmtwo$ normal for $\to$.
		By postponement hypothesis (\Cref{p:irrelevance-abstract-postpone}), $\tm \to_{2}^* \tmthree \to_{1}^* \tmtwo$ where $\tmthree$ is $\to_{2}$-normal by preservation hypothesis (\Cref{p:irrelevance-abstract-preserve-normal}).
		Therefore, $\to_{2}$ is weakly normalizing on $\tm$.
		
		If $\to_{2}$ is weakly normalizing on $\tm$ then $\tm \to_{2}^* \tmthree$ for some $\to_{2}$-normal $\tmthree$.
		Since $\to_{1}$ is strongly normalizing (\Cref{p:irrelevance-abstract-strong-normalizing}), $\tmthree \to_{1}^* \tmtwo$ for some $\tmtwo$ that is $\to_{1}$-normal.
		By preservation hypothesis (\Cref{p:irrelevance-abstract-preserve-normal}), $\tmtwo$ is also $\to_{2}$-normal and hence $\to$-normal. 
		So, $\to$ is weakly normalizing on~$\tm$.
\end{proof}

\section{Proofs of Section \ref{sect:vsc} (Value Substitution Calculus)}

\begin{lemma}[Basic Properties of $\VSC$]
	\label{l:basic-value-substitution}
	\hfill
	\begin{enumerate}
		\item\label{p:basic-value-substitution-tom-toe-terminates} $\tom$ and $\toe$ are both confluent and strongly normalizing (separately) \cite{AccattoliPaolini12}.
		
		\item\label{p:basic-value-substitution-toevar-diamond} $\toeabs$ is confluent, and $\toevar$ is diamond.
		
		\item\label{p:basic-value-substitution-tom-toe-commute-open}  $\tomo$ and $\toeabso$ and $\toevaro$ are pairwise strongly commuting; 
		$\tomo$ and $\toeo$ strongly commute.
		
		\item\label{p:basic-value-substitution-tom-toe-diamond-open} $\tomo$ and $\toeabso$ and $\toevaro$ and $\toeo$ are diamond (separately).
	\end{enumerate}
\end{lemma}

\begin{proof}
	The statements of \reflemma{basic-value-substitution} are a refinement of some results proved in \cite{AccattoliPaolini12}, where $\tovsubo$ is denoted by $\to_\mathsf{w}$.
	\begin{enumerate}
		\item See \cite[Lemma~3]{AccattoliPaolini12}, where $\tom$ is denoted by $\to_\mathsf{db}$, and $\toe$ is denoted by $\to_\mathsf{vs}$.
		
		\smallskip
		
		\item Looking at the proof of confluence of $\toe$ in \cite[Lemma~3]{AccattoliPaolini12}, we have that:
		\begin{itemize}
			\item if $\tmtwo \lRew{\evarsym} \tm \toevar \tmthree$ with $\tmtwo \neq \tmthree$, then $\tmtwo \toevar \tm' \lRew{\evarsym} \tmthree$ for some $\tm$, hence $\toevar$ is diamond;
			
			\item if $\tmtwo \lRew{\eabssym} \tm \toeabs \tmthree$ then $\tmtwo \toeabs^* \tm' {\ }^{*\!\!\!}_{\eabssym\!\!}\leftarrow \tmthree$ for some $\tm'$ (\emph{weak confluence}), hence $\toeabs$ is confluent by Newman's lemma, since $\toeabs \,\subseteq\, \toe$ is strongly normalizing (\Cref{l:basic-value-substitution}.\ref{p:basic-value-substitution-tom-toe-terminates}).
		\end{itemize}  
		
		\smallskip 
		
		\item \emph{Multiplicative vs. Exponential.} 
		For every $\Rule \in \{\varsym, \abssym\}$, we show that $\toeruleo$ and $\tomo$ strongly commute, \ie if $\tmtwo \lRew{\wsym\erulesym} \tm \tomo \tmthree$, then $\tmtwo \neq \tmthree$ and there is $\tmp \in \Lambda_\vsub$ such that $\tmtwo \tomo \tmp \lRew{\wsym\erulesym} \tmthree$. 
		Since $\toeo \,=\, \toeabso \cup \toevaro$, this also means that $\tomo$ and $\toeo$ strongly commute.
		 
		The proof is by induction on the definition of $\tm \toeruleo \tmtwo$. 
		The proof that $\tmtwo \neq \tmthree$ is left to the reader.
		Since the $\toeruleo$ and $\tomo$ cannot reduce under abstractions, all values are $\omsym$-normal and $\osym\erulesym$-normal. So, there are the following cases.
		\begin{itemize}
			\item \emph{Step at the Root for $\tm \toeruleo \tmtwo$} and \emph{\ES left for $\tm \tomo \tmthree$}, \ie given a value $\val$, where $\val$ is an abstraction if $\Rule = \abssym$, and a variable if $\Rule = \varsym$, $\tm \defeq \tmfour\esub\varthree{\sctxp\val} \toeruleo \sctxp{\tmfour\isub\varthree{\val}} \eqdef \tmtwo$ and $\tm \tomo \tmfourp\esub\varthree{\sctxp{\val}} \eqdef \tmthree$ with $\tmfour \tomo \tmfourp$: then 
			$$
			\tmtwo \tomo \sctxp{\tmfourp\isub\varthree{\val}} \lRew{\wsym\erulesym} \tmtwo.
			$$
			
			\item \emph{Step at the Root for $\tm \toeruleo \tmtwo$} and \emph{\ES right for $\tm \tomo \tmthree$}, \ie given a value $\val$, where $\val$ is an abstraction if $\Rule = \abssym$, and a variable if $\Rule = \varsym$,
			\begin{align*}
			\tm &\defeq \tmfour\esub\varthree{\val\esub{\var_1}{\tm_1}\dots\esub{\var_n}{\tm_n}} 
			\toeruleo \tmfour\isub\varthree{\val}\esub{\var_1}{\tm_1}\dots\esub{\var_n}{\tm_n} \eqdef \tmtwo
			\end{align*}
			and $\tm \tomo \tmfour\esub\varthree{\val\esub{\var_1}{\tm_1}\dots\esub{\var_j}{\tmp_j}\dots\esub{\var_n}{\tm_n}} \eqdef \tmthree$
			for some $n > 0$, and $\tm_j \tomo \tmp_j$ for some $1 \leq j \leq n$: then, $\tmtwo \tomo \tmfour\isub\varthree{\val}\esub{\var_1}{\tm_1}\dots\esub{\var_j}{\tmp_j}\dots\esub{\var_n}{\tm_n} \lRew{\wsym\erulesym} \tmthree$.
			
			\item \emph{Application Left for $\tm \toeruleo \tmtwo$} and \emph{Application Right for $\tm \tomo \tmthree$}, \ie $\tm \defeq \tmfour\tmfive \toeruleo \tmfourp\tmfive \eqdef \tmtwo$ and $\tm \tomo \tmfour\tmfivep \eqdef \tmthree$ with $\tmfour \toeruleo \tmfourp$ and $\tmfive \tomo \tmfivep$: then, $\tm \tomo \tmfourp\tmfivep \lRew{\wsym\erulesym} \tmtwo$.
			
			\item \emph{Application Right for $\tm \toeruleo \tmtwo$} and \emph{Application Left for $\tm \tomo \tmthree$}, \ie $\tm \defeq \tmfive\tmfour \toeruleo \tmfive\tmfourp \eqdef \tmtwo$ and $\tm \tomo \tmfivep\tmfour \eqdef \tmthree$ with $\tmfour \toeruleo \tmfourp$ and $\tmfive \tomo \tmfivep$: 
			analogous to the previous point.
			
			\item \emph{Application Left for both $\tm \toeruleo \tmtwo$ and $\tm \tomo \tmthree$}, \ie $\tm \defeq \tmfour\tmfive \toeruleo \tmfourp\tmfive \eqdef \tmtwo$ and $\tm \tomo \tmfour''\tmfive \eqdef \tmthree$ with $\tmfourp \lRew{\wsym\erulesym} \tmfour \tomo \tmfour''$: 
			by \ih, there exists $\tmsix \in \Lambda_\vsub$ such that $\tmfourp \tomo \tmsix \lRew{\wsym\erulesym} \tmfour''$, hence $\tmtwo \tomo \tmsix\tmfive \lRew{\wsym\erulesym} \tmthree$.
			
			\item \emph{Application Right for both $\tm \toeruleo \tmtwo$ and $\tm \tomo \tmthree$}, \ie $\tm \defeq \tmfive\tmfour \toeruleo \tmfive\tmfourp \eqdef \tmtwo$ and $\tm \tomo \tmfive\tmfour'' \eqdef \tmthree$ with $\tmfourp \lRew{\wsym\erulesym} \tmfour \tomo \tmfour''$:
			analogous to the previous point. 
			
			\item \emph{Application Left for $\tm \toeruleo \tmtwo$} and \emph{Step at the Root for $\tm \tomo \tmthree$}, \ie $\tm \defeq (\la\var\tmfive){\esub{\var_1}{\tm_1}\dots\esub{\var_n}{\tm_n}}\tmfour \allowbreak\toeruleo (\la\var\tmfive){\esub{\var_1}{\tm_1}\dots\esub{\var_j}{\tmp_j}\dots\esub{\var_n}{\tm_n}}\tmfour \eqdef \tmtwo$ with $n > 0$ and $\tm_j \toeruleo \tmp_j$ for some $1 \leq j \leq n$, and 
			$$
			\tm \tomo \allowbreak \tmfive\esub\var\tmfour{\esub{\var_1}{\tm_1}\dots\esub{\var_n}{\tm_n}} \eqdef \tmthree
			$$ 
			Then,
			\begin{equation*}
			\tmtwo \tomo \tmfive\esub\var\tmfour{\esub{\var_1}{\tm_1}\dots\esub{\var_j}{\tmp_j}\dots\esub{\var_n}{\tm_n}} \lRew{\wsym\erulesym} \tmthree.
			\end{equation*}
			
			\item \emph{Application Right for $\tm \toeruleo \tmtwo$} and \emph{Step at the Root for $\tm \tomo \tmthree$}, \ie $\tm \defeq \sctxp{\la\var\tmfive}\tmfour \toeruleo \sctxp{\la\var\tmfive}\tmfourp \eqdef \tmtwo$ with $\tmfour \toeruleo \tmfourp$\!, and $\tm \tomo \allowbreak \sctxp{\tmfive\esub\var\tmfour} \eqdef \tmthree$: then, $\tmtwo \tomo \sctxp{\tmfive\esub\var\tmfourp} \lRew{\wsym\erulesym} \tmthree$.
			
			\item \emph{\ES left for $\tm \toeruleo \tmtwo$} and \emph{\ES right for $\tm \tomo \tmthree$}, \ie $\tm \defeq \tmfour\esub\var\tmfive \toeruleo \tmfourp\esub\var\tmfive \eqdef \tmtwo$ and $\tm \tomo \tmfour\esub\var\tmfivep \eqdef \tmthree$ with $\tmfour \toeruleo \tmfourp$ and $\tmfive \tomo \tmfivep$: then, $\tmtwo \tomo \tmfourp\esub\var\tmfivep \lRew{\wsym\erulesym} \tmthree$.
			\item \emph{\ES right for $\tm \toeruleo \tmtwo$} and \emph{\ES left for $\tm \tomo \tmthree$}, \ie $\tm \defeq \tmfive\esub\var\tmfour \toeruleo \tmfive\esub\var\tmfourp \eqdef \tmtwo$ and $\tm \tomo \tmfivep\esub\var\tmfour \eqdef \tmthree$ with $\tmfour \toeruleo \tmfourp$ and $\tmfive \tomo \tmfivep$: 
			analogous to the previous point.
			
			\item \emph{\ES left for both $\tm \toeruleo \tmtwo$ and $\tm \tomo \tmthree$}, \ie $\tm \defeq \tmfour\esub\var\tmfive \toeruleo \tmfourp\esub\var\tmfive \eqdef \tmtwo$ and $\tm \tomo \tmfour''\esub\var\tmfive \eqdef \tmthree$ with $\tmfourp \lRew{\wsym\erulesym} \tmfour \tomo \tmfour''$: by \ih,there is $\tmsix \in \Lambda_\vsub$ such that $\tmfourp \tomo \tmsix \lRew{\wsym\erulesym} \tmfour''$, hence $\tmtwo \tomo \tmsix\esub\var\tmfive \lRew{\wsym\erulesym} \tmthree$.
			
			\item \emph{\ES right for both $\tm \toeruleo \tmtwo$ and $\tm \tomo \tmthree$}, \ie $\tm \defeq \tmfive\esub\var\tmfour \toeruleo \tmfive\esub\var{\tmfourp} \eqdef \tmtwo$ and $\tm \tomo \tmfive\esub\var{\tmfour''} \eqdef \tmthree$ with $\tmfour \lRew{\wsym\erulesym} \tmfourp \tomo \tmfour''$: 
			analogous to the previous point.
		\end{itemize}
	
		\smallskip
		
		\emph{Exponential variable vs. Exponential abstraction.} 
		We show that $\toevaro$ and $\toeabso$ strongly commute, \ie if $\tmtwo \lRew{\osym\evarsym} \tm \toeabso \tmthree$, then $\tmtwo \neq \tmthree$ and there is $\tmp \in \Lambda_\vsub$ such that $\tmtwo \toeabso \tmp \lRew{\wsym\evarsym} \tmthree$. 
		The proof is by induction on the definition of $\tm \toevaro \tmtwo$. 
		The proof that $\tmtwo \neq \tmthree$ is left to the reader.
		Since the $\toevaro$ and $\toeabso$ cannot reduce under abstractions, all values are $\osym\eabssym$-normal and $\osym\evarsym$-normal. So, there are the following cases.
		\begin{itemize}
			\item \emph{Step at the Root for $\tm \toevaro \tmtwo$} and \emph{\ES left for $\tm \toeabso \tmthree$}, \ie~$\tm \defeq \tmfour\esub\varthree{\sctxp\vartwo} \toevaro \sctxp{\tmfour\isub\varthree{\vartwo}} \eqdef \tmtwo$ and $\tm \toeabso \tmfourp\esub\varthree{\sctxp{\vartwo}} \eqdef \tmthree$ with $\tmfour \toeabso \tmfourp$: then 
			$$
			\tmtwo \toeabso \sctxp{\tmfourp\isub\varthree{\vartwo}} \lRew{\wsym\evarsym} \tmthree.
			$$
			
			\item \emph{Step at the Root for $\tm \toeabso \tmtwo$} and \emph{\ES left for $\tm \toevaro \tmthree$}, \ie~$\tm \defeq \tmfour\esub\varthree{\sctxp{\la\vartwo\tmfive}} \toeabso \sctxp{\tmfour\isub\varthree{\la\vartwo\tmfive}} \eqdef \tmtwo$ and $\tm \toevaro \tmfourp\esub\varthree{\sctxp{\la\vartwo\tmfive}} \eqdef \tmthree$ with $\tmfour \toevaro \tmfourp$: 
			analogous to the previous point.
			
			\item \emph{Step at the Root for $\tm \toevaro \tmtwo$} and \emph{\ES right for $\tm \toeabso \tmthree$}, \ie 
			\begin{align*}
			\tm &\defeq \tmfour\esub\varthree{\vartwo\esub{\var_1}{\tm_1}\dots\esub{\var_n}{\tm_n}} 
			\toevaro \tmfour\isub\varthree{\vartwo}\esub{\var_1}{\tm_1}\dots\esub{\var_n}{\tm_n} \eqdef \tmtwo
			\end{align*}
			and $\tm \toeabso \tmfour\esub\varthree{\vartwo\esub{\var_1}{\tm_1}\dots\esub{\var_j}{\tmp_j}\dots\esub{\var_n}{\tm_n}} \eqdef \tmthree$
			for some $n > 0$, and $\tm_j \toeabso \tmp_j$ for some $1 \leq j \leq n$: 
			then, $\tmtwo \toeabso \tmfour\isub\varthree{\vartwo}\esub{\var_1}{\tm_1}\dots\esub{\var_j}{\tmp_j}\dots\esub{\var_n}{\tm_n} \lRew{\wsym\evarsym} \tmthree$.
			
			\item \emph{Step at the Root for $\tm \toeabso \tmtwo$} and \emph{\ES right for $\tm \toevaro \tmthree$}, \ie 
			\begin{align*}
			\tm &\defeq \tmfour\esub{\varthree}{(\la{\vartwo}{\tmfive})\esub{\var_1}{\tm_1}\dots\esub{\var_n}{\tm_n}} 
			\toeabso \tmfour\isub\varthree{\la\vartwo\tmfive}\esub{\var_1}{\tm_1}\dots\esub{\var_n}{\tm_n} \eqdef \tmtwo
			\end{align*}
			and $\tm \toevaro \tmfour\esub\varthree{(\la{\vartwo}\tmfive)\esub{\var_1}{\tm_1}\dots\esub{\var_j}{\tmp_j}\dots\esub{\var_n}{\tm_n}} \eqdef \tmthree$
			for some $n > 0$, and $\tm_j \toevaro \tmp_j$ for some $1 \leq j \leq n$: 
			analogous to the previous point.
			\item \emph{Application Left for $\tm \toevaro \tmtwo$} and \emph{Application Right for $\tm \toeabso \tmthree$}, \ie $\tm \defeq \tmfour\tmfive \toevaro \tmfourp\tmfive \eqdef \tmtwo$ and $\tm \toeabso \tmfour\tmfivep \eqdef \tmthree$ with $\tmfour \toevaro \tmfourp$ and $\tmfive \toeabso \tmfivep$: then, $\tm \toeabso \tmfourp\tmfivep \lRew{\wsym\evarsym\!} \tmtwo$.
			
			\item \emph{Application Right for $\tm \toevaro \tmtwo$} and \emph{Application Left for $\tm \toeabso \tmthree$}, \ie $\tm \defeq \tmfive\tmfour \toevaro \tmfive\tmfourp \eqdef \tmtwo$ and $\tm \toeabso \tmfivep\tmfour \eqdef \tmthree$ with $\tmfour \toevaro \tmfourp$ and $\tmfive \toeabso \tmfivep$: 
			analogous to the previous point.
			
			\item \emph{Application Left for both $\tm \toevaro \tmtwo$ and $\tm \toeabso \tmthree$}, \ie $\tm \defeq \tmfour\tmfive \toevaro \tmfourp\tmfive \eqdef \tmtwo$ and $\tm \toeabso \tmfour''\tmfive \eqdef \tmthree$ with $\tmfourp \lRew{\wsym\evarsym\!} \tmfour \toeabso \tmfour''$: 
			by \ih, there exists $\tmsix \in \Lambda_\vsub$ such that $\tmfourp \toeabso \tmsix \lRew{\wsym\evarsym\!} \tmfour''$, hence $\tmtwo \toeabso \tmsix\tmfive \lRew{\wsym\evarsym\!} \tmthree$.
			
			\item \emph{Application Right for both $\tm \toevaro \tmtwo$ and $\tm \toeabso \tmthree$}, \ie $\tm \defeq \tmfive\tmfour \toevaro \tmfive\tmfourp \eqdef \tmtwo$ and $\tm \toeabso \tmfive\tmfour'' \eqdef \tmthree$ with $\tmfourp \lRew{\wsym\evarsym} \tmfour \toeabso \tmfour''$:
			analogous to the previous point. 
			
			\item \emph{\ES left for $\tm \toevaro \tmtwo$} and \emph{\ES right for $\tm \toeabso \tmthree$}, \ie $\tm \defeq \tmfour\esub\var\tmfive \toevaro \tmfourp\esub\var\tmfive \eqdef \tmtwo$ and $\tm \toeabso \tmfour\esub\var\tmfivep \eqdef \tmthree$ with $\tmfour \toevaro \tmfourp$ and $\tmfive \toeabso \tmfivep$: then, $\tmtwo \toeabso \tmfourp\esub\var\tmfivep \lRew{\wsym\evarsym\!} \tmthree$.
			\item \emph{\ES right for $\tm \toevaro \tmtwo$} and \emph{\ES left for $\tm \toeabso \tmthree$}, \ie $\tm \defeq \tmfive\esub\var\tmfour \toevaro \tmfive\esub\var\tmfourp \eqdef \tmtwo$ and $\tm \toeabso \tmfivep\esub\var\tmfour \eqdef \tmthree$ with $\tmfour \toevaro \tmfourp$ and $\tmfive \toeabso \tmfivep$: 
			analogous to the previous point.
			
			\item \emph{\ES left for both $\tm \toevaro \tmtwo$ and $\tm \toeabso \tmthree$}, \ie $\tm \defeq \tmfour\esub\var\tmfive \toevaro \tmfourp\esub\var\tmfive \eqdef \tmtwo$ and $\tm \toeabso \tmfour''\esub\var\tmfive \eqdef \tmthree$ with $\tmfourp \lRew{\wsym\evarsym\!} \tmfour \toeabso \tmfour''$: by \ih,there is $\tmsix \in \Lambda_\vsub$ such that $\tmfourp \toeabso \tmsix \lRew{\wsym\evarsym\!} \tmfour''$, hence $\tmtwo \toeabso \tmsix\esub\var\tmfive \lRew{\wsym\evarsym\!} \tmthree$.
			
			\item \emph{\ES right for both $\tm \toevaro \tmtwo$ and $\tm \toeabso \tmthree$}, \ie $\tm \defeq \tmfive\esub\var\tmfour \toevaro \tmfive\esub\var{\tmfourp} \eqdef \tmtwo$ and $\tm \toeabso \tmfive\esub\var{\tmfour''} \eqdef \tmthree$ with $\tmfour \lRew{\wsym\evarsym\!} \tmfourp \toeabso \tmfour''$: 
			analogous to the previous point.
		\end{itemize}
	
		\smallskip 
		
		\item \emph{Multiplicative.}
		We prove that $\tomo$ is diamond, \ie if $\tmtwo \lRew{\wmsym} \tm \tomo \tmthree$ with $\tmtwo \neq \tmthree$ then there is $\tmp \in \Lambda_\vsub$ such that $\tmtwo \tomo \tmp \lRew{\wmsym} \tmthree$.
		The proof is by induction on the definition of $\tomo$. 
		Since $\tm \tomo \tmthree \neq \tmtwo$ and $\tomo$ does not reduce under abstractions, there are only eight cases:
		\begin{itemize}
			\item \emph{Step at the Root for $\tm \!\tomo\! \tmtwo$ and Application Right for $\tm \!\tomo\! \tmthree$}, \ie $\tm \defeq \sctxp{\la\var\tmfive}\tmfour \rtom \sctxp{\tmfive\esub{\var}{\tmfour}} \eqdef \tmtwo$ and $\tm \!\rtom\! \sctxp{\la\var\tmfive}\tmfourp\! \eqdef \tmthree$ with $\tmfour \!\tomo\! \tmfourp$: then, $\tmtwo \!\tomo\! \sctxp{\tmfive\esub{\var}{\tmfourp}} \!\lRew{\wmsym}\! \tmthree$;
			
			\item \emph{Step at the Root for $\tm \tomo \tmtwo$ and Application Left for $\tm \tomo \tmthree$}, \ie, for some $n > 0$, 
			$$
			\tm \defeq (\la\var\tmfive)\esub{\var_1}{\tm_1}\dots\esub{\var_n}{\tm_n}\tmfour \allowbreak\rtom \tmfive\esub{\var}{\tmfour}\esub{\var_1}{\tm_1}\dots\esub{\var_n}{\tm_n} \eqdef \tmtwo
			$$
			whereas $\tm \tomo \allowbreak (\la\var\tmfive)\esub{\var_1}{\tm_1}\dots\esub{\var_j}{\tmp_j}\dots\esub{\var_n}{\tm_n}\tmfour \eqdef \tmthree$ with $\tm_j \tomo \tmp_j$ for some $1 \leq j \leq n$: then, 
			\begin{align*}
			\tmtwo \tomo \allowbreak \tmfive\esub{\var}{\tmfour}\esub{\var_1}{\tm_1}\dots\esub{\var_j}{\tmp_j}\dots\esub{\var_n}{\tm_n} \lRew{\wmsym} \tmthree;
			\end{align*}
			\item \emph{Application Left for $\tm \tomo \tmtwo$ and Application Right for $\tm \tomo \tmthree$}, \ie $\tm \defeq \tmfour\tmfive \tomo \tmfourp\tmfive \eqdef \tmtwo$ and $\tm \tomo \tmfour\tmfivep \eqdef \tmthree$ with $\tmfour \tomo \tmfourp$ and $\tmfive \tomo \tmfivep$: then, $\tmtwo \tomo \tmfourp\tmfivep\! \lRew{\wmsym} \tmthree$;
			\item \emph{Application Left for both $\tm \tomo \tmtwo$ and $\tm \tomo \tmthree$}, \ie $\tm \defeq \tmfour\tmfive \tomo \tmfourp\tmfive \eqdef \tmtwo$ and $\tm \tomo \tmfour''\tmfive \eqdef \tmthree$ with $\tmfourp \lRew{\wmsym} \tmfour \tomo \tmfour''$: by \ih, there exists $\tmfour_0 \in \Lambda_\vsub$ such that $\tmfourp \tomo \tmfour_0 \lRew{\msym} \tmfour''$, hence $\tmtwo \tomo \tmfour_0\tmfive \lRew{\msym} \tmthree$;
			\item \emph{Application Right for both $\tm \tomo \tmtwo$ and $\tm \tomo \tmthree$}, \ie $\tm \defeq \tmfive\tmfour \tomo \tmfive\tmfourp \eqdef \tmtwo$ and $\tm \tomo \tmfive\tmfour'' \eqdef \tmthree$ with $\tmfourp \lRew{\wmsym} \tmfour \tomo \tmfour''$: by \ih, there exists $\tmfour_0 \in \Lambda_\vsub$ such that $\tmfourp \tomo \tmfour_0 \lRew{\wmsym} \tmfour''$, hence $\tmtwo \tomo \tmfive\tmfour_0 \lRew{\wmsym} \tmthree$;
			\item \emph{$\mathsf{ES}$ left for $\tm \tomo \tmtwo$ and $\mathsf{ES}$ right for $\tm \tomo \tmthree$}, \ie $\tm \defeq \tmfour\esub\var\tmfive \tomo \tmfourp\esub\var\tmfive \eqdef \tmtwo$ and $\tm \tomo \tmfour\esub\var\tmfivep \eqdef \tmthree$ with $\tmfour \tomo \tmfourp$ and $\tmfive \tomo \tmfivep$: then, 
			$$
			\tmtwo \tomo \tmfourp\esub\var\tmfivep\! \lRew{\wmsym} \tmthree
			$$
			\item \emph{$\mathsf{ES}$ left for both $\tm \tomo \tmtwo$ and $\tm \tomo \tmthree$}, \ie $\tm \defeq \tmfour\esub\var\tmfive \tomo \tmfourp\esub\var\tmfive \eqdef \tmtwo$ and $\tm \tomo \tmfour''\esub\var\tmfive \eqdef \tmthree$ with $\tmfourp \lRew{\wmsym} \tmfour \tomo \tmfour''$: by \ih, there exists $\tmfour_0 \in \Lambda_\vsub$ such that $\tmfourp \tomo \tmfour_0 \lRew{\wmsym} \tmfour''$, hence $\tmtwo \tom \tmfour_0\esub\var\tmfive \lRew{\wmsym} \tmthree$;
			\item \emph{$\mathsf{ES}$ right for both $\tm \tomo \tmtwo$ and $\tm \tomo \tmthree$}, \ie $\tm \defeq \tmfive\esub\var\tmfour \tomo \tmfive\esub\var\tmfourp \eqdef \tmtwo$ and $\tm \tomo \tmfive\esub\var{\tmfour''} \eqdef \tmthree$ with $\tmfourp \lRew{\wmsym} \tmfour \tom \tmfour''$: by \ih, there exists $\tmfour_0 \in \Lambda_\vsub$ such that $\tmfourp \tomo \tmfour_0 \lRew{\wmsym} \tmfour''$, hence $\tmtwo \tom \tmfive\esub\var{\tmfour_0} \lRew{\wmsym} \tmthree$.
		\end{itemize}
		
		\smallskip
		
		\emph{Exponential variable.} 
		We prove that $\toevaro$ is diamond, \ie if $\tmtwo \lRew{\wsym\evarsym} \tm \toevaro \tmthree$ with $\tmtwo \neq \tmthree$ then there is $\tm' \in \Lambda_\vsub$ such that $\tmtwo \toevaro \tmp \lRew{\wsym\evarsym} \tmthree$.
		The proof is by induction on the definition of $\toevaro$. 
		As $\tm \toevaro \tmthree \neq \tmtwo$ and $\toevaro$ does not fire under abstractions, there are only eight cases:
		\begin{itemize}
			\item \emph{Step at the Root for $\tm \!\toevaro\! \tmtwo$} and \emph{\ES left for $\tm \!\toevaro\! \tmthree$}, \ie $\tm \defeq \tmfour\esub\var{\sctxp{\vartwo}} \rtoevar \sctxp{\tmfour\isub{\var}{\vartwo}} \eqdef \tmtwo$ and $\tm \!\toevaro\! \tmfourp\esub\var{\sctxp{\vartwo}}\! \eqdef \tmthree$ with $\tmfour \!\toevaro\! \tmfourp$: then, 
			$$
			\tmtwo \!\toevaro\! \sctxp{\tmfourp\isub{\var}{\vartwo}} \!\lRew{\wsym\evarsym}\! \tmthree
			$$
			
			\item \emph{Step at the Root for $\tm \toevaro \tmtwo$} and \emph{\ES right for $\tm \toevaro \tmthree$}, \ie, for some $n > 0$, $\tm \defeq \tmfour\esub\var{\vartwo\esub{\var_1}{\tm_1}\dots\esub{\var_n}{\tm_n}} \allowbreak\rtoevar \tmfour\isub{\var}{\vartwo}\esub{\var_1}{\tm_1}\dots\esub{\var_n}{\tm_n} \eqdef \tmtwo$, whereas $\tm \toevaro \allowbreak \tmfour\esub{\var}{\vartwo\esub{\var_1}{\tm_1}\dots\esub{\var_j}{\tmp_j}\dots\esub{\var_n}{\tm_n}} \eqdef \tmthree$ with $\tm_j \toevaro \tmp_j$ for some $1 \leq j \leq n$: then, 
			\begin{align*}
			\tmtwo \toevaro \allowbreak \tmfour\isub{\var}{\vartwo}\esub{\var_1}{\tm_1}\dots\esub{\var_j}{\tmp_j}\dots\esub{\var_n}{\tm_n} \lRew{\wsym\evarsym} \tmthree;
			\end{align*}
			
			\item \emph{Application Left for $\tm \toevaro \tmtwo$} and \emph{Application Right for $\tm \toevaro \tmthree$}, \ie $\tm \defeq \tmfour\tmfive \toevaro \tmfourp\tmfive \eqdef \tmtwo$ and $\tm \toevaro \tmfour\tmfivep \eqdef \tmthree$ with $\tmfour \toevaro \tmfourp$ and $\tmfive \toevaro \tmfivep$: then, $\tmtwo \toevaro \tmfourp\tmfivep\! \lRew{\wesym} \tmthree$.
			\item \emph{\ES left for $\tm \toevaro \tmtwo$} and \emph{\ES right for $\tm \toevaro \tmthree$}, \ie $\tm \defeq \tmfour\esub\var\tmfive \toevaro \tmfourp\esub\var\tmfive \eqdef \tmtwo$ and $\tm \toevaro \tmfour\esub\var\tmfivep \eqdef \tmthree$ with $\tmfour \toevaro \tmfourp$ and $\tmfive \toevaro \tmfivep$: then, $\tmtwo \toevaro \tmfourp\esub\var\tmfivep\! \lRew{\wsym\evarsym} \tmthree$.
			
			\item \emph{Application Left for both $\tm \toevaro \tmtwo$ and $\tm \toevaro \tmthree$}, \ie $\tm \defeq \tmfour\tmfive \toevaro \tmfourp\tmfive \eqdef \tmtwo$ and $\tm \toevaro \tmfour''\tmfive \eqdef \tmthree$ with $\tmfourp \lRew{\wsym\evarsym} \tmfour \toevaro \tmfour''$: 
			by \ih, there exists $\tmfour_0 \in \Lambda_\vsub$ such that $\tmfourp \toevaro \tmfour_0 \lRew{\wsym\evarsym} \tmfour''$, hence $\tmtwo \toevaro \tmfour_0\tmfive \lRew{\wsym\evarsym} \tmthree$.
			\item \emph{Application Right for both $\tm \toevaro \tmtwo$ and $\tm \toevaro \tmthree$}, \ie $\tm \defeq \tmfive\tmfour \toevaro \tmfive\tmfourp \eqdef \tmtwo$ and $\tm \toevaro \tmfive\tmfour'' \eqdef \tmthree$ with $\tmfourp \lRew{\wsym\evarsym} \tmfour \toevaro \tmfour''$: 
			analogous to the previous point.

			\item \emph{\ES left for both $\tm \toevaro \tmtwo$ and $\tm \toevaro \tmthree$}, \ie $\tm \defeq \tmfour\esub\var\tmfive \toevaro \tmfourp\esub\var\tmfive \eqdef \tmtwo$ and $\tm \toevaro \tmfour''\esub\var\tmfive \eqdef \tmthree$ with $\tmfourp \lRew{\wsym\evarsym} \tmfour \toevaro \tmfour''$: by \ih, there exists $\tmfour_0 \in \Lambda_\vsub$ such that $\tmfourp \toevaro \tmfour_0 \lRew{\esym} \tmfour''$, hence $\tmtwo \toevaro \tmfour_0\esub\var\tmfive \lRew{\wsym\evarsym} \tmthree$.
			
			\item \emph{\ES right for both $\tm \toevaro \tmtwo$ and $\tm \toevaro \tmthree$}, \ie $\tm \defeq \tmfive\esub\var\tmfour \toevaro \tmfive\esub\var\tmfourp \eqdef \tmtwo$ and $\tm \toevaro \tmfive\esub\var{\tmfour''} \eqdef \tmthree$ with $\tmfourp \lRew{\wsym\evarsym} \tmfour \toevaro \tmfour''$: 
			analogous to the previous point.
		\end{itemize}
		
		\smallskip
		\emph{Exponential abstraction.} 
		We prove that $\toeabso$ is diamond, \ie if $\tmtwo \lRew{\wsym\eabssym} \tm \toeabso \tmthree$ with $\tmtwo \neq \tmthree$ then there exists $\tm' \in \Lambda_\vsub$ such that $\tmtwo \toeabso \tmp \lRew{\wsym\eabssym} \tmthree$.
		The proof is by induction on the definition of $\toeabso$ and,  
		Since $\toeabso$ does not fire under abstractions, it analogous to the previous point (exponential variable, see above). 
		The only differences are that $\toevaro$ is replaced by $\toeabso$, and that in the cases with a step at the root, the variable $\vartwo$ is replaced by an abstraction $\la{\vartwo}{\tmfive}$.
		
		\smallskip
		\emph{Exponential.}
		Diamond of $\toeo$ follows immediately from the diamond of $\toevaro$ and of $\toeabso$ (separately, see the previous points in \Cref{l:basic-value-substitution}.\ref{p:basic-value-substitution-tom-toe-diamond-open}), and from the strong commutation of $\toevaro$ and $\toeabso$ (\Cref{l:basic-value-substitution}.\ref{p:basic-value-substitution-tom-toe-commute-open}).

		\smallskip
		
		Note that in \cite[Lemma~11]{AccattoliPaolini12} it has only been proved the diamond of $\tovsubo$ (denoted by $\to_{\mathsf{w}}$ there), not of $\tomo$ or $\toevaro$ or $\toeabso$.
		\qedhere
	\end{enumerate}
\end{proof}

\begin{proposition}[Properties of open reduction]
	\label{propappendix:properties-open-reduction}
	\NoteState{prop:properties-open-reduction}
	\begin{enumerate}
				
		\item \label{pappendix:properties-open-reduction-commutation} 
		\emph{Strong commutation:} reductions $\tomo$, $\toeabso$, and $\toevaro$ are pairwise strongly commuting.
			
		\item \label{pappendix:properties-open-reduction-diamond}
		\emph{Diamond}: reductions $\tovsubo$ and $\tovsubonvar$ are diamond (separately).
			
		\item \label{pappendix:properties-open-redction-harmony}
		\emph{Normal forms}:  $\tm$ is $\onvarsym$-normal if and only if $\tm$ is a fireball.
		If $\tm$ is $\osym$-normal then it is a fireball.
	\end{enumerate}
\end{proposition}

\begin{proof}
	\begin{enumerate}
		\item See \Cref{l:basic-value-substitution}.\ref{p:basic-value-substitution-tom-toe-commute-open}.
		
		\item 
		As $\tovsubo \,=\, \tomo \cup \toeo$, diamond of $\tovsubo$ follows immediately from strong commutation of $\tomo$ and $\toeo$ (\Cref{l:basic-value-substitution}.\ref{p:basic-value-substitution-tom-toe-commute-open}), and from diamond of $\tomo$ and $\toeo$ (separately, \Cref{l:basic-value-substitution}.\ref{p:basic-value-substitution-tom-toe-diamond-open}). 

		As $\tovsubonvar \,=\, \tomo \cup \toeabso$, diamond of $\tovsubonvar$ follows immediately from strong commutation of $\tomo$ and $\toeabso$ (\Cref{l:basic-value-substitution}.\ref{p:basic-value-substitution-tom-toe-commute-open}), and from diamond of $\tomo$ and $\toeabso$ (separately, \Cref{l:basic-value-substitution}.\ref{p:basic-value-substitution-tom-toe-diamond-open}).

		\item The second statement (``If a term is $\osym$-normal then it is a fireball'') follows immediately from the first one (``A term is $\onvarsym$-normal if and only if it is a  fireball''), since $\tovsubonvar \,\subseteq\, \tovsubo$.
		
		We prove the two directions of the equivalence in the first statement separately.
		\begin{enumerate}
			\item \label{p:normal-to-fireball} Let $\tm$ be $\onvarsym$-normal. 
			We prove that $\tm$ is a  fireball by induction on $\tm$.
			Cases:
			\begin{itemize}
				\item \emph{Variable}, \ie $\tm = \var$ for some variable $\var$. 
				So, $\var$ is a inert term and hence a fireball.
				
				\item \emph{Abstraction}, \ie $\tm = \la{\var}{\tmtwo}$. 
				As $\tm$ is $\onvarsym$-normal, so is $\tmtwo$.
				By \ih, $\tmtwo$ is a fireball, and then so~is~$\tm$.
				
				\item \emph{Application}, \ie $\tm = \tm_{1} \tm_{2}$. 
				Since $\tm$ is $\onvarsym$-normal,  so are $\tm_1$ and $\tm_2$.
				By \ih, $\tm_{1}$ and $\tm_{2}$ are fireballs. 
				Note that  $\tm_{1} \neq \sctxp{\la{\var}{\tmtwo}}$, otherwise $\tm \rtom \sctxp{\tmtwo \esub{\var}{\tm_{2}}}$ which contradicts $\onvarsym$-normality of $\tm$. 
				So, according to the definition of  fireball, $\tm_{1}$ is a inert term, thus $\tm$ is a  fireball.
				
				\item \emph{Explicit substitution}; \ie, $\tm = \tm_{1} \esub{\var}{\tm_{2}}$. 
				Since $\tm$ is $\onvarsym$-normal, then so are $\tm_1$ and $\tm_2$.
				By \ih, $\tm_{1}$ and $\tm_{2}$ are solved fireballs. 
				Note that $\tm_{2} \neq \sctxp{\la{\var}{\tmtwo}}$, otherwise $\tm \rtoeabs \sctxp{\tm_{1} \isub{\var}{\la{\var}{\tmtwo}}}$ which contradicts the $\onvarsym$-normality of $\tm$. 
				Thus, according to the definition of fireball, $\tm_{2}$ is a inert term, and so $\tm$ is a  fireball.
			\end{itemize}
			
			\item\label{p:fireball-to-normal} Let $\tm$ be a fireball. 
			We prove that $\tm$ is $\onvarsym$-normal by induction on the definition of fireball.
			\begin{itemize}
				\item \emph{Variable}, \ie $\tm = \var$ for some variable $\var$. Clearly, $\tm$ is $\onvarsym$-normal.
				
				\item \emph{Abstraction}, \ie~$\tm = \la{\var}{\tmtwo}$. 
				As $\tovsubo$ cannot reduce under abstractions, $\tm$ is $\onvarsym$-normal.
				
				\item \emph{Application}, \ie~$\tm = \itm \fire$ for some inert term $\itm$ and fireball $\fire$.
				By \ih, $\itm $ and $\fire$  are $\onvarsym$-normal. 
				Since $\itm$ is not of the form $\sctxp{\la{\var}{\tmtwo}}$, then $\tm$ is also $\onvarsym$-normal.
				
				\item \emph{Explicit substitution}, \ie~$\tm = \fire \esub{\var}{\itm}$ for some inert term $\itm$ and solved fireball $\fire$ (it includes the case when $\fire$ is a inert term). 
				By \ih, both $\fire$ and $\itm$ are $\onvarsym$-normal. 
				Since $\itm$ is not of the form $\sctxp{\la{\var}{\tmtwo}}$, then $\tm$ is also $\onvarsym$-normal.
				\qedhere
			\end{itemize}
		\end{enumerate}
	\end{enumerate}
	
\end{proof}

%

\begin{proposition}[Further properties of open reduction]
	\label{propappendix:properties-open-extra} 
	\NoteState{prop:properties-open-extra}
	\hfill
	\begin{enumerate} 
		
		\item\label{pappendix:properties-open-extra-valuability} \emph{Valuability \cite{AccattoliPaolini12}:} if $\tm \tovsub^{*}\val$ for some value $\val$, then $\tm \tovsubo^{*} \valtwo$ for some value $\valtwo$.
		
		\item\label{pappendix:properties-open-extra-normalization} \emph{Normalization:} if $\tm \tovsub^{*} \tmtwo$ for some $\osym$-normal $\tmtwo$, then $\tm \tovsubo^{*} \tmthree$ for some $\osym$-normal $\tmthree$.
		
		\item\label{pappendix:properties-open-extra-normalization-bis} \emph{Normalization 2:} if $\tm \tovsubnvar^{*} \fire$ for some fireball $\fire$, then $\tm \tovsubonvar^{*} \firetwo$ for some fireball $\firetwo$.
		
	\end{enumerate}
\end{proposition}	

\begin{proof}
	\begin{enumerate}
		
		\item See \cite[Corollary 5.1]{AccattoliPaolini12}, where $\tovsubo$ is denoted by $\to_\mathsf{w}$.
		
		\item This is exactly our \Cref{thm:normalization}.\ref{p:normalization-open}, which is proved independently by type-theoretical means.
		
		
		\item 
		As $\toevaro \subseteq \toe$ is strongly normalizing  (\Cref{l:basic-value-substitution}.\ref{p:basic-value-substitution-tom-toe-terminates}), there is a $\evarsym$-normal term $\tmfour$ such that $\fire \toevaro^* \tmfour$. 
		As $\fire$ is $\onvarsym$-normal (\Cref{prop:properties-open-reduction}.\ref{p:properties-open-reduction-harmony}), $\tmfour$ is $\onvarsym$-normal too by forthcoming  \Cref{l:preservation-normal}.\ref{p:preservation-normal-open} (proved independently below).
		By the normalization property for $\tovsubo$ (\Cref{prop:properties-open-extra}.\ref{p:properties-open-extra-normalization}, since $\tm \tovsub^* \tmfour$ with $\tmfour$ $\osym$-normal), $\tm \tovsubo^{*} \tmthree'$ for some  $\tmthree'$ that is $\osym$-normal and in particular $\onvarsym$-normal.
		By postponement of $\toevaro$ (\Cref{cor:postponement}.\ref{p:postponement-open}, proved independently below), $\tm \tovsubonvar^* \tmthree \toevaro^* \tmthree'$ for some $\tmthree$ which is $\onvarsym$-normal by \Cref{l:preservation-normal}.\ref{p:preservation-normal-open}. 
		According to \Cref{prop:properties-open-reduction}.\ref{p:properties-open-reduction-harmony}, $\tmthree$ is a fireball.
		\qedhere
	\end{enumerate}
\end{proof}


\begin{lemma}[Sizes for inert terms]
	\label{l:sizes-inert}
	For every inert term $\itm$, $\sizeo{\itm} = \sizes{\itm}$.
\end{lemma}

\begin{proof}
	By induction of the inert term $\itm$.
	Cases:
	\begin{itemize}
		\item \emph{Variable}, \ie $\itm = \var$ for some variable $\var$. 
		Then, $\sizeo{\itm} = 0 = \sizes{\itm}$.
		
%
		\item \emph{Application}, \ie $\itm = \itmtwo \fire$ for some inert term $\itmtwo$ and fireball $\fire$. 
		By \ih, $\sizeo{\itmtwo} = \sizes{\itmtwo}$.
		Therefore, $\sizeo{\itm} = 1 + \sizeo{\itmtwo} + \sizeo{\fire} = 1 + \sizes{\itmtwo} + \sizeo{\fire} = \sizes{\itm}$.

		\item \emph{Explicit substitution}, \ie $\itm = \itmtwo \esub{\var}{\itmthree}$ for some inert terms $\itmtwo$ and $\itmthree$. 
		By \ih, $\sizeo{\itmtwo} = \sizes{\itmtwo}$.
		Thus, $\sizeo{\itm} = \sizeo{\itmtwo} + \sizeo{\itmthree} = \sizes{\itmtwo} + \sizeo{\itmthree} = \sizes{\itm}$.
		\qedhere
	\end{itemize}
\end{proof}

\begin{proposition}[Properties of the full reduction]
	\label{propappendix:properties-full-reduction}
	\NoteState{prop:properties-full-reduction}
	\hfill
	\begin{enumerate}
		\item\label{pappendix:properties-full-reduction-confluence} \emph{Confluence \cite{AccattoliPaolini12}:} reductions $\tovsub$ and $\tovsubnvar$ are confluent.
		
		\item\label{pappendix:properties-full-reduction-harmony} \emph{Normal forms:} A term is $\vsubnvarsym$-normal if and only if it is a \full fireball. 
		If a term is $\vsub$-normal then it is a \full fireball.
	\end{enumerate}
\end{proposition}

\begin{proof}
	\begin{enumerate}
		\item Confluence of $\tovsub$ is proven in \cite[Corollary 1]{AccattoliPaolini12}.
		In that proof, the crucial point is the following property, called \emph{local commutation} of $\tom$ and $\toe$: if $\tmtwo \lRew{\msym} \tm \toe \tmthree$ then $\tmtwo \toe \tm' {\ }^{*\!\!\!}_{\msym\!\!}\leftarrow \tmthree$ for some $\tm'$ \cite[Lemma 4]{AccattoliPaolini12}.
		From that, it follows that $\tom$ and $\toe$ \emph{commute}, that is, if $\tmtwo {\ }^{*\!\!\!}_{\msym\!\!}\leftarrow \tm \toe^* \tmthree$ then $\tmtwo \toe^* \tm' {\ }^{*\!\!\!}_{\msym\!\!}\leftarrow \tmthree$ for some $\tm'$, which implies that $\tovsub \,=\, \tom \cup \toe$ is confluent by the Hindley-Rosen lemma \cite[Prop. 3.3.5]{Barendregt84}, since $\tom$ and $\toe$ are separately confluent (\Cref{l:basic-value-substitution}.\ref{p:basic-value-substitution-tom-toe-terminates}).
		
		Confluence of $\tovsubnvar$ is proven in a similar way. 
		Looking at the proof of \cite[Lemma 4]{AccattoliPaolini12}, we actually have that $\tom$ and $\toeabs$ \emph{local commute}: if $\tmtwo \lRew{\msym} \tm \toeabs \tmthree$ then $\tmtwo \toeabs \tm' {\ }^{*\!\!\!}_{\msym\!\!}\leftarrow \tmthree$ for some $\tm'$.
		From that, it follows that $\tom$ and $\toeabs$ \emph{commute}, that is, if $\tmtwo {\ }^{*\!\!\!}_{\msym\!\!}\leftarrow \tm \toeabs^* \tmthree$ then $\tmtwo \toeabs^* \tm' {\ }^{*\!\!\!}_{\msym\!\!}\leftarrow \tmthree$ for some $\tm'$, which implies that $\tovsubnvar \,=\, \tom \cup \toeabs$ is confluent by the Hindley-Rosen lemma \cite[Prop. 3.3.5]{Barendregt84}, since $\tom$ and $\toeabs$ are separately confluent (\Cref{l:basic-value-substitution}.\ref{p:basic-value-substitution-tom-toe-terminates} and \Cref{l:basic-value-substitution}.\ref{p:basic-value-substitution-toevar-diamond}).
		
		\item The second statement (``If a term is $\vsub$-normal then it is a \full fireball'') follows immediately from the first one (``A term is $\vsubnvarsym$-normal if and only if it is a \full fireball''), since $\tovsubnvar \,\subseteq\, \tovsub$.
		
		We prove the two directions of the equivalence in the first statement separately.
		\begin{enumerate}
			\item \label{p:normal-to-fullfireball} Let $\tm$ be $\vsubnvarsym$-normal. 
			We prove that $\tm$ is a \full fireball by induction on $\tm$.
			Cases:
			\begin{itemize}
				\item \emph{Variable}; \ie $\tm = \var$ for some variable $\var$. 
				So, $\var$ is a \full value and hence a \full fireball.
				
				\item \emph{Abstraction}; \ie, $\tm = \la{\var}{\tmtwo}$. 
				As $\tm$ is $\vsubnvarsym$-normal, so is $\tmtwo$.
				By \ih, $\tmtwo$ is a \full fireball, and then so~is~$\tm$.
				
				\item \emph{Application}; \ie, $\tm = \tm_{1} \tm_{2}$. 
				Since $\tm$ is $\vsubnvarsym$-normal,  so are $\tm_1$ and $\tm_2$.
				By \ih, $\tm_{1}$ and $\tm_{2}$ are \full fireballs. 
				Note that  $\tm_{1} \neq \sctxp{\la{\var}{\tmtwo}}$, otherwise $\tm \rtom \sctxp{\tmtwo \esub{\var}{\tm_{2}}}$ which contradicts $\vsubnvarsym$-normality of $\tm$. 
				So, according to the definition of \full fireball, $\tm_{1}$ is a \full inert term, thus $\tm$ is a \full fireball.
				
				\item \emph{Explicit substitution}; \ie, $\tm = \tm_{1} \esub{\var}{\tm_{2}}$. 
				Since $\tm$ is $\vsubnvarsym$-normal, then so are $\tm_1$ and $\tm_2$.
				By \ih, $\tm_{1}$ and $\tm_{2}$ are \full fireballs. 
				Note that $\tm_{2} \neq \sctxp{\la{\var}{\tmtwo}}$, otherwise $\tm \rtoeabs \sctxp{\tm_{1} \isub{\var}{\la{\var}{\tmtwo}}}$ which contradicts the $\vsubnvarsym$-normality of $\tm$. 
				Thus, according to the definition of \full fireball, $\tm_{2}$ is a \full inert term, and so $\tm$ is a \full fireball.
			\end{itemize}
			
			\item\label{p:fullfireball-to-normal} Let $\tm$ be a \full fireball. 
			We prove that $\tm$ is $\vsubnvarsym$-normal by induction on the definition of \full fireball.
			\begin{itemize}
				\item \emph{Variable}; \ie $\tm = \var$ for some variable $\var$. Clearly, $\tm$ is $\vsubnvarsym$-normal.
				
				\item \emph{Abstraction}; \ie, $\tm = \la{\var}{\sfire}$ for some \full fireball $\sfire$. 
				By \ih, $\sfire$ is $\vsubnvarsym$-normal, and hence so is $\tm$.
				
				\item \emph{Application}; \ie, $\tm = \sitm \sfire$ for some \full inert term $\sitm$ and \full fireball $\sfire$.
				By \ih, $\sitm $ and $\sfire$  are $\vsubnvarsym$-normal. 
				Since $\sitm$ is not of the form $\sctxp{\la{\var}{\tmtwo}}$, then $\tm$ is also $\vsubnvarsym$-normal.
				
				\item \emph{Explicit substitution}; \ie, $\tm = \sfire \esub{\var}{\sitm}$ for some \full inert term $\sitm$ and \full fireball $\sfire$ (it includes the case when $\sfire$ is a \full  inert term). 
				By \ih, both $\sfire$ and $\sitm$ are $\vsubnvarsym$-normal. 
				Since $\sitm$ is not of the form $\sctxp{\la{\var}{\tmtwo}}$, then $\tm$ is also $\vsubnvarsym$-normal.
				\qedhere
			\end{itemize}
		\end{enumerate}
	\end{enumerate}
\end{proof}

\begin{lemma}[Separate diamonds]
	\label{l:diamond-tom-toe-solv} Reductions $\tomsolv$, $\toeabssolv$, $\toevarsolv$, $\toesolv$ are diamond (separately).
\end{lemma}

\begin{proof}
	Given $\Rule \in \{\msym, \evarsym, \eabssym, \esym\}$, the proof that $\Rew{\solvsym\Rule}$ is diamond is analogous to the proof that $\Rew{\wsym\Rule}$ is diamond (\Cref{l:basic-value-substitution}.\ref{p:basic-value-substitution-tom-toe-diamond-open}), except for the following new case:
	\begin{itemize}
		\item \emph{Abstraction for both $\tm \Rew{\solvsym\Rule} \tmtwo$ and $\tm \Rew{\solvsym\Rule} \tmthree$}, \ie $\tm \defeq \la{\var}\tmfour \Rew{\solvsym\Rule} \la{\var}\tmfourp \eqdef \tmtwo$ and $\tm \Rew{\solvsym\Rule} \la{\var}\tmfour'' \eqdef \tmthree$ with $\tmfourp \lRew{\solvsym\Rule} \tmfour \toeabssolv \tmfour''$: 
		by \ih, there exists $\tmfive \in \Lambda_\vsub$ such that $\tmfourp \Rew{\solvsym\Rule} \tmfive \lRew{\solvsym\Rule} \tmfour''$, hence $\tmtwo \Rew{\solvsym\Rule} \la{\var}\tmfive \lRew{\solvsym\Rule} \tmthree$.
		\qedhere
	\end{itemize} 
\end{proof}

\begin{proposition}[Properties of the solving reduction]
	\label{propappendix:properties-solvable-reduction}
	\NoteState{prop:properties-solvable-reduction}
	\hfill
	\begin{enumerate}
		\item \label{pappendix:properties-solvable-reduction-commutation} 
		\emph{Strong commutation:} reductions $\tomsolv$, $\toeabssolv$, and $\toevarsolv$ are pairwise strongly commuting.
	
		\item \label{pappendix:properties-solvable-reduction-diamond}
		\emph{Diamond}: reductions $\tovsubsolv$ and $\tosolvnvar$ are diamond (separately).
		
		\item \label{pappendix:properties-solvable-reduction-harmony} 
		\emph{Normal forms}: a term is $\solvnvarsym$-normal if and only if it is a solved fireball.
		If a term is $\solvredsym$-normal then it is a solved fireball.
		
	\end{enumerate}
\end{proposition}

\begin{proof}
	In \cite{AccattoliPaolini12}, $\tovsubsolv$ is noted  $\to_\mathsf{sw}$, $\tomsolv$ is noted  $\to_\mathsf{db}$, and $\toesolv$ is noted  $\to_\mathsf{vs}$.
	\begin{enumerate}

		\item \cite[Lemma 4]{AccattoliPaolini12} proves that $\tomsolv$ and $\toesolv$ strongly commute. 
		Looking at that proof, we actually have that $\tomsolv$ and $\toevarsolv$ strongly commute, and that $\tomsolv$ and $\toeabssolv$ strongly commute. 
		Strong commutation of $\toeabssolv$ and $\toevarsolv$ is a straightforward generalization of strong commutation of $\toeabso$ and $\toevaro$. The proof is analogous, except for the new case:
		\begin{itemize}
			\item \emph{Abstraction for both $\tm \toevarsolv \tmtwo$ and $\tm \toeabssolv \tmthree$}, \ie $\tm \defeq \la{\var}\tmfour \toevarsolv \la{\var}\tmfourp \eqdef \tmtwo$ and $\tm \toeabssolv \la{\var}\tmfour'' \eqdef \tmthree$ with $\tmfourp \lRew{\solvsym\evarsym\!} \tmfour \toeabssolv \tmfour''$: 
			by \ih, there exists $\tmfive \in \Lambda_\vsub$ such that $\tmfourp \toeabssolv \tmfive \lRew{\solvsym\evarsym\!} \tmfour''$, hence $\tmtwo \toeabssolv \la{\var}\tmfive \lRew{\solvsym\evarsym\!} \tmthree$.
		\end{itemize}
		
		\item As $\tovsubsolv \,=\, \tomsolv \cup \toesolv$, diamond of $\tovsubsolv$ follows immediately from strong commutation of $\tomsolv$ and $\toesolv$ (\Cref{prop:properties-solvable-reduction}.\ref{p:properties-solvable-reduction-commutation}), and from diamond of $\tomsolv$ and $\toesolv$ (separately, \Cref{l:diamond-tom-toe-solv}). 
		
		As $\tosolvnvar \,=\, \tomsolv \cup \toeabssolv$, diamond of $\tosolvnvar$ follows immediately from strong commutation of $\tomsolv$ and $\toeabssolv$ (\Cref{prop:properties-solvable-reduction}.\ref{p:properties-solvable-reduction-commutation}), and from diamond of $\tomsolv$ and $\toeabssolv$ (separately, \Cref{l:diamond-tom-toe-solv}). 

		\item The second statement (``If a term is $\solvsym$-normal then it is a solved fireball'') follows immediately from the first one (``A term is $\solvnvarsym$-normal if and only if it is a solved fireball''), since $\tosolvnvar \,\subseteq\, \tosolv$.
		
		We prove the two directions of the equivalence in the first statement separately.
		\begin{enumerate}
			\item \label{p:normal-to-solvfireball} Let $\tm$ be $\solvnvarsym$-normal. 
			We prove that $\tm$ is a solved fireball by induction on $\tm$.
			Cases:
			\begin{itemize}
				\item \emph{Variable}, \ie $\tm = \var$ for some variable $\var$. 
				So, $\var$ is a inert term and hence a solved fireball.
				
				\item \emph{Abstraction}, \ie $\tm = \la{\var}{\tmtwo}$. 
				As $\tm$ is $\solvnvarsym$-normal, so is $\tmtwo$.
				By \ih, $\tmtwo$ is a solved fireball, and then so~is~$\tm$.
				
				\item \emph{Application}, \ie $\tm = \tm_{1} \tm_{2}$. 
				Since $\tm$ is $\solvnvarsym$-normal,  so are $\tm_1$ and $\tm_2$.
				By \ih, $\tm_{1}$ and $\tm_{2}$ are solved fireballs. 
				Note that  $\tm_{1} \neq \sctxp{\la{\var}{\tmtwo}}$, otherwise $\tm \rtom \sctxp{\tmtwo \esub{\var}{\tm_{2}}}$ which contradicts $\solvnvarsym$-normality of $\tm$. 
				So, according to the definition of solved fireball, $\tm_{1}$ is a inert term, thus $\tm$ is a solved fireball.
				
				\item \emph{Explicit substitution}; \ie, $\tm = \tm_{1} \esub{\var}{\tm_{2}}$. 
				Since $\tm$ is $\solvnvarsym$-normal, then so are $\tm_1$ and $\tm_2$.
				By \ih, $\tm_{1}$ and $\tm_{2}$ are solved fireballs. 
				Note that $\tm_{2} \neq \sctxp{\la{\var}{\tmtwo}}$, otherwise $\tm \rtoeabs \sctxp{\tm_{1} \isub{\var}{\la{\var}{\tmtwo}}}$ which contradicts the $\solvnvarsym$-normality of $\tm$. 
				Thus, according to the definition of solved fireball, $\tm_{2}$ is a inert term, and so $\tm$ is a solved fireball.
			\end{itemize}
			
			\item\label{p:solvfireball-to-normal} Let $\tm$ be a solved fireball. 
			We prove that $\tm$ is $\solvnvarsym$-normal by induction on the definition of solved fireball.
			\begin{itemize}
				\item \emph{Variable}; \ie $\tm = \var$ for some variable $\var$. Clearly, $\tm$ is $\solvnvarsym$-normal.
				
				\item \emph{Abstraction}; \ie, $\tm = \la{\var}{\solvnf}$ for some solved fireball $\solvnf$. 
				By \ih, $\solvnf$ is $\solvnvarsym$-normal, and hence so is $\tm$.
				
				\item \emph{Application}; \ie, $\tm = \itm \solvnf$ for some inert term $\itm$ and solved fireball $\solvnf$.
				By \ih, $\itm $ and $\solvnf$  are $\solvnvarsym$-normal. 
				Since $\itm$ is not of the form $\sctxp{\la{\var}{\tmtwo}}$, then $\tm$ is also $\solvnvarsym$-normal.
				
				\item \emph{Explicit substitution}; \ie, $\tm = \solvnf \esub{\var}{\itm}$ for some inert term $\itm$ and solved fireball $\solvnf$ (it includes the case when $\solvnf$ is a inert term). 
				By \ih, both $\solvnf$ and $\itm$ are $\solvnvarsym$-normal. 
				Since $\itm$ is not of the form $\sctxp{\la{\var}{\tmtwo}}$, then $\tm$ is also $\solvnvarsym$-normal.
				\qedhere
			\end{itemize}
		\end{enumerate}
	\end{enumerate}
\end{proof}

\begin{proposition}[Further properties of solving reduction]
	\label{propappendix:properties-solving}
	\NoteState{prop:properties-solving}
	\hfill
	\begin{enumerate}\setcounter{enumi}{-1}
		\item\label{pappendix:properties-solving-factorization} \emph{Factorization \cite{AccattoliPaolini12}:}
		if $\tm \tovsub^* \tmtwo$ then $\tm \tosolv^* \!\cdot\! \tonsolv^* \tmtwo$, where $\tonsolv \,\defeq\, \tovsub \smallsetminus \tosolv$.
		
		\item \emph{Normalization:} if $\tm \tovsub^{*} \tmtwo$ for some $\solvsym$-normal $\tmtwo$, then $\tm \tosolv^{*} \tmthree$ for some $\solvsym$-normal $\tmthree$.

		\item \emph{Normalization 2:} if $\tm \tovsubnvar^{*}\! \solvfire$ with $\solvfire$ solved fireball, then $\tm \tosolvnvar^{*} \solvfire'$ for some solved fireball $\solvfire'$\!.
		
		\item \emph{Stability by extraction from a head context:} if $\hctxp{\tm} \tosolv^* \tmtwo$ for some head context $\hctx$ and $\solvsym$-normal $\tmtwo$, then $\tm \tosolv^* \tmthree$ for some $\solvsym$-normal $\tmthree$.
	\end{enumerate}
\end{proposition}	

\begin{proof}
		\begin{enumerate}\setcounter{enumi}{-1}
			\item See \cite[Theorem 1.2]{AccattoliPaolini12}, where $\tosolv$ and $\tonsolv$ are noted $\to_\mathsf{sw}$ and $\to_{\lnot \mathsf{sw}}$, respectively.
			
			\item By factorization (\Cref{propappendix:properties-solving}.\ref{pappendix:properties-solving-factorization}), from $\tm \tovsub^{*} \tmtwo$ it follows that $\tm \tosolv^{*}\tmthree \tonsolv^{*}\tmtwo$ for some $\tmthree$. 
			Now, if in $\tmthree$ there is a redex out of non-head abstractions, it cannot be eliminated by reducing (via $\tonsolv$) into head abstractions. 
			Thus, if $\tmthree$ is not $\solvsym$-normal then $\tmtwo$ is not $\solvsym$-normal, absurd. 
			Therefore, $\tmthree$ is $\solvsym$-normal.	
			

			\item As $\toevarsolv \subseteq \toe$ is strongly normalizing  (\Cref{l:basic-value-substitution}.\ref{p:basic-value-substitution-tom-toe-terminates}), there is a $\evarsym$-normal term $\tmfour$ such that $\solvfire \toevarsolv^* \tmfour$. 
			As $\solvfire$ is $\solvnvarsym$-normal (\Cref{prop:properties-solvable-reduction}.\ref{p:properties-solvable-reduction-harmony}), $\tmfour$ is $\solvnvarsym$-normal too by forthcoming  \Cref{l:preservation-normal}.\ref{p:preservation-normal-solv} (proved independently below).
			By the normalization property for $\tosolv$ (\Cref{prop:properties-solving}.\ref{p:properties-solving-normalization}, since $\tm \tovsub^* \tmfour$ with $\tmfour$ $\solvsym$-normal), $\tm \tosolv^{*} \tmthree'$ for some  $\tmthree'$ that is $\solvsym$-normal and in particular $\solvnvarsym$-normal.
			By postponement of $\toevarsolv$ (\Cref{cor:postponement}.\ref{p:postponement-solv}, proved independently below), $\tm \tosolvnvar^* \tmthree \toevarsolv^* \tmthree'$ for some $\tmthree$ which is $\solvnvarsym$-normal by \Cref{l:preservation-normal}.\ref{p:preservation-normal-solv}.
		According to \Cref{prop:properties-solvable-reduction}.\ref{p:properties-solvable-reduction-harmony}, $\tmthree$ is a solved fireball.
			
			\item By contradiction, suppose that there is no $\tosolv$-normal $\tmthree$ such that $\tm \tosolv^* \tmthree$.
			Hence, for every $n \geq 0$, there is an evaluation $\deriv_n \colon \tm \tosolv^* \tm_n$ with $\size{\deriv_n} = n$.
			By \cite[Lemma 10]{AccattoliPaolini12}, for every $n \geq 0$, there is a reduction sequence $\deriv_n' \colon \hctxp{\tm} \tosolv^* \hctxp{\tm_n}$ with $\size{\deriv_n'} = n$.
			Since $\tosolv$ is diamond (\Cref{prop:properties-solvable-reduction}.\ref{p:properties-solvable-reduction-diamond}), there is no $\solvsym$-normal $\tmtwo$ such that $\hctxp{\tm} \tosolv^* \tmthree$. Absurd.
			\qedhere
		\end{enumerate}
\end{proof}

\paragraph{Irrelevance of Variable Exponential Steps.}

\begin{lemma}[Preservation of $\tovsubnvar$-normal forms]
	\label{l:preservation-normal}
 Let $\tm \toevar \tmtwo$. Then:
	\begin{enumerate}
		\item\label{p:preservation-normal-vsub} $\tm$ is $\vsubnvarsym$-normal if and only if $\tmtwo$ is $\vsubnvarsym$-normal;
		\item\label{p:preservation-normal-open}  $\tm$ is $\onvarsym$-normal if and only if $\tmtwo$ is $\onvarsym$-normal;
		\item\label{p:preservation-normal-solv} $\tm$ is $\solvnvarsym$-normal if and only if $\tmtwo$ is $\solvnvarsym$-normal.
	\end{enumerate}	
\end{lemma}

\begin{proof}
	Intuitively, if $\tm \toevar \tmtwo$, then $\tmtwo$ is roughly obtained from $\tm$ by erasing one of its \ES and replacing a variable with another variable. 
	Clearly, this operation cannot create new redexes for $\tovsubnvar$ (\resp $\tovsubonvar$; $\tosolvnvar$) in $\tmtwo$, or erase some redexes for $\tovsubnvar$ (\resp $\tovsubonvar$; $\tosolvnvar$) from $\tm$.
	Formally, the proofs are by induction on the definition of $\tm \toevar \tmtwo$ for all \Cref{p:preservation-normal-vsub,p:preservation-normal-open,p:preservation-normal-solv}. 
\end{proof}

\begin{lemma}[Swaps]
	\label{l:swaps}
	\hfill
	\begin{enumerate}
		\item\label{p:swaps-mult-local} $\toevar\tom \,\subseteq\, \tom\toevar$; 
		and $\toevarsolv \tomsolv \,\subseteq\, \tomsolv \toevarsolv$;
		and $\toevaro \tomo \,\subseteq\, \tomo \toevaro$;
		
		\item\label{p:swaps-exp-local} $\toevar\toeabs \,\subseteq\, \toeabs\toevar^* \!\cup \toeabs\toeabs$; 
		and $\toevarsolv\toeabssolv \,\subseteq\, \toeabssolv\toevarsolv^* \!\cup \toeabssolv\toeabssolv$;
		and $\toevaro\toeabso \,\subseteq\, \toeabso\toevaro^* \!\cup \toeabso\toeabso$;
		
		\item\label{p:swaps-exp-global} $\toevar^*\toeabs^+ \,\subseteq\, \toeabs^+\toevar^*$; 
		and $\toevarsolv^*\toeabssolv^+ \,\subseteq\, \toeabssolv^+\toevarsolv^*$;
		and $\toevaro^*\toeabso^+ \,\subseteq\, \toeabso^+\toevaro^*$.
	\end{enumerate}
\end{lemma}

\begin{proof}
	\begin{enumerate}
		\item By induction on the term $\tmtwo$ such that $\tm \toevar \tmtwo \tom \tmthree$ (or $\tm \toevarsolv \tmtwo \tomsolv \tmthree$).
		\item By induction on the term $\tmtwo$ such that $\tm \toevar \tmtwo \toeabs \tmthree$ (or $\tm \toevarsolv \tmtwo \toeabssolv \tmthree$).
		\item By \Cref{l:swaps-abstract}, since its hypotheses hold (\Cref{l:basic-value-substitution}.\ref{p:basic-value-substitution-tom-toe-terminates} and \Cref{l:swaps}.\ref{p:swaps-exp-local}).
		\qedhere
	\end{enumerate}
\end{proof}
%

\begin{corollary}[Postponement of variable steps]
	\label{cor:postponement}
	\hfill
	\begin{enumerate}
		\item\label{p:postponement-vsub} \emph{In $\tovsub$:} if $\deriv \colon \tm \,\tovsub^*\, \tmtwo$ then there is $\derivp \colon \tm \tovsubnvar^* \!\cdot \toevar^* \, \tmtwo$ with $\sizem\derivp = \sizem \deriv$;
		\item\label{p:postponement-open} \emph{In $\tovsubo$:} if $\deriv \colon \tm \,\tovsubo^*\, \tmtwo$ then there is $\derivp \colon \tm \tovsubonvar^* \!\cdot \toevaro^* \, \tmtwo$ with $\sizem\derivp = \sizem \deriv$;
		\item\label{p:postponement-solv} \emph{In $\tovsubsolv$:} if $\deriv \colon \tm \,\tosolv^*\, \tmtwo$ then there is $\derivp \colon \tm \tosolvnvar^* \!\cdot \toevarsolv^* \, \tmtwo$ with $\sizem\derivp = \sizem \deriv$.
	\end{enumerate}
\end{corollary}

\begin{proof}
	\begin{enumerate}
		\item By induction on the length $\size{\deriv}$ of the reduction $\deriv \colon \tm\tovsub^*\tmtwo$, using \Cref{l:swaps}.
		
		If $\size{\deriv} = 0$ then the statement trivially holds.
		
		Otherwise, $\size{\deriv} > 0$ and $\deriv$ is the concatenation of the reduction $\deriv_1 \colon \tm \tovsub^* \tmthree$ and $\tmthree \tovsub \tmtwo$, with $\size{\deriv_1} = \size{\deriv} -1$. 
		By \ih, there exists a reduction $\deriv_1' \colon \tm \tovsubnvar^* \tmfour \toevar^* \tmthree$ with $\sizem{\deriv_1'} = \sizem{\deriv_1}$.
		Note that $\sizem{\deriv_1'}$ is also the number of $\msym$-steps in $\tm \tovsubnvar^* \tmfour$, since there are no $\msym$-steps in $\tmfour \toevar^* \tmthree$.
		There are three sub-cases, depending on $\tmthree \tovsub \tmtwo$. 
		\begin{enumerate}
			\item \emph{Multiplicative:} if $\tmthree \tom \tmtwo$, from $\tmfour \toevar^* \tmthree \tom \tmtwo$ we have $\tmfour \tom \tmthree' \toevar^* \tmtwo$ by applying \Cref{l:swaps}.\ref{p:swaps-mult-local} $n$ times, where $n$ is the number of steps in the reduction $\tmfour \toevar^* \tmthree$.
			Therefore, there exists $\derivp \colon \tm \tovsubnvar^* \tmfour \tom \tmthree'` \toevar^* \tmtwo$ with $\sizem{\derivp} = \sizem{\deriv_1'} + 1 = \sizem{\deriv_1} + 1  =\sizem{\deriv}$.
			
			\item \emph{Exponential variable:} if $\tmthree \toevar \tmtwo$, then we have $\deriv' \colon \tm \tovsubnvar \tmfour \toevar^* \tmthree \toevar \tmtwo$  with $\sizem{\deriv'} = \sizem{\deriv_1'} = \sizem{\deriv_1} = \sizem{\deriv}$.
			
			\item \emph{Exponential abstraction:} if $\tmthree \toeabs \tmtwo$, from $\tmfour \toevar^* \tmthree \toeabs \tmtwo$ we have $\tmfour \toeabs^+ \tmthree' \toevar^* \tmtwo$ by applying \Cref{l:swaps}.\ref{p:swaps-exp-global}.
			Therefore, there exists $\derivp \colon \tm \tovsubnvar^* \tmfour \toeabs^+ \tmthree' \toevar^* \tmtwo$ with $\sizem{\derivp} = \sizem{\deriv_1'} = \sizem{\deriv_1} = \sizem{\deriv}$.
			
		\end{enumerate}
	
		\item By induction on the length $\size{\deriv}$ of the reduction $\deriv \colon \tm\tovsubo^*\tmtwo$, similarly to \Cref{p:postponement-vsub}.
		\item By induction on the length $\size{\deriv}$ of the reduction $\deriv \colon \tm\tovsubsolv^*\tmtwo$, similarly to \Cref{p:postponement-vsub}.
		\qedhere
	\end{enumerate}
\end{proof}

\begin{proposition}[Irrelevance of $\toevar$]
	\label{propappendix:irrelevance}
	\NoteState{prop:irrelevance}
	\begin{enumerate}
		\item\label{pappendix:irrelevance-vsub} Reduction $\toevar$ is $\tovsubnvar$-irrelevant.
				\item\label{pappendix:irrelevance-open} Reduction $\toevaro$ is $\tovsubonvar$-irrelevant.
		\item\label{pappendix:irrelevance-solv} Reduction $\toevarsolv$ is $\tosolvnvar$-irrelevant.
	\end{enumerate}
\end{proposition}

\begin{proof}
	\begin{enumerate}
		\item By \Cref{l:irrelevance-abstract}, since its hypotheses hold (\Cref{cor:postponement}.\ref{p:postponement-vsub}, \Cref{l:basic-value-substitution}.\ref{p:basic-value-substitution-tom-toe-terminates} and \Cref{l:preservation-normal}.\ref{p:preservation-normal-vsub}).
		Note that $\toevar$ is strongly normalizing because so is $\toe$ (\Cref{l:basic-value-substitution}.\ref{p:basic-value-substitution-tom-toe-terminates}), and $\toevar \,\subseteq\, \toe$.
		
		\item By \Cref{l:irrelevance-abstract}, since its hypotheses hold (\Cref{cor:postponement}.\ref{p:postponement-open}, \Cref{l:basic-value-substitution}.\ref{p:basic-value-substitution-tom-toe-terminates} and \Cref{l:preservation-normal}.\ref{p:preservation-normal-open}).
		Note that $\toevaro$ is strongly normalizing because so is $\toe$ (\Cref{l:basic-value-substitution}.\ref{p:basic-value-substitution-tom-toe-terminates}), and $\toevaro \,\subseteq\, \toe$.
		
		\item By \Cref{l:irrelevance-abstract}, since its hypotheses hold (\Cref{cor:postponement}.\ref{p:postponement-solv}, \Cref{l:basic-value-substitution}.\ref{p:basic-value-substitution-tom-toe-terminates} and \Cref{l:preservation-normal}.\ref{p:preservation-normal-solv}).
		Note that $\toevarsolv$ is strongly normalizing because so is $\toe$ (\Cref{l:basic-value-substitution}.\ref{p:basic-value-substitution-tom-toe-terminates}), and $\toevarsolv \,\subseteq\, \toe$.
		\qedhere
	\end{enumerate}
\end{proof}

\begin{lemma}[Stability of reductions without $\toevar$ under substitution]
	\label{lappendix:stability-substitution}
	\NoteState{l:stability-substitution}
	Let $a \in \{\onvarsym, \solvnvarsym, \vsubnvarsym\}$.
	If $\tm \Rew{a} \tmtwo$ then $\tm\isub\var\tmthree \Rew{a} \tmtwo\isub\var\tmthree$ for every $\tmthree$.
\end{lemma}

\begin{proof}
By induction on $\tm \Rew{a} \tmtwo$, with $a \in \{\onvarsym, \solvnvarsym, \vsubnvarsym\}$. 
Note that $\sctx\isub\var\tmthree$ is still a substitution context. Cases:
\begin{itemize}
\item \emph{Root multiplicative}: $\tm = \sctxp{\la\vartwo\tm'} \tmfour \tom\sctxp{\tm'\esub\vartwo{\tmfour}} = \tmtwo$. Then:
\[\tm\isub\var\tmthree = \sctx\isub\var\tmthree\ctxholep{\la\vartwo\tm'\isub\var\tmthree} (\tmfour\isub\var\tmthree) \tom  \sctx\isub\var\tmthree \ctxholep{\tm'\isub\var\tmthree\esub\vartwo{\tmfour\isub\var\tmthree}} = \tmtwo\isub\var\tmthree.\]

\item \emph{Root exponential-abstraction}: $\tm = \tm' \esub\vartwo {\sctxp{\la\varthree\tmfour}}\rtoeabs \sctxp{\tm' \isub\vartwo {\la\varthree\tmfour}} = \tmtwo$. Then:
\[\tm\isub\var\tmthree = \tm'\isub\var\tmthree \esub\vartwo {\sctx\isub\var\tmthree\ctxholep{\la\varthree\tmfour\isub\var\tmthree}}
\rtoeabs
 \sctx\isub\var\tmthree\ctxholep{\tm'\isub\var\tmthree \isub\vartwo {\la\varthree\tmfour\isub\var\tmthree}} = \tmtwo\isub\var\tmthree.\]

\item \emph{Inductive cases}: they follow from the \ih and the immediate fact that open, solving and full contexts are stable by substitution.
\end{itemize}
\end{proof}
\section{Proofs of Section \ref{sect:other-calculi} (Plotkin, Shuffling, and Moggi)}

\begin{proposition}[Simulation]
	\label{propappendix:plotkin-vsc}
	\NoteState{prop:plotkin-vsc}
	Let $\tm$ and $\tm'$ be terms without \ES. 
	If $\tm \toobvplot \tm'$ then $\tm \tomo \!\cdot \toeo  \tm'$; and if $\tm \tofbvplot \tm'$ then $\tm \tom \!\cdot \toe  \tm'$.
\end{proposition}

\begin{proof}
	The proof that $\tm \Rew{\wsym\betaplot} \tm'$ implies $\tm \tomo \!\cdot\! \toeo \tm'$ (\resp $\tm \Rew{\fullsym\betaplot} \tm'$ implies $\tm \tom \!\cdot\! \toe \tm'$)
	is by induction on the context $\weakctx$ (\resp $\fctx$) such that $\tm = \weakctxp{\tmthree} \toobvplot \weakctxp{\tmthree'} = \tm'$ (\resp $\tm = \fctxp{\tmthree} \tofbvplot \fctxp{\tmthree'} = \tm'$) with $\tmthree \rtobvplot \tmthree'$.
			
	If $\weakctx = \ctxhole$ (\resp $\fctx = \ctxhole$), then we have the root-step $\tm = (\la{\var}\tmthree)\val \rtobvplot \tmthree \isub{\var}{\val} = \tm'$.  
	Then, $\tm \rtom \tmthree \esub{\var}{\val} \rtoe \tmthree \isub{\var}{\val} = \tm'$.
			
	The cases where $\weakctx \neq \ctxhole$ (\resp $\fctx \neq \ctxhole$) follow  from the \ih easily.
\end{proof}

A term is \emph{\shallow} if it is of the form $\tm_0\esub{\var_1}{\tm_1}\dots \esub{\var_n}{\tm_n}$ for some $n \geq 0$ and some terms $\tm_0, \dots, \tm_n$ without \ES.
	
	\begin{lemma}[Preservation of \shallow under open reduction]
		\label{l:preserve-shallow}
		If $\tm$ is \shallow and $\tm \tovsubo^* \tmtwo$, then $\tmtwo \eqstruct \tmthree$ for some \shallow $\tmthree$.
	\end{lemma}
	
	\begin{proof}
		It suffices to prove the statement for $\tm \tovsubo \tmtwo$. 
		We proceed by induction on $\tm \tovsubo \tmtwo$.
		Cases:
		\begin{itemize}
			\item \emph{Multiplicative root step}, \ie $\tm \defeq \subctxp{\la{\var}{\tm_1}}\tm_2 \rtom \subctxp{\tm_1 \esub{\var}{\tm_2}} \eqdef \tmtwo$.
			Since $\tm$ is \shallow, $\subctx = \ctxhole$ and $\tm_1$ and $\tm_2$ are without \ES.
			Hence, $\tmtwo$ is \shallow. 
			
			\item \emph{Exponential root step}, \ie $\tm \defeq \tm'\esub{\var}{\subctxp{\val}} \rtoe \subctxp{\tm' \isub{\var}{\val}} \eqdef \tmtwo$.
			Since $\tm$ is \shallow, $\subctx = \ctxhole$ and $\tm'$ and $\val$ are without \ES.
			Hence, $\tmtwo$ is without \ES and in particular \shallow.
			
			\item \emph{Application left}, \ie $\tm \defeq \tm_1\tm_2 \tovsubo \tm_1' \tm_2 \eqdef \tmtwo$ with $\tm_1 \tovsubo \tm_1'$. 
			Since $\tm$ is \shallow, $\tm_1$ and $\tm_2$ are without \ES and in particular \shallow.
			By \ih, $\tm_1' \eqstruct \tmthree'$ for some \shallow $\tmthree' \defeq \tmthree_0 \esub{\var_1}{\tmthree_1}\dots\esub{\var_n}{\tmthree_n}$ where $n \geq 0$ and $\tmthree_0, \dots, \tmthree_n$ are without \ES. 
			Thus, $\tm \tovsubo \tm_1'\tm_2 \eqstruct \tmthree'\tm_2 \eqstruct (\tmthree_0\tm_2)\esub{\var_1}{\tmthree_1}\dots\esub{\var_n}{\tmthree_n}$, which is \shallow.
			
			\item \emph{Application right}, \ie $\tm \defeq \tm_1\tm_2 \tovsubo \tm_1 \tm_2' \eqdef \tmtwo$ with $\tm_2 \tovsubo \tm_2'$.
			Analogous to the previous case.
			
			\item \emph{\ES left}, \ie $\tm \defeq \tm_1\esub{\vartwo}{\tm_2} \tovsubo \tm_1' \esub{\vartwo}{\tm_2} \eqdef \tmtwo$ with $\tm_1 \tovsubo \tm_1'$.
			Since $\tm$ is \shallow, $\tm_1$ and $\tm_2$ are without \ES and in particular \shallow.
			By \ih, $\tm_1' \eqstruct \tmthree'$ for some \shallow $\tmthree'$. 
			Thus, $\tm \tovsubo \tm_1'\esub{\vartwo}{\tm_2} \eqstruct \tmthree'\esub{\vartwo}{\tm_2}$, which is \shallow.
			
			\item \emph{\ES right}, \ie $\tm \defeq \tm_1\esub{\vartwo}{\tm_2} \tovsubo \tm_1 \esub{\vartwo}{\tm_2'} \eqdef \tmtwo$ with $\tm_2 \tovsubo \tm_2'$.
			Since $\tm$ is \shallow, $\tm_1$ and $\tm_2$ are without \ES and in particular \shallow.
			By \ih, $\tm_2' \eqstruct \tmthree'$ for some \shallow $\tmthree' \defeq \tmthree_0 \esub{\var_1}{\tmthree_1}\dots\esub{\var_n}{\tmthree_n}$ where $n \geq 0$ and $\tmthree_0, \dots, \tmthree_n$ are without \ES. 
			Thus, $\tm \tovsubo \tm_1\esub{\vartwo}{\tm_2'} \eqstruct \tm_1\esub{\vartwo}{\tmthree'} \eqstruct (\tm_1\tmthree_0)\esub{\var_1}{\tmthree_1}\dots\esub{\var_n}{\tmthree_n}$, which is \shallow.
			\qedhere
		\end{itemize}
\end{proof}

\begin{lemma}[Lifting valuability]
	\label{lappendix:valuing-sequences-lift-to-plotkin}
	\NoteState{l:valuing-sequences-lift-to-plotkin}
		If $\tm \tovsubo^{*} \val$ and $\tm$ is without \ES, then $\val$ is without \ES and $\tm\tobvploto^{*}\!\!\val$.
\end{lemma}

\begin{proof}
	Let us first prove that $\val$ is without \ES.
		As $\val$ is a value,  $\val$ is either a variable or an abstraction.
		In the first case, clearly $\val$ is without \ES. 
		In the second case, by preservation of \shallow under open reduction (\Cref{l:preserve-shallow}, since $\tm$ is without \ES and in particular \shallow), $\tm \tovsubo \val \eqstruct \tmtwo$ for some \shallow $\tmtwo$. 
		By definition of $\eqstruct$, $\tmtwo$ is also an abstraction and hence without \ES (as it is \shallow). 
		Therefore, from $\val \eqstruct  \tmtwo$ it follows that $\val = \tmtwo$.
	
	Let us now prove by contradiction that $\tm \tobvploto^* \val$. 
	First, it cannot be that $\tm\tobvploto^{*}\tmtwo$ for some $\osym$-normal  $\tmtwo \neq \val$, because $\tobvploto^{*}$ is a sub-relation of $\tovsubo^{*}$, which is confluent (\Cref{prop:properties-open-reduction}.\ref{p:properties-open-reduction-diamond}), and $\val$ and $\tmtwo$ would be two different $\osym$-normal forms of $\tm$ (\Cref{prop:properties-open-reduction}.\ref{p:properties-open-reduction-harmony}), which is absurd. 
	Note also that $\tobvploto$ cannot diverge on $\tm$ because $\tobvploto^*$ is a sub-relation of $\tovsubo^*$ which by hypothesis terminates on $\tm$ (since $\tovsubo$ is diamond, \Cref{prop:properties-open-reduction}.\ref{p:properties-open-reduction-diamond}).
	
	Hence, it must be that $\tobvploto$ terminates on a $\obetaplot$-normal $\tmtwo$ that is not $\osym$-normal.  
	By \Cref{l:equiv-solv-inert-redex}, $\tmtwo = \weakctxp{(\la\var\tmthree)\itm}$. 
	Since $\tovsubo$ is diamond (\Cref{prop:properties-open-reduction}.\ref{p:properties-open-reduction-diamond}) and all the steps of $\tobvploto$ can be seen as $\tovsubo$-sequences of length 2 (\Cref{prop:plotkin-vsc}), we have $\tm \tovsubo^{*} \tmtwo = \weakctxp{(\la\var\tmthree)\itm}$. 
	Then $\tmtwo \tomo \weakctxp{\tmthree\esub\var\itm} \tovsubo^{*} \val$ by the diamond property. Note however that no $\tovsubo$ step can erase, or substitute, or move under abstraction the \ES $\esub\var\itm$, so that it survives out of abstractions to the $\tovsubo^{*}$ sequence, and it must then be (out of abstractions) in the normal form $\val$, absurd.
\end{proof}

\begin{lemma}[Stability of contextual equivalence by \ES expansion]
	\label{lappendix:es-exp-and-ctx-eq}
	\NoteState{l:es-exp-and-ctx-eq}
	Let $\tm, \tmtwo \in \vsubterms$: one has
	$\tm \ctxeq^{\VSC} \tmtwo$ if and only if $\tm^\bullet \ctxeq^{\VSC} \tmtwo^\bullet$ (if and only if $\tm^\bullet \ctxeq^{\plotcalc} \tmtwo^\bullet$).
\end{lemma}

\begin{proof}
	It simply follows from the fact that, for $\tm \in \vsubterms$, we have $\tm^\bullet \tom^* \tm$ and so $\tm \ctxeq^{\VSC} \tm^\bullet$ since $\ctxeq^{\VSC}$ contains the symmetric closure of $\tom$. 
\end{proof}

\begin{proposition}[$\eqstruct$ is a strong bisimulation]
	\label{propappendix:strong-bisimulation}
	\NoteState{prop:strong-bisimulation} 
	If $\tm\eqstruct\tmtwo$ and $\tm\Rew{\mathsf{a}}\tmp$ then there exists $\tmtwop \in \vsubterms$ such that 
	$\tmtwo\Rew{\mathsf{a}}\tmtwop$ and $\tmp\eqstruct\tmtwop$, for $\mathsf{a}\in\set{\msym,\esym, 
		\osym\msym,\osym\esym, \ssym\msym,\ssym\esym}$. 
\end{proposition}     

%

\begin{proof}
	The proof that $\eqstruct$ is a strong bisimulation with respect to $\Rew{a}$ where $\mathsf{a}\in\set{\msym,\esym, 
		\osym\msym,\osym\esym, \allowbreak \ssym\msym,\ssym\esym}$ is in \cite[Lemma 12]{AccattoliPaolini12}.
\end{proof}

\begin{lemma}[Reduction modulo $\streq$ is confluent]
	\label{lappendix:confluence-modulo-eqstruct}
	\NoteState{l:confluence-modulo-eqstruct} %
	Reduction $\Rew{\vsub/\eqstruct}$ is confluent.
\end{lemma}

\begin{proof}
	It follows from confluence of $\tovsub$ and strong bisimulation of $\streq$. Indeed, let $\tm \Rew{\vsub/\eqstruct}^{*} \tmtwo_{1}$ and $\tm \Rew{\vsub/\eqstruct}^{*} \tmtwo_{2}$. By strong bisimulation, $\tm \tovsub^{*} \tmtwo_{1}'\streq\tmtwo_{1}$ and $\tm \tovsub^{*} \tmtwo_{2}'\streq \tmtwo_{2}$ for some $\tmtwo_{1}'$ and $\tmtwo_{2}'$. By confluence of $\tovsub$, $\tmtwo_{1}'\tovsub^{*} \tmthree$ and $\tmtwo_{2}'\tovsub^{*} \tmthree$ (that is, $\tmtwo_{1}\streq\tovsub^{*} \tmthree$ and $\tmtwo_{2}\streq\tovsub^{*} \tmthree$)  for some $\tmthree$.
\end{proof}

\paragraph{Proofs about Moggi's calculus and the glueing rule.}

\begin{lemma}[Strong normalization of $\toglue$]
	\label{l:strongly-normalizing-glue}
	The reduction $\toglue $ is strongly normalizing.
\end{lemma}

\begin{proof}
	Each $\toglue$-step strictly decreases the number of \ES.
\end{proof}

\begin{lemma}[Preservation of normal forms]
	\label{l:preservation-normal-glue}
	Let $\tm \toglue \tmtwo$. Then, $\tm$ is $\vsub$-normal if and only if $\tmtwo$ is $\vsub$-normal.
\end{lemma}

\begin{proof}
	By induction on the definition of $\tm \toglue \tmtwo$. 
\end{proof}

\begin{lemma}[Swaps for $\toglue$]
	\label{l:swaps-glue}
	\hfill
	\begin{enumerate}
		\item $\toglue\tom \subseteq \tom\toglue$;
		\item $\toglue\toe \subseteq \toe\toglue$.
	\end{enumerate}
\end{lemma}

\begin{proof}
	By straightforward inspection of all possible cases.
\end{proof}

\begin{corollary}[Postponement of glue steps]
	\label{cor:postponement-glue}
	If $\deriv \colon \tm \,(\tovsub \!\cup \toglue)^*\, \tmtwo$ then there is $\derivp \colon \tm \tovsub^* \!\cdot \toglue^* \, \tmtwo$ with $\sizem\derivp = \sizem \deriv$.
\end{corollary}

\begin{proof}
	By induction on the length $\size{\deriv}$ of the reduction sequence $\deriv:\tm \,(\tovsub \!\cup \toglue)^* \tmtwo$, using \Cref{l:swaps-glue}.
\end{proof}

\begin{proposition}[$\toglue$ Irrelevance]
	\label{propappendix:glue-irrelevance}
	\NoteState{prop:glue-irrelevance}
	The reduction $\toglue$ is $\tovsub$-irrelevant.
\end{proposition}

\begin{proof}
	By \Cref{l:irrelevance-abstract}, since its hypotheses hold (\Cref{cor:postponement-glue}, \Cref{l:strongly-normalizing-glue} and \Cref{l:preservation-normal-glue}).
\end{proof}

\begin{proposition}[$\eqstruct$ is a strong bisimulation for $\toglue$]
	\label{propappendix:strong-bisimulation-glue}
	\NoteState{prop:strong-bisimulation-glue} 
	If $\tm\eqstruct\tmtwo$ and $\tm\toglue\tmp$ then there exists $\tmtwop \in \vsubterms$ such that 
	$\tmtwo\toglue\tmtwop$ and $\tmp\eqstruct\tmtwop$. 
\end{proposition}     
\begin{proof}
	To prove that $\eqstruct$ is a strong bisimulation with respect to $\toglue$, first consider the relation
	\begin{equation}\label{eq:structural}
		\weakctxp{\tm\esub{\var}{\tmtwo}} \sim \weakctxp{\tm}\esub{\var}{\tmtwo}
	\end{equation}
	where $\weakctx$ does not capture the free variables of $\tmtwo$ and $\var$ does not occur free in $\weakctx$.
	Note that $\eqstruct$ can be equivalently defined as the closure of $\sim$ by reflexivity, symmetry, transitivity and full context $\fctx$. 
	We prove that $\sim$ is a strong bisimulation with respect to $\toglue$ by induction on $\weakctx$. 
	From that, the fact that $\eqstruct$ is a strong bisimulation with respect to $\toglue$ follows immediately.
\end{proof}

\section{Proofs of Section \ref{sect:solving-strat} (Call-by-Value Solvability and Scrutability)}
\label{sect:solving-strat-proofs}

\begin{lemma}[Reducing to values]
	\label{l:reduce-to-value}
	\hfill
	\begin{enumerate}
		\item\label{p:reduce-to-value-app} If $\tm\tmtwo\tmthree \tovsub^* \val$ for some value $\val$, then $\tm \tovsub^* \valtwo$ for some value $\valtwo$.
		\item\label{p:reduce-to-value-beta} If $(\la{\var}\tm)\tmtwo \tovsub^* \val$ for some value $\val$, then $\tmtwo \tovsub^* \valtwo$ for some value $\valtwo$.
		
		\item\label{p:reduce-to-value-compose} Let $n \geq 0$. 
		If $\tm\subs{\var_1}{\val_1}{\var_n}{\val_n} \allowbreak\tovsub^* \val$ for some value $\val$, then $\tm\subs{\var_1}{\val_1}{\var_n}{\val_n} \allowbreak\tovsub^* \valtwo$ for some value $\valtwo$, where $\valtwo_i \defeq \val_i \subs{\var_{i+1}}{\valtwo_{i+1}}{\var_n}{\valtwo_n}$ for all $1 \leq i \leq n$ (so, $\valtwo_n \defeq \val_n$).
	\end{enumerate}
\end{lemma}

\begin{proof}
	\begin{enumerate}
		\item If $\tm\tmtwo$ cannot \VSC-reduce to a value, then for every $\tmfive$ such that $\tm\tmtwo\tmthree \tovsub^* \tmfive$ we have $\tmfive = \tmfour' \tmthree'$ with $\tm\tmtwo \tovsub^* \tmfour'$ and $\tmthree \tovsub^* \tmthree'$. 
		Hence, $\tm\tmtwo\tmthree$ cannot \VSC-reduce to a value.
		
		\item  If $\tmtwo$ cannot \VSC-reduce to a value, then for every $\tmthree$ such that $(\la{\var}\tm)\tmtwo \tovsub^* \tmthree$ we have either $\tmthree = (\la{\var}\tm')\tmtwo'$ or $\tmthree = \tm' \esub{\var}{\tmtwo'}$, with $\tm \tovsub^* \tm'$ and $\tmtwo \tovsub^* \tmtwo'$. 
		In any case, $(\la{\var}\tm)\tmtwo$ cannot \VSC-reduce to a value.
		
		\item Proof by induction on $n \geq 0$. 
		
		If $n = 0$ then $\tm\subs{\var_1}{\val_1}{\var_n}{\val_n} = \tm$  and the property holds with $\valtwo = \val$.
		
		Assume that $\tm\subs{\var_1}{\valtwo_1}{\var_n}{\valtwo_n} \allowbreak\tovsub^* \val$ for some value $\val$ and some $n \geq 0$ (\ih), where $\valtwo_i \defeq \val_i \subs{\var_{i+1}}{\valtwo_{i+1}}{\var_n}{\valtwo_n}$ for all $1 \leq i \leq n$.
		Let $\valthree_{n+1} \defeq \val_{n+1}$ and $\valthree_i \defeq \valtwo_i \subs{\var_{i+1}}{\valthree_{i+1}}{\var_{n+1}}{\valthree_{n+1}}$ for all $1 \leq i \leq n$.
		So, $\tm\subs{\var_1}{\valthree_1}{\var_{n+1}}{\valthree_{n+1}} = \tm\subs{\var_1}{\valtwo_1}{\var_n}{\valtwo_n} \isub{\var_{n+1}}{\val_{n+1}} \allowbreak \tovsub^* \val\isub{\var_{n+1}}{\val_{n+1}}$ which is a value.
		\qedhere
	\end{enumerate}
\end{proof}

\begin{proposition}[Equivalent definition of \VSC-scrutability]
	\label{prop:equiv-def-scrutability}
	A term $\tm$ is \VSC-scrutable if and only if there are variables $\var_1, \dots, \var_n$ and 
	values $\val, \val_1, \dots, \val_n$ (with $n \geq 0$) such that $\tm\subs{\var_1}{\val_1}{\var_n}{\val_n} \allowbreak\tovsub^* \val$ (where $\tm\subs{\var_1}{\val_1}{\var_n}{\val_n}$ stands for the simultaneous substitution).
\end{proposition} 

\begin{proof}\hfill
	\begin{description}
		\item[$\Leftarrow$:] If $\tm\subs{\var_1}{\val_1}{\var_n}{\val_n} \allowbreak\tovsub^* \val$ then take the \balanced context $\bctx \defeq (\la{\var_n} \dots (\la{\var_1} \ctxhole)\val_1 \dots )\val_n$:  so $\bctxp{\tm}  = (\la{\var_n} \dots (\la{\var_1} \tm)\val_1 \dots )\val_n \tovsub^* \tm\isub{\var_1}{\val_1}\dots \isub{\var_n}{\val_n} = \tm\subs{\var_1}{\valtwo_1}{\var_n}{\valtwo_n}$ where $\valtwo_i \defeq \val_i \subs{\var_{i+1}}{\valtwo_{i+1}}{\var_n}{\valtwo_n}$ for all $1 \leq i \leq n$ (thus, $\valtwo_n \defeq \val_n$).
		By \Cref{l:reduce-to-value}.\ref{p:reduce-to-value-compose}, $\tm\subs{\var_1}{\valtwo_1}{\var_n}{\valtwo_n} \tovsub^* \valtwo$ for some value $\valtwo$.
		Hence, $\bctxp{\tm} \tovsub^* \valtwo$, \ie $\tm$ is \VSC-scrutable.
		
		\item[$\Rightarrow$:] If $\tm$ is \VSC-scrutable, then there is a \balanced context $\bctx$ such that $\bctxp{\tm} \tovsub^* \val$ for some value $\val$.
		By \Cref{l:reduce-to-value}.\ref{p:reduce-to-value-app}-\ref{p:reduce-to-value-beta}, we can assume without loss of generality $\bctx \defeq (\la{\var_n} \dots (\la{\var_1} \ctxhole)\val_1 \dots )\val_n$ for some variables $\var_1, \dots, \var_n$ and some values $\val_1, \dots, \val_n$.
		So, $\bctxp{\tm} \tovsub^* \tm\isub{\var_1}{\val_1}\dots\isub{\var_n}{\val_n} \allowbreak= \tm\subs{\var_1}{\valtwo_1}{\var_n}{\valtwo_n}$ where $\valtwo_i \defeq \val_i \subs{\var_{i+1}}{\valtwo_{i+1}}{\var_n}{\valtwo_n}$ (which is a value) for all $1 \leq i \leq n$.
		Since $\tovsub^*$ is confluent, $\val \tovsub^* \tmtwo$ and $\tm\subs{\var_1}{\valtwo_1}{\var_n}{\valtwo_n} \tovsub^* \tmtwo$ for some $\tmtwo$, which is a value because so is $\val$.
		\qedhere
	\end{description} 
\end{proof}

\begin{corollary}[Operational characterization of \cbv solvability/scrutability, Bis]
	\label{propappendix:operational-characterization-cbv-novar}
	\NoteState{prop:operational-characterization-cbv-novar}
\begin{enumerate}
	\item\label{pappendix:operational-characterization-cbv-novar-scrut} \emph{\VSC-Scrutability via  $\tovsubonvar$:} a term $\tm$ is \VSC-scrutable if and only if $\tovsubonvar$ terminates on $\tm$.
	
	\item\label{pappendix:operational-characterization-cbv-novar-solv} \emph{\VSC-Solvability via  $\tosolvnvar$:} a term $\tm$ is \VSC-solvable if and only if $\tosolvnvar$ terminates on $\tm$.
\end{enumerate}
\end{corollary}

\begin{proof}
	\begin{enumerate}
		\item By $\tovsubonvar$-irrelevance of $\toevaro$  (\Cref{{prop:irrelevance}}) and  
		the operational characterizations of \VSC-scrutability (\Cref{prop:operational-characterization-cbv-var}.\ref{p:operational-characterization-cbv-var-scrut}), $\tovsubonvar$ terminates on $\tm$ if and only if $\tovsubo$ terminates on $\tm$ if and only if $\tm$ is \VSC-scrutable.
		
		\item By $\tosolvnvar$-irrelevance of $\toevarsolv$  (\Cref{{prop:irrelevance}}) and  
		the operational characterizations of \VSC-solvability (\Cref{prop:operational-characterization-cbv-var}.\ref{p:operational-characterization-cbv-var-solv}), $\tosolvnvar$ terminates on $\tm$ if and only if $\tosolv$ terminates on $\tm$ if and only if $\tm$ is \VSC-solvable.
		\qedhere
	\end{enumerate}
	
\end{proof}

\paragraph{Equivalence with Solvability and Scrutability in Plotkin's Calculus.}

\begin{lemma}
\label{l:equiv-solv-inert-redex}
Let $\tm$ be a term without \ES. 
If $\tm$ is $\obetaplot$-normal but not $\osym$-normal, then $\tm = \weakctxp{(\la\var\tmtwo)\itm}$ for some term $\tmtwo$ and some non-variable inert term $\itm$, both without \ES.
\end{lemma}

\begin{proof}
If $\tm$ is $\obetaplot$-normal but not $\osym$-normal, then it has a $\tomo$-redex, because being without \ES it cannot have $\toeo$-redexes. 
Among the various $\tomo$-redexes in $\tm$, there must be one such that its argument contains no other $\tomo$-redexes, that is, $\tm = \weakctxp{(\la\var\tmtwo)\tmthree}$ for some terms $\tmtwo$ and $\tmthree$ without \ES, with $\tmthree$ $\osym$-normal. 
By \Cref{prop:properties-open-reduction}.\ref{p:properties-open-reduction-harmony}, $\tmthree$ is a fireball. 
It cannot be an abstraction nor a variable, otherwise $\tm$ would have a $\tobvploto$-redex. 
So, $\tmthree$ is a inert term.
\end{proof}

%

\begin{theorem}[Robustness of \cbv solvability and scrutability]
	\label{thmappendix:robust}
	\NoteState{thm:robust}
	Let $\tm$ be a term without \ES.
	
	\begin{enumerate}
		\item \label{pappendix:robust-scrutable} \emph{\cbv Scrutability:} $\tm$ is \VSC-scrutable if and only if $\tm$ is $\plotcalc$-scrutable.
		
		\item \label{pappendix:robust-solvable} \emph{\cbv Solvability:} $\tm$ is \VSC-solvable if and only if $\tm$ is $\plotcalc$-solvable.
		
		\item \label{pappendix:robust-ES} \emph{With/without ES}: for every term $\tm \in \vsubterms$, $\tm$ is \cbv scrutable (resp. solvable) if and only if $\tm^\bullet$ is \cbv scrutable (resp. solvable). 
	\end{enumerate} 
\end{theorem}

\begin{proof}
	\begin{enumerate}
		\item See the proof of \Cref{thm:robust}.\ref{p:robust-scrutable} on p. \pageref{thm:robust}.
		
		\item See the proof of \Cref{thm:robust}.\ref{p:robust-solvable} on p. \pageref{thm:robust}.
		
		\item If $\tm$ is \cbv scrutable, then there is a \balanced context $\bctx$ and a value $\val$ such that $\bctxp{\tm} \tovsub^* \val$. 
		As $\tm^\bullet \tom^* \tm$, one has $\bctxp{\tm^*} \tom^* \bctxp{\tm}$ and hence $\tm^*$ is \cbv scrutable.
		
		Conversely, if $\tm^\bullet$ is \cbv scrutable, then  there is a \balanced head context $\bctx$ and a value $\val$ such that $\bctxp{\tm^\bullet} \tovsub^* \val$. 
		As $\tm^\bullet \tom^* \tm$, one has $\bctxp{\tm^\bullet} \tom^* \bctxp{\tm}$. 
		By confluence of $\tovsub$ (\Cref{prop:properties-full-reduction}), $ \bctxp{\tm} \tovsub^* \tmtwo$ and $\val \tovsub^* \tmtwo$ for some $\tmtwo$ which is necessarily a value. 
		So, $\tm$ is \cbv scrutable.  
		
		As for the claim with \cbv solvability, we prove both directions with the same arguments, just replace the \balanced context $\bctx$ with a head context $\hctx$ and the value $\val$ with the identity~$\Id$.
		\qedhere
	\end{enumerate}
\end{proof}

%
%

\begin{corollary}[Yet another definition of \cbv solvability]
	\label{cor:op-char-cbv-solv-open}
	A term $\tm$ is \VSC-solvable if and only if
	\begin{itemize}
		\item \emph{SOL-IN$'$}: there is a head context $\hctx$ and an inert term $\itm$ such that $\hctxp\tm \tovsubonvar^{*} \itm$.
	\end{itemize}
\end{corollary}

\begin{proof}
	According to the characterization of \VSC-solvability given in \Cref{prop:op-char-cbv-solv-alt}, it is enough to prove that, for every term $\tm$, every head context $\hctx$ and every inert term $\itm$, one has
	\begin{center}
$\hctxp\tm \tovsub^{*} \itm$ for some inert term $\itm$ if and only if $\hctxp\tm \tovsubonvar^{*} \itmtwo$ for some inert term $\itmtwo$.
	\end{center} 
	The right-to-left implication is obvious since $\tovsubonvar \subseteq \tovsub$.
	Let us prove the left-to-right direction.
	By postponement of variable steps (\Cref{cor:postponement}.\ref{p:postponement-vsub}), $\hctxp\tm \tovsubnvar^{*} \tmtwo \toevar^* \itm$  for some $\tmtwo$.
	As $\itm$ is $\onvarsym$-normal (\Cref{prop:properties-open-reduction}.\ref{p:properties-open-reduction-harmony}), $\tmtwo$ is $\onvarsym$-normal too (\Cref{l:preservation-normal}.\ref{p:preservation-normal-open}) and it is not a value (otherwise $\itm$ would be a value), hence $\tmtwo$ is a inert term.
	By normalization (\Cref{prop:properties-open-extra}.\ref{p:properties-open-extra-normalization-bis}), $\hctxp\tm \tovsubonvar^{*} \fire$ for some fireball $\fire$.
	Since $\tovsub$ is confluent (\Cref{prop:properties-full-reduction}.\ref{p:properties-full-reduction-confluence}), $\itm \tovsub^* \tmthree$ and $\fire \tovsub^* \tmthree$ for some term $\tmthree$. 
	As $\itm$ is inert, so is $\tmthree$, and hence the fireball $\fire$ is an inert term (otherwise $\fire$ would be a value and then $\tmthree$ too).
\end{proof}

\section{Proofs of Section \ref{sect:collapsibility} ((Non-)Collapsibility)}

\begin{proposition}[$\plotcalc$ contextual equivalence is scrutable]
	\label{propappendix:contextual-scrutable-plotkin}
	\NoteState{prop:contextual-scrutable-plotkin}
	$\ctxeq^{\plotcalc}$ is a scrutable $\plotcalc$-theory.
\end{proposition}

\begin{proof}
	In \cite{DBLP:journals/fuin/EgidiHR92,DBLP:books/cu/12/Pitts12}, it is proved that $\ctxeq^{\plotcalc}$ coincides with \cbv applicative similarity (the definition of which is here omitted, see \citet{DBLP:books/cu/12/Pitts12}), that is, that---roughly---one can test \cbv contextual equivalence on closed terms by restricting to applicative contexts, and on open terms by first closing them via a substitution of values for its free variables. Since both applicative contexts and substitutions of values are special cases of \balanced contexts, for a \cbv inscrutable term there are no applicative contexts that send it to a value. Then, all \cbv inscrutable are equated by \cbv applicative bisimilarity, that is, by $\ctxeq^{\plotcalc}$.
\end{proof}

\begin{proposition}[VSC contextual equivalence is scrutable]
	\label{corappendix:contextual-scrutable-vsc}
	\NoteState{cor:contextual-scrutable-vsc}
	$\ctxeq^{\VSC}$ is a scrutable \VSC-theory.
\end{proposition}

\begin{proof}	
	If $\tm, \tmtwo \in \vsubterms$ are \cbv inscrutable, then so are $\tm^\bullet$ and $\tmtwo^\bullet$ (\Cref{thm:robust}.\ref{p:robust-ES}). 
	As $\ctxeq^{\plotcalc}$ is a scrutable $\plotcalc$-theory (\Cref{prop:contextual-scrutable-plotkin}), we have $\tm^\bullet \ctxeq^{\plotcalc}\tmtwo^\bullet$, and since $\ctxeq^{\plotcalc}$ and $\ctxeq^{\VSC}$ coincide on terms without ES (\Cref{prop:consistency-ctx-eq}), we have $\tm^\bullet \ctxeq^{\VSC}\tmtwo^\bullet$. 
	By stability of $\ctxeq^{\VSC}$ by ES expansion (\Cref{l:es-exp-and-ctx-eq}), $\tm \ctxeq^{\VSC}  \tmtwo$.
\end{proof}
\section{Proofs of Section \ref{sect:types} (Multi Types by Value)}

First, we observe the following property: given a derivation for a term $\tm$, all variables associated with a non-empty multi type in the type context are free variables~of~ $\tm$.
\begin{remark}
	\label{rmk:free-variables}
	If $\namedtyjp{\tderiv}{}{\tm}{\typctx}{\mtype}$ then $\dom{\typctx} \subseteq \fv{\tm}$.
	The proof is by straightforward induction on the derivation $\tderiv$.
\end{remark}

\begin{lemma}[Typing of values: splitting]
	\label{l:typing-value-splitting}
	Let $\namedtyjp{\tderiv}{}{\val}{\typctx}{\mtype}$ (for $\val$ value).
	\begin{enumerate}
		\item \label{p:typing-value-splitting-one} If $\mtype = \emptytype$, then $\dom{\typctx} = \emptyset$ and 
		$\sizem{\tderiv} = 0 = \size{\tderiv}$. 
		
		\item \label{p:typing-value-splitting-two} For every splitting $\mtype = \mtype_{1} \mplus \mtype_{2}$, there exist 
		type derivations $\namedtyjp{\tderiv_{1}}{}{\val}{\typctx_{1}}{\mtype_{1}}$ and 
		$\namedtyjp{\tderiv_{2}}{}{\val}{\typctx_{2}}{\mtype_{2}}$ such that $\typctx = \typctx_{1} \mplus \typctx_{2}$, 
		$\sizem{\tderiv} = \sizem{\tderiv_{1}} + \sizem{\tderiv_{2}}$ and $\size{\tderiv} = \size{\tderiv_{1}} + 
		\size{\tderiv_{2}}$.
		
	\end{enumerate}
\end{lemma}

\begin{proof}\hfill
	\begin{enumerate}
		\item By a simple inspection of the typing rules, $\mtype = \emptytype$ and the fact that $\val$ is a value imply 
		that 
		$\tderiv$ is of the form
		\begin{prooftree}
			\hypo{}
			\infer1[\footnotesize$\ruleMany$]{\tyjp{}{\val}{}{\emptytype}}
		\end{prooftree}
		where $\dom{\typctx} = \emptyset$ and $\sizem{\tderiv} = 0 = \size{\tderiv}$.
		
		\item Let
		\begin{equation*}
		\tderiv = 
		\begin{prooftree}
		\hypo{}
		\ellipsis{$\tderiv_{i}$}{\tyjp{}{\val}{\typctx_{i}}{\ltype_{i}}}
		\delims{\left(}{\right)_{\iI}}
		\infer1[\footnotesize$\ruleMany$]{\tyjp{}{\val}{\bigmplus_{\iI} \typctx_{i}}{\mult{\ltype}_{\iI}}}
		\end{prooftree}
		\end{equation*}
		with $\bigmplus_{\iI} \typctx_{i} = \typctx$ and $\mult{\ltype}_{\iI} = \mtype = \mtype_{1} \mplus \mtype_{2}$. Let 
		$I_{1}$ and $I_{2}$ be sets of indices such that $I = I_{1} \cup I_{2}$, $\mtype_{1} = \mult{\ltype_{i}}_{i \in I_{1}}$ 
		and $\mtype_{2} = \mult{\ltype_{i}}_{i \in I_{2}}$. 
		As $\val$ is a value, 	we can then derive
		\begin{equation*}
		\tderiv_{1} = 
		\begin{prooftree}
		\hypo{}
		\ellipsis{$\tderiv_{i}$}{\tyjp{}{\val}{\typctx_{i}}{\ltype_{i}}}
		\delims{\left(}{\right)_{i \in I_1}}
		\infer1[\footnotesize$\ruleMany$]{\tyjp{}{\val}{\bigmplus_{i \in I_{1}} \typctx_{i}}{\mult{\ltype}_{i \in I_{1}}}}
		\end{prooftree}
		\end{equation*}
		and 
		\begin{equation*}
		\tderiv_{2} = 
		\begin{prooftree}
		\hypo{}
		\ellipsis{$\tderiv_{i}$}{\tyjp{}{\val}{\typctx_{i}}{\ltype_{i}}}
		\delims{\left(}{\right)_{i \in I_2}}
		\infer1[\footnotesize$\ruleMany$]{\tyjp{}{\val}{\bigmplus_{i \in I_{2}} \typctx_{i}}{\mult{\ltype}_{i \in I_{2}}}}
		\end{prooftree}
		\end{equation*}
		noting that 
		$$
		\typctx = \bigmplus_{\iI} \typctx_{i} = \left( \bigmplus_{i \in I_{1}} \typctx_{i} \right) \mplus 
		\left(\bigmplus_{i \in I_{2}} \typctx_{i} \right)
		$$
		with 
		$$
		\sizem{\tderiv} = \sum_{\iI} \sizem{\tderiv_{i}} = \left( \sum_{i \in I_{1}} \sizem{\tderiv_{i}} \right) + \left( 
		\sum_{i \in I_{2}} \sizem{\tderiv_{i}} \right) = \sizem{\tderiv_{1}} + \sizem{\tderiv_{2}}
		$$
		and 
		$$
		\size{\tderiv} = \sum_{\iI} \size{\tderiv_{i}} = \left( \sum_{i \in I_{1}} \size{\tderiv_{i}} \right) + \left( 
		\sum_{i \in I_{2}} \size{\tderiv_{i}} \right) = \size{\tderiv_{1}} + \size{\tderiv_{2}}
		$$
		
	\end{enumerate}
\end{proof}

\begin{lemma}[Substitution]
	\label{lappendix:substitution}	
	\NoteProof{l:substitution}
	Let $\tm$ be a term, $\val$ be a value and $\namedtyjp{\tderiv}{}{\tm}{\typctx, \var \hastype 
		\mtypetwo}{\mtype}$ and $\namedtyjp{\tderivtwo}{}{\val}{\typctxtwo}{\mtypetwo}$ be derivations.
	Then there is a derivation $\namedtyjp{\tderivthree}{}{\tm \isub{\var}{\val}}{\typctx \mplus \typctxtwo}{\mtype}$ 
	with $\sizem{\tderivthree} = \sizem{\tderiv} + \sizem{\tderivtwo}$ and $\size{\tderivthree} \leq \size{\tderiv} + 
	\size{\tderivtwo}$. 
\end{lemma}

\begin{proof}
	By induction on the term $\tm$.
	Cases:
	\begin{itemize}
		\item \emph{Variable}, then are two sub-cases:
		\begin{enumerate}
			\item $\tm = \var$, then $\tm \isub{\var}{\val} = \var \isub{\var}{\val} = \val$ and 			
			$\sizem{\tderiv} = 0$ and $\size{\tderiv} = 1$.
			
			the derivation $\tderiv$ has necessarily the form (for some $n \in \nat$)
			\begin{equation*}
			\tderiv = 
			\begin{prooftree}
			\infer0[\footnotesize$\Ax$]{\tyjp{}{\var}{\var \hastype \mset{\ltype_1}}{\ltype_1}}
			\hypo{\overset{n \in \nat}{\ldots}}
			\infer0[\footnotesize$\Ax$]{\tyjp{}{\var}{\var \hastype \mset{\ltype_n}}{\ltype_n}}
			\infer3[\footnotesize$\ruleManyVar$]{\tyjp{}{\var}{\var \hastype 
					\mset{\ltype_1,\dots,\ltype_n}}{\mset{\ltype_1,\dots,\ltype_n}}}
			\end{prooftree}
			\end{equation*}
			with $\mtype = \mset{\ltype_1, \dots, \ltype_n} = \mtypetwo$ and $\dom{\typctx} = \emptyset$.
			Thus, $\sizem{\tderiv} = 0$ and $\size{\tderiv} = n$.
			Let $\tderivthree = \tderivtwo$: so, $\namedtyjp{\tderivthree}{}{\tm \isub{\var}{\val}}{\typctx \mplus 
				\typctxtwo}{\mtype}$ (since $\typctx \mplus \typctxtwo = \typctxtwo$) with $\sizem{\tderivthree} = \sizem{\tderivtwo} = 
			\sizem{\tderivtwo} + \sizem{\tderiv}$ and $\size{\tderivthree} = \size{\tderivtwo} \leq \size{\tderivtwo} + 
			\size{\tderiv}$ (note that $\size{\tderivthree} = \size{\tderiv} + \size{\tderivtwo}$ if and only if $n=0$).
			
			\item $\tm = \varthree \neq \var$, then $\tm \isub{\var}{\val} = \varthree$ and 
			$\sizem{\tderiv} = 0$, $\size{\tderiv} = 1$, 
			the derivation $\tderiv$ has necessarily the form (for some $n~\in~\nat$)
			\begin{equation*}
			\tderiv = 
			\begin{prooftree}
			\infer0[\footnotesize$\Ax$]{\tyjp{}{\varthree}{\varthree \hastype \mset{\ltype_1}}{\ltype_1}}
			\hypo{\overset{n \in \nat}{\ldots}}
			\infer0[\footnotesize$\Ax$]{\tyjp{}{\varthree}{\varthree \hastype \mset{\ltype_n}}{\ltype_n}}
			\infer3[\footnotesize$\ruleManyVar$]{\tyjp{}{\varthree}{\varthree \hastype 
					\mset{\ltype_1,\dots,\ltype_n}}{\mset{\ltype_1,\dots,\ltype_n}}}
			\end{prooftree}
			\end{equation*}
			where $\mtype = \mset{\ltype_1, \dots, \ltype_n}$ and  $\mtypetwo = \emptymset$ and $\typctx = \varthree \hastype 
			\mtype$ (while $\typctx(\var) = \emptymset$).
			Thus, $\sizem{\tderiv} = 0$ and $\size{\tderiv} = n$.
			By \reflemmap{typing-value-splitting}{one}, from $\namedtyjp{\tderivtwo}{}{\val}{\typctxtwo}{\emptytype}$ it 
			follows that $\sizem{\tderivtwo} = 0 = \size{\tderivtwo}$ and $\dom{\typctxtwo} = \emptyset$.  
			Therefore, $\typctx \mplus \typctxtwo = \typctx$.
			Let $\tderivthree = \tderiv$: so, $\namedtyjp{\tderivthree}{}{\tm \isub{\var}{\val}}{\typctx \mplus 
				\typctxtwo}{\mtype}$  with $\sizem{\tderivthree} = \sizem{\tderiv} = \sizem{\tderiv} + \sizem{\tderivtwo}$ and 
			$\size{\tderivthree} = \size{\tderiv} = \size{\tderiv} + \size{\tderivtwo}$.
		\end{enumerate}
		
		\item \emph{Application}, \ie $\tm = \tmtwo\tmthree$. 
		Then $\tm \isub{\var}{\val} = \tmtwo \isub{\var}{\val} \tmthree \isub{\var}{\val}$ and necessarily
		\begin{equation*}
		\tderiv = 
		\begin{prooftree}
		\hypo{}
		\ellipsis{$\tderiv_{1}$}{\typctx_1, \var \hastype \mtypetwo_1 \vdash \tmtwo \hastype 
			\mset{\larrow{\mtypethree}{\mtype}}}
		\hypo{}
		\ellipsis{$\tderiv_{2}$}{\typctx_2, \var \hastype \mtypetwo_2 \vdash \tmthree \hastype \mtypethree}
		\infer2[\footnotesize$\ruleAp$]{\typctx, \var \hastype \mtypetwo \vdash \tmtwo \tmthree \hastype \mtype}
		\end{prooftree}
		\end{equation*}
		with $\sizem{\tderiv} = \sizem{\tderiv_{1}} + \sizem{\tderiv_{2}} + 1$, $\size{\tderiv} = \size{\tderiv_{1}} + 
		\size{\tderiv_{2}} + 1$, $\typctx = \typctx_1 \mplus \typctx_2$ and $\mtypetwo = \mtypetwo_2 \mplus \mtypetwo_2$. 
		According to \reflemmap{typing-value-splitting}{two} applied to $\tderivtwo$ and to the decomposition $\mtypetwo = 
		\mtypetwo_1 \mplus \mtypetwo_2$, there are contexts $\typctxtwo_1, \typctxtwo_2$ and derivations 
		$\namedtyjp{\tderivtwo_{1}}{}{\val}{\typctxtwo_1}{\mtypetwo_1}$ and 
		$\namedtyjp{\tderivtwo_{2}}{}{\val}{\typctxtwo_2}{\mtypetwo_2}$ such that $\typctxtwo = \typctxtwo_{1} \mplus 
		\typctxtwo_2$, $\sizem{\tderivtwo} = \sizem{\tderivtwo_1} + \sizem{\tderivtwo_2}$ and $\size{\tderivtwo} = 
		\size{\tderivtwo_{1}} + \size{\tderivtwo_{2}}$.
		
		By \ih, there are derivations $\namedtyjp{\tderivthree_1}{}{\tmtwo \isub{\var}{\val}}{\typctx_1 \mplus 
			\typctxtwo_1}{\mult{\larrow{\mtypethree}{\mtype}}}$ and $\namedtyjp{\tderivthree_2}{}{\tmthree 
			\isub{\var}{\val}}{\typctx_2 \mplus \typctxtwo_2}{\mtypethree}$ such that $\sizem{\tderivthree_{i}} = 
		\sizem{\tderiv_{i}} + \sizem{\tderivtwo_{i}}$ and $\size{\tderivthree_{i}} \leq \size{\tderiv_{i}} + 
		\size{\tderivtwo_{i}}$ for all $i \in \{1,2\}$.
		Since $\typctx \mplus \typctxtwo = \typctx_1 \mplus \typctxtwo_1 \mplus \typctx_2 \mplus \typctxtwo_2$, we can 
		build the derivation
		\begin{equation*}
		\tderivthree = 
		\begin{prooftree}
		\hypo{}
		\ellipsis{$\tderivthree_1$}{\typctx_1 \mplus \typctxtwo_1 \vdash \tmtwo\isub{\var}{\val} \hastype 
			\mset{\larrow{\mtypethree}{\mtype}}}
		\hypo{}
		\ellipsis{$\tderivthree_2$}{\typctx_2 \mplus \typctxtwo_2 \vdash \tmthree\isub{\var}{\val} \hastype \mtypethree}
		\infer2[\footnotesize$\ruleAp$]{\typctx \mplus \typctxtwo \vdash \tmtwo \isub{\var}{\val} \tmthree 
			\isub{\var}{\val} \hastype \mtype}
		\end{prooftree}
		\end{equation*}
		where $\sizem{\tderivthree} = \sizem{\tderivthree_{1}} + \size{\tderivthree_{2}} + 1 = \sizem{\tderiv_{1}} + 
		\sizem{\tderivtwo_{1}} + \sizem{\tderiv_{2}} + \sizem{\tderivtwo_{2}} + 1 = \sizem{\tderiv} + \sizem{\tderivtwo}$ and 
		$\size{\tderivthree} = \size{\tderivthree_{1}} + \size{\tderivthree_{2}} + 1 \leq \sizem{\tderiv_{1}} + 
		\sizem{\tderivtwo_{1}} + \sizem{\tderiv_{2}} + \size{\tderivtwo_{2}} + 1 = \size{\tderiv} + \size{\tderivtwo}$.
		
		\item \emph{Abstraction}, \ie $\tm = \la{\vartwo}{\tmtwo}$.
		We can suppose without loss of generality that $\vartwo \notin \fv{\val} \cup \{\var \}$, hence $\tm 
		\isub{\var}{\val} = \la{\vartwo}{\tmtwo\isub{\var}\val}$ and  $\tderiv$ is necessarily of the form (for some $n \in 
		\nat$) 
		\begin{equation*}
		\begin{prooftree}[separation=1em]
		\hypo{}
		\ellipsis{$\tderiv_{i}$}{\typctx_{i}, \vartwo \hastype \mtypethree_{i}, \var \hastype \mtypetwo_{i} \vdash \tmtwo 
			\hastype \mtype_{i}}
		\infer1[\footnotesize$\ruleFun$]{\tyjp{}{\la{\vartwo}{\tmtwo}}{\typctx_{i}, \var \hastype 
				\mtypetwo_{i}}{\ty{\mtypethree_{i}\!}{\!\mtype_{i}}}}
		\delims{\left(}{\right)_{1\leq i \leq n}}
		\infer1[\footnotesize$\ruleManyVal$]{\tyjp{}{\la{\vartwo}{\tmtwo}}{\bigmplus_{i=1}^{n} \typctx_{i} , \var \hastype 
				\bigmplus_{i=1}^{n} \mtypetwo_i}{\bigmplus_{i=1}^{n} \mset{\larrow{\mtypethree_i}{\mtype_i}}}}
		\end{prooftree}
		\end{equation*}
		with $\sizem{\tderiv} = \sum_{i=1}^{n} (\sizem{\tderiv_{i}} + 1)$ and $\size{\tderiv} = \sum_{i=1}^n 
		(\size{\tderiv_i} + 1)$.
		Since $\vartwo \notin \fv{\val}$, then $\vartwo \notin \domain{\typctxtwo}$ (\refrmk{free-variables}), and so 
		$\namedtyjp{\tderivtwo}{}{\val}{\typctxtwo, \vartwo \hastype \emptymset}{\mtypetwo}$. 
		Now, there are two subcases:
		\begin{itemize}
			\item \emph{Empty multi type}: If $n = 0$,  then $\mtypetwo = \emptymset = \mtype$ and $\dom{\typctx} = 
			\emptyset$, with $\sizem{\tderiv} = 0 = \size{\tderiv}$. 
			According to \reflemmap{typing-value-splitting}{one} applied to $\tderivtwo$, $\dom{\typctxtwo} = \emptyset$ with 
			$\sizem{\tderivtwo} = 0 = \size{\tderivtwo}$.
			We can then build the derivation 
			\begin{equation*}
			\tderivthree = 
			\begin{prooftree}
			\infer0[\footnotesize$\ruleManyVal$]{\tyjp{}{\la{\vartwo}(\tmtwo\isub{\var}{\val})}{}{\emptymset}}
			\end{prooftree}
			\end{equation*}
			where $\sizem{\tderivthree} = 0 = \sizem{\tderiv} + \sizem{\tderivtwo}$ and $\size{\tderivthree} = 0 = 
			\size{\tderiv} + \size{\tderivtwo}$, and $\concl{\tderivthree}{\typctx \mplus 
				\typctxtwo}{\tm\isub{\var}{\val}}{\mtype}$ since $\dom{\typctx \mplus \typctxtwo} = \emptyset$.
			
			\item\emph{Non-empty multi type}: If $n > 0$, then we can decompose $\tderivtwo$ according to the partitioning 
			$\mtypetwo = \biguplus_{i=1}^n \mtypetwo_i$ by repeatedly applying \reflemmap{typing-value-splitting}{two}, and hence 
			for all $1 \leq i \leq n$ there are context $\typctxtwo_{i}$ and a derivation 
			$\namedtyjp{\tderivtwo_i}{}{\val}{\typctxtwo_{i} ; \vartwo \hastype \emptytype}{\mtypetwo_i}$ such that 
			$\sizem{\tderivtwo} = \sum_{i=1}^n \sizem{\tderivtwo_{i}}$ and $\size{\tderivtwo} = \sum_{i=1}^{n} 
			\size{\tderivtwo_{i}}$.
			By \ih, for all $1 \leq i \leq n$, there is a derivation 
			$\namedtyjp{\tderivthree_{i}}{}{\tmtwo\isub{\var}{\val}}{\typctx_i \mplus \typctxtwo_i, \vartwo \hastype 
				\mtypethree_i}{\mtype_i}$ such that $\sizem{\tderivthree_{i}} = \sizem{\tderiv_{i}} + \sizem{\tderivtwo_{i}}$ and 
			$\size{\tderivthree_{i}} \leq \size{\tderiv_{i}} + \size{\tderivtwo_{i}}$.
			Since $\typctx \mplus \typctxtwo = \bigmplus_{i=1}^n (\typctx_i \mplus \typctxtwo_i)$, we can build 
			$\tderivthree$ as
			\begin{equation*}
			\begin{prooftree}[separation=1em]
			\hypo{}
			\ellipsis{$\tderivthree_i$}{\typctx_i \mplus \typctxtwo_i, \vartwo \hastype \mtypethree_i \vdash \tmtwo 
				\isub{\var}{\val} \hastype \mtype_i}
			\infer1[\footnotesize$\ruleFun$]{\tyjp{}{\la{\vartwo}{(\tmtwo \isub{\var}{\val})}}{\typctx_{i} \mplus 
					\typctxtwo_{i}}{\ty{\mtypethree_{i}\!}{\!\mtype_{i}}}}
			\delims{\left(}{\right)_{1\leq i \leq n}}
			\hypo{}
			\infer2[\footnotesize$\ruleManyVal$]{\typctx \mplus \typctxtwo \vdash \la{\vartwo}{(\tmtwo \isub{\var}{\val})} 
				\hastype \mtype}
			\end{prooftree}
			\end{equation*}
			noting that $\sizem{\tderivthree} = \sum_{i=1}^n (\sizem{\tderivthree_{i}} + 1) = \sum_{i=1}^n 
			(\sizem{\tderiv_{i}} + \sizem{\tderivtwo_{i}} + 1) = \sum_{i=1}^{n} (\sizem{\tderiv_{i}} + 1) + \sum_{i=1}^{n} 
			\sizem{\tderivtwo_{i}} = \sizem{\tderiv} + \sizem{\tderivtwo} $ and that $\size{\tderivthree} = \sum_{i=1}^{n} 
			(\size{\tderivthree_{i}} + 1) \leq \sum_{i=1}^{n} (\size{\tderiv_{i}} + \size{\tderivtwo_{i}} + 1) = \sum_{i=1}^{n} 
			(\size{\tderiv_{i}} + 1) + \sum_{i=1}^{n} \size{\tderivtwo_{i}} = \size{\tderiv} + \size{\tderivtwo}$.
		\end{itemize}
		
		\item \emph{Explicit substitution}, \ie $\tm = \tmtwo \esub{\vartwo}{\tmthree}$. 
		We can suppose without loss of generality that $\vartwo \notin \fv{\val} \cup \{\var \}$, hence $\tm 
		\isub{\var}{\val} = \tmtwo\isub{\var}{\val} \esub{\vartwo}{\tmthree\isub{\var}\val}$ and necessarily
		\begin{equation*}
		\tderiv = 
		\begin{prooftree}
		\hypo{}
		\ellipsis{$\tderiv_{1}$}{\typctx_1 , \var \hastype \mtypetwo_1 , \vartwo \hastype \mtypethree \vdash \tmtwo 
			\hastype \mtype}
		\hypo{}
		\ellipsis{$\tderiv_{2}$}{\typctx_2, \var \hastype \mtypetwo_2 \vdash \tmthree \hastype \mtypethree}
		\infer2[\footnotesize$\Es$]{\typctx, \var \hastype \mtypetwo \vdash \tmtwo \esub{\vartwo}{\tmthree} \hastype \mtype}
		\end{prooftree}
		\end{equation*}
		with $\sizem{\tderiv} = \sizem{\tderiv_{1}} + \sizem{\tderiv_{2}}$, $\size{\tderiv} = \size{\tderiv_{1}} + 
		\size{\tderiv_{2}} + 1$, $\typctx = \typctx_1 \mplus \typctx_2$ and $\mtypetwo = \mtypetwo_2 \mplus \mtypetwo_2$. 
		According to \reflemmap{typing-value-splitting}{two} applied to $\tderivtwo$ and to the decomposition $\mtypetwo = 
		\mtypetwo_1 \mplus \mtypetwo_2$, there are contexts $\typctxtwo_1, \typctxtwo_2$ and derivations 
		$\namedtyjp{\tderivtwo_{1}}{}{\val}{\typctxtwo_{1}}{\mtypetwo_{1}}$ and 
		$\namedtyjp{\tderivtwo_{2}}{}{\val}{\typctxtwo_{2}}{\mtypetwo_{2}}$ such that $\typctxtwo = \typctxtwo_{1} \mplus 
		\typctxtwo_{2}$, $\sizem{\tderivtwo} = \sizem{\tderivtwo_{1}} + \sizem{\tderivtwo_{2}}$ and $\size{\tderivtwo} = 
		\size{\tderivtwo_{1}} + \size{\tderivtwo_{2}}$.
		
		By \ih, there are derivations $\namedtyjp{\tderivthree_{1}}{}{\tmtwo\isub{\var}{\val}}{\typctx_{1} \mplus 
			\typctxtwo_{1}, \vartwo \hastype \mtypethree}{\mtype}$ and $\namedtyjp{\tderivthree_{2}}{}{\tmthree 
			\isub{\var}{\val}}{\typctx_{2} \mplus \typctxtwo_{2}}{\mtypethree}$ such that $\sizem{\tderivthree_{i}} = 
		\sizem{\tderiv_{i}} + \sizem{\tderivtwo_{i}}$ and $\size{\tderivthree_{i}} \leq \size{\tderiv_{i}} + 
		\size{\tderivtwo_{i}}$ for all $i \in \{1,2\}$.
		Since $\typctx \mplus \typctxtwo = \typctx_1 \mplus \typctxtwo_1 \mplus \typctx_2 \mplus \typctxtwo_2$, we can 
		build the derivation
		\begin{equation*}
		\tderivthree = 
		\begin{prooftree}
		\hypo{}
		\ellipsis{$\tderivthree_1$}{\tyjp{}{\tmtwo\isub{\var}{\val}}{\typctx_{1} \mplus \typctxtwo_{1}, \vartwo \hastype 
				\mtypethree}{\mtype}}
		\hypo{}
		\ellipsis{$\tderivthree_2$}{\typctx_2 \mplus \typctxtwo_2 \vdash \tmthree\isub{\var}{\val} \hastype \mtypethree}
		\infer2[\footnotesize$\Es$]{\typctx \mplus \typctxtwo \vdash \tmtwo \isub{\var}{\val} \esub{\vartwo} {\tmthree 
				\isub{\var}{\val}} \hastype \mtype}
		\end{prooftree}
		\end{equation*}
		verifying that $\sizem{\tderivthree} = \sizem{\tderivthree_{1}} + \sizem{\tderivthree_{2}} = \sizem{\tderiv_{1}} + 
		\sizem{\tderivtwo_{1}} + \sizem{\tderiv_{2}} + \sizem{\tderivtwo_{2}} = \sizem{\tderiv} + \sizem{\tderivtwo}$ and 
		$\size{\tderivthree} = 1 + \size{\tderivthree_{1}} + \size{\tderivthree_{2}} \leq 1 + (\size{\tderiv_{1}} + 
		\size{\tderivtwo_{1}}) + (\size{\tderiv_{2}} + \size{\tderivtwo_{2}}) = \size{\tderiv} + \size{\tderivtwo}$.
		\qedhere
	\end{itemize}	
\end{proof}

\begin{lemma}[Typing of values: merging]
	\label{l:typing-value-complete} 
	Let $\val$ be a value.
	\begin{enumerate}
		\item \label{p:typing-value-complete-empty} There is a derivation $\namedtyjp{\tderiv}{}{\val}{}{\zero}$ with 
		$\sizem{\tderiv} = 0 = \size{\tderiv}$.
		
		\item \label{p:typing-value-complete-merge} For every derivations 
		$\namedtyjp{\tderiv_{1}}{}{\val}{\typctx_{1}}{\mtype_{1}}$ and 
		$\namedtyjp{\tderiv_{2}}{}{\val}{\typctx_{2}}{\mtype_{2}}$, there exists a  derivation 
		$\namedtyjp{\tderiv}{}{\val}{\typctx_{1} \mplus \typctx_2}{\mtype_{1} \mplus \mtype_{2}}$ such that $\sizem{\tderiv} = 
		\sizem{\tderiv_{1}} + \sizem{\tderiv_{2}}$ and $\size{\tderiv} = \size{\tderiv_{1}} + \size{\tderiv_{2}}$.
		
	\end{enumerate}
\end{lemma}

\begin{proof}
	\begin{enumerate}
		\item 
		Let $\tderiv$ be the following type derivation (applying the rule $\ruleMany$ with $n = 0$) such that 
		$\sizem{\tderiv} = 0 = \size{\tderiv}$:
		\begin{equation*}
		\tderiv =
		\begin{prooftree}
		\hypo{}
		\infer1[\footnotesize$\ruleMany$]{\tyjp{}{\val}{}{\emptytype}}
		\end{prooftree}\ .
		\end{equation*}
		
		\item Let 
		\begin{equation*}
		\tderiv_{1} =
		\begin{prooftree}
		\hypo{}
		\ellipsis{$\tderiv_{i}$}{\tyjp{}{\val}{\typctx_{i}}{\ltype_{i}}}
		\delims{ \left( }{ \right)_{\iI} }
		\infer1[\footnotesize$\ruleMany$]{\tyjp{}{\val}{\bigmplus_{\iI} \typctx_{i}}{\bigmplus_{\iI}\mult{\ltype_{i}}}}
		\end{prooftree}
		\end{equation*}
		with $\typctx_{1} = \bigmplus_{\iI} \typctx_{i}$ and $\mtype_{1} = \bigmplus_{\iI}\mult{\ltype_{i}}$, and let
		\begin{equation*}
		\tderiv_{2} =
		\begin{prooftree}
		\hypo{}
		\ellipsis{$\tderiv_{j}$}{\tyjp{}{\val}{\typctx_{j}}{\ltype_{j}}}
		\delims{ \left( }{ \right)_{\jJ} }
		\infer1[\footnotesize$\ruleMany$]{\tyjp{}{\val}{\bigmplus_{\jJ} \typctx_{j}}{\bigmplus_{\jJ}\mult{\ltype_{j}}}}
		\end{prooftree}
		\end{equation*}
		with $\typctx_{2} = \bigmplus_{\jJ} \typctx_{j}$ and $\mtype_{2} = \bigmplus_{\jJ}\mult{\ltype_{j}}$.
		
		We can derive $\tderiv$ by setting $K = I \cup J$ and then
		\begin{equation*}
		\tderiv =
		\begin{prooftree}
		\hypo{}
		\ellipsis{$\tderiv_{k}$}{\tyjp{}{\val}{\typctx_{k}}{\ltype_{k}}}
		\delims{ \left( }{ \right)_{\kK} }
		\infer1[\footnotesize$\ruleMany$]{\tyjp{}{\val}{\bigmplus_{\kK} \typctx_{k}}{\bigmplus_{\kK}\mult{\ltype_{k}}}}
		\end{prooftree}
		\end{equation*}
		trivially verifying the statement.
		\qedhere
	\end{enumerate}
\end{proof}

\begin{lemma}[Removal]
	\label{lappendix:anti-substitution}
	\NoteProof{l:anti-substitution}
	Let $\tm$ be a term, $\val$ be a value, and 
	$\namedtyjp{\tderiv}{}{\tm\isub{\var}{\val}}{\typctx}{\mtype}$
	be a type derivation. Then there are two  derivations $\namedtyjp{\tderivtwo}{}{\tm}{\typctxtwo, \var \hastype 
		\mtypetwo}{\mtype}$ 
	and $\namedtyjp{\tderivthree}{}{\val}{\typctxthree}{\mtypetwo}$ such that $\typctx = \typctxtwo \mplus \typctxthree$ 
	with $\sizem{\tderiv} = \sizem{\tderivtwo} + \sizem{\tderivthree}$ and $\size{\tderiv} \leq \size{\tderivtwo} + 
	\size{\tderivthree}$.
\end{lemma}

\begin{proof}
	By induction on the term $\tm$.
	Cases:
	\begin{itemize}
		\item \emph{Variable}, then are two sub-cases (let $\mtype = \mset{\ltype_1, \dots, \ltype_n}$ for some $n \in 
		\nat$):
		\begin{enumerate}
			\item $\tm = \var$, then $\tm \isub{\var}{\val} = \val$. 
			Let $\typctxthree = \typctx$, let $\typctxtwo$ be the empty context (\ie $\dom{\typctxtwo} = \emptyset$), let 
			$\mtypetwo = \mtype$ and let $\tderivtwo$ be the derivation 
			\begin{equation*}
			\tderivtwo = 
			\begin{prooftree}
			\infer0[\footnotesize$\Ax$]{\tyjp{}{\var}{\var \hastype \mset{\ltype_1}}{\ltype_1}}
			\hypo{\overset{n \in \nat}{\ldots}}
			\infer0[\footnotesize$\Ax$]{\tyjp{}{\var}{\var \hastype \mset{\ltype_n}}{\ltype_n}}
			\infer3[\footnotesize$\ruleManyVar$]{\tyjp{}{\var}{\var \hastype 
					\mset{\ltype_1,\dots,\ltype_n}}{\mset{\ltype_1,\dots,\ltype_n}}}
			\end{prooftree}
			\end{equation*}
			Thus, $\concl{\tderivtwo}{\typctxtwo, \var \hastype \mtypetwo}{\tm}{\mtype}$ with $\sizem{\tderivtwo} = 0$ and 
			$\size{\tderivtwo} = n$.
			Let $\tderivthree = \tderiv$: so, $\namedtyjp{\tderivthree}{}{\val}{\typctxthree}{\mtypetwo}$ and $\typctxthree 
			\mplus \typctxtwo = \typctx$ with $\sizem{\tderiv} = \sizem{\tderivthree} = \sizem{\tderivtwo} + \sizem{\tderivthree}$ 
			and $\size{\tderiv} = \size{\tderivthree} \leq \size{\tderivtwo} + \size{\tderivthree}$.
			
			\item $\tm = \varthree \neq \var$, then $\tm \isub{\var}{\val} = \varthree$ and 
			the derivation $\tderiv$ has necessarily the form (for some $n~\in~\nat$)
			\begin{equation*}
			\tderiv = 
			\begin{prooftree}
			\infer0[\footnotesize$\Ax$]{\tyjp{}{\varthree}{\varthree \hastype \mset{\ltype_1}}{\ltype_1}}
			\hypo{\overset{n \in \nat}{\ldots}}
			\infer0[\footnotesize$\Ax$]{\tyjp{}{\varthree}{\varthree \hastype \mset{\ltype_n}}{\ltype_n}}
			\infer3[\footnotesize$\ruleManyVar$]{\tyjp{}{\varthree}{\varthree \hastype 
					\mset{\ltype_1,\dots,\ltype_n}}{\mset{\ltype_1,\dots,\ltype_n}}}
			\end{prooftree}
			\end{equation*}
			where $\mtype = \mset{\ltype_1, \dots, \ltype_n}$ and $\typctx = \varthree \hastype \mtype$ (while $\typctx(\var) 
			= \emptymset$).
			Thus, $\sizem{\tderiv} = 0$ and $\size{\tderiv} = n$.
			Let $\typctxthree$ be the empty context (\ie $\dom{\typctxthree} = 0$) and  $\mtypetwo = \emptymset$.
			By \reflemmap{typing-value-complete}{empty}, there is a derivation 
			$\namedtyjp{\tderivthree}{}{\val}{}{\emptytype}$ (and hence 
			$\namedtyjp{\tderivthree}{}{\val}{\typctxthree}{\mtypetwo}$) such that $\sizem{\tderivthree} = 0 = 
			\size{\tderivthree}$.  
			Let $\typctxtwo = \typctx$  and $\tderivtwo = \tderiv$:
			therefore, $\typctxthree \mplus \typctxtwo = \typctx$
			and $\namedtyjp{\tderivtwo}{}{\tm }{\typctxtwo, \var \hastype \mtypetwo}{\mtype}$  with $\sizem{\tderiv} = 
			\sizem{\tderivtwo} = \sizem{\tderivtwo} + \sizem{\tderivthree}$ and $\size{\tderiv} = \size{\tderivtwo} \leq 
			\size{\tderivtwo} + \size{\tderivthree}$.
		\end{enumerate}
		
		\item \emph{Application}, \ie $\tm = \tm_1\tm_2$. 
		Then $\tm \isub{\var}{\val} = \tm_1 \isub{\var}{\val} \tm_2 \isub{\var}{\val}$ and necessarily
		\begin{equation*}
		\tderiv = 
		\begin{prooftree}
		\hypo{}
		\ellipsis{$\tderiv_{1}$}{\typctx_1 \vdash \tm_1\isub{\var}{\val} \hastype \mset{\larrow{\mtypethree}{\mtype}}}
		\hypo{}
		\ellipsis{$\tderiv_{2}$}{\typctx_2 \vdash \tm_2\isub{\var}{\val} \hastype \mtypethree}
		\infer2[\footnotesize$\ruleAp$]{\typctx \vdash \tm_1\isub{\var}{\val} \tm_2\isub{\var}{\val} \hastype \mtype}
		\end{prooftree}
		\end{equation*}
		with $\sizem{\tderiv} = \sizem{\tderiv_{1}} + \sizem{\tderiv_{2}} + 1$, $\size{\tderiv} = \size{\tderiv_{1}} + 
		\size{\tderiv_{2}} + 1$ and $\typctx = \typctx_1 \mplus \typctx_2$. 		
		By \ih, for all $i \in \{1,2\}$, there are derivations $\namedtyjp{\tderivtwo_i}{}{\tm_i}{\typctxtwo_i, \var 
			\hastype \mtypetwo_i}{\mult{\larrow{\mtypethree}{\mtype}}}$ and $\namedtyjp{\tderivthree_i}{}{\val 
		}{\typctxthree_i}{\mtypetwo_i}$ such that $\typctx_i = \typctxtwo_i \mplus \typctxthree_i$ with $\sizem{\tderiv_{i}} = 
		\sizem{\tderivtwo_{i}} + \sizem{\tderivthree_{i}}$ and $\size{\tderiv_{i}} \leq \size{\tderivtwo_{i}} + 
		\size{\tderivthree_{i}}$.
		According to \reflemmap{typing-value-complete}{merge} applied to $\tderivthree_1$ and $\tderivthree_2$, there is a 
		derivation $\namedtyjp{\tderivthree}{}{\val}{\typctxthree}{\mtypetwo}$ where $\mtypetwo = \mtypetwo_1 \mplus 
		\mtypetwo_2$ and $\typctxthree = \typctxthree_1 \mplus \typctxthree_2$, such that  $\sizem{\tderivtwo} = 
		\sizem{\tderivtwo_1} + \sizem{\tderivtwo_2}$ and $\size{\tderivtwo} = \size{\tderivtwo_{1}} + \size{\tderivtwo_{2}}$.
		We can build the derivation (where $\typctxtwo = \typctxtwo_1 \mplus \typctxtwo_2$)
		\begin{equation*}
		\tderivtwo = 
		\begin{prooftree}
		\hypo{}
		\ellipsis{$\tderivtwo_1$}{\typctxtwo_1, \var \hastype \mtypetwo_1 \vdash \tm_1 \hastype 
			\mset{\larrow{\mtypethree}{\mtype}}}
		\hypo{}
		\ellipsis{$\tderivtwo_2$}{\typctxtwo_2, \var \hastype \mtypetwo_2 \vdash \tm_2 \hastype \mtypethree}
		\infer2[\footnotesize$\ruleAp$]{\typctxtwo, \var \hastype \mtypetwo \vdash \tm \hastype \mtype}
		\end{prooftree}
		\end{equation*}
		with $\sizem{\tderiv} = \sizem{\tderiv_{1}} + \size{\tderiv_{2}} + 1 = \sizem{\tderivtwo_{1}} + 
		\sizem{\tderivthree_{1}} + \sizem{\tderivtwo_{2}} + \sizem{\tderivthree_{2}} + 1 = \sizem{\tderivtwo} + 
		\sizem{\tderivthree}$ 
		and
		$\size{\tderiv} = \size{\tderiv_{1}} + \size{\tderiv_{2}} + 1 \leq \sizem{\tderivtwo_{1}} + 
		\sizem{\tderivthree_{1}} + \sizem{\tderivtwo_{2}} + \size{\tderivthree_{2}} + 1 = \size{\tderivtwo} + 
		\size{\tderivthree}$.
		
		\item \emph{Abstraction}, \ie $\tm = \la{\vartwo}{\tmtwo}$.
		We can suppose without loss of generality that $\vartwo \notin \fv{\val} \cup \{\var \}$, hence $\tm 
		\isub{\var}{\val} = \la{\vartwo}{\tmtwo\isub{\var}\val}$ and necessarily, for some $n \in \nat$,
		\begin{equation*}
		\tderiv =
		\begin{prooftree}[separation=1em]
		\hypo{}
		\ellipsis{$\tderiv_i$}{\typctx_i, \vartwo \hastype \mtypethree_i \vdash \tmtwo\isub{\var}{\val} \hastype \mtype_i}
		
		\infer1[\footnotesize$\ruleFun$]{\tyjp{}{\la{\vartwo}{\tmtwo\isub{\var}{\val}}}{\typctx_{i}}{\ty{\mtypethree_{i}\!}{
					\!\mtype_{1}}}}
		\delims{\left(}{\right)_{1 \leq i \leq n}}
		\hypo{}
		\infer2[\footnotesize$\ruleManyVal$]{\tyjp{}{\la{\vartwo}{\tmtwo\isub{\var}{\val}}}{\typctx}{\mtype}}
		\end{prooftree}
		\end{equation*}
		with $\typctx = \bigmplus_{i=1}^{n} \typctx_i$, $\mtype = \bigmplus_{i=1}^{n} 
		\mset{\larrow{\mtypethree_i\!}{\!\mtype_i}}$, $\sizem{\tderiv} = n + \sum_{i=1}^{n} \sizem{\tderiv_{i}}$ and 
		$\size{\tderiv} = n + \sum_{i=1}^n \size{\tderiv_i} $.
		%
		$\namedtyjp{\tderivtwo}{}{\val}{\typctxtwo, \vartwo \hastype \emptymset}{\mtypetwo}$. 
		There are two subcases:
		\begin{itemize}
			\item \emph{Empty multi type}: If $n = 0$,  then $\mtype = \emptymset$ and $\dom{\typctx} = \emptyset$, with 
			$\sizem{\tderiv} = 0 = \size{\tderiv}$. 
			We can  build the derivation 
			\begin{equation*}
			\tderivtwo = 
			\begin{prooftree}
			\infer0[\footnotesize$\ruleManyVal$]{\tyjp{}{\la{\vartwo}\tmtwo}{}{\emptymset}}
			\end{prooftree}
			\end{equation*}
			where $\sizem{\tderivtwo} = 0 = \size{\tderivtwo}$.
			Let $\mtypetwo = \emptymset$ and $\typctxtwo$ be the empty context (\ie $\dom{\typctxtwo} = \emptyset$): then 
			$\concl{\tderivtwo}{\typctxtwo, \var \hastype \mtypetwo}{\tm}{\mtype}$.			
			According to \reflemmap{typing-value-complete}{empty}, there is a derivation 
			$\concl{\tderivthree}{}{\val}{\emptymset}$ with $\sizem{\tderivthree} = 0 = \size{\tderivthree}$. 
			Let $\typctxthree$ be the empty context (\ie $\dom{\typctxthree} = \emptyset$): so, 
			$\concl{\tderivthree}{\typctxthree}{\val}{\mtypetwo}$ with $\typctx = \typctxtwo \mplus \typctxthree$ and 
			$\sizem{\tderiv} = 0 = \sizem{\tderivtwo} + \sizem{\tderivthree}$ and $\size{\tderiv} = 0 \leq \size{\tderivtwo} + 
			\size{\tderivthree}$.
			
			\item\emph{Non-empty multi type}: If $n > 0$ then by \ih, for all $1 \leq i \leq n$, there are derivations 
			$\concl{\tderivtwo_i}{\typctxtwo_i, \vartwo \hastype \mtypethree_i, \var \hastype \mtypetwo_i}{\tmtwo}{\mtype_i}$ and 
			$\concl{\tderivthree_i}{\typctxthree_i}{\val}{\mtypetwo_i}$ such that $\typctx_i = \typctxtwo_i \mplus \typctxthree_i$ 
			with $\sizem{\tderiv_i} = \sizem{\tderivtwo_i} + \sizem{\tderivthree_i}$ and $\size{\tderiv_i} \leq \size{\tderivtwo_i} 
			+ \size{\tderivthree_i}$.
			We can build the derivation
			\begin{equation*}
			\tderivtwo = 
			\begin{prooftree}[separation=1em]
			\hypo{}
			\ellipsis{$\tderivtwo_i$}{\typctxtwo_i ; \vartwo \hastype \mtypethree_i ; \var \hastype \mtypetwo_i \vdash \tmtwo 
				\hastype \mtype_i}
			\infer1[\footnotesize$\ruleFun$]{\tyjp{}{\la{\vartwo}{\tmtwo}}{\typctxtwo_{i} ; \var \hastype 
					\mtypetwo_{i}}{\ty{\mtypethree_{i}\!}{\!\mtype_{i}}}}
			\delims{ \left( }{ \right)_{1 \leq i \leq n} }
			\infer1[\footnotesize$\ruleManyVal$]{\tyjp{}{\la{\vartwo}{\tmtwo}}{\bigmplus_{i=1}^n \typctxtwo_i ; \var \hastype 
					\mplus_{i=1}^n \mtypetwo_i}{\bigmplus_{i=1}^{n} \mult{\ty{\mtypethree_{i}\!}{\!\mtype_{i}}}}}
			\end{prooftree}
			\end{equation*}
			Thus, $\sizem{\tderivtwo} = n + \sum_{i=1}^{n} \sizem{\tderivtwo_{i}}$ and $\size{\tderivtwo} = n + \sum_{i=1}^n 
			\size{\tderivtwo_i} $.
			By repeatedly applying \reflemmap{typing-value-complete}{merge}, there is a derivation 
			$\concl{\tderivthree}{\typctxthree}{\val}{\mtypetwo}$ with $\typctxthree = \bigmplus_{i=1}^n \typctxthree_i$ such that 
			$\sizem{\tderivthree} = \sum_{i=1}^n\sizem{\tderivthree_i}$ and $\size{\tderivthree} = 
			\sum_{i=1}^n\size{\tderivthree_i}$.
			So, $\typctx = \bigmplus_{i=1}^n \typctx_i = \bigmplus_{i=1}^n (\typctxtwo_i \mplus \typctxthree_i) = \typctxtwo 
			\mplus \typctxthree$ with $\sizem{\tderiv} = n + \sum_{i=1}^n \sizem{\tderiv_i} = n + \sum_{i=1}^n 
			(\sizem{\tderivtwo_i} + \sizem{\tderivthree_i}) = \sizem{\tderivtwo} + \sizem{\tderivthree}$ 
			and $\size{\tderiv} = n + \sum_{i=1}^n \size{\tderiv_i} \leq n + \sum_{i=1}^n (\size{\tderivtwo_i} + 
			\size{\tderivthree_i}) = \size{\tderivtwo} + \size{\tderivthree}$.
		\end{itemize}
		
		\item \emph{Explicit substitution}, \ie $\tm = \tmtwo \esub{\vartwo}{\tmthree}$. 
		We can suppose without loss of generality that $\vartwo \notin \fv{\val} \cup \{\var \}$, hence $\tm 
		\isub{\var}{\val} = \tmtwo\isub{\var}{\val} \esub{\vartwo}{\tmthree\isub{\var}\val}$ and necessarily
		\begin{equation*}
		\tderiv = 
		\begin{prooftree}
		\hypo{}
		\ellipsis{$\tderiv_1$}{\tyjp{}{\tmtwo\isub{\var}{\val}}{\typctx_{1}, \vartwo \hastype \mtypethree}{\mtype}}
		\hypo{}
		\ellipsis{$\tderiv_2$}{\typctx_2 \vdash \tmthree\isub{\var}{\val} \hastype \mtypethree}
		\infer2[\footnotesize$\Es$]{\typctx \vdash \tmtwo \isub{\var}{\val} \esub{\vartwo} {\tmthree \isub{\var}{\val}} 
			\hastype \mtype}
		\end{prooftree}
		\end{equation*}
		with $\sizem{\tderiv} = \sizem{\tderiv_{1}} + \sizem{\tderiv_{2}}$, $\size{\tderiv} = \size{\tderiv_{1}} + 
		\size{\tderiv_{2}} + 1$ and $\typctx = \typctx_1 \mplus \typctx_2$. 
		By \ih applied to $\tderiv_1$ and \refrmk{free-variables}, there are derivations 
		$\concl{\tderivtwo_1}{\typctxtwo_1, \vartwo \hastype \mtypethree, \var \hastype \mtypetwo_1}{\tmtwo}{\mtype}$ and
		$\concl{\tderivthree_1}{\typctxthree_1}{\val}{\mtypetwo_1}$ with $\typctx_1 = \typctxtwo_1 \mplus \typctxthree_1$ 
		such that $\sizem{\tderiv_1} = \sizem{\tderivtwo_1} + \sizem{\tderivthree_1}$ and $\size{\tderiv_1} \leq 
		\size{\tderivtwo_1} + \size{\tderivthree_1}$.
		By \ih applied to $\tderiv_2$ , there are derivations $\concl{\tderivtwo_2}{\typctxtwo_2, \var \hastype 
			\mtypetwo_2}{\tmthree}{\mtype}$ and
		$\concl{\tderivthree_2}{\typctxthree_2}{\val}{\mtypetwo_2}$ with $\typctx_2 = \typctxtwo_2 \mplus \typctxthree_2$ 
		such that $\sizem{\tderiv_2} = \sizem{\tderivtwo_2} + \sizem{\tderivthree_2}$ and $\size{\tderiv_2} \leq 
		\size{\tderivtwo_2} + \size{\tderivthree_2}$.
		According to \reflemmap{typing-value-complete}{merge}, there is a derivation 
		$\namedtyjp{\tderivthree}{}{\val}{\typctxthree}{\mtypetwo}$ with $\typctxthree = \typctxthree_{1} \mplus 
		\typctxthree_{2}$ and $\mtypetwo = \mtypetwo_1 \mplus \mtypetwo_2$ such that $\sizem{\tderivthree} = 
		\sizem{\tderivthree_{1}} + \sizem{\tderivthree_{2}}$ and $\size{\tderivthree} = \size{\tderivthree_{1}} + 
		\size{\tderivthree_{2}}$.
		We can build the derivation (where $\typctxtwo = \typctxtwo_1 \mplus \typctxtwo_2$)
		\begin{equation*}
		\tderivtwo = 
		\begin{prooftree}
			\hypo{}
			\ellipsis{$\tderivtwo_{1}$}{\typctxtwo_1 , \var \hastype \mtypetwo_1 , \vartwo \hastype \mtypethree \vdash \tmtwo 
				\hastype \mtype}
			\hypo{}
			\ellipsis{$\tderivtwo_{2}$}{\typctxtwo_2, \var \hastype \mtypetwo_2 \vdash \tmthree \hastype \mtypethree}
			\infer2[\footnotesize$\Es$]{\typctxtwo, \var \hastype \mtypetwo \vdash \tmtwo \esub{\vartwo}{\tmthree} \hastype 
				\mtype}
		\end{prooftree}
		\end{equation*}
		verifying that $\typctx = \typctx_1 \mplus \typctx_2 = \typctxtwo_1 \mplus \typctxthree_1 \mplus \typctxtwo_2 
		\mplus \typctxthree_2 = \typctxtwo \mplus \typctxthree$ and $\sizem{\tderiv} = \sizem{\tderiv_{1}} + \sizem{\tderiv_{2}} 
		= \sizem{\tderivtwo_{1}} + \sizem{\tderivthree_{1}} + \sizem{\tderivtwo_{2}} + \sizem{\tderivthree_{2}} = 
		\sizem{\tderivtwo} + \sizem{\tderivthree}$ and $\size{\tderiv} = 1 + \size{\tderiv_{1}} + \size{\tderiv_{2}} \leq 1 + 
		(\size{\tderivtwo_{1}} + \size{\tderivthree_{1}}) + (\size{\tderivthree_{2}} + \size{\tderivthree_{2}}) = 
		\size{\tderivtwo} + \size{\tderivthree}$.
		\qedhere
	\end{itemize}	
\end{proof}

\begin{lemma}[Linear substitution and removal]
	\label{l:linear-substitution-removal}
	Let $\tm$ be a term and $\weakctx$ be an open context such that $\var \notin \fv{\weakctx} \cup \bv{\weakctx}$.
	There is a derivation $\namedtyjp{\tderiv}{}{\weakctxp{\var}\esub{\var}{\tm}}{\typctx}{\mtype}$ if and only if there is a derivation $\namedtyjp{\tderiv'}{}{\weakctxp{\tm}}{\typctx}{\mtype}$.
\end{lemma}

\begin{proof}
	\begin{description}
		\item[$\Rightarrow$ (linear substitution):]
		By induction on $\weakctx$. 
		Cases:
		\begin{itemize}
			\item $\weakctx = \ctxhole$, \ie $\weakctxp{\var}\esub{\var}{\tm} = \var\esub{\var}{\tm}$ and $\weakctxp{\tm} = \tm$.
			Then, the derivation $\tderiv$ has the form below, for some $n \geq 0$ (where $\mtype = \mset{\ltype_1, \dots, \ltype_n}$)
			\begin{equation*}
			\tderiv = 
			\begin{prooftree}
			\infer0[\footnotesize$\ruleAx$]{\tyjp{}{\var}{\var \hastype \mset{\ltype_{1}}}{\ltype_1}}
			\hypo{\cdots}
			\infer0[\footnotesize$\ruleAx$]{\tyjp{}{\var}{\var \hastype \mset{\ltype_{n}}}{\ltype_n}}
			\infer3[\footnotesize$\ruleMany$]{\tyjp{}{\var}{\var \hastype \mset{\ltype_1, \dots, \ltype_n}}{\mset{\ltype_1, \dots, \ltype_n}}}
			\hypo{}
			\ellipsis{$\tderivp$}{\tyjp{}{\tm}{\typctx}{\mset{\ltype_1, \dots, \ltype_n}}}
			\infer2[\footnotesize$\ruleES$]{\tyjp{}{\var\esub{\var}{\tm}}{\typctx}{\mset{\ltype_1, \dots, \ltype_n}}}
			\end{prooftree}
			\end{equation*}
			Thus, we have the derivation $\namedtyjp{\tderivp}{}{\tm}{\typctx}{\mset{\ltype_1, \dots, \ltype_n}}$.
			
			\item $\weakctx = \weakctxtwo\tmthree$, \ie~$\weakctxp{\var}\esub{\var}{\tm} = (\weakctxtwop{\var} \tmthree)\esub{\var}{\tm}$ and $\weakctxp{\tm} =  \weakctxtwop{\tm} \tmthree$  because $\var \notin \fv{\weakctx} \cup \bv{\weakctx}$.
			The derivation $\tderiv$ is necessarily (with $\var\notin \dom{\typctxthree}$ by \Cref{rmk:free-variables}, since $\var\notin \fv{\tmthree}$)
			\begin{equation*}
			\tderiv = 
			\begin{prooftree}
			\hypo{}
			\ellipsis{$\tderivtwo$}{\tyjp{}{\weakctxtwop{\var}}{\typctxtwo, \var \hastype \mtype'}{\mult{\ty{\mtypetwo}{\mtype}}}}
			\hypo{}
			\ellipsis{$\tderivthree$}{\tyjp{}{\tmthree}{\typctxthree}{\mtypetwo}}
			\infer2[\footnotesize$\ruleAp$]{\tyjp{}{\weakctxtwop{\var} \tmthree}{\typctxtwo \mplus \typctxthree, \var \hastype \mtype'}{\mtype}}
			\hypo{}
			\ellipsis{$\tderivfour$}{\tyjp{}{\tm}{\typctxfour}{\mtype'}}
			\infer2[\footnotesize$\ruleES$]{\tyjp{}{(\weakctxtwop{\var} \tmthree)\esub{\var}{\tm}}{\typctxtwo \mplus \typctxthree \mplus \typctxfour}{\mtype}}
			\end{prooftree}
			\end{equation*}
			where $\typctx = \typctxtwo \mplus \typctxthree \mplus \typctxfour$.
			Consider the derivation 
			\begin{equation*}
			\tderiv_0 = 
			\begin{prooftree}
			\hypo{}
			\ellipsis{$\tderivtwo$}{\tyjp{}{\weakctxtwop{\var}}{\typctxtwo, \var \hastype \mtype'}{\mult{\ty{\mtypetwo}{\mtype}}}}
			\hypo{}
			\ellipsis{$\tderivfour$}{\tyjp{}{\tm}{\typctxfour}{\mtype'}}
			\infer2[\footnotesize$\ruleES$]{\tyjp{}{\weakctxtwop{\var} \esub{\var}{\tm}}{\typctxtwo \mplus \typctxfour}{\mult{\ty{\mtypetwo}{\mtype}}}}
			\end{prooftree}
			\end{equation*}
			
			By \ih applied to $\tderiv_0$, there is a derivation $\namedtyjp{\tderiv_0'}{}{\weakctxtwop{\tm}}{\typctxtwo \mplus \typctxfour}{\mult{\ty{\mtypetwo}{\mtype}}}$.
			We can then build the derivation       
			\begin{equation*}                      
			\tderiv' =                             
			\begin{prooftree}                      
			\hypo{}                                
			\ellipsis{$\tderiv_0'$}{\tyjp{}{\weakctxtwop{\tm}}{\typctxtwo \mplus \typctxfour}{\mult{\ty{\mtypetwo}{\mtype}}}}
			\hypo{}                                
			\ellipsis{$\tderivthree$}{\tyjp{}{\tmthree}{\typctxthree}{\mtypetwo}}
			\infer2[\footnotesize$\ruleAp$]{\tyjp{}{\weakctxtwop{\tm} \tmthree}{\typctxtwo \mplus \typctxthree \mplus \typctxfour}{\mtype}}
			\end{prooftree}                        
			\end{equation*}                        
			
			\item \emph{Application right}, \ie $\weakctx = \tmthree \weakctxtwo$.
			Analogous to the previous case.        
			
			\item \emph{Explicit substitution left}, \ie $\weakctx = \weakctxtwo\esub{\vartwo}{\tmthree}$. 
			Then, $\weakctxp{\var}\esub{\var}{\tm} = \weakctxtwop{\var} \esub{\vartwo}{\tmthree} \esub{\var}{\tm} $ and $\weakctxp{\tm} = \weakctxtwop{\tm}\esub{\vartwo}{\tmthree}$ because $\var \notin \fv{\weakctx} \cup \bv{\weakctx}$.
			The derivation $\tderiv$ is necessarily (with $\var\notin \dom{\typctxthree}$ by \Cref{rmk:free-variables}, since $\var\notin \fv{\tmthree}$)
			\begin{equation*}                      
			\tderiv =                              
			\begin{prooftree}                      
			\hypo{}                                
			\ellipsis{$\tderivtwo$}{\tyjp{}{\weakctxtwop{\var}}{\typctxtwo, \var \hastype \mtype', \vartwo \hastype \mtypetwo}{\mtype}}
			\hypo{}                                
			\ellipsis{$\tderivthree$}{\tyjp{}{\tmthree}{\typctxthree}{\mtypetwo}}
			\infer2[\footnotesize$\Es$]{\tyjp{}{\weakctxtwop{\var} \esub{\vartwo}{\tmthree}}{\typctxtwo \mplus \typctxthree, \var \hastype \mtype'}{\mtype}}
			\hypo{}
			\ellipsis{$\tderivfour$}{\tyjp{}{\tm}{\typctxfour}{\mtype'}}
			\infer2[\footnotesize$\ruleES$]{\tyjp{}{\weakctxtwop{\var} \esub{\vartwo}{\tmthree} \esub{\var}{\tm}}{\typctxtwo \mplus \typctxthree \mplus \typctxfour}{\mtype}}
			\end{prooftree}
			\end{equation*}
			where $\typctx = \typctxtwo \mplus \typctxthree \mplus \typctxfour$.
			Consider the derivation 
			\begin{equation*}
			\tderiv_0 = 
			\begin{prooftree}
			\hypo{}
			\ellipsis{$\tderivtwo$}{\tyjp{}{\weakctxtwop{\var}}{\typctxtwo, \var \hastype \mtype', \vartwo \hastype \mtypetwo}{\mtype}}
			\hypo{}
			\ellipsis{$\tderivfour$}{\tyjp{}{\tm}{\typctxfour}{\mtype'}}
			\infer2[\footnotesize$\ruleES$]{\tyjp{}{\weakctxtwop{\var} \esub{\var}{\tm}}{\typctxtwo \mplus \typctxfour, \vartwo \hastype \mtypetwo}{\mtype}}
			\end{prooftree}
			\end{equation*}
			By \ih applied to $\tderiv_0$, there is a derivation $\namedtyjp{\tderiv_0'}{}{\weakctxtwop{\tm}}{\typctxtwo \mplus \typctxfour, \vartwo \hastype \mtypetwo}{\mtype}$.
			We can then build the derivation 
			\begin{equation*}
			\tderiv' = 
			\begin{prooftree}
			\hypo{}
			\ellipsis{$\tderiv_0'$}{\tyjp{}{\weakctxtwop{\tm}}{\typctxtwo \mplus \typctxfour, \vartwo \hastype \mtypetwo}{\mtype}}
			\hypo{}
			\ellipsis{$\tderivthree$}{\tyjp{}{\tmthree}{\typctxthree}{\mtypetwo}}
			\infer2[\footnotesize$\Es$]{\typctxtwo \mplus \typctxthree \mplus \typctxfour \vdash \weakctxtwop{\tm} \esub{\vartwo}{\tmthree} \hastype \mtype}
			\end{prooftree}
			\end{equation*}
			
			\item \emph{Explicit substitution right}, \ie $\weakctx = \tmthree \esub{\var}{\weakctxtwo}$. 
			Analogous to the previous case.
		\end{itemize}				
		
		\item[$\Leftarrow$ (linear removal):]
		By induction on $\weakctx$. 
		Cases:
		\begin{itemize}
			\item $\weakctx = \ctxhole$, \ie $\weakctxp{\var}\esub{\var}{\tm} = \var\esub{\var}{\tm}$ and $\weakctxp{\tm} = \tm$.
			Given the derivation $\namedtyjp{\tderivp}{}{\tm}{\typctx}{\mset{\ltype_1, \dots, \ltype_n}}$  for some $n \geq 0$ (where $\mtype = \mset{\ltype_1, \dots, \ltype_n}$), we can build the derivation 
			\begin{equation*}
			\tderiv = 
			\begin{prooftree}
			\infer0[\footnotesize$\ruleAx$]{\tyjp{}{\var}{\var \hastype \mset{\ltype_{1}}}{\ltype_1}}
			\hypo{\cdots}
			\infer0[\footnotesize$\ruleAx$]{\tyjp{}{\var}{\var \hastype \mset{\ltype_{n}}}{\ltype_n}}
			\infer3[\footnotesize$\ruleMany$]{\tyjp{}{\var}{\var \hastype \mset{\ltype_1, \dots, \ltype_n}}{\mset{\ltype_1, \dots, \ltype_n}}}
			\hypo{}
			\ellipsis{$\tderivp$}{\tyjp{}{\tm}{\typctx}{\mset{\ltype_1, \dots, \ltype_n}}}
			\infer2[\footnotesize$\ruleES$]{\tyjp{}{\var\esub{\var}{\tm}}{\typctx}{\mset{\ltype_1, \dots, \ltype_n}}}
			\end{prooftree}
			\end{equation*}
			
			\item $\weakctx = \weakctxtwo\tmthree$, \ie~$\weakctxp{\var}\esub{\var}{\tm} = (\weakctxtwop{\var} \tmthree)\esub{\var}{\tm}$ and $\weakctxp{\tm} =  \weakctxtwop{\tm} \tmthree$  because $\var \notin \fv{\weakctx} \cup \bv{\weakctx}$.
			The derivation $\tderiv'$ is necessarily of the form below (with $\var\notin \dom{\typctxthree}$ by \Cref{rmk:free-variables}, since $\var\notin \fv{\tmthree}$)
			\begin{equation*}                      
			\tderiv' =                             
			\begin{prooftree}                      
			\hypo{}                                
			\ellipsis{$\tderiv_0'$}{\tyjp{}{\weakctxtwop{\tm}}{\typctxtwo \mplus \typctxfour}{\mult{\ty{\mtypetwo}{\mtype}}}}
			\hypo{}                                
			\ellipsis{$\tderivthree$}{\tyjp{}{\tmthree}{\typctxthree}{\mtypetwo}}
			\infer2[\footnotesize$\ruleAp$]{\tyjp{}{\weakctxtwop{\tm} \tmthree}{\typctxtwo \mplus \typctxthree \mplus \typctxfour}{\mtype}}
			\end{prooftree}                        
			\end{equation*}       
			where $\typctx = \typctxtwo \mplus \typctxthree \mplus \typctxfour$.
			By \ih applied to $\tderiv_0'$, there is a derivation $\namedtyjp{\tderiv_0}{}{\weakctxtwop{\var} \esub{\var}{\tm}}{\typctxtwo \mplus \typctxfour}{\mult{\ty{\mtypetwo}{\mtype}}}$, which is necessarily of the form 
			\begin{equation*}
			\tderiv_0 = 
			\begin{prooftree}
			\hypo{}
			\ellipsis{$\tderivtwo$}{\tyjp{}{\weakctxtwop{\var}}{\typctxtwo, \var \hastype \mtype'}{\mult{\ty{\mtypetwo}{\mtype}}}}
			\hypo{}
			\ellipsis{$\tderivfour$}{\tyjp{}{\tm}{\typctxfour}{\mtype'}}
			\infer2[\footnotesize$\ruleES$]{\tyjp{}{\weakctxtwop{\var} \esub{\var}{\tm}}{\typctxtwo \mplus \typctxfour}{\mult{\ty{\mtypetwo}{\mtype}}}}
			\end{prooftree}
			\end{equation*}
			
			We can then build the derivation  
			\begin{equation*}
			\tderiv = 
			\begin{prooftree}
			\hypo{}
			\ellipsis{$\tderivtwo$}{\tyjp{}{\weakctxtwop{\var}}{\typctxtwo, \var \hastype \mtype'}{\mult{\ty{\mtypetwo}{\mtype}}}}
			\hypo{}
			\ellipsis{$\tderivthree$}{\tyjp{}{\tmthree}{\typctxthree}{\mtypetwo}}
			\infer2[\footnotesize$\ruleAp$]{\tyjp{}{\weakctxtwop{\var} \tmthree}{\typctxtwo \mplus \typctxthree, \var \hastype \mtype'}{\mtype}}
			\hypo{}
			\ellipsis{$\tderivfour$}{\tyjp{}{\tm}{\typctxfour}{\mtype'}}
			\infer2[\footnotesize$\ruleES$]{\tyjp{}{(\weakctxtwop{\var} \tmthree)\esub{\var}{\tm}}{\typctxtwo \mplus \typctxthree \mplus \typctxfour}{\mtype}}
			\end{prooftree}
			\end{equation*}                      
			
			\item \emph{Application right}, \ie $\weakctx = \tmthree \weakctxtwo$.
			Analogous to the previous case.        
			
			\item \emph{Explicit substitution left}, \ie $\weakctx = \weakctxtwo\esub{\vartwo}{\tmthree}$. 
			Then, $\weakctxp{\var}\esub{\var}{\tm} = \weakctxtwop{\var} \esub{\vartwo}{\tmthree} \esub{\var}{\tm} $ and $\weakctxp{\tm} = \weakctxtwop{\tm}\esub{\vartwo}{\tmthree}$ because $\var \notin \fv{\weakctx} \cup \bv{\weakctx}$.
			The derivation $\tderiv'$ is necessarily of the form below (with $\var\notin \dom{\typctxthree}$ by \Cref{rmk:free-variables}, since $\var\notin \fv{\tmthree}$)
			\begin{equation*}
			\tderiv' = 
			\begin{prooftree}
			\hypo{}
			\ellipsis{$\tderiv_0'$}{\tyjp{}{\weakctxtwop{\tm}}{\typctxtwo \mplus \typctxfour, \vartwo \hastype \mtypetwo}{\mtype}}
			\hypo{}
			\ellipsis{$\tderivthree$}{\tyjp{}{\tmthree}{\typctxthree}{\mtypetwo}}
			\infer2[\footnotesize$\Es$]{\typctxtwo \mplus \typctxthree \mplus \typctxfour \vdash \weakctxtwop{\tm} \esub{\vartwo}{\tmthree} \hastype \mtype}
			\end{prooftree}
			\end{equation*}
			where $\typctx = \typctxtwo \mplus \typctxthree \mplus \typctxfour$.
			By \ih applied to $\tderiv_0'$, there is a derivation $\namedtyjp{\tderiv_0}{}{\weakctxtwop{\var}\esub{\var}{\tm}}{\typctxtwo \mplus \typctxfour, \vartwo \hastype \mtypetwo}{\mtype}$, which is necessarily of the form
			\begin{equation*}
			\tderiv_0 = 
			\begin{prooftree}
			\hypo{}
			\ellipsis{$\tderivtwo$}{\tyjp{}{\weakctxtwop{\var}}{\typctxtwo, \var \hastype \mtype', \vartwo \hastype \mtypetwo}{\mtype}}
			\hypo{}
			\ellipsis{$\tderivfour$}{\tyjp{}{\tm}{\typctxfour}{\mtype'}}
			\infer2[\footnotesize$\ruleES$]{\tyjp{}{\weakctxtwop{\var} \esub{\var}{\tm}}{\typctxtwo \mplus \typctxfour, \vartwo \hastype \mtypetwo}{\mtype}}
			\end{prooftree}
			\end{equation*}
			We can then build the derivation 
			\begin{equation*}                      
			\tderiv =                              
			\begin{prooftree}                      
			\hypo{}                                
			\ellipsis{$\tderivtwo$}{\tyjp{}{\weakctxtwop{\var}}{\typctxtwo, \var \hastype \mtype', \vartwo \hastype \mtypetwo}{\mtype}}
			\hypo{}                                
			\ellipsis{$\tderivthree$}{\tyjp{}{\tmthree}{\typctxthree}{\mtypetwo}}
			\infer2[\footnotesize$\Es$]{\tyjp{}{\weakctxtwop{\var} \esub{\vartwo}{\tmthree}}{\typctxtwo \mplus \typctxthree, \var \hastype \mtype'}{\mtype}}
			\hypo{}
			\ellipsis{$\tderivfour$}{\tyjp{}{\tm}{\typctxfour}{\mtype'}}
			\infer2[\footnotesize$\ruleES$]{\tyjp{}{\weakctxtwop{\var} \esub{\vartwo}{\tmthree} \esub{\var}{\tm}}{\typctxtwo \mplus \typctxthree \mplus \typctxfour}{\mtype}}
			\end{prooftree}
			\end{equation*}
			
			\item \emph{Explicit substitution right}, \ie $\weakctx = \tmthree \esub{\var}{\weakctxtwo}$. 
			Analogous to the previous case.
			\qedhere
		\end{itemize}				
	\end{description}
\end{proof}

The next proposition is slightly more general than the corresponding one in the body of the paper (\Cref{prop:qual-subject}) because it shows subject reduction and expansion not only with respect to $\tovsub \cup \eqstruct$, but also to the glueing rule $\toglue$ related to Moggi's calculus and discussed in \Cref{sect:app-moggi}.

\begin{proposition}[Qualitative subject reduction and expansion]
	\label{propappendix:qual-subject}
	\NoteState{prop:qual-subject}
	Let $\tm \,(\tovsub \!\cup \toglue \!\cup \eqstruct)\, \tm'$.
	There is a derivation $\namedtyjp{\tderiv}{}{\tm}{\typctx}{\mtype}$ if and only if there is a derivation $\namedtyjp{\tderiv'}{}{\tm'}{\typctx}{\mtype}$.
\end{proposition}

\begin{proof}
	\emph{Subject reduction.}
	We prove that if $\namedtyjp{\tderiv}{}{\tm}{\typctx}{\mtype}$ then there is a derivation $\namedtyjp{\tderiv'}{}{\tm'}{\typctx}{\mtype}$, by induction on the context $\fctx$ such that $\tm = \fctxp{\tmtwo} \,(\tovsub \!\cup \toglue \!\cup \eqstruct)\, \fctxp{\tmtwop} = \tmp$ with $\tmtwo \rootRew{} \tmtwop$, where $\rootRew{} \,\defeq\, \rtom \!\cup \rtoe \!\cup \rtoglue \!\cup \tostructcom \!\cup \tostructapr \!\cup \tostructapl \!\cup \tostructes$ (more precisely, for $\tostructapr$, $ \tostructapl$ and $\tostructes$, we take their symmetric closure). 
	Indeed, $\tovsub \!\cup \toglue \!\cup \eqstruct$ is the closure under full contexts of $\rootRew{}$.
	\begin{itemize}
		\item \emph{Empty context}; \ie, $\fctx = \ctxhole$. There are several sub-cases.
		\begin{itemize}
			\item \emph{Multiplicative}, \ie $\tm = \subctxp{\la\var\tmtwo}\tmthree \rtom  \subctxp{\tmtwo \esub{\var}{\tmthree}} = \tm'$.
			Then $\tderiv$ has the form:
			\begin{equation*}
			\tderiv = 
			\begin{prooftree}[separation = 1em]
			\hypo{}
			\ellipsis{$\tderivtwo$}{\typctx', \var \hastype \mtypetwo \vdash \tmtwo \hastype \mtype}
			\infer1[\footnotesize$\ruleFun$]{\typctx' \vdash \la\var\tmtwo \hastype \larrow{\mtypetwo}{\mtype}}
			\infer1[\footnotesize$\ruleManyVal$]{\typctx' \vdash \la\var\tmtwo \hastype \mset{\larrow{\mtypetwo}{\mtype}}}
			\hypo{}
			\ellipsis{$\tderiv_1$}{\quad}
			\infer2[\footnotesize$\ruleES$]{}
			\ellipsis{}{\quad}
			\hypo{}
			\ellipsis{$\tderiv_n$}{\quad}
			\infer2[\footnotesize$\ruleES$]{\typctx \vdash \subctxp{\la\var\tmtwo} \hastype \mset{\larrow{\mtypetwo}{\mtype}}}
			\hypo{}
			\ellipsis{$\tderivthree$}{\typctxtwo\vdash\tmthree \hastype \mtypetwo}
			\infer2[\footnotesize$\ruleApp$]{\typctx \uplus \typctxtwo \vdash \subctxp{\la\var\tmtwo}\tmthree \hastype \mtype}
			\end{prooftree}
			\end{equation*}
			
			where $n \geq 0$ is the length of the list $\subctx$ of \ES.
			We can then build $\tderiv'$ as follows:
			\begin{equation*}
			\tderiv' = 
			\begin{prooftree}
			\hypo{}
			\ellipsis{$\tderivtwo$}{\tyjp{}{\tmtwo}{\typctx' ; \var \hastype \mtypetwo}{\mtype}}
			\hypo{}
			\ellipsis{$\tderivthree$}{\tyjp{}{\tmthree}{\typctxtwo}{\mtypetwo}}
			\infer2[\footnotesize$\ruleES$]{\typctx'\uplus\typctxtwo \vdash\tmtwo\esub\var\tmthree \hastype \mtype}
			\hypo{}
			\ellipsis{$\tderiv_1$}{\quad}
			\infer2[\footnotesize$\ruleES$]{}
			\ellipsis{}{\quad}
			\hypo{}
			\ellipsis{$\tderiv_n$}{\quad}
			\infer2[\footnotesize$\Es$]{\typctx\uplus\typctxtwo \vdash \subctxp{\tmtwo\esub\var\tmthree} \hastype \mtype}
			\end{prooftree}
			\end{equation*}
			
			\item \emph{Exponential}, \ie $\tm = \tmtwo\esub\var{\subctxp{\val}} \rtoe \subctxp{\tmtwo \isub{\var}{\val}} = \tmp$.
			Then the derivation $\tderiv$ has the form (where $n \geq 0$ is the length of the list $\subctx$ of \ES):
			\begin{equation*}
			\tderiv = 
			\begin{prooftree}
			\hypo{}
			\ellipsis{$\tderivtwo$}{\tyjp{}{\tmtwo}{\typctxtwo, \var \hastype \mtypetwo}{\mtype}}
			\hypo{}
			\ellipsis{$\tderivthree$}{\tyjp{}{\val}{\typctxthree'}{\mtypetwo}}
			\hypo{}
			\ellipsis{$\tderiv_1$}{\quad}
			\infer2[\footnotesize$\Es$]{}
			\ellipsis{}{\quad}
			\hypo{}
			\ellipsis{$\tderiv_n$}{\quad}
			\infer2[\footnotesize$\Es$]{\typctxthree \vdash \subctxp{\val} \hastype \mtypetwo}
			\infer2[\footnotesize$\Es$]{\typctxtwo\uplus\typctxthree \vdash\tmtwo\esub\var{\subctxp{\val}}\hastype \mtype}
			\end{prooftree}
			\end{equation*}
			
			By the substitution lemma (\Cref{l:substitution}), there is a derivation $\namedtyjp{\tderiv''}{}{\tmtwo\isub{\var}{\val}}{\typctxtwo \mplus \typctxthree'}{\mtype}$.
			We can then build the following derivation $\tderiv'$:
			\begin{equation*}
			\tderiv' = 
			\begin{prooftree}
			\hypo{}
			\ellipsis{$\tderiv''$}{\typctxtwo \mplus\typctxthree' \vdash \tmtwo\isub\var \val \hastype \mtype}
			\hypo{}
			\ellipsis{$\tderiv_1$}{\quad}
			\infer2[\footnotesize{$\Es$}]{}
			\ellipsis{}{}
			\hypo{}
			\ellipsis{$\tderiv_n$}{\quad}
			\infer2[\footnotesize$\Es$]{\typctxtwo\mplus\typctxthree \vdash \subctxp{\tmtwo\isub\var \val} \hastype \mtype}
			\end{prooftree}
			\end{equation*}
			
			\item \emph{Glue}, \ie $\tm = \weakctxp\var\esub\var\aptm  \rtoglue    \weakctxp\aptm = \tmp$ where $\var \notin \fv{\weakctx} \cup \bv{\weakctx}$ and $\aptm$ is an application.
			Given the derivation $\namedtyjp{\tderiv}{}{\weakctxp\var\esub\var\aptm}{\typctx}{\mtype}$,  there is a derivation $\namedtyjp{\tderiv'}{}{\weakctxp\aptm}{\typctx}{\mtype}$ by linear substitution (\Cref{l:linear-substitution-removal}).
			
			\item $\tostructcom$, \ie~$\tm = \tmthree \esub{\vartwo}{\tmfour} \esub{\varthree}{\tmfive} \tostructcom \tmthree \esub{\varthree}{\tmfive} \esub{\vartwo}{\tmfour} = \tmp$, with $\vartwo \notin \fv{\tmfive}$ and $\varthree \notin \fv{\tmfour}$. 
			Then $\tderiv$ is of the form	
			\begin{prooftree*}
				\hypo{}
				\ellipsis{$\tderivtwo$}{\tyjp{}{\tmthree}{\typctxtwo, \vartwo \hastype \mtypetwo ; \varthree \hastype \mtypethree}{\mtype}}
				\hypo{}
				\ellipsis{$\tderivthree$}{\tyjp{}{\tmfour}{\typctxthree}{\mtypetwo}}
				\infer2[\footnotesize$\ruleES$]{\tyjp{}{\tmthree \esub{\vartwo}{\tmfour}}{\typctxtwo \mplus \typctxthree, \varthree \hastype \mtypethree}{\mtype}}
				\hypo{}
				\ellipsis{$\tderivfour$}{\tyjp{}{\tmfive}{\typctxfour}{\mtypethree}}
				\infer2[\footnotesize$\ruleES$]{\tyjp{}{\tmthree \esub{\vartwo}{\tmfour} \esub{\varthree}{\tmfive}}{\typctxtwo \mplus \typctxthree \mplus \typctxfour}{\mtype}}
			\end{prooftree*}
			
			noting that since $\vartwo \notin \fv{\tmfive}$ and $\varthree \notin \fv{\tmfour}$ then $\vartwo \notin \dom{\typctxfour}$ and $\varthree \notin \dom{\typctxthree}$ by \Cref{rmk:free-variables}. Thus, we can build a derivation $\tderivp$ as
			\begin{prooftree*}
				\hypo{}
				\ellipsis{$\tderivtwo$}{\tyjp{}{\tmthree}{\typctxtwo, \vartwo \hastype \mtypetwo ; \varthree \hastype \mtypethree}{\mtype}}
				\hypo{}
				\ellipsis{$\tderivfour$}{\tyjp{}{\tmfive}{\typctxfour}{\mtypethree}}
				\infer2[\footnotesize$\ruleES$]{\tyjp{}{\tmthree \esub{\varthree}{\tmfive}}{\typctxtwo \mplus \typctxfour, \vartwo \hastype \mtypetwo}{\mtype}}
				\hypo{}
				\ellipsis{$\tderivthree$}{\tyjp{}{\tmfour}{\typctxthree}{\mtypetwo}}
				\infer2[\footnotesize$\ruleES$]{\tyjp{}{\tmthree \esub{\varthree}{\tmfive} \esub{\vartwo}{\tmfour}}{\typctxtwo \mplus \typctxthree \mplus \typctxfour}{\mtype}}
			\end{prooftree*}
			
			\item $\tostructapr$ left-to-right, \ie~$\tm = \tmthree (\tmfour \esub{\vartwo}{\tmfive}) \tostructapr (\tmthree \tmfour) \esub{\vartwo}{\tmfive} = \tmp$, with $\vartwo \notin \fv{\tmthree}$. Then $\tderiv$ is of the form
			\begin{prooftree*}
				\hypo{}
				\ellipsis{$\tderivtwo$}{\tyjp{}{\tmthree}{\typctxtwo}{\mult{\ty{\mtypetwo}{\mtype}}}}
				\hypo{}
				\ellipsis{$\tderivthree$}{\tyjp{}{\tmfour}{\typctxthree, \vartwo \hastype \mtypethree}{\mtypetwo}}
				\hypo{}
				\ellipsis{$\tderivfour$}{\tyjp{}{\tmfive}{\typctxfour}{\mtypethree}}
				\infer2[\footnotesize$\ruleES$]{\tyjp{}{\tmfour \esub{\vartwo}{\tmfive}}{\typctxthree \mplus \typctxfour}{\mtypetwo}}
				\infer2[\footnotesize$\ruleApp$]{\tyjp{}{\tmthree (\tmfour \esub{\vartwo}{\tmfive})}{\typctxtwo \mplus \typctxthree \mplus \typctxfour}{\mtype}}
			\end{prooftree*}
			
			noting that $\vartwo \notin \dom{\typctxtwo}$ by \Cref{rmk:free-variables}. 
			Thus, we can build a derivation $\tderivp$ as 
			\begin{prooftree*}
				\hypo{}
				\ellipsis{$\tderivtwo$}{\tyjp{}{\tmthree}{\typctxtwo}{\mult{\ty{\mtypetwo}{\mtype}}}}
				\hypo{}
				\ellipsis{$\tderivfour$}{\tyjp{}{\tmfour}{\typctxthree, \vartwo \hastype \mtypethree}{\mtypetwo}}
				\infer2[\footnotesize$\ruleApp$]{\tyjp{}{\tmthree \tmfour}{\typctxtwo \mplus \typctxthree, \vartwo \hastype \mtypethree}{\mtype}}
				\hypo{}
				\ellipsis{$\tderivthree$}{\tyjp{}{\tmfive}{\typctxfour}{\mtypethree}}
				\infer2[\footnotesize$\ruleES$]{\tyjp{}{(\tmthree \tmfour) \esub{\vartwo}{\tmfive}}{\typctxtwo \mplus \typctxfour \mplus \typctxthree}{\mtype}}
			\end{prooftree*}
			
			\item $\tostructapr$ right-to-left, \ie~$\tm = (\tmthree \tmfour) \esub{\vartwo}{\tmfive} \tostructapr \tmthree (\tmfour \esub{\vartwo}{\tmfive}) = \tmp$, with $\vartwo \notin \fv{\tmthree}$. 
			Then $\tderiv$ is of the form
			\begin{prooftree*}
				\hypo{}
				\ellipsis{$\tderivtwo$}{\tyjp{}{\tmthree}{\typctxtwo}{\mult{\ty{\mtypetwo}{\mtype}}}}
				\hypo{}
				\ellipsis{$\tderivfour$}{\tyjp{}{\tmfour}{\typctxthree, \vartwo \hastype \mtypethree}{\mtypetwo}}
				\infer2[\footnotesize$\ruleApp$]{\tyjp{}{\tmthree \tmfour}{\typctxtwo \mplus \typctxthree, \vartwo \hastype \mtypethree}{\mtype}}
				\hypo{}
				\ellipsis{$\tderivthree$}{\tyjp{}{\tmfive}{\typctxfour}{\mtypethree}}
				\infer2[\footnotesize$\ruleES$]{\tyjp{}{(\tmthree \tmfour) \esub{\vartwo}{\tmfive}}{\typctxtwo \mplus \typctxfour \mplus \typctxthree}{\mtype}}
			\end{prooftree*}
			
			noting that $\vartwo \notin \dom{\typctxtwo}$ by \Cref{rmk:free-variables}. 
			Thus, we can build a derivation $\tderivp$ as 
			\begin{prooftree*}
				\hypo{}
				\ellipsis{$\tderivtwo$}{\tyjp{}{\tmthree}{\typctxtwo}{\mult{\ty{\mtypetwo}{\mtype}}}}
				\hypo{}
				\ellipsis{$\tderivthree$}{\tyjp{}{\tmfour}{\typctxthree, \vartwo \hastype \mtypethree}{\mtypetwo}}
				\hypo{}
				\ellipsis{$\tderivfour$}{\tyjp{}{\tmfive}{\typctxfour}{\mtypethree}}
				\infer2[\footnotesize$\ruleES$]{\tyjp{}{\tmfour \esub{\vartwo}{\tmfive}}{\typctxthree \mplus \typctxfour}{\mtypetwo}}
				\infer2[\footnotesize$\ruleApp$]{\tyjp{}{\tmthree (\tmfour \esub{\vartwo}{\tmfive})}{\typctxtwo \mplus \typctxthree \mplus \typctxfour}{\mtype}}
			\end{prooftree*}
			
			\item $\tostructapl$ left-to-right, \ie~$\tm = (\tmthree \esub{\vartwo}{\tmfour}) \tmfive \tostructapl (\tmthree \tmfive) \esub{\vartwo}{\tmfour} = \tmp$, with $\vartwo \notin \fv{\tmfive}$. Then $\tderiv$ is of the form
			\begin{prooftree*}
				\hypo{}
				\ellipsis{$\tderivtwo$}{\tyjp{}{\tmthree}{\typctxtwo, \vartwo \hastype \mtypethree}{\mult{\ty{\mtypetwo}{\mtype}}}}
				\hypo{}
				\ellipsis{$\tderivthree$}{\tyjp{}{\tmfour}{\typctxthree}{\mtypethree}}
				\infer2[\footnotesize$\ruleES$]{\tyjp{}{\tmthree \esub{\vartwo}{\tmfour}}{\typctxtwo \mplus \typctxthree}{\mult{\ty{\mtypetwo}{\mtype}}}}
				\hypo{}
				\ellipsis{$\tderivfour$}{\tyjp{}{\tmfive}{\typctxfour}{\mtypetwo}}
				\infer2[\footnotesize$\ruleApp$]{\tyjp{}{(\tmthree \esub{\vartwo}{\tmfour}) \tmfive}{\typctxtwo \mplus \typctxthree \mplus \typctxfour}{\mtype}}
			\end{prooftree*}
			
			noting that $\vartwo \notin \dom{\typctxfour}$ by \Cref{rmk:free-variables}. Thus, we can build a derivation $\tderivp$ as 
			\begin{prooftree*}
				\hypo{}
				\ellipsis{$\tderivtwo$}{\tyjp{}{\tmthree}{\typctxtwo, \vartwo \hastype \mtypethree}{\mult{\ty{\mtypetwo}{\mtype}}}}
				\hypo{}
				\ellipsis{$\tderivfour$}{\tyjp{}{\tmfive}{\typctxfour}{\mtypetwo}}
				\infer2[\footnotesize$\ruleApp$]{\tyjp{}{\tmthree \tmfive}{\typctxtwo \mplus \typctxfour, \vartwo \hastype \mtypethree}{\mtype}}
				\hypo{}
				\ellipsis{$\tderivthree$}{\tyjp{}{\tmfour}{\typctxthree}{\mtypethree}}
				\infer2[\footnotesize$\ruleES$]{\tyjp{}{(\tmthree \tmfive) \esub{\vartwo}{\tmfour}}{\typctxtwo \mplus \typctxthree \mplus \typctxfour}{\mtype}}
			\end{prooftree*}
			
			\item $\tostructapr$ right-to-left, \ie~$\tm = (\tmthree \tmfive) \esub{\vartwo}{\tmfour} \tostructapl (\tmthree \esub{\vartwo}{\tmfour}) \tmfive = \tmp$, with $\vartwo \notin \fv{\tmfive}$. 
			Then $\tderiv$ is of the form
			\begin{prooftree*}
				\hypo{}
				\ellipsis{$\tderivtwo$}{\tyjp{}{\tmthree}{\typctxtwo, \vartwo \hastype \mtypethree}{\mult{\ty{\mtypetwo}{\mtype}}}}
				\hypo{}
				\ellipsis{$\tderivfour$}{\tyjp{}{\tmfive}{\typctxfour}{\mtypetwo}}
				\infer2[\footnotesize$\ruleApp$]{\tyjp{}{\tmthree \tmfive}{\typctxtwo \mplus \typctxfour, \vartwo \hastype \mtypethree}{\mtype}}
				\hypo{}
				\ellipsis{$\tderivthree$}{\tyjp{}{\tmfour}{\typctxthree}{\mtypethree}}
				\infer2[\footnotesize$\ruleES$]{\tyjp{}{(\tmthree \tmfive) \esub{\vartwo}{\tmfour}}{\typctxtwo \mplus \typctxthree \mplus \typctxfour}{\mtype}}
			\end{prooftree*}
			
			noting that $\vartwo \notin \dom{\typctxfour}$ by \Cref{rmk:free-variables}. 
			Thus, we can build a derivation $\tderivp$ as 
			\begin{prooftree*}
				\hypo{}
				\ellipsis{$\tderivtwo$}{\tyjp{}{\tmthree}{\typctxtwo, \vartwo \hastype \mtypethree}{\mult{\ty{\mtypetwo}{\mtype}}}}
				\hypo{}
				\ellipsis{$\tderivthree$}{\tyjp{}{\tmfour}{\typctxthree}{\mtypethree}}
				\infer2[\footnotesize$\ruleES$]{\tyjp{}{\tmthree \esub{\vartwo}{\tmfour}}{\typctxtwo \mplus \typctxthree}{\mult{\ty{\mtypetwo}{\mtype}}}}
				\hypo{}
				\ellipsis{$\tderivfour$}{\tyjp{}{\tmfive}{\typctxfour}{\mtypetwo}}
				\infer2[\footnotesize$\ruleApp$]{\tyjp{}{(\tmthree \esub{\vartwo}{\tmfour}) \tmfive}{\typctxtwo \mplus \typctxthree \mplus \typctxfour}{\mtype}}
			\end{prooftree*}
			
			\item $\tostructes$ left-to-right, \ie~$\tm = \tmthree \esub{\vartwo}{\tmfour \esub{\varthree}{\tmfive}} \tostructes \tmtwo \esub{\vartwo}{\tmfour} \esub{\varthree}{\tmfive} = \tmp$, with $\varthree \notin \fv{\tmthree}$. 
			Then $\tderiv$ is of the form
			\begin{prooftree*}
				\hypo{}
				\ellipsis{$\tderivtwo$}{\tyjp{}{\tmthree}{\typctxtwo, \vartwo \hastype \mtypetwo}{\mtype}}
				\hypo{}
				\ellipsis{$\tderivthree$}{\tyjp{}{\tmfour}{\typctxthree, \varthree \hastype \mtypethree}{\mtypetwo}}
				\hypo{}
				\ellipsis{$\tderivfour$}{\tyjp{}{\tmfive}{\typctxfour}{\mtypethree}}
				\infer2[\footnotesize$\ruleES$]{\tyjp{}{\tmfour \esub{\varthree}{\tmfive}}{\typctxthree \mplus \typctxfour}{\mtypetwo}}
				\infer2[\footnotesize$\ruleES$]{\tyjp{}{\tmthree \esub{\vartwo}{\tmfour \esub{\varthree}{\tmfive}}}{\typctxtwo \mplus \typctxthree \mplus \typctxfour}{\mtype}}
			\end{prooftree*}
			
			noting that $\varthree \notin \dom{\typctxtwo}$ by \Cref{rmk:free-variables}. 
			Thus, we can build a derivation $\tderivp$ as 
			\begin{prooftree*}
				\hypo{}
				\ellipsis{$\tderivtwo$}{\tyjp{}{\tmthree}{\typctxtwo, \vartwo \hastype \mtypetwo}{\mtype}}
				\hypo{}
				\ellipsis{$\tderivthree$}{\tyjp{}{\tmfour}{\typctxthree, \varthree \hastype \mtypethree}{\mtypetwo}}
				\infer2[\footnotesize$\ruleES$]{\tyjp{}{\tmthree \esub{\vartwo}{\tmfour}}{\typctxtwo \mplus \typctxthree, \varthree \hastype \mtypethree}{\mtype}}
				\hypo{}
				\ellipsis{$\tderivfour$}{\tyjp{}{\tmfive}{\typctxfour}{\mtypethree}}
				\infer2[\footnotesize$\ruleES$]{\tyjp{}{\tmthree \esub{\vartwo}{\tmfour} \esub{\varthree}{\tmfive}}{\typctxtwo \mplus \typctxthree \mplus \typctxfour}{\mtype}}
			\end{prooftree*}
			
			\item $\tostructes$ right-to-left, \ie~$\tm = \tmtwo \esub{\vartwo}{\tmfour} \esub{\varthree}{\tmfive} \tostructes \tmthree \esub{\vartwo}{\tmfour \esub{\varthree}{\tmfive}} = \tmp$, with $\varthree \notin \fv{\tmthree}$.
			Then $\tderiv$ is of the form
			\begin{prooftree*}
				\hypo{}
				\ellipsis{$\tderivtwo$}{\tyjp{}{\tmthree}{\typctxtwo, \vartwo \hastype \mtypetwo}{\mtype}}
				\hypo{}
				\ellipsis{$\tderivthree$}{\tyjp{}{\tmfour}{\typctxthree, \varthree \hastype \mtypethree}{\mtypetwo}}
				\infer2[\footnotesize$\ruleES$]{\tyjp{}{\tmthree \esub{\vartwo}{\tmfour}}{\typctxtwo \mplus \typctxthree, \varthree \hastype \mtypethree}{\mtype}}
				\hypo{}
				\ellipsis{$\tderivfour$}{\tyjp{}{\tmfive}{\typctxfour}{\mtypethree}}
				\infer2[\footnotesize$\ruleES$]{\tyjp{}{\tmthree \esub{\vartwo}{\tmfour} \esub{\varthree}{\tmfive}}{\typctxtwo \mplus \typctxthree \mplus \typctxfour}{\mtype}}
			\end{prooftree*}
			
			noting that $\varthree \notin \dom{\typctxtwo}$ by \Cref{rmk:free-variables}. 
			Thus, we can build a derivation $\tderivp$ as 
			\begin{prooftree*}
				\hypo{}
				\ellipsis{$\tderivtwo$}{\tyjp{}{\tmthree}{\typctxtwo, \vartwo \hastype \mtypetwo}{\mtype}}
				\hypo{}
				\ellipsis{$\tderivthree$}{\tyjp{}{\tmfour}{\typctxthree, \varthree \hastype \mtypethree}{\mtypetwo}}
				\hypo{}
				\ellipsis{$\tderivfour$}{\tyjp{}{\tmfive}{\typctxfour}{\mtypethree}}
				\infer2[\footnotesize$\ruleES$]{\tyjp{}{\tmfour \esub{\varthree}{\tmfive}}{\typctxthree \mplus \typctxfour}{\mtypetwo}}
				\infer2[\footnotesize$\ruleES$]{\tyjp{}{\tmthree \esub{\vartwo}{\tmfour \esub{\varthree}{\tmfive}}}{\typctxtwo \mplus \typctxthree \mplus \typctxfour}{\mtype}}
			\end{prooftree*}
		\end{itemize}
		
		\item \emph{Application right}, \ie~$\fctx = \tmthree \fctxtwo$. 
		Then $\tm = \tmthree \fctxtwop{\tmtwo} \,(\tovsub \!\cup \toglue \!\cup \eqstruct)\, \tmthree \fctxtwop{\tmtwop} = \tmp$ with $\tmtwo \rootRew{} \tmtwop$, and $\tderiv$ is of the form
		\begin{prooftree*}
			\hypo{}
			\ellipsis{$\tderivtwo$}{\tyjp{}{\tmthree}{\typctxtwo}{\mult{\ty{\mtypetwo}{\mtype}}}}
			\hypo{}
			\ellipsis{$\tderivthree$}{\tyjp{}{\fctxtwop{\tmtwo}}{\typctxthree}{\mtypetwo}}
			\infer2[\footnotesize$\ruleApp$]{\tyjp{}{\tmthree \fctxtwop{\tmtwo}}{\typctxtwo \mplus \typctxthree}{\mtype}}
		\end{prooftree*}
		
		By applying the \ih to $\tderivthree$, there exists $\namedtyjp{\tderivthreep}{}{\fctxtwop{\tmtwop}}{\typctxthree}{\mtypetwo}$. 
		Thus, we can build a  derivation $\tderivp$ as follows
		\begin{prooftree*}
			\hypo{}
			\ellipsis{$\tderivtwo$}{\tyjp{}{\tmthree}{\typctxtwo}{\mult{\ty{\mtypetwo}{\mtype}}}}
			\hypo{}
			\ellipsis{$\tderivthreep$}{\tyjp{}{\fctxtwop{\tmtwop}}{\typctxthree}{\mtypetwo}}
			\infer2[\footnotesize$\ruleApp$]{\tyjp{}{\tmthree \fctxtwop{\tmtwop}}{\typctxtwo \mplus \typctxthree}{\mtype}}
		\end{prooftree*}
		
		\item \emph{Application left}, \ie~$\fctx = \fctxtwo \tmthree$. Then $\tm = \fctxtwop{\tmtwo} \tmthree \,(\tovsub \!\cup \toglue \!\cup \eqstruct)\, \fctxtwop{\tmtwop} \tmthree = \tmp$ with $\tmtwo \rootRew{} \tmtwop$, and $\tderiv$ is of the form
		\begin{prooftree*}
			\hypo{}
			\ellipsis{$\tderivtwo$}{\tyjp{}{\fctxtwop{\tmtwo}}{\typctxtwo}{\mult{\ty{\mtypetwo}{\mtype}}}}
			\hypo{}
			\ellipsis{$\tderivthree$}{\tyjp{}{\tmthree}{\typctxthree}{\mtypetwo}}
			\infer2[\footnotesize$\ruleApp$]{\tyjp{}{\fctxtwop{\tmtwo} \tmthree}{\typctxtwo \mplus \typctxthree}{\mtype}}
		\end{prooftree*}
		
		By applying the \ih to $\tderivtwo$, there exists $\namedtyjp{\tderivtwop}{}{\fctxtwop{\tmtwop}}{\typctxtwo}{\mult{\ty{\mtypetwo}{\mtype}}}$. Thus, we can build a derivation $\tderivp$ as follows
		\begin{prooftree*}
			\hypo{}
			\ellipsis{$\tderivtwop$}{\tyjp{}{\fctxtwop{\tmtwop}}{\typctxtwo}{\mult{\ty{\mtypetwo}{\mtype}}}}
			\hypo{}
			\ellipsis{$\tderivthree$}{\tyjp{}{\tmthree}{\typctxthree}{\mtypetwo}}
			\infer2[\footnotesize$\ruleApp$]{\tyjp{}{\fctxtwop{\tmtwop} \tmthree}{\typctxtwo \mplus \typctxthree}{\mtype}}
		\end{prooftree*}
		
		\item \emph{Abstraction}, \ie~$\fctx = \la{\var}{\fctxtwo}$. 
		Then $\tm = \la{\var}{\fctxtwop{\tmtwo}} \,(\tovsub \!\cup \toglue \!\cup \eqstruct)\, \la{\var}{\fctxtwop{\tmtwop}} = \tmp$ with $\tmtwo \rootRew{} \tmtwop$, and $\tderiv$ is of the form (for some $n \geq 0$)
		\begin{prooftree*}
			\hypo{}
			\ellipsis{$\tderiv_{j}$}{\tyjp{}{\fctxtwop{\tmtwo}}{\typctx_{j}, \var \hastype \mtypethree_{j}}{\mtypetwo_{j}}}
			\infer1[\footnotesize$\ruleFun$]{\tyjp{}{\la{\var}{\fctxtwop{\tmtwo}}}{\typctx_{j}}{\mult{\ty{\mtypethree_{j}}{\mtypetwo_{j}}}}}
			\delims{ \left( }{ \right)_{1 \leq j \leq n} }
			\infer1[\footnotesize$\ruleMany$]{\tyjp{}{\la{\var}{\fctxtwop{\tmtwo}}}{\bigmplus_{j=1}^{n} \typctx_{j}}{\bigmplus_{j=1}^{n} \mult{\ty{\mtypethree_{j}}{\mtypetwo_{j}}}}}
		\end{prooftree*}
		
		By applying the \ih to $\tderiv_{j}$, for each $1 \leq j \leq n$, there exists $\namedtyjp{\tderivp_{j}}{}{\fctxtwop{\tmtwop}}{\typctx_{j}, \var \hastype \mtypethree_{j}}{\mtypetwo_{j}}$. 
		Thus, we can build a derivation $\tderivp$ as follows
		\begin{prooftree*}
			\hypo{}
			\ellipsis{$\tderivp_{j}$}{\tyjp{}{\fctxtwop{\tmtwop}}{\typctx_{j}, \var \hastype \mtypethree_{j}}{\mtypetwo_{j}}}
			\infer1[\footnotesize$\ruleFun$]{\tyjp{}{\la{\var}{\fctxtwop{\tmtwop}}}{\typctx_{j}}{\mult{\ty{\mtypethree_{j}}{\mtypetwo_{j}}}}}
			\delims{ \left(}{\right)_{1 \leq j \leq n} }
			\infer1[\footnotesize$\ruleMany$]{\tyjp{}{\la{\var}{\fctxtwop{\tmtwop}}}{\bigmplus_{j=1}^{n} \typctx_{j}}{\bigmplus_{j=1}^{n} \mult{\ty{\mtypethree_{j}}{\mtypetwo_{j}}}}}
		\end{prooftree*}
		
		\item \emph{\ES right}, \ie~$\fctx = \tmthree \esub{\var}{\fctxtwo}$. Then $\tm = \tmthree \esub{\var}{\fctxtwop{\tmtwo}} \,(\tovsub \!\cup \toglue \!\cup \eqstruct)\, \tmthree \esub{\var}{\fctxtwop{\tmtwop}} = \tmp$ with $\tmtwo \rootRew{} \tmtwop$, and $\tderiv$ is of the form
		\begin{prooftree*}
			\hypo{}
			\ellipsis{$\tderivtwo$}{\tyjp{}{\tmthree}{\typctxtwo, \var \hastype \mtypetwo}{\mtype}}
			\hypo{}
			\ellipsis{$\tderivthree$}{\tyjp{}{\fctxtwop{\tmtwo}}{\typctxthree}{\mtypetwo}}
			\infer2[\footnotesize$\ruleES$]{\tyjp{}{\tmthree \esub{\var}{\fctxtwop{\tmtwo}}}{\typctxtwo \mplus \typctxthree}{\mtype}}
		\end{prooftree*}
		
		By applying the \ih to $\tderivthree$, there exists $\namedtyjp{\tderivthreep}{}{\fctxtwop{\tmtwop}}{\typctxthree}{\mtypetwo}$. Thus, we can build a derivation $\tderivp$ as follows
		\begin{prooftree*}
			\hypo{}
			\ellipsis{$\tderivtwo$}{\tyjp{}{\tmthree}{\typctxtwo, \var \hastype \mtypetwo}{\mtype}}
			\hypo{}
			\ellipsis{$\tderivthreep$}{\tyjp{}{\fctxtwop{\tmtwop}}{\typctxthree}{\mtypetwo}}
			\infer2[\footnotesize$\ruleES$]{\tyjp{}{\tmthree \esub{\var}{\fctxtwop{\tmtwo}}}{\typctxtwo \mplus \typctxthree}{\mtype}}
		\end{prooftree*}
		
		\item \emph{$\ES$ left}, \ie~$\fctx = \fctxtwo \esub{\var}{\tmthree }$. Then $\tm = \fctxtwop{\tmtwo} \esub{\var}{\tmthree} \,(\tovsub \!\cup \toglue \!\cup \eqstruct)\, \fctxtwop{\tmtwo} \esub{\var}{\tmthree} = \tmp$ with $\tmtwo \rootRew{} \tmtwop$, and $\tderiv$ is of the form
		\begin{prooftree*}
			\hypo{}
			\ellipsis{$\tderivtwo$}{\tyjp{}{\fctxtwop{\tmtwo}}{\typctxtwo, \var \hastype \mtypetwo}{\mtype}}
			\hypo{}
			\ellipsis{$\tderivthree$}{\tyjp{}{\tmthree}{\typctxthree}{\mtypetwo}}
			\infer2[\footnotesize$\ruleES$]{\tyjp{}{\fctxtwop{\tmtwo} \esub{\var}{\tmthree}}{\typctxtwo \mplus \typctxthree}{\mtype}}
		\end{prooftree*}
		
		By applying the \ih to $\tderivtwo$, there exists $\namedtyjp{\tderivtwop}{}{\fctxtwop{\tmtwop}}{\typctxtwo, \var \hastype \mtypetwo}{\mtype}$. Thus, we can build a derivation $\tderivp$ as follows
		\begin{prooftree*}
			\hypo{}
			\ellipsis{$\tderivtwop$}{\tyjp{}{\fctxtwop{\tmtwop}}{\typctxtwo, \var \hastype \mtypetwo}{\mtype}}
			\hypo{}
			\ellipsis{$\tderivthree$}{\tyjp{}{\tmthree}{\typctxthree}{\mtypetwo}}
			\infer2[\footnotesize$\ruleES$]{\tyjp{}{\fctxtwop{\tmtwo} \esub{\var}{\tmthree}}{\typctxtwo \mplus \typctxthree}{\mtype}}
		\end{prooftree*}
	\end{itemize}
	
	\emph{Subject expansion.}
	We prove that if $\namedtyjp{\tderiv'}{}{\tm'}{\typctx}{\mtype}$ then there is a derivation $\namedtyjp{\tderiv}{}{\tm}{\typctx}{\mtype}$, by induction on the context $\fctx$ such that $\tm = \fctxp{\tmtwo} \,(\tovsub \!\cup \toglue \!\cup \eqstruct)\, \fctxp{\tmtwop} = \tmp$ with $\tmtwo \rootRew{} \tmtwop$, where $\rootRew{} \,\defeq\, \rtom \!\cup \rtoe \!\cup \rtoglue \!\cup \tostructcom \!\cup \tostructapr \!\cup \tostructapl \!\cup \tostructes$ (more precisely, for $\tostructapr$, $ \tostructapl$ and $\tostructes$, we take their symmetric closure). 
	Indeed, $\tovsub \!\cup \toglue \!\cup \eqstruct$ is the closure under full contexts of $\rootRew{}$.
	By symmetry of $\eqstruct$, the case when $\tm \eqstruct \tm'$ is already proved in the subject reduction (see above). 
	Therefore, we can only consider $\rootRew{} \,\defeq\, \rtom \!\cup \rtoe \!\cup \rtoglue$.
	\begin{itemize}
		\item \emph{Empty context}; \ie, $\fctx = \ctxhole$. There are three sub-cases.
		\begin{itemize}
			\item \emph{Multiplicative}, \ie $\tm = \subctxp{\la\var\tmtwo}\tmthree \rtom  \subctxp{\tmtwo \esub{\var}{\tmthree}} = \tm'$.
			Then $\tderiv'$ has the form:
			\begin{equation*}
			\tderiv' = 
			\begin{prooftree}
			\hypo{}
			\ellipsis{$\tderivtwo$}{\tyjp{}{\tmtwo}{\typctx' ; \var \hastype \mtypetwo}{\mtype}}
			\hypo{}
			\ellipsis{$\tderivthree$}{\tyjp{}{\tmthree}{\typctxtwo}{\mtypetwo}}
			\infer2[\footnotesize$\ruleES$]{\typctx'\uplus\typctxtwo \vdash\tmtwo\esub\var\tmthree \hastype \mtype}
			\hypo{}
			\ellipsis{$\tderiv_1$}{\quad}
			\infer2[\footnotesize$\ruleES$]{}
			\ellipsis{}{\quad}
			\hypo{}
			\ellipsis{$\tderiv_n$}{\quad}
			\infer2[\footnotesize$\Es$]{\typctx\uplus\typctxtwo \vdash \subctxp{\tmtwo\esub\var\tmthree} \hastype \mtype}
			\end{prooftree}
			\end{equation*}
			where $n \geq 0$ is the length of the list $\subctx$ of \ES.
			We can then build $\tderiv$ as follows:
			\begin{equation*}
			\tderiv = 
			\begin{prooftree}[separation = 1em]
			\hypo{}
			\ellipsis{$\tderivtwo$}{\typctx', \var \hastype \mtypetwo \vdash \tmtwo \hastype \mtype}
			\infer1[\footnotesize$\ruleFun$]{\typctx' \vdash \la\var\tmtwo \hastype \larrow{\mtypetwo}{\mtype}}
			\infer1[\footnotesize$\ruleManyVal$]{\typctx' \vdash \la\var\tmtwo \hastype \mset{\larrow{\mtypetwo}{\mtype}}}
			\hypo{}
			\ellipsis{$\tderiv_1$}{\quad}
			\infer2[\footnotesize$\ruleES$]{}
			\ellipsis{}{\quad}
			\hypo{}
			\ellipsis{$\tderiv_n$}{\quad}
			\infer2[\footnotesize$\ruleES$]{\typctx \vdash \subctxp{\la\var\tmtwo} \hastype \mset{\larrow{\mtypetwo}{\mtype}}}
			\hypo{}
			\ellipsis{$\tderivthree$}{\typctxtwo\vdash\tmthree \hastype \mtypetwo}
			\infer2[\footnotesize$\ruleApp$]{\typctx \uplus \typctxtwo \vdash \subctxp{\la\var\tmtwo}\tmthree \hastype \mtype}
			\end{prooftree}
			\end{equation*}
			
			\item \emph{Exponential}, \ie $\tm = \tmtwo\esub\var{\subctxp{\val}} \rtoe \subctxp{\tmtwo \isub{\var}{\val}} = \tmp$.
			Then the derivation $\tderiv'$ has the form (where $n \geq 0$ is the length of the list $\subctx$ of \ES):
			\begin{equation*}
			\tderiv' = 
			\begin{prooftree}
			\hypo{}
			\ellipsis{$\tderiv''$}{\typctxtwo \mplus\typctxthree' \vdash \tmtwo\isub\var \val \hastype \mtype}
			\hypo{}
			\ellipsis{$\tderiv_1$}{\quad}
			\infer2[\footnotesize{$\Es$}]{}
			\ellipsis{}{}
			\hypo{}
			\ellipsis{$\tderiv_n$}{\quad}
			\infer2[\footnotesize$\Es$]{\typctxtwo\mplus\typctxthree \vdash \subctxp{\tmtwo\isub\var \val} \hastype \mtype}
			\end{prooftree}
			\end{equation*}
			
			By the removal lemma (\Cref{l:anti-substitution}), there are derivations $\namedtyjp{\tderivtwo}{}{\tmtwo}{\typctxtwo, \var \hastype \mtypetwo}{\mtype}$ and $\namedtyjp{\tderivthree}{}{\val}{\typctxthree'}{\mtypetwo}$.
			We can then build the following derivation $\tderiv$:
			\begin{equation*}
			\tderiv = 
			\begin{prooftree}
			\hypo{}
			\ellipsis{$\tderivtwo$}{\tyjp{}{\tmtwo}{\typctxtwo, \var \hastype \mtypetwo}{\mtype}}
			\hypo{}
			\ellipsis{$\tderivthree$}{\tyjp{}{\val}{\typctxthree'}{\mtypetwo}}
			\hypo{}
			\ellipsis{$\tderiv_1$}{\quad}
			\infer2[\footnotesize$\Es$]{}
			\ellipsis{}{\quad}
			\hypo{}
			\ellipsis{$\tderiv_n$}{\quad}
			\infer2[\footnotesize$\Es$]{\typctxthree \vdash \subctxp{\val} \hastype \mtypetwo}
			\infer2[\footnotesize$\Es$]{\typctxtwo\uplus\typctxthree \vdash\tmtwo\esub\var{\subctxp{\val}}\hastype \mtype}
			\end{prooftree}
			\end{equation*}
			
			\item \emph{Glue}, \ie $\tm = \weakctxp\var\esub\var\aptm  \rtoglue    \weakctxp\aptm = \tmp$ where $\var \notin \fv{\weakctx}$ and $\aptm$ is an application.
			We proceed by induction on $\weakctx$.
			Given the derivation $\namedtyjp{\tderiv'}{}{\weakctxp\aptm}{\typctx}{\mtype}$, there is a derivation $\namedtyjp{\tderiv}{}{\weakctxp\var\esub\var\aptm}{\typctx}{\mtype}$ by linear removal (\Cref{l:linear-substitution-removal}).
		\end{itemize}

		\item \emph{Application right}, \ie~$\fctx = \tmthree \fctxtwo$. 
		Then $\tm = \tmthree \fctxtwop{\tmtwo} \,(\tovsub \!\cup \toglue \!\cup \eqstruct)\, \tmthree \fctxtwop{\tmtwop} = \tmp$ with $\tmtwo \rootRew{} \tmtwop$, and $\tderiv'$ is of the form
		\begin{prooftree*}
			\hypo{}
			\ellipsis{$\tderivtwo$}{\tyjp{}{\tmthree}{\typctxtwo}{\mult{\ty{\mtypetwo}{\mtype}}}}
			\hypo{}
			\ellipsis{$\tderivthreep$}{\tyjp{}{\fctxtwop{\tmtwop}}{\typctxthree}{\mtypetwo}}
			\infer2[\footnotesize$\ruleApp$]{\tyjp{}{\tmthree \fctxtwop{\tmtwop}}{\typctxtwo \mplus \typctxthree}{\mtype}}
		\end{prooftree*}
		
		By applying the \ih to $\tderivthree'$, there exists $\namedtyjp{\tderivthree}{}{\fctxtwop{\tmtwo}}{\typctxthree}{\mtypetwo}$. 
		Thus, we can build a  derivation $\tderiv$ as follows
		\begin{prooftree*}
			\hypo{}
			\ellipsis{$\tderivtwo$}{\tyjp{}{\tmthree}{\typctxtwo}{\mult{\ty{\mtypetwo}{\mtype}}}}
			\hypo{}
			\ellipsis{$\tderivthree$}{\tyjp{}{\fctxtwop{\tmtwo}}{\typctxthree}{\mtypetwo}}
			\infer2[\footnotesize$\ruleApp$]{\tyjp{}{\tmthree \fctxtwop{\tmtwo}}{\typctxtwo \mplus \typctxthree}{\mtype}}
		\end{prooftree*}
		
		\item \emph{Application left}, \ie~$\fctx = \fctxtwo \tmthree$. 
		Then $\tm = \fctxtwop{\tmtwo} \tmthree \,(\tovsub \!\cup \toglue \!\cup \eqstruct)\, \fctxtwop{\tmtwop} \tmthree = \tmp$ with $\tmtwo \rootRew{} \tmtwop$, and $\tderiv'$ is of the form
		\begin{prooftree*}
			\hypo{}
			\ellipsis{$\tderivtwop$}{\tyjp{}{\fctxtwop{\tmtwop}}{\typctxtwo}{\mult{\ty{\mtypetwo}{\mtype}}}}
			\hypo{}
			\ellipsis{$\tderivthree$}{\tyjp{}{\tmthree}{\typctxthree}{\mtypetwo}}
			\infer2[\footnotesize$\ruleApp$]{\tyjp{}{\fctxtwop{\tmtwop} \tmthree}{\typctxtwo \mplus \typctxthree}{\mtype}}
		\end{prooftree*}
		
		By applying the \ih to $\tderivtwop$, there exists $\namedtyjp{\tderivtwo}{}{\fctxtwop{\tmtwo}}{\typctxtwo}{\mult{\ty{\mtypetwo}{\mtype}}}$. 
		Thus, we can build a derivation $\tderivp$ as follows
		\begin{prooftree*}
			\hypo{}
			\ellipsis{$\tderivtwo$}{\tyjp{}{\fctxtwop{\tmtwo}}{\typctxtwo}{\mult{\ty{\mtypetwo}{\mtype}}}}
			\hypo{}
			\ellipsis{$\tderivthree$}{\tyjp{}{\tmthree}{\typctxthree}{\mtypetwo}}
			\infer2[\footnotesize$\ruleApp$]{\tyjp{}{\fctxtwop{\tmtwo} \tmthree}{\typctxtwo \mplus \typctxthree}{\mtype}}
		\end{prooftree*}

		\item \emph{Abstraction}, \ie~$\fctx = \la{\var}{\fctxtwo}$. 
		Then $\tm = \la{\var}{\fctxtwop{\tmtwo}} \,(\tovsub \!\cup \toglue \!\cup \eqstruct)\, \la{\var}{\fctxtwop{\tmtwop}} = \tmp$ with $\tmtwo \rootRew{} \tmtwop$, and $\tderivp$ is of the form (for some $n \geq 0$)
		\begin{prooftree*}
			\hypo{}
			\ellipsis{$\tderivp_{j}$}{\tyjp{}{\fctxtwop{\tmtwop}}{\typctx_{j}, \var \hastype \mtypethree_{j}}{\mtypetwo_{j}}}
			\infer1[\footnotesize$\ruleFun$]{\tyjp{}{\la{\var}{\fctxtwop{\tmtwop}}}{\typctx_{j}}{\mult{\ty{\mtypethree_{j}}{\mtypetwo_{j}}}}}
			\delims{ \left(}{\right)_{1 \leq j \leq n} }
			\infer1[\footnotesize$\ruleMany$]{\tyjp{}{\la{\var}{\fctxtwop{\tmtwop}}}{\bigmplus_{j=1}^{n} \typctx_{j}}{\bigmplus_{j=1}^{n} \mult{\ty{\mtypethree_{j}}{\mtypetwo_{j}}}}}
		\end{prooftree*}
		
		By applying the \ih to $\tderivp_{j}$, for each $1 \leq j \leq n$, there exists $\namedtyjp{\tderiv_{j}}{}{\fctxtwop{\tmtwo}}{\typctx_{j}, \var \hastype \mtypethree_{j}}{\mtypetwo_{j}}$. 
		Thus, we can build a derivation $\tderiv$ as follows
		\begin{prooftree*}
			\hypo{}
			\ellipsis{$\tderiv_{j}$}{\tyjp{}{\fctxtwop{\tmtwo}}{\typctx_{j}, \var \hastype \mtypethree_{j}}{\mtypetwo_{j}}}
			\infer1[\footnotesize$\ruleFun$]{\tyjp{}{\la{\var}{\fctxtwop{\tmtwo}}}{\typctx_{j}}{\mult{\ty{\mtypethree_{j}}{\mtypetwo_{j}}}}}
			\delims{ \left( }{ \right)_{1 \leq j \leq n} }
			\infer1[\footnotesize$\ruleMany$]{\tyjp{}{\la{\var}{\fctxtwop{\tmtwo}}}{\bigmplus_{j=1}^{n} \typctx_{j}}{\bigmplus_{j=1}^{n} \mult{\ty{\mtypethree_{j}}{\mtypetwo_{j}}}}}
		\end{prooftree*}
		
		\item \emph{\ES right}, \ie~$\fctx = \tmthree \esub{\var}{\fctxtwo}$. Then $\tm = \tmthree \esub{\var}{\fctxtwop{\tmtwo}} \,(\tovsub \!\cup \toglue \!\cup \eqstruct)\, \tmthree \esub{\var}{\fctxtwop{\tmtwop}} = \tmp$ with $\tmtwo \rootRew{} \tmtwop$, and $\tderivp$ is of the form
		\begin{prooftree*}
			\hypo{}
			\ellipsis{$\tderivtwo$}{\tyjp{}{\tmthree}{\typctxtwo, \var \hastype \mtypetwo}{\mtype}}
			\hypo{}
			\ellipsis{$\tderivthreep$}{\tyjp{}{\fctxtwop{\tmtwop}}{\typctxthree}{\mtypetwo}}
			\infer2[\footnotesize$\ruleES$]{\tyjp{}{\tmthree \esub{\var}{\fctxtwop{\tmtwo}}}{\typctxtwo \mplus \typctxthree}{\mtype}}
		\end{prooftree*}
		
		By applying the \ih to $\tderivthreep$, there exists $\namedtyjp{\tderivthree}{}{\fctxtwop{\tmtwo}}{\typctxthree}{\mtypetwo}$. 
		Thus, we can build a derivation $\tderiv$ as follows
		\begin{prooftree*}
			\hypo{}
			\ellipsis{$\tderivtwo$}{\tyjp{}{\tmthree}{\typctxtwo, \var \hastype \mtypetwo}{\mtype}}
			\hypo{}
			\ellipsis{$\tderivthree$}{\tyjp{}{\fctxtwop{\tmtwo}}{\typctxthree}{\mtypetwo}}
			\infer2[\footnotesize$\ruleES$]{\tyjp{}{\tmthree \esub{\var}{\fctxtwop{\tmtwo}}}{\typctxtwo \mplus \typctxthree}{\mtype}}
		\end{prooftree*}
		
		\item \emph{$\ES$ left}, \ie~$\fctx = \fctxtwo \esub{\var}{\tmthree }$. Then $\tm = \fctxtwop{\tmtwo} \esub{\var}{\tmthree} \,(\tovsub \!\cup \toglue \!\cup \eqstruct)\, \fctxtwop{\tmtwo} \esub{\var}{\tmthree} = \tmp$ with $\tmtwo \rootRew{} \tmtwop$, and $\tderivp$ is of the form
		\begin{prooftree*}
			\hypo{}
			\ellipsis{$\tderivtwop$}{\tyjp{}{\fctxtwop{\tmtwop}}{\typctxtwo, \var \hastype \mtypetwo}{\mtype}}
			\hypo{}
			\ellipsis{$\tderivthree$}{\tyjp{}{\tmthree}{\typctxthree}{\mtypetwo}}
			\infer2[\footnotesize$\ruleES$]{\tyjp{}{\fctxtwop{\tmtwo} \esub{\var}{\tmthree}}{\typctxtwo \mplus \typctxthree}{\mtype}}
		\end{prooftree*}
		
		By applying the \ih to $\tderivtwop$, there exists $\namedtyjp{\tderivtwo}{}{\fctxtwop{\tmtwo}}{\typctxtwo, \var \hastype \mtypetwo}{\mtype}$. Thus, we can build a derivation $\tderiv$ as follows
		\begin{prooftree*}
			\hypo{}
			\ellipsis{$\tderivtwo$}{\tyjp{}{\fctxtwop{\tmtwo}}{\typctxtwo, \var \hastype \mtypetwo}{\mtype}}
			\hypo{}
			\ellipsis{$\tderivthree$}{\tyjp{}{\tmthree}{\typctxthree}{\mtypetwo}}
			\infer2[\footnotesize$\ruleES$]{\tyjp{}{\fctxtwop{\tmtwo} \esub{\var}{\tmthree}}{\typctxtwo \mplus \typctxthree}{\mtype}}
		\end{prooftree*}
		\qedhere
	\end{itemize}
\end{proof}

\section{Proofs of Section \ref{sect:open} (Multi Types for Open \cbv)}

\begin{remark}[Merging and splitting inertness]
	\label{rmk:merge-split-inert}
	Let $\mtype, \mtypetwo, \mtypethree$ be multi types with $\mtype = \mtypetwo \mplus \mtypethree$; then, $\mtype$ is inert iff $\mtypetwo$ and $\mtypethree$ are inert.
	Similarly for type contexts.
\end{remark}

\begin{lemma}
	[Spreading of inertness on judgments]
	\label{lappendix:spread-inert}
	\NoteState{l:spread-inert}
	Let $\concl{\tderiv}{\typctx}{\itm}{\mtype}$ be a inert derivation and $\itm$ be an inert term. 
	Then, $\mtype$ is a inert multi type.
\end{lemma}

\begin{proof}
	By induction on the definition of inert terms $\itm$.
	Cases:
	\begin{itemize}
		\item \emph{Variable}, \ie $\itm = \var$. Then necessarily, for some $n \in \nat$,
		\begin{equation*}
		\tderiv = 
		\begin{prooftree}
		\infer0[\footnotesize$\ruleAx$]{\tyjp{}{\var}{\var \hastype \mset{\ltype_1}}{\ltype_1}}
		\hypo{\overset{n \in \nat}{\ldots}}
		\infer0[\footnotesize$\ruleAx$]{\tyjp{}{\var}{\var \hastype \mset{\ltype_n}}{\ltype_n}}
		\infer3[\footnotesize$\ruleManyVar$]{\tyjp{}{\var}{\typctx}{\mtype}}
		\end{prooftree}
		\end{equation*}
		where $\mtype = \mset{\ltype_1, \dots, \ltype_n}$ and $\typctx = \var \hastype \mtype$. 
		As $\typctx$ is inert (by hypothesis), so is $\mtype$.
		
		\item \emph{Application}, \ie $\itm = \itmtwo \fire$. Then necessarily
		\begin{equation*}
		\tderiv = 
		\begin{prooftree}
		\hypo{}
		\ellipsis{$\tderivtwo$}{\typctxtwo \vdash \itmtwo \hastype \mult{\ty{\mtypetwo}{\mtype}}}
		\hypo{}
		\ellipsis{$\tderivthree$}{\typctxthree \vdash \fire \hastype\mtypetwo}
		\infer2[\footnotesize$\ruleApp$]{\tyjp{}{\itmtwo \fire}{\typctxtwo \mplus \typctxthree}{\mtype}}
		\end{prooftree}
		\end{equation*}
		where $\typctx = \typctxtwo \mplus \typctxthree$.
		As $\typctx$ is inert (because $\tderiv$ is a inert derivation), $\typctxtwo$ is also inert, according to \Cref{rmk:merge-split-inert}, and thus $\tderivtwo$ is a inert derivation.
		By \ih applied to $\tderivtwo$, $\mult{\ty{\mtypetwo}{\mtype}}$ is inert and hence $\mtype$ is inert.
		
		\item \emph{Explicit substitution on inert}, \ie $\itm = \itmtwo \esub{\var}{\itmthree}$. Then necessarily
		\begin{equation*}
		\tderiv = 
		\begin{prooftree}
		\hypo{}
		\ellipsis{$\tderivtwo$}{\typctxtwo, \var \hastype \mtypetwo \vdash \itmtwo \hastype \mtype}
		\hypo{}
		\ellipsis{$\tderivthree$}{\typctxthree \vdash \itmthree \hastype \mtypetwo}
		\infer2[\footnotesize$\ruleES$]{\tyjp{}{\itmtwo \esub{\var}{\itmthree}}{\typctxtwo \mplus \typctxthree}{\mtype}}
		\end{prooftree}
		\end{equation*}
		where $\typctx = \typctxtwo \mplus \typctxthree$.
		As $\typctx$ is inert (because $\tderiv$ is a inert derivation), so are $\typctxtwo$ and $\typctxthree$, according to \Cref{rmk:merge-split-inert}, hence $\tderivthree$ is a inert derivation.
		By \ih applied to $\tderivthree$, $\mtypetwo$ is inert. 
		Therefore, $\typctxtwo, \var \hastype \mtypetwo$ is a inert type context and so $\tderivtwo$ is a inert derivation.
		By \ih applied to $\tderivtwo$, the multi type $\mtype$ is inert.
		\qedhere
	\end{itemize}
\end{proof}

\paragraph*{Correctness}

\begin{lemma}[Size of fireballs]
	\label{lappendix:size-fireballs}
	\NoteState{l:size-fireballs}
Let $\namedtyjp{\tderiv}{}{\tm}{\typctx}{\mtype}$.
%
\begin{enumerate}
	\item\label{pappendix:size-fireballs-inert} If $\tm = \itm$ is an inert term then $\sizem{\tderiv} \geq \sizeo{\itm}$. 
	If moreover $\tderiv$ is inert, then  $\sizem{\tderiv} = \sizeo{\itm}$.
	
	\item\label{pappendix:size-fireballs-tight} If $\tm =\fire$ is a fireball then $\sizem{\tderiv} \geq \sizeo{\fire}$. 
	If moreover $\tderiv$ is tight, then $\sizem{\tderiv} = \sizeo{\fire}$.
	
\end{enumerate}
\end{lemma}

\begin{proof}
	We prove simultaneously both points by mutual induction on the definition of fireballs $\fire$ and inert terms $\itm$.
	Cases for inert terms:
	\begin{itemize}
		\item \emph{Variable}, \ie $\tm = \var$. Then necessarily, for some $n \in \nat$,
		\begin{equation*}
		\tderiv = 
		\begin{prooftree}
		\infer0[\footnotesize$\ruleAx$]{\tyjp{}{\var}{\var \hastype \mset{\ltype_1}}{\ltype_1}}
		\hypo{\overset{n \in \nat}{\ldots}}
		\infer0[\footnotesize$\ruleAx$]{\tyjp{}{\var}{\var \hastype \mset{\ltype_n}}{\ltype_n}}
		\infer3[\footnotesize$\ruleManyVar$]{\tyjp{}{\var}{\typctx}{\mtype}}
		\end{prooftree}
		\end{equation*}
		where $\mtype = \mset{\ltype_1, \dots, \ltype_n}$ and $\typctx = \var \hastype \mtype$. 
		Therefore, $\sizeo{\tm} = 0 = \sizem{\tderiv}$.
		
		\item \emph{Application}, \ie $\tm = \itm \firetwo$. Then necessarily
		\begin{equation*}
		\tderiv = 
		\begin{prooftree}
		\hypo{}
		\ellipsis{$\tderivtwo$}{\typctxtwo \vdash \itm \hastype \mult{\ty{\mtypetwo}{\mtype}}}
		\hypo{}
		\ellipsis{$\tderivthree$}{\typctxthree \vdash \firetwo \hastype\mtypetwo}
		\infer2[\footnotesize$\ruleApp$]{\tyjp{}{\itm \firetwo}{\typctxtwo \mplus \typctxthree}{\mtype}}
		\end{prooftree}
		\end{equation*}
		where $\typctx = \typctxtwo \mplus \typctxthree$.
		By \ih applied to both premises, $\sizem{\tderivtwo} \geq \size{\itm}$ and $\sizem{\tderivthree} \geq 
\size{\firetwo}$.
		Therefore, $\sizeo{\tm} = \sizeo{\itm} + \sizeo{\firetwo} + 1 \leq \sizem{\tderivtwo} + \sizem{\tderivthree} + 1 = \sizem{\tderiv}$.
		
		If, moreover, $\tderiv$ is a inert derivation, then $\typctx$ is inert and hence so are $\typctxtwo$ and $\typctxthree$, according to \Cref{rmk:merge-split-inert}, thus $\tderivtwo$ and $\tderivthree$ are inert derivations.
		By \ih~for inert terms applied to $\tderivtwo$, $\sizem{\tderivtwo} = \sizeo{\itm}$.
		By spreading of inertness (\Cref{l:spread-inert}), $\mset{\larrow{\mtypetwo}{\mtype}}$ is inert, hence $\mtypetwo = \emptytype$ and so $\tderivthree$ is tight.
		By \ih~for fireballs applied to $\tderivthree$, $\sizem{\tderivthree} = \sizeo{\firetwo}$.
		Therefore, $\size{\tm} = \size{\itm} + \size{\firetwo} + 1 = \sizem{\tderivtwo} + \sizem{\tderivthree} + 1 = \sizem{\tderiv}$.
		
		\item \emph{Explicit substitution on inert}, \ie $\tm = \itm \esub{\var}{\itmtwo}$. Then necessarily
		\begin{equation*}
		\tderiv = 
		\begin{prooftree}
		\hypo{}
		\ellipsis{$\tderivtwo$}{\typctxtwo, \var \hastype \mtypetwo \vdash \itm \hastype \mtype}
		\hypo{}
		\ellipsis{$\tderivthree$}{\typctxthree \vdash \itmtwo \hastype \mtypetwo}
		\infer2[\footnotesize$\ruleES$]{\tyjp{}{\itm \esub{\var}{\itmtwo}}{\typctxtwo \mplus \typctxthree}{\mtype}}
		\end{prooftree}
		\end{equation*}
		where $\typctx = \typctxtwo \mplus \typctxthree$.
		We can then apply \ih to both premises: $\sizem{\tderivtwo} \geq \sizeo{\itm}$ and $\sizem{\tderivthree} \geq \sizeo{\itmtwo}$. 
		Therefore, $\sizeo{\tm} = \sizeo{\itm} + \sizeo{\itmtwo} \leq \sizem{\tderivtwo} + \sizem{\tderivthree} = \sizem{\tderiv}$.	
		
		If, moreover, $\tderiv$ is a inert derivation, then $\typctx$ is inert and hence so are $\typctxtwo$ and $\typctxthree$, according to \Cref{rmk:merge-split-inert}, thus $\tderivtwo$ and $\tderivthree$ are inert derivations.
		By spreading of inertness (\Cref{l:spread-inert}), $\mtypetwo$ is inert, hence $\typctxtwo, \var \hastype \mtypetwo$ is a inert type context.
		By \ih for inert terms applied to $\tderivtwo$ and $\tderivthree$, we have $\sizem{\tderivtwo} = \sizeo{\itm}$.
		and $\sizem{\tderivthree} = \sizeo{\itmtwo}$.
	 	So, $\sizeo{\tm} = \sizeo{\itm} + \sizeo{\itmtwo} = \sizem{\tderivtwo} + \sizem{\tderivthree} = \sizem{\tderiv}$.
		
	\end{itemize}

	Cases for fireballs that may not be inert terms:
	\begin{itemize}
		\item \emph{Explicit substitution on fireball}, \ie $\tm = \firetwo \esub{\var}{\itm}$. Then necessarily
		\begin{equation*}
		\tderiv = 
		\begin{prooftree}
		\hypo{}
		\ellipsis{$\tderivtwo$}{\typctxtwo, \var \hastype \mtypetwo \vdash \firetwo \hastype \mtype}
		\hypo{}
		\ellipsis{$\tderivthree$}{\typctxthree \vdash \itm \hastype \mtypetwo}
		\infer2[\footnotesize$\ruleES$]{\tyjp{}{\firetwo \esub{\var}{\itm}}{\typctxtwo \mplus \typctxthree}{\mtype}}
		\end{prooftree}
		\end{equation*}
		where $\typctx = \typctxtwo \mplus \typctxthree$.
		We can then apply \ih to both premises: $\sizem{\tderivtwo} \geq \sizeo{\firetwo}$ and $\sizem{\tderivthree} \geq \sizeo{\itm}$. 
		Therefore, $\sizeo{\tm} = \sizeo{\firetwo} + \sizeo{\itm} \leq \sizem{\tderivtwo} + \sizem{\tderivthree} = 
\sizem{\tderiv}$ 
		
		If, moreover, $\tderiv$ is tight, then $\typctx$ is inert and $\mtype$ is ground inert, so $\typctxtwo$ and $\typctxthree$ are inert, according to \Cref{rmk:merge-split-inert}, and thus $\tderivtwo$ and $\tderivthree$ are inert derivations.
		By \ih for inert terms applied to $\tderivthree$, $\sizem{\tderivthree} = \sizeo{\itm}$.
		By spreading of inertness (\Cref{l:spread-inert}), $\mtypetwo$ is inert, hence $\typctxtwo, \var \hastype \mtypetwo$ is a inert type context and so $\tderivtwo$ is tight.
		By \ih for fireballs applied to $\tderivtwo$, $\sizem{\tderivtwo} = \sizeo{\firetwo}$.
		Thus, $\sizeo{\tm} = \sizeo{\firetwo} + \sizeo{\itm} = \sizem{\tderivtwo} + \sizem{\tderivthree} = \sizem{\tderiv}$.
		
		\item \emph{Abstraction}, \ie $\tm = \la{\var}{\tmtwo}$. 
		Then necessarily, for some $n \in \nat$,
		\begin{equation*}
		\tderiv = 
		\begin{prooftree}[separation = 1em]
		\hypo{}
		\ellipsis{$\tderivtwo_1$}{\typctx_1, \var \hastype \mtypethree_1 \vdash \tmtwo \hastype \mtypetwo_1}
		\infer1[\footnotesize$\ruleFun$]{\tyjp{}{\la{\var}{\tmtwo}}{\typctx_1}{\ty{\mtypethree_1}{\mtypetwo_1}}}
		\hypo{\overset{n \in \nat}{\ldots}}
		\hypo{}
		\ellipsis{$\tderivtwo_n$}{\typctx_n, \var \hastype \mtypethree_n \vdash \tmtwo \hastype \mtypetwo_n}
		\infer1[\footnotesize$\ruleFun$]{\tyjp{}{\la{\var}{\tmtwo}}{\typctx_n}{\ty{\mtypethree_n}{\mtypetwo_n}}}
		\infer3[\footnotesize$\ruleManyVal$]{\tyjp{}{\la{\var}{\tmtwo}}{\typctx}{\mtype}}
		\end{prooftree}
		\end{equation*}
		where $\mtype = \bigmplus_{i=1}^n\mset{\larrow{\mtypethree_i}{\mtypetwo_i}}$ and $\typctx = 
\bigmplus_{i=1}^n\typctx_i$. 
		Thus, $\sizeo{\tm} = 0 \leq \sum_{i=1}^n(\sizem{\tderivtwo_i} + 1) = \sizem{\tderiv}$.
		
		If, moreover, $\tderiv$ is tight, then $\mtype$ is ground inert, so necessarily $\mtype = \emptytype$ and $n = 0$, hence $\tderiv$ consist of the rule $\ruleManyVal$ with $0$ premises.
		Therefore, $\sizeo{\tm}= 0 = \sizem{\tderiv}$. 
		\qedhere
	\end{itemize}
\end{proof}

\begin{proposition}[Open quantitative subject reduction]
	\label{propappendix:weak-subject-reduction}
	\NoteState{prop:weak-subject-reduction}
	Let $\namedtyjp{\tderiv}{}{\tm}{\typctx}{\mtype}$ be a derivation.
	\begin{enumerate}
		\item If $\tm \towm \tm'$ then there exists a derivation $\namedtyjp{\tderiv'}{}{\tm'}{\typctx}{\mtype}$ such that
		$\sizem{\tderiv'} = \sizem{\tderiv} - 2$ and $\size{\tderiv'} = \size{\tderiv} - 1$; 
		\item If $\tm \towe \tm'$ then there exists a derivation $\namedtyjp{\tderiv'}{}{\tm'}{\typctx}{\mtype}$ such that
		$\sizem{\tderiv'} = \sizem{\tderiv}$ and $\size{\tderiv'} < \size{\tderiv}$.
	\end{enumerate}
\end{proposition}

\begin{proof}
	By induction on the open context $\weakctx$ such that $\tm = \weakctxp{\tmtwo} \tow \weakctxp{\tmtwo'} = \tm'$ with $\tmtwo \rtom \tmtwo'$ or $\tmtwo' \rtoe \tmtwo'$. 
	Cases for $\weakctx$:
	\begin{itemize}
		\item \emph{Hole}, \ie $\weakctx = \ctxhole$.
		Then there are two sub-cases:
		\begin{enumerate}
			\item \emph{Multiplicative}, \ie $\tm = \subctxp{\la\var\tmtwo}\tmthree \rtom  \subctxp{\tmtwo \esub{\var}{\tmthree}} = \tm'$.
			Then $\tderiv$ has necessarily the form:
			\begin{equation*}
			\tderiv = 
			\begin{prooftree}[separation = 1em]
			\hypo{}
			\ellipsis{$\tderivtwo$}{\typctx', \var \hastype \mtypetwo \vdash \tmtwo \hastype \mtype}
			\infer1[\footnotesize$\ruleFun$]{\typctx' \vdash \la\var\tmtwo \hastype \larrow{\mtypetwo}{\mtype}}
			\infer1[\footnotesize$\ruleManyVal$]{\typctx' \vdash \la\var\tmtwo \hastype \mset{\larrow{\mtypetwo}{\mtype}}}
			\hypo{}
			\ellipsis{$\tderiv_1$}{\quad}
			\infer2[\footnotesize$\ruleES$]{}
			\ellipsis{}{\quad}
			\hypo{}
			\ellipsis{$\tderiv_n$}{\quad}
			\infer2[\footnotesize$\ruleES$]{\typctx \vdash \subctxp{\la\var\tmtwo} \hastype \mset{\larrow{\mtypetwo}{\mtype}}}
			\hypo{}
			\ellipsis{$\tderivthree$}{\typctxtwo\vdash\tmthree \hastype \mtypetwo}
			\infer2[\footnotesize$\ruleApp$]{\typctx \uplus \typctxtwo \vdash \subctxp{\la\var\tmtwo}\tmthree \hastype \mtype}
			\end{prooftree}
			\end{equation*}
			
			with $\sizem{\tderiv} = 2 + \sizem{\tderivtwo} + \sizem{\tderivthree} + \sum_{i=1}^{n} \sizem{\tderiv_{i}}$ and $\size{\tderiv} = 2 + n + \size{\tderivtwo} + \size{\tderivthree} + \sum_{i=1}^{n} \size{\tderiv_{i}}$.
			We can then build $\tderiv'$ as follows:
			\begin{equation*}
			\tderiv' = 
			\begin{prooftree}
			\hypo{}
			\ellipsis{$\tderivtwo$}{\tyjp{}{\tmtwo}{\typctx' ; \var \hastype \mtypetwo}{\mtype}}
			\hypo{}
			\ellipsis{$\tderivthree$}{\tyjp{}{\tmthree}{\typctxtwo}{\mtypetwo}}
			\infer2[\footnotesize$\ruleES$]{\typctx'\uplus\typctxtwo \vdash\tmtwo\esub\var\tmthree \hastype \mtype}
			\hypo{}
			\ellipsis{$\tderiv_1$}{\quad}
			\infer2[\footnotesize$\ruleES$]{}
			\ellipsis{}{\quad}
			\hypo{}
			\ellipsis{$\tderiv_n$}{\quad}
			\infer2[\footnotesize$\Es$]{\typctx\uplus\typctxtwo \vdash \subctxp{\tmtwo\esub\var\tmthree} \hastype \mtype}
			\end{prooftree}
			\end{equation*}
			where $\sizem{\tderiv'} = \sizem{\tderivtwo} + \sizem{\tderivthree} + \sum_{i=1}^{n} \sizem{\tderiv_{i}} = \sizem{\tderiv} - 2$ and $\size{\tderiv'} = 1+ n + \size{\tderivtwo} + \size{\tderivthree} + \sum_{i=1}^{n} \size{\tderiv_{i}} = \size{\tderiv} - 1$.
			
			\item \emph{Exponential}, \ie $\tm = \tmtwo\esub\var{\subctxp{\val}} \rtoe \subctxp{\tmtwo \isub{\var}{\val}} = \tmp$.
			Then the derivation $\tderiv$ has necessarily the form:
			\begin{equation*}
			\tderiv = 
			\begin{prooftree}
			\hypo{}
			\ellipsis{$\tderivtwo$}{\tyjp{}{\tmtwo}{\typctxtwo, \var \hastype \mtypetwo}{\mtype}}
			\hypo{}
			\ellipsis{$\tderivthree$}{\tyjp{}{\val}{\typctxthree'}{\mtypetwo}}
			\hypo{}
			\ellipsis{$\tderiv_1$}{\quad}
			\infer2[\footnotesize$\Es$]{}
			\ellipsis{}{\quad}
			\hypo{}
			\ellipsis{$\tderiv_n$}{\quad}
			\infer2[\footnotesize$\Es$]{\typctxthree \vdash \subctxp{\val} \hastype \mtypetwo}
			\infer2[\footnotesize$\Es$]{\typctxtwo\uplus\typctxthree \vdash\tmtwo\esub\var{\subctxp{\val}}\hastype \mtype}
			\end{prooftree}
			\end{equation*}
			where $\typctx = \typctxtwo \uplus \typctxthree$, $\sizem{\tderiv} = \sizem{\tderivtwo} + \sizem{\tderivthree} + \sum_{i=1}^{n} \sizem{\tderiv_{i}}$ and $\size{\tderiv} = 1 + n + \size{\tderivtwo} + \size{\tderivthree} + \sum_{i=1}^{n} \size{\tderiv_{i}}$.
			By the substitution lemma (\reflemma{substitution}), there is a derivation $\namedtyjp{\tderiv''}{}{\tmtwo\isub{\var}{\val}}{\typctxtwo \mplus \typctxthree'}{\mtype}$
			such that $\sizem{\tderiv''} = \sizem{\tderivtwo} + \sizem{\tderivthree}$ and $\size{\tderiv''} \leq \size{\tderivtwo} + \size{\tderivthree}$.
			We can then build the following derivation $\tderiv'$:
			\begin{equation*}
			\tderiv' = 
			\begin{prooftree}
			\hypo{}
			\ellipsis{$\tderiv''$}{\typctxtwo \mplus\typctxthree' \vdash \tmtwo\isub\var \val \hastype \mtype}
			\hypo{}
			\ellipsis{$\tderiv_1$}{\quad}
			\infer2[\footnotesize{$\Es$}]{}
			\ellipsis{}{}
			\hypo{}
			\ellipsis{$\tderiv_n$}{\quad}
			\infer2[\footnotesize$\Es$]{\typctxtwo\mplus\typctxthree \vdash \subctxp{\tmtwo\isub\var \val} \hastype \mtype}
			\end{prooftree}
			\end{equation*}
			where $\sizem{\tderiv'} = \sizem{\tderiv''} + \sum_{i=1}^{n} \sizem{\tderiv_{i}} = \sizem{\tderivtwo} + \sizem{\tderivthree} + \sum_{i=1}^{n} \sizem{\tderiv_{i}} = \sizem{\tderiv}$ and $\size{\tderiv'} = n + \size{\tderiv''} + \sum_{i=1}^{n} \size{\tderiv_{i}} \leq n + \size{\tderivtwo} + \size{\tderivthree} + \sum_{i=1}^{n} \size{\tderiv_{i}} < 1 + n + \size{\tderivtwo} + \size{\tderivthree} + \sum_{i=1}^{n} \size{\tderiv_{i}} = \size{\tderiv}$ ($\tderiv'$ contains at least one rule $\Es$ less than $\tderiv$).
		\end{enumerate}
		
		\item \emph{Application left}, \ie $\weakctx = \weakctxtwo\tmthree$.
		Then, $\tm = \weakctxp{\tmtwo} = \weakctxtwop{\tmtwo} \tmthree \rootRew{a} \weakctxtwop{\tmtwo'} \tmthree = \weakctxp{\tmtwo'} = \tm'$ with $\tmtwo \rootRew{a} \tmtwo'$ and $a \in \{\msym, \esym\}$.
		The derivation $\tderiv$ is necessarily
		\begin{equation*}
		\tderiv = 
		\begin{prooftree}
		\hypo{}
		\ellipsis{$\tderivtwo$}{\tyjp{}{\weakctxtwop{\tmtwo}}{\typctxtwo}{\mult{\ty{\mtypetwo}{\mtype}}}}
		\hypo{}
		\ellipsis{$\tderivthree$}{\tyjp{}{\tmthree}{\typctxthree}{\mtypetwo}}
		\infer2[\footnotesize$\ruleAp$]{\tyjp{}{\weakctxtwop{\tmtwo} \tmthree}{\typctxtwo \mplus \typctxthree}{\mtype}}
		\end{prooftree}
		\end{equation*}
		where $\typctx = \typctxtwo \uplus \typctxthree$, $\sizem{\tderiv} = 1 + \sizem{\tderivtwo} + \sizem{\tderivthree}$ and $\size{\tderiv} = 1 + \size{\tderivtwo} + \size{\tderivthree}$.
		By \ih, there is a derivation $\namedtyjp{\tderivtwo'}{}{\weakctxtwop{\tmtwo'}}{\typctxtwo}{\mult{\ty{\mtypetwo}{\mtype}}}$ with: 
		\begin{enumerate}
			\item $\sizem{\tderivtwo'} = \sizem{\tderivtwo} - 2$ and $\size{\tderivtwo'} = \size{\tderivtwo} - 1$ if $\tmtwo \rtom \tmtwo'$; 
			\item $\sizem{\tderivtwo'} = \sizem{\tderivtwo}$ and $\size{\tderivtwo'} < \size{\tderivtwo}$ if $\tmtwo \rtoe \tmtwo'$.
		\end{enumerate}
		We can then build the derivation 
		\begin{equation*}
		\tderiv' = 
		\begin{prooftree}
		\hypo{}
		\ellipsis{$\tderivtwo'$}{\tyjp{}{\weakctxtwop{\tmtwo'}}{\typctxtwo}{\mult{\ty{\mtypetwo}{\mtype}}}}
		\hypo{}
		\ellipsis{$\tderivthree$}{\tyjp{}{\tmthree}{\typctxthree}{\mtypetwo}}
		\infer2[\footnotesize$\ruleAp$]{\tyjp{}{\weakctxtwop{\tmtwo'} \tmthree}{\typctxtwo \mplus \typctxthree}{\mtype}}
		\end{prooftree}
		\end{equation*}
		noting that
		\begin{enumerate}
			\item If $\tmtwo \rtom \tmtwo'$ then $\sizem{\tderiv'} = 1 + \sizem{\tderivtwo'} + \sizem{\tderivthree} = 1 + (\sizem{\tderivtwo} - 2) + \sizem{\tderivthree} = \sizem{\tderiv} - 2$ and $\size{\tderiv'} = 1 + \size{\tderivtwo'} + \size{\tderivthree} = 1 + (\size{\tderivtwo} - 1) + \size{\tderivthree} = \size{\tderiv} - 1$; 
			\item If $\tmtwo \rtoe \tmtwo'$ then $\sizem{\tderiv'} = 1 + \sizem{\tderivtwo'} + \sizem{\tderivthree} = 1 + \sizem{\tderivtwo} + \sizem{\tderivthree} = \sizem{\tderiv}$ and $\size{\tderiv'} = 1 + \size{\tderivtwo'} + \size{\tderivthree} < 1 + \size{\tderivtwo} + \size{\tderivthree} = \size{\tderiv}$.
		\end{enumerate}
		
		\item \emph{Application right}, \ie $\weakctx = \tmthree \weakctxtwo$.
		Analogous to the previous case.
		
		\item \emph{Explicit substitution left}, \ie $\weakctx = \weakctxtwo\esub{\var}{\tmthree}$. 
		Then, $\tm = \weakctxp{\tmtwo} = \weakctxtwop{\tmtwo} \esub{\var}{\tmthree} \rootRew{a} \weakctxtwop{\tmtwo'}\esub{\var}{\tmthree} = \weakctxp{\tmtwo'} = \tm'$ with $\tmtwo \rootRew{a} \tmtwo'$ and $a \in \{\msym, \esym\}$.
		The derivation $\tderiv$ is necessarily
		\begin{equation*}
		\tderiv = 
		\begin{prooftree}
		\hypo{}
		\ellipsis{$\tderivtwo$}{\tyjp{}{\weakctxtwop{\tmtwo}}{\typctxtwo ; \var \hastype \mtypetwo}{\mtype}}
		\hypo{}
		\ellipsis{$\tderivthree$}{\tyjp{}{\tmthree}{\typctxthree}{\mtypetwo}}
		\infer2[\footnotesize$\Es$]{\tyjp{}{\weakctxtwop{\tmtwo} \esub{\var}{\tmthree}}{\typctxtwo \mplus \typctxthree}{\mtype}}
		\end{prooftree}
		\end{equation*}
		where $\typctx = \typctxtwo \uplus \typctxthree$, $\sizem{\tderiv} = \sizem{\tderivtwo} + \sizem{\tderivthree}$ and $\size{\tderiv} = 1 + \size{\tderivtwo} + \size{\tderivthree}$.
		By \ih, there is a derivation $\namedtyjp{\tderivtwo'}{}{\weakctxtwop{\tmtwo'}}{\typctxtwo, \var \hastype \mtypetwo}{\mtype}$ with: 
		\begin{enumerate}
			\item $\sizem{\tderivtwo'} = \sizem{\tderivtwo} - 2$ and $\size{\tderivtwo'} = \size{\tderivtwo} - 1$ if $\tmtwo \rtom \tmtwo'$; 
			\item $\sizem{\tderivtwo'} = \sizem{\tderivtwo}$ and $\size{\tderivtwo'} < \size{\tderivtwo}$ if $\tmtwo \rtoe \tmtwo'$.
		\end{enumerate}
		We can then build the derivation 
		\begin{equation*}
		\tderiv' = 
		\begin{prooftree}
		\hypo{}
		\ellipsis{$\tderivtwo'$}{\tyjp{}{\weakctxtwop{\tmtwo'}}{\typctxtwo ; \var \hastype \mtypetwo}{\mtype}}
		\hypo{}
		\ellipsis{$\tderivthree$}{\tyjp{}{\tmthree}{\typctxthree}{\mtypetwo}}
		\infer2[\footnotesize$\Es$]{\typctxtwo \uplus \typctxthree \vdash \weakctxtwop{\tmtwo'} \esub{\var}{\tmthree} \hastype \mtype}
		\end{prooftree}
		\end{equation*}
		noting that 
		\begin{enumerate}
			\item If $\tmtwo \rtom \tmtwo'$ then $\sizem{\tderiv'} = \sizem{\tderivtwo'} + \sizem{\tderivthree} = (\sizem{\tderivtwo} - 2) + \sizem{\tderivthree} = \sizem{\tderiv} - 2$ and $\size{\tderiv'} = \size{\tderivtwo'} + \size{\tderivthree} = (\size{\tderivtwo} - 1) + \size{\tderivthree} = \size{\tderiv} - 1$; 
			\item If $\tmtwo \rtoe \tmtwo'$ then $\sizem{\tderiv'} = \sizem{\tderivtwo'} + \sizem{\tderivthree} = \sizem{\tderivtwo} + \sizem{\tderivthree} = \sizem{\tderiv}$ and $\size{\tderiv'} = \size{\tderivtwo'} + \size{\tderivthree} < \size{\tderivtwo} + \size{\tderivthree} = \size{\tderiv}$.
		\end{enumerate}
		
		\item \emph{Explicit substitution right}, \ie $\weakctx = \tmthree \esub{\var}{\weakctxtwo}$. 
		Analogous to the previous case.
		\qedhere
	\end{itemize}
\end{proof}

\paragraph*{Completeness}

\begin{proposition}[Tight typability of open normal forms]
	\label{propappendix:precise-open-typability-nf}
	\NoteState{prop:precise-open-typability-nf}
	\begin{enumerate}
		\item \emph{Inert:}\label{pappendix:precise-open-typability-nf-inert} 
		if $\tm$ is an inert term then, for any multi type $\mtype$, there is a type context $\typctx$ and a derivation $\concl{\tderiv}{\typctx}{\tm}{\mtype}$; if,  moreover, $\mtype$ is inert then $\tderiv$ is inert.
		
		\item \emph{Fireball:}\label{pappendix:precise-open-typability-nf-fireball} if $\tm$ is a fireball then there is a tight derivation $\concl{\tderiv}{\typctx}{\tm}{\emptytype}$.		
	\end{enumerate}
\end{proposition}

\begin{proof}
	We prove simultaneously \Cref{p:precise-open-typability-nf-fireball,p:precise-open-typability-nf-inert} by 
		mutual induction on the definition of fireball and inert term.
		Note that the ``moreover'' part of \Cref{p:precise-open-typability-nf-inert} is stronger than \Cref{p:precise-open-typability-nf-fireball} and amounts to proving that $\typctx$ is inert.
		Cases for inert terms:
		\begin{itemize}
			\item \emph{Variable}, \ie $\tm = \var$, which is an inert term. 
			Let $\mtype$ be a multi type: hence, $\mtype = \mset{\ltype_1, \dots, \ltype_n}$ for some $n \in \nat$ and some $\ltype_1, \dots, \ltype_n$ linear (\resp inert linear) types.
			We can then build  the derivation (with $\Gamma = \var \hastype \mtype$)
			\begin{equation*}
			\tderiv = 
			\begin{prooftree}[separation = 1em]
			\infer0[\footnotesize$\ruleAx$]{\var \hastype \mset{\ltype_1} \vdash \var \hastype \ltype_1}
			\hypo{\dots}
			\infer0[\footnotesize$\ruleAx$]{\var \hastype \mset{\ltype_n} \vdash \var \hastype \ltype_n}
			\infer3[\footnotesize$\ruleManyVar$]{\var \hastype \mset{\ltype_1, \dots, \ltype_n} \vdash \var \hastype \mset{\ltype_1, \dots, \ltype_n}}
			\end{prooftree}
			\end{equation*}
			If, moreover, $\mtype$ is an inert multi type, then $\Gamma = \var \hastype \mtype$ is an inert type context.

			\item \emph{Inert application}, \ie $\tm = \itm \fire$ for some inert term $\itm$ and fireball $\fire$.
			Let $\mtype$ be a multi (\resp an inert multi) type.
			By \ih for fireballs, there is a derivation $\concl{\tderivthree}{\typctxthree}{\fire}{\emptytype}$  for some inert type context $\typctxthree$.
			By \ih for inert terms, since $\mset{\larrow{\emptytype}{\mtype}}$ is a multi (\resp an inert multi) type, there is a derivation $\concl{\tderivtwo}{\typctxtwo}{\itm}{\mset{\larrow{\emptytype}{\mtype}}}$ for some type (\resp inert type) context $\typctxtwo$.
			We can then build the derivation 
			\begin{equation*}
			\tderiv  =
			\begin{prooftree}
			\hypo{}
			\ellipsis{$\tderivtwo$}{\typctxtwo \vdash \itm \hastype \mset{\larrow{\emptytype}{\mtype}}}
			\hypo{}
			\ellipsis{$\tderivthree$}{\typctxthree \vdash \fire \hastype \emptytype}
			\infer2[\footnotesize$\ruleApp$]{\typctxtwo \mplus \typctxthree \vdash \itm \fire \hastype \mtype}				
			\end{prooftree}
			\end{equation*}
			where $\typctx = \typctxtwo \mplus \typctxthree$ (\resp where $\typctx = \typctxtwo \mplus \typctxthree$ is an inert type context, by \Cref{rmk:merge-split-inert}).
			
			\item \emph{Explicit substitution on inert}, \ie $\tm = \itm \esub{\var}{\itmtwo}$ for some inert terms $\itm$ and $\itmtwo$.
			Let $\mtype$ be a multi (\resp an inert multi) type.
			By \ih for inert terms applied to $\itm$, there is a derivation $\concl{\tderivtwo}{\typctxtwo, \var \hastype \mtypetwo}{\itm}{\mtype}$ for some multi (\resp inert multi) type $\mtypetwo$ and type (\resp inert type) context $\typctxtwo$.
			By \ih for inert terms applied to $\itmtwo$, there is a derivation $\concl{\tderivthree}{\typctxthree}{\itmtwo}{\mtypetwo}$  for some type (\resp inert type) context $\typctxthree$.
			We can then build the derivation 
			\begin{equation*}
				\tderiv =
				\begin{prooftree}
				\hypo{}
				\ellipsis{$\tderivtwo$}{\typctxtwo, \var \hastype \mtypetwo \vdash \itm \hastype \mtype}
				\hypo{}
				\ellipsis{$\tderivthree$}{\typctxthree \vdash \itmtwo \hastype \mtypetwo}
				\infer2[\footnotesize$\ruleES$]{\typctxtwo \mplus \typctxthree \vdash \itm \esub{\var} {\itmtwo} \hastype \mtype}				
				\end{prooftree}
			\end{equation*}
			where $\typctx = \typctxtwo \mplus \typctxthree$ (\resp where $\typctx = \typctxtwo \mplus \typctxthree$ is an inert type context, by \Cref{rmk:merge-split-inert}).
		\end{itemize}
	
		Cases for fireballs that may not be inert terms:
		\begin{itemize}
			\item \emph{Abstraction}, \ie $\tm  = \la{\var}{\tmtwo}$.
			We can then build the derivation
			\begin{equation*}
				\tderiv = 
				\begin{prooftree}
				\infer0[\footnotesize$\ruleManyVal$]{\vdash \la{\var}{\tmtwo} \hastype \emptytype}
			\end{prooftree}
			\end{equation*}
			where the type context $\typctx$ is empty and hence inert, thus $\tderiv$ is tight.
			
			\item \emph{Explicit substitution on fireball}, \ie $\tm = \fire \esub{\var}{\itm}$ for some fireball $\fire$ and inert term $\itm$.
			By \ih for fireballs applied to $\fire$, there is a derivation $\concl{\tderivtwo}{\typctxtwo, \var \hastype \mtypetwo}{\fire}{\emptytype}$ for some inert multi type $\mtypetwo$ and inert type context $\typctxtwo$.
			By \ih for inert terms applied to $\itm$, there is a derivation $\concl{\tderivthree}{\typctxthree}{\itm}{\mtypetwo}$  for some inert type context $\typctxthree$.
			We can build the derivation 
			\begin{equation*}
			\tderiv = 
			\begin{prooftree}
			\hypo{}
			\ellipsis{$\tderivtwo$}{\typctxtwo, \var \hastype \mtypetwo \vdash \fire \hastype \emptytype}
			\hypo{}
			\ellipsis{$\tderivthree$}{\typctxthree \vdash \itm \hastype \mtypetwo}
			\infer2[\footnotesize$\ruleES$]{\typctxtwo \mplus \typctxthree \vdash \fire \esub{\var} {\itm} \hastype \emptytype}				
			\end{prooftree}
			\end{equation*}
			where $\typctx = \typctxtwo \mplus \typctxthree$ is an inert type context, by \Cref{rmk:merge-split-inert}.
			Therefore, $\tderiv$ is tight.
			\qedhere
		\end{itemize}
\end{proof}

\section{Proofs of Section \ref{sect:solvable} (Multi Types for \cbv Solvability)}

\paragraph*{Correctness}

\begin{lemma}[Size of solved fireballs]
	\label{lappendix:size-solvable-nf}
	\NoteState{l:size-solvable-nf}
	Let $\solvnf$ be a solved fireball and $\namedtyjp{\tderiv}{}{\solvnf}{\typctx}{\mtype}$.
	\begin{enumerate}
		\item\label{pappendix:size-solvable-nf-bound} \emph{Bounds}: if $\mtype$ is solvable then $\sizem{\tderiv} \geq \sizes{\solvnf}$.
		\item\label{pappendix:size-solvable-nf-exact} \emph{Exact bounds}: if $\typctx$ is inert and $\mtype$ is precisely solvable then $\sizem{\tderiv} = \sizes{\solvnf}$.
	\end{enumerate}
\end{lemma}

\begin{proof}
	Both \Cref{p:size-solvable-nf-bound} and \Cref{p:size-solvable-nf-exact} are proved by induction on the definition of solved fireball $\solvnf$.
	For each case, we shall first prove \Cref{p:size-solvable-nf-bound}, and then we shall prove \Cref{p:size-solvable-nf-exact}.
	Cases:
	
	\begin{itemize}
		\item \emph{Inert}, \ie $\solvnf$ is an inert term.
		According to \Cref{l:size-fireballs}, $\sizem{\tderiv} \geq \sizes{\solvnf}$.
		If, moreover, $\typctx$ is inert then $\sizem{\tderiv} = \sizeo{\solvnf} = \sizes{\solvnf}$ by \Cref{l:size-fireballs} and \Cref{l:sizes-inert}.		
		Note that in this case the hypothesis that $\mtype$ is a (precisely) solvable multi type is not used.
		
		\item \emph{Explicit substitution on solvable normal form}, \ie $\solvnf = \solvnftwo \esub{\var}{\itm}$ for some solved fireball $\solvnftwo$ and inert term $\itm$. 
		Then necessarily
		\begin{equation*}
		\tderiv = 
		\begin{prooftree}
		\hypo{}
		\ellipsis{$\tderivtwo$}{\typctxtwo, \var \hastype \mtypetwo \vdash \solvnftwo \hastype \mtype}
		\hypo{}
		\ellipsis{$\tderivthree$}{\typctxthree \vdash \itm \hastype \mtypetwo}
		\infer2[\footnotesize$\ruleES$]{\tyjp{}{\solvnftwo \esub{\var}{\itm}}{\typctxtwo \mplus \typctxthree}{\mtype}}
		\end{prooftree}
		\end{equation*}
		where $\typctx = \typctxtwo \mplus \typctxthree$.
		According to \Cref{l:size-fireballs} for inert terms, $\sizem{\tderivthree} \geq \sizeo{\itm}$.
		By \ih applied to $\tderivtwo$, $\sizem{\tderivtwo} \geq \sizes{\solvnftwo}$. 
		Therefore, $\sizes{\solvnf} = \sizes{\solvnftwo} + \sizeo{\itm} \leq \sizem{\tderivtwo} + \sizem{\tderivthree} = 
		\sizem{\tderiv}$ 
		
		If, moreover, $\typctx$ is inert and $\mtype$ is precisely solvable, then $\typctxtwo$ and $\typctxthree$ are inert, according to \Cref{rmk:merge-split-inert}.
		By \Cref{l:size-fireballs} for inert terms, $\sizem{\tderivthree} = \sizeo{\itm}$.
		By spreading of inertness (\Cref{l:spread-inert}), $\mtypetwo$ is inert and hence $\typctxtwo, \var \hastype \mtypetwo$ is a inert type context.
		By \ih applied to $\tderivtwo$, $\sizem{\tderivtwo} = \sizes{\solvnftwo}$.
		Therefore, $\sizes{\solvnf} = \sizes{\solvnftwo} + \sizeo{\itm} = \sizem{\tderivtwo} + \sizem{\tderivthree} = 
		\sizem{\tderiv}$.
		
		\item \emph{Abstraction}, \ie $\solvnf = \la{\var}{\solvnftwo}$ for some solved fireball $\solvnftwo$. 
		Then $\mtype = \bigmplus_{i=1}^n\mset{\larrow{\mtypethree_i}{\mtypetwo_i}}$ for some $n > 0$ (as $\mtype$ is solvable), and
		\begin{equation*}
		\tderiv = 
		\begin{prooftree}[separation = 1em]
		\hypo{}
		\ellipsis{$\tderivtwo_1$}{\typctx_1, \var \hastype \mtypethree_1 \vdash \solvnftwo \hastype \mtypetwo_1}
		\infer1[\footnotesize$\ruleFun$]{\tyjp{}{\la{\var}{\solvnftwo}}{\typctx_1}{\ty{\mtypethree_1}{\mtypetwo_1}}}
		\hypo{\overset{n \in \nat}{\ldots}}
		\hypo{}
		\ellipsis{$\tderivtwo_n$}{\typctx_n, \var \hastype \mtypethree_n \vdash \solvnftwo \hastype \mtypetwo_n}
		\infer1[\footnotesize$\ruleFun$]{\tyjp{}{\la{\var}{\solvnftwo}}{\typctx_n}{\larrow{\mtypethree_n}{\mtypetwo_n}}}
		\infer3[\footnotesize$\ruleManyVal$]{\tyjp{}{\la{\var}{\solvnftwo}}{\typctx}{\mtype}}
		\end{prooftree}
		\end{equation*}
		where $\typctx = \bigmplus_{i=1}^n\typctx_i$. 
		By \ih, $\sizes{\solvnftwo} \leq \sizem{\tderivtwo_i}$ for all $1 \leq i \leq n$.
		Therefore, $\sizes{\solvnf} =  1 + \sizes{\solvnftwo} = 1 + \sum_{i=1}^n\sizem{\tderivtwo_i} \leq \sum_{i=1}^n(\sizem{\tderivtwo_i} + 1) = \sizem{\tderiv}$, where the inequality holds because $n > 0$.
		
		If, moreover, $\typctx$ is inert and $\mtype$ is precisely solvable, then $n = 1$ and $\typctx = \typctx_1$ and $\mtype = \mset{\larrow{\mtypethree_1}{\mtypetwo_1}}$ with $\mtypethree_1$ inert. 
		Thus, $\typctx_1, \var \hastype \mtypethree_1$ is an inert type context.
		By \ih, $\sizem{\tderivtwo_1} = \sizes{\solvnftwo}$.
		So, $\sizes{\solvnf}= \sizes{\solvnftwo} + 1 = \sizem{\tderivtwo_1} + 1 = \sizem{\tderiv}$. 
		\qedhere
	\end{itemize}
\end{proof}

\begin{proposition}[Solvable quantitative subject reduction]
	\label{propappendix:solvable-subject-reduction}
	\NoteState{prop:solvable-subject-reduction}
	Let $\concl{\tderiv}{\typctx}{\tm}{\mtype}$ with $\mtype$ solvable.
	\begin{enumerate}
		\item \emph{Multiplicative step:} if $\tm \tosolvm \tm'$ then there is a derivation 
$\concl{\tderiv'}{\typctx}{\tm'}{\mtype}$ such that $\sizem{\tderiv'} \leq \sizem{\tderiv}-2$ and $\size{\tderiv'} < 
\size{\tderiv}$. If moreover  $\mtype$ is unitary solvable then 
$\sizem{\tderiv'} = \sizem{\tderiv}-2$ and $\size{\tderiv'} = \size{\tderiv}-1$.
		
		\item \emph{Exponential step:} if $\tm \tosolve \tm'$ then there is a derivation 
$\concl{\tderiv'}{\typctx}{\tm'}{\mtype}$ such that
		$\sizem{\tderiv'} = \sizem{\tderiv}$ and $\size{\tderiv'} < \size{\tderiv}$.
	\end{enumerate}
\end{proposition}

\begin{proof}
	First, we prove the unitary solvable version of the claim (\ie under the hypothesis that $\mtype$ is unitary solvable).
	The proof is by induction on the solvable context $\solvctx$ such that $\tm = \solvctxp{\tmtwo} \tosolv \solvctxp{\tmtwo'} = \tm'$ with $\tmtwo \tomo \tmtwo'$ or $\tmtwo' \toeo \tmtwo'$. 
	Cases for $\solvctx$:
	\begin{itemize}
		\item \emph{Hole context}, \ie{} $\solvctx = \ctxhole$ and $\tm \Rew{\wsym a}  \tm'$ with $a \in \{\msym, \esym\}$.
		According to subject reduction for $\tow$ (\Cref{prop:weak-subject-reduction}),
		\begin{itemize}
			\item if $\tm \towm \tm'$ then there exists a derivation $\concl{\tderiv'}{\typctx}{\tm'}{\mtype}$ such that
			$\sizem{\tderiv'} = \sizem{\tderiv}-2$ and $\size{\tderiv'} = \size{\tderiv}-1$;
			\item if $\tm \towe \tm'$ then there is a derivation $\concl{\tderiv'}{\typctx}{\tm'}{\mtype}$ such that 
$\sizem{\tderiv'} = \sizem{\tderiv}$ and
			$\size{\tderiv'} < \size{\tderiv}$.
		\end{itemize}
		Note that in this case the hypothesis that $\mtype$ is a (unitary) solvable multi type is not used.
		
		\item \emph{Abstraction}, \ie $\solvctx = \la{\var}{\solvctxtwo}$. 
		So, $\tm = \solvctxp{\tmtwo} = \la{\var}{\solvctxtwop{\tmtwo}} \Rew{\solvredsym a} \la{\var}{\solvctxtwop{\tmtwo'}} = 
\solvctxp{\tmtwo'} = \tm'$ with $\tmtwo \Rew{\wsym a} \tmtwo'$ and $a \in \{\msym, \esym\}$.
		Since $\mtype$ is a unitary solvable multi type by hypothesis and hence it has the form $\mtype = \mset{\larrow{\mtypethree}{\smtypetwo}}$ where $\smtypetwo$ is unitary solvable.
		Thus, the derivation $\tderiv$ is necessarily
		\begin{equation*}
		\tderiv = 
		\begin{prooftree}
		\hypo{}
		\ellipsis{$\tderivtwo$}{\typctx, \var \hastype \mtypethree \vdash \solvctxtwop{\tmtwo} \hastype \smtypetwo}
		\infer1[\footnotesize$\lambda$]{\typctx \vdash \la{\var}\solvctxtwop{\tmtwo} \hastype \larrow{\mtypethree}{\smtypetwo}}
		\infer1[\footnotesize$\ruleManyVal$]{\typctx \vdash \la{\var}\solvctxtwop{\tmtwo} \hastype 
\mset{\larrow{\mtypethree}{\smtypetwo}}}
		\end{prooftree}
		\end{equation*}
		By \ih, there is a derivation $\concl{\tderivtwo'}{\typctx, \var \hastype 
\mtypetwo}{\solvctxtwop{\tmtwo'}}{\smtypetwo}$ with: 
		\begin{enumerate}
			\item $\sizem{\tderivtwo'} = \sizem{\tderivtwo} - 2$ and $\size{\tderivtwo'} = \size{\tderivtwo}-1$ if $\tmtwo 
\tomo \tmtwo'$ ; 
			\item $\sizem{\tderivtwo'} = \sizem{\tderivtwo}$ and $\size{\tderivtwo'} < \size{\tderivtwo}$ if $\tmtwo \toeo 
\tmtwo'$.
		\end{enumerate}
		We can then build the derivation 
		\begin{equation*}
		\tderiv' = 
		\begin{prooftree}
		\hypo{}
		\ellipsis{$\tderivtwo'$}{\typctx, \var \hastype \mtypethree \vdash \solvctxtwop{\tmtwo'} \hastype \smtypetwo}
		\infer1[\footnotesize$\lambda$]{\typctx \vdash \la{\var}\solvctxtwop{\tmtwo'} \hastype \larrow{\mtypethree}{\smtypetwo}}
		\infer1[\footnotesize$\ruleManyVal$]{\typctx \vdash \la{\var}\solvctxtwop{\tmtwo'} \hastype 
\mset{\larrow{\mtypethree}{\smtypetwo}}}
		\end{prooftree}
		\end{equation*}
		where
		\begin{enumerate}
			\item $\sizem{\tderiv'} = \sizem{\tderivtwo'} +1 = \sizem{\tderivtwo} + 1 - 2 = \sizem{\tderiv} - 2$ and 
$\size{\tderiv'} = \size{\tderivtwo'} +1 = \size{\tderivtwo} + 1 - 1 = \size{\tderiv} - 1$ if $\tmtwo \tomo \tmtwo'$ (\ie $\tm \tosolvm \tm'$); 
			\item $\sizem{\tderiv'} = \sizem{\tderivtwo'} +1 = \sizem{\tderivtwo} + 1 = \sizem{\tderiv}$ and $\size{\tderiv'} 
= \size{\tderivtwo'} +1 < \size{\tderivtwo} + 1 = \size{\tderiv}$ if $\tmtwo \toeo \tmtwo'$ (\ie $\tm \tosolve \tm'$).
		\end{enumerate}

		\item \emph{Explicit substitution}, \ie $\solvctx = \solvctxtwo\esub{\var}{\tmthree}$. 
		Then, $\tm = \solvctxp{\tmtwo} = \solvctxtwop{\tmtwo} \esub{\var}{\tmthree} \Rew{\solvredsym a} 
\solvctxtwop{\tmtwo'}\esub{\var}{\tmthree} = \solvctxp{\tmtwo'} = \tm'$ with $\tmtwo \Rew{\wsym a} \tmtwo'$ and $a \in 
\{\msym, \esym\}$.
		The derivation $\tderiv$ is necessarily
		\begin{equation*}
		\tderiv = 
		\begin{prooftree}
		\hypo{}
		\ellipsis{$\tderivtwo$}{\typctxtwo, \var \hastype \mtypetwo \vdash \solvctxtwop{\tmtwo} \hastype \mtype}
		\hypo{}
		\ellipsis{$\tderivthree$}{\typctxthree \vdash \tmthree \hastype \mtypetwo}
		\infer2[\footnotesize$\Es$]{\typctxtwo \uplus \typctxthree \vdash \solvctxtwop{\tmtwo} \esub{\var}{\tmthree} \hastype \mtype}
		\end{prooftree}
		\end{equation*}
		where $\typctx = \typctxtwo \uplus \typctxthree$.
		By \ih applied to $\tderivtwo$, there is a derivation $\concl{\tderivtwo'}{\typctxtwo, \var \hastype \mtypetwo}{\solvctxtwop{\tmtwo'}}{\mtype}$ with: 
		\begin{enumerate}
			\item $\sizem{\tderivtwo'} = \sizem{\tderivtwo} - 2$ and $\size{\tderivtwo'} = \size{\tderivtwo} - 1$ if $\tmtwo 
\tomo \tmtwo'$; 
			\item $\sizem{\tderivtwo'} = \sizem{\tderivtwo}$ and $\size{\tderivtwo'} < \size{\tderivtwo}$ if $\tmtwo \toeo \tmtwo'$.
		\end{enumerate}
		We can then build the derivation 
		\begin{equation*}
		\tderiv' = 
		\begin{prooftree}
		\hypo{}
		\ellipsis{$\tderivtwo'$}{\typctxtwo, \var \hastype \mtypetwo \vdash \solvctxtwop{\tmtwo'} \hastype \mtype}
		\hypo{}
		\ellipsis{$\tderivthree$}{\typctxthree \vdash \tmthree \hastype \mtypetwo}
		\infer2[\footnotesize$\Es$]{\typctxtwo \uplus \typctxthree \vdash \solvctxtwop{\tmtwo'} \esub{\var}{\tmthree} \hastype \mtype}
		\end{prooftree}
		\end{equation*}
		where $\typctx = \typctxtwo \uplus \typctxthree$ and
		\begin{enumerate}
			\item $\sizem{\tderiv'} = \sizem{\tderivtwo'} + \sizem{\tderivthree} = \sizem{\tderivtwo} + \sizem{\tderivthree} - 2 = \sizem{\tderiv} - 2$ and $\size{\tderiv'} = \size{\tderivtwo'} + \size{\tderivthree} +1 = \size{\tderivtwo} -1 + 
\size{\tderivthree} +1 = \size{\tderiv} - 1$ if $\tmtwo \tomo \tmtwo'$ (\ie $\tm \tosolvm \tm'$); 
			\item $\sizem{\tderiv'} = \sizem{\tderivtwo'} + \sizem{\tderivthree} = \sizem{\tderivtwo} + \sizem{\tderivthree} = \sizem{\tderiv}$ and $\size{\tderiv'} = \size{\tderivtwo'} + \size{\tderivthree} +1 < \size{\tderivtwo} + 
\size{\tderivthree} +1 = \size{\tderiv}$ if $\tmtwo \toeo \tmtwo'$ (\ie $\tm \tosolve \tm'$).
		\end{enumerate}

		\item \emph{Application}, \ie $\solvctx = \solvctxtwo\tmthree$. 
		Then, $\tm = \solvctxp{\tmtwo} = \solvctxtwop{\tmtwo} \tmthree \Rew{\solvredsym a} \solvctxtwop{\tmtwo'} \tmthree = \solvctxp{\tmtwo'} = \tm'$ with $\tmtwo \Rew{\wsym a} \tmtwo'$ and $a \in \{\msym, \esym\}$.
		The derivation $\tderiv$ is necessarily
		\begin{equation*}
		\tderiv = 
		\begin{prooftree}
		\hypo{}
		\ellipsis{$\tderivtwo$}{\typctxtwo \vdash \solvctxtwop{\tmtwo} \hastype \mset{\larrow{\mtypetwo}{\mtype}}}
		\hypo{}
		\ellipsis{$\tderivthree$}{\typctxthree \vdash \tmthree \hastype \mtypetwo}
		\infer2[\footnotesize$\ruleAp$]{\typctxtwo \uplus \typctxthree \vdash \solvctxtwop{\tmtwo} \tmthree \hastype \mtype}
		\end{prooftree}
		\end{equation*}
		where $\typctx = \typctxtwo \mplus \typctxthree$.
		By \ih applied to $\tderivtwo$ (since $\mset{\larrow{\mtypetwo}{\mtype}}$ is a unitary solvable multi type), there 
is 
a derivation $\concl{\tderivtwo'}{\typctxtwo}{\solvctxtwop{\tmtwo'}}{\mset{\larrow{\mtypetwo}{\mtype}}}$ with: 
		\begin{enumerate}
			\item $\sizem{\tderivtwo'} = \sizem{\tderivtwo} - 2$ and $\size{\tderivtwo'} = \size{\tderivtwo} - 1$ if $\tmtwo 
\tomo \tmtwo'$; 
			\item $\sizem{\tderivtwo'} = \sizem{\tderivtwo}$ and $\size{\tderivtwo'} < \size{\tderivtwo}$ if $\tmtwo \toeo 
\tmtwo'$.
		\end{enumerate}
		We can then build the derivation 
		\begin{equation*}
		\tderiv' = 
		\begin{prooftree}
		\hypo{}
		\ellipsis{$\tderivtwo'$}{\typctxtwo \vdash \solvctxtwop{\tmtwo'} \hastype \mset{\larrow{\mtypetwo}{\mtype}}}
		\hypo{}
		\ellipsis{$\tderivthree$}{\typctxthree \vdash \tmthree \hastype \mtypetwo}
		\infer2[\footnotesize$\ruleAp$]{\typctxtwo \uplus \typctxthree \vdash \solvctxtwop{\tmtwo'} \tmthree \hastype 
\mtype}
		\end{prooftree}
		\end{equation*}
		where $\typctx = \typctxtwo \uplus \typctxthree$ and
		\begin{enumerate}
			\item $\sizem{\tderiv'} = \sizem{\tderivtwo'} + \sizem{\tderivthree} +1 = \sizem{\tderivtwo} - 2 + 
\sizem{\tderivthree} +1 = \sizem{\tderiv} - 2$ and $\size{\tderiv'} = \size{\tderivtwo'} + \size{\tderivthree} +1 = 
\size{\tderivtwo} - 1 + \size{\tderivthree} +1 = \size{\tderiv} - 1$ if $\tmtwo \tomo \tmtwo'$ (\ie $\tm \tosolvm \tm'$); 
			\item $\sizem{\tderiv'} = \sizem{\tderivtwo'} + \sizem{\tderivthree} +1 = \sizem{\tderivtwo} + \sizem{\tderivthree} +1 = \sizem{\tderiv}$ and $\size{\tderiv'} = \size{\tderivtwo'} + \size{\tderivthree} +1 <  \size{\tderivtwo} + \size{\tderivthree} +1 = \size{\tderiv}$ if $\tmtwo \toeo \tmtwo'$ (\ie $\tm \tosolve \tm'$).
		\end{enumerate}
	\end{itemize}
	
	This completes the proof for the unitary solvable case.
	
	In the solvable case (\ie under the weaker hypothesis that  $\mtype$ is a solvable multi type), the proof is 
analogous to the one for unitary solvable, except for the \emph{Abstraction} case.
	Indeed, in the base case (\emph{Hole context}) unitary solvability does not play any role, and the other cases follow 
from the \ih in a way analogous to unitary solvable.
	Let us see the only substantially different case: 
	\begin{itemize}
		\item \emph{Abstraction}, \ie $\solvctx = \la{\var}{\solvctxtwo}$. 
		So, $\tm = \solvctxp{\tmtwo} = \la{\var}{\solvctxtwop{\tmtwo}} \Rew{\solvredsym a} \la{\var}{\solvctxtwop{\tmtwo'}} = 
\solvctxp{\tmtwo'} = \tm'$ with $\tmtwo \Rew{\wsym a} \tmtwo'$ and $a \in \{\msym, \esym\}$.
		Since $\mtype$ is a solvable multi type by hypothesis, it has the form $\mtype = \mset{\larrow{\mtype_1}{\smtype_1}, \dots, \larrow{\mtype_n}{\smtype_n}}$ for some $n > 0$, where $\smtype_j$ is solvable for all $1 \leq j \leq n$.
		Thus, the derivation $\tderiv$ is necessarily
		\begin{equation*}
		\tderiv = 
		\begin{prooftree}[separation=1em]
		\hypo{}
		\ellipsis{$\tderivtwo_j$}{\typctx_j, \var \hastype \mtype_j \vdash \solvctxtwop{\tmtwo} \hastype \smtype_j}
		\infer1[\footnotesize$\lambda$]{\typctx_j \vdash \la{\var}\solvctxtwop{\tmtwo} \hastype 
\larrow{\mtype_j}{\smtype_j}}
		\delims{ \left( }{ \right)_{1 \leq j \leq n} }
		\infer1[\footnotesize$\ruleManyVal$]{ \bigmplus_{j=1}^n \typctx_{j} \vdash \la{\var}\solvctxtwop{\tmtwo} \hastype  
\bigmplus_{j=1}^n \mset{\larrow{\mtype_j}{\smtype_j}}}
		\end{prooftree}
		\end{equation*}
		For all $1 \leq j \leq n$, by \ih, there is a derivation $\concl{\tderivtwo_j'}{\typctx_j, \var \hastype 
\mtype_j}{\solvctxtwop{\tmtwo'}}{\smtype_j}$ with: 
		\begin{enumerate}
			\item $\sizem{\tderivtwo_j'} \leq \sizem{\tderivtwo_j} - 2$ and $\size{\tderivtwo_j'} < \size{\tderivtwo_j}$ 
if $\tmtwo \tomo \tmtwo'$ ; 
			\item $\sizem{\tderivtwo_j'} = \sizem{\tderivtwo_j}$ and $\size{\tderivtwo_j'} < \size{\tderivtwo_j}$ if $\tmtwo 
\toeo \tmtwo'$.
		\end{enumerate}
		We can then build the derivation 
		\begin{equation*}
		\tderiv' = 
		\begin{prooftree}[separation=1em]
		\hypo{}
		\ellipsis{$\tderivtwop_j$}{\typctx_j, \var \hastype \mtype_j \vdash \solvctxtwop{\tmtwo'} \hastype \smtype_i}
		\infer1[\footnotesize$\lambda$]{\typctx_j \vdash \la{\var}\solvctxtwop{\tmtwo'} \hastype 
\larrow{\mtype_j}{\smtype_j}}
		\delims{ \left( }{ \right)_{1 \leq j \leq n} }
		\infer1[\footnotesize$\ruleManyVal$]{ \bigmplus_{j=1}^n \typctx_{j} \vdash \la{\var}\solvctxtwop{\tmtwop} \hastype  
\bigmplus_{j=1}^n \mset{\larrow{\mtype_j}{\smtype_j}}}
		\end{prooftree}
		\end{equation*}
		where
		\begin{enumerate}
			\item $\sizem{\tderiv'} = \sum_{j=1}^n(\sizem{\tderivtwo_j'} +1) \leq \sum_{j=1}^n(\sizem{\tderivtwo_j} + 1 - 2) = \sizem{\tderiv} - 2n \leq \sizem{\tderiv} - 2$ (where the last inequality holds because $n >0$) and $\size{\tderiv'} =  \sum_{j=1}^n(\size{\tderivtwo_j'} +1) =  \sum_{j=1}^n(\size{\tderivtwo_j} + 1 - 1) \leq \size{\tderiv} - n < \size{\tderiv}$ (the last inequality holds because $n >0$) if $\tmtwo \tomo \tmtwo'$, \ie $\tm \tosolvm \tm'$; 
			\item $\sizem{\tderiv'} = \sum_{j=1}^n(\sizem{\tderivtwo_j'} +1) = \sum_{j=1}^n(\sizem{\tderivtwo_j} + 1) = \sizem{\tderiv}$ and $\size{\tderiv'} = \sum_{j=1}^n(\size{\tderivtwo_j'} +1) < \sum_{j=1}^n(\size{\tderivtwo_j} + 1) = \size{\tderiv}$ if $\tmtwo \toeo \tmtwo'$, \ie $\tm \tosolve \tm'$ (the inequality holds because $n > 0$).
			\qedhere
		\end{enumerate}
	\end{itemize}	
\end{proof}

\paragraph*{Completeness}

\begin{lemma}[Precisely solvable typability of solved fireballs]
	\label{propappendix:precise-solvable-typability-nf}
	\NoteState{prop:precise-solvable-typability-nf}
	If $\tm$ is a solved fireball, then there is a derivation 
$\concl{\tderiv}{\typctx}{\tm}{\mtype}$ with $\typctx$ inert type context and $\mtype$ precisely solvable.
\end{lemma}

\begin{proof}		
		By induction on the definition of solved fireball. 
		Cases:
		\begin{itemize}
			\item \emph{Inert}, \ie $\tm$ is an inert term. According to \Cref{prop:precise-open-typability-nf}.\ref{p:precise-open-typability-nf-inert}, since the precisely solvable multi type $\mset{\ground}$ is also inert, there is a derivation 
			$\concl{\tderiv}{\typctx}{\tm}{\mset{\ground}}$ for some inert type context $\typctx$.
			
			\item \emph{Abstraction}, \ie $\tm = \la{\var}{\solvnf}$ for some solved fireball $\solvnf$.
			By \ih, there is a derivation $\concl{\tderivtwo}{\typctx, \var \hastype \imtype}{\solvnf}{\psmtypetwo}$ for some unitary precisely solvable multi type $\psmtypetwo$, inert multi type $\imtype$ and inert type context $\typctx$. 
			We can then build the derivation 
			\begin{equation*}
			\tderiv = 
			\begin{prooftree}
			\hypo{}
			\ellipsis{$\tderivtwo$}{\typctx, \var \hastype \imtype \vdash \solvnf \hastype \psmtypetwo}
			\infer1[\footnotesize$\lambda$]{\typctx \vdash \la{\var}\solvnf \hastype \larrow{\imtype}{\psmtypetwo}}
			\infer1[\footnotesize$\ruleManyVal$]{\typctx \vdash \la{\var}\solvnf \hastype \mset{\larrow{\imtype}{\psmtypetwo}}}
			\end{prooftree}
			\end{equation*}
			where $\mtype = \mset{\larrow{\imtype}{\psmtypetwo}}$ is a precisely solvable multi type.
			
			\item \emph{Explicit substitution}, \ie $\tm = {\solvnf}\esub{\var}{\itm}$ for some solved fireball $\solvnf$ and inert term $\itm$.
			By \ih, there is a derivation $\concl{\tderivtwo}{\typctxtwo, \var \hastype \imtype}{\solvnf}{\psmtypetwo}$ for some inert type context $\typctxtwo$, inert multi type $\imtype$ and unitary precisely solvable multi type $\smtypetwo$.
			By \Cref{prop:precise-open-typability-nf}.\ref{p:precise-open-typability-nf-inert}, there is a derivation $\concl{\tderivthree}{\typctxthree}{\itm}{\imtype}$ for some inert type context $\typctxthree$.
			We can build the derivation 
			\begin{equation*}
			\tderiv = 
			\begin{prooftree}
			\hypo{}
			\ellipsis{$\tderivtwo$}{\typctxtwo, \var \hastype \imtype \vdash \solvnf \hastype \smtype}
			\hypo{}
			\ellipsis{$\tderivthree$}{\typctxthree \vdash \itm \hastype \imtype}
			\infer2[\footnotesize$\ruleES$]{\typctxtwo \uplus \typctxthree \vdash \solvnf \esub{\var}{\itm} \hastype \psmtype}
			\end{prooftree}
			\end{equation*}
			where $\typctx = \typctxtwo \mplus \typctxthree$ is a inert type context, by \Cref{rmk:merge-split-inert}.
			\qedhere
		\end{itemize}
\end{proof}

\end{document}